\documentclass[11pt,a4paper,reqno,oneside]{amsart}

\usepackage{geometry}
\usepackage{orcidlink}
\usepackage[foot]{amsaddr}
\newcommand\blfootnote[1]{%
  \begingroup
  \renewcommand\thefootnote{}\footnote{#1}%
  \addtocounter{footnote}{-1}%
  \endgroup
}

\usepackage[english]{babel} 

\usepackage{stix2}

\usepackage{microtype} 

\usepackage{ragged2e} 

\usepackage[autostyle=true]{csquotes} 

\usepackage{graphicx} 
\graphicspath{{./figures/}}

\usepackage[listformat=empty]{caption} 
\usepackage{subcaption}

\usepackage[shortlabels]{enumitem} 

\usepackage{amsthm}

\usepackage[
backend=bibtex,
style=numeric-comp,
giveninits=true,
maxbibnames=299,
natbib=true,
doi=false,
isbn=false,
url=false,
sorting=none,
maxbibnames=8,
eprint=false
]
{biblatex}
  \renewbibmacro{in:}{}
\DeclareFieldFormat[misc]{title}{``#1''\isdot}
\DeclareFieldFormat[online]{title}{``#1''\isdot}
\DeclareFieldFormat{journaltitle}{#1\isdot}
\DeclareFieldFormat[article,periodical]{volume}{\mkbibbold{#1}}
\DeclareFieldFormat{pages}{#1}
 \AtEveryBibitem{%
 	\clearfield{month}}
    \AtEveryBibitem{%
 	\clearfield{day}}
      \AtEveryBibitem{%
 	 \clearfield{number}}
\AtEveryBibitem{%
 	 \clearfield{pagetotal}}
\renewbibmacro*{addendum+pubstate}{} 

\AtEveryBibitem{%
	\clearfield{number}}    \AtEveryBibitem{\clearfield{pagetotal}}    \AtEveryBibitem{\clearfield{volumes}}
\AtEveryBibitem{\clearfield{language}}

    \AtEveryBibitem{%
        \iffieldundef{doi}
        {
            \iffieldundef{issn}
            {}
            {\clearfield{isbn}}
        }
        {
            \clearfield{issn}%
            \clearfield{isbn}
        }%
    }

\usepackage{appendix}

\usepackage{mathtools} 
\usepackage{amsfonts} 

\usepackage{dsfont} 

\usepackage{nicefrac}
\usepackage{braket}

\usepackage{diffcoeff}

\usepackage{siunitx}
\usepackage{cancel}

\usepackage{xcolor} 
\definecolor{royalpurple}{rgb}{0.47, 0.32, 0.66}
\definecolor{dblue}  {RGB}{20,66,129}
\definecolor{ddblue} {RGB}{11,36,69}
\definecolor{jblue}  {RGB}{20,50,100}
\definecolor{cblue}{rgb}{0.16, 0.32, 0.75}
\definecolor{cred}{rgb}{0.7, 0.11, 0.11}

\usepackage{todonotes} 

\usepackage{hyperref} 
\hypersetup{%
    colorlinks = true,
	linkcolor=cred,
	citecolor=cblue,
	urlcolor=black,
}

\usepackage[capitalize]{cleveref}

\difdef { f } {}
{
  outer-Ldelim = \left . ,
  outer-Rdelim = \right |,
  sub-nudge = 0 mu
}

\NewDocumentCommand \dfl { E{^}{1} m } { \,\difc[#1]{#2}{{}} }
\difdef { f, s } { gd }
{ op-symbol = \delta }

\newcommand{\id}{\mathds{1}}

\DeclareMathOperator{\Real}{Re}
\DeclareMathOperator{\Imag}{Im}
\DeclareMathOperator{\mspan}{Span}

\DeclarePairedDelimiterX{\norm}[1]\lVert\rVert{
  \ifblank{#1}{\:\cdot\:}{#1}
}

\newcommand{\iu}{\mathrm{i}\mkern1mu}
\newcommand{\e}{\mathrm{e}}

\newcommand{\nnum}{\mathbb{N}}
\newcommand{\rnum}{\mathbb{R}}
\newcommand{\cnum}{\mathbb{C}}

\newcommand{\znum}{\mathbb{Z}}

\newcommand{\hilbert}{\mathcal{H}}

\newcommand{\domain}{\mathcal{D}}

\newcommand{\Jop}{\mathcal{J}}
\newcommand{\Jmax}{J_{\mathrm{max}}}
\newcommand{\Jmin}{J_{\mathrm{min}}}
\newcommand{\Top}{\mathcal{T}}
\newcommand{\Tmax}{T_{\mathrm{max}}}
\newcommand{\Tmin}{T_{\mathrm{min}}}
\newcommand{\brnum}{\overline{\rnum}}
\newcommand{\bcnum}{\overline{\cnum}}

\DeclareMathOperator{\Ai}{Ai}
\DeclareMathOperator{\Bi}{Bi}
\DeclareMathOperator{\arccosh}{arccosh}

\newcommand{\recessSol}{\mathrm{r}}
\newcommand{\domSol}{\mathrm{d}}
\newcommand{\qSol}{\mathrm{q}}

\numberwithin{equation}{section}

\theoremstyle{plain}
\newtheorem{theorem}{Theorem}[section]

\newtheorem{corollary}[theorem]{Corollary} 
\newtheorem{lemma}[theorem]{Lemma} 
\newtheorem{proposition}[theorem]{Proposition}

\theoremstyle{definition} 
\newtheorem{definition}[theorem]{Definition}

\theoremstyle{remark} 
\newtheorem{remark}[theorem]{Remark}

\theoremstyle{remark}

\theoremstyle{plain} 

\theoremstyle{plain} 

\theoremstyle{plain} 
\newtheorem{hypothesis}[theorem]{Hypothesis}

\AddToHook{env/definition/begin}{\crefalias{theorem}{definition}}
\AddToHook{env/lemma/begin}{\crefalias{theorem}{lemma}}
\AddToHook{env/proposition/begin}{\crefalias{theorem}{proposition}}
\AddToHook{env/corollary/begin}{\crefalias{theorem}{corollary}}
\AddToHook{env/remark/begin}{\crefalias{theorem}{remark}}
\AddToHook{env/example/begin}{\crefalias{theorem}{example}}
\AddToHook{env/question/begin}{\crefalias{theorem}{question}}
\AddToHook{env/conjecture/begin}{\crefalias{theorem}{conjecture}}
\AddToHook{env/hypothesis/begin}{\crefalias{theorem}{hypothesis}}

\makeatletter
\AddToHook{cmd/appendix/before}{\def\cref@section@alias{appendix}}
\makeatother

\usepackage{cancel}

\title[Essentially singular limits of Jacobi operators and applications to higher-order squeezing]{Essentially singular limit of Jacobi operators\\ and applications to higher-order squeezing}

\author{Felix Fischer\textsuperscript{\,1,*}\hspace{2pt}\orcidlink{0009-0001-3040-8184}\hspace{1pt}}

\author{Daniel Burgarth\textsuperscript{\,1}\hspace{2pt}\orcidlink{0000-0003-4063-1264}\hspace{1pt}}

\author{Davide Lonigro\textsuperscript{\,1}\hspace{2pt}\orcidlink{0000-0002-0792-8122}\hspace{1pt}}

\address{\footnotesize \textsuperscript{1}Institute of Theoretical Physics, Friedrich-Alexander-Universität Erlangen-Nürnberg, Staudtstraße 7, 91058 Erlangen, Germany}

\renewcommand*{\thefootnote}{\fnsymbol{footnote}}

\begin{document}
\begin{abstract}
   We study a family of Jacobi operators in which the diagonal entries are multiplied by a coupling parameter $\lambda\geq0$. Under suitable conditions, the operator is self-adjoint for every $\lambda>0$, while the formal limit at $\lambda=0$ is a symmetric Jacobi operator admitting a one-parameter family of self-adjoint extensions. A central ingredient of our analysis is the derivation of uniform bounds for square-summable generalized eigenvectors in the small-$\lambda$ regime, which combines discrete WKB methods with Airy-function asymptotics. Using these estimates, we analyze the limiting behavior $\lambda\to0$ in the strong resolvent sense, proving that for every sequence $\lambda_j\to0$ one can extract a subsequence along which the corresponding Jacobi operators converge to some self-adjoint extension of the limiting operator; conversely, every such extension can be obtained in this way. We call this behavior an essentially singular limit, by analogy with essential singularities in complex analysis.

As an application, we study higher-order squeezing operators arising in quantum optics. Using the connection with Jacobi operators, we show that when the relative strength of the free-field term tends to zero, different self-adjoint extensions of the squeezing operator are selected along different sequences. In particular, this limit does not single out a physically distinguished self-adjoint extension, but instead identifies a distinguished subclass of extensions compatible with the underlying symmetry.
\end{abstract}

\maketitle
\thispagestyle{empty}

\footnotetext{*Corresponding author: \href{mailto:felix.o.fischer@fau.de}{\texttt{felix.o.fischer@fau.de}}.}

\blfootnote{2020 \textit{Mathematics Subject Classification}. 47B36, 47B25, 81Q10, 81Q12, 46N50.}

\vspace{-0.9cm}

\noindent \small \textbf{Keywords}: Jacobi operators; Self-adjoint extensions; Strong resolvent convergence; Turning point analysis; Higher-order squeezing.\normalsize

 \vspace{-0.35cm}

\section{Introduction}\label{sec:intro}
The relation between symmetric operators on Hilbert spaces and their self-adjoint extensions is a central topic in operator theory and spectral analysis. A symmetric operator that fails to be self-adjoint can admit a family of self-adjoint extensions with different spectral and dynamical properties. Understanding how such extensions arise and whether they can be selected by natural limiting procedures is therefore an important problem in operator theory. Questions of this type play a prominent role in quantum mechanics, where self-adjointness is required for the generators of dynamics, and different extensions may correspond to different boundary conditions or interactions.

Situations in which such questions arise naturally occur when studying limits of parameter-dependent operators. A family of operators depending on a coupling parameter $\lambda$ may be self-adjoint for all nonzero values of the parameter, while the formal limit obtained by removing the perturbation yields a symmetric operator with nontrivial deficiency indices. The behavior of the resolvent in this limit is therefore of particular interest, as it determines whether a specific self-adjoint extension is selected or whether different extensions arise along different limiting sequences.

In this work we analyze this phenomenon for a class of Jacobi operators. Such operators play an important role in spectral theory, appearing for instance in the study of orthogonal polynomials, discrete Schrödinger operators, and representations of bosonic operators in Fock space. We consider a family of Jacobi operators in which the diagonal entries are multiplied by a parameter $\lambda>0$,
\begin{equation}
    \Jop(\lambda) = \begin{pmatrix}
        \lambda f_0 & a_0 & 0 & 0 & 0 & \dots \\
        a_0 & \lambda f_1 & a_1 & 0 & 0 &  \dots \\
        0 & a_1 & \lambda f_2 & a_2 & 0 & \dots \\
        0 & 0 & a_2 & \lambda f_3 & a_3 & \dots \\
        \vdots & \vdots & \vdots & \vdots & \vdots & \ddots
    \end{pmatrix}.
\end{equation}
Under suitable assumptions, the resulting operator on the Hilbert space $\ell^2(\nnum)$, defined on the subspace of compactly supported sequences, is essentially self-adjoint for every nonzero value of $\lambda$, whereas the formal limit obtained by setting $\lambda=0$ yields a symmetric Jacobi operator with deficiency indices $(1,1)$, and therefore a one-parameter family of self-adjoint extensions.

Our main result (\cref{thm:main-result}) concerns the behavior of this family in the limit $\lambda \to 0$. We show that the strong resolvent limit depends on the way in which the parameter tends to zero. More precisely, for every sequence $\lambda_j \to 0$ one can extract a subsequence along which the corresponding operators converge in the strong resolvent sense to a self-adjoint extension of the limiting operator. Conversely, every self-adjoint extension of the limiting operator arises as a strong resolvent limit along a suitable sequence $\lambda_j \to 0$. In this sense, this limit parametrizes the full family of self-adjoint extensions. We refer to this phenomenon as an essentially singular limit, in analogy with the notion of essential singularities in complex analysis.

A useful geometric picture of this phenomenon is shown in \cref{fig:eigenvalue-limit-b}. For each $\lambda>0$, one can associate to the Jacobi operator $\Jop(\lambda)$ a complex quantity encoding its boundary behavior at infinity. As $\lambda \to 0$, these values do not converge to a single distinguished limit. Instead, they move along a closed curve associated with the limiting non-self-adjoint operator. Different sequences $\lambda_j \to 0$ approach different points on this curve, and our main result shows that every point can arise in this way. This picture is formalized later in terms of Weyl $m$-functions (\cref{def:green}) and the corresponding limit circle of the limiting Jacobi operator.

\begin{figure}[ht]
    \centering
    \includegraphics[width=0.6\linewidth]{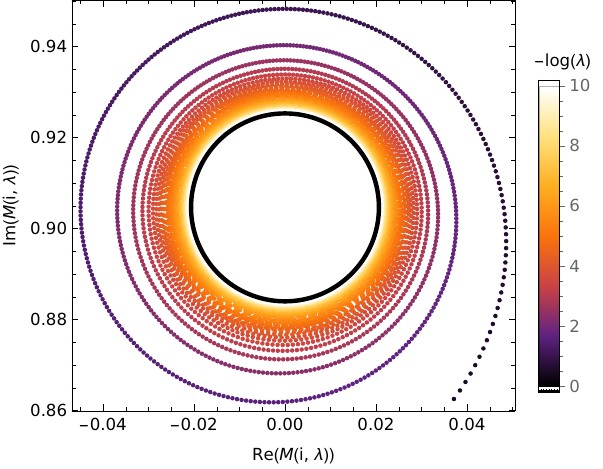}
    \caption{
Numerical plot of the values of the Weyl $m$-function $M(z,\lambda)$ at $z=\iu$ associated with the Jacobi operator $\Jop(\lambda)$ for different values of $\lambda>0$, for the choice
$a_n = 3^{-3/2}\sqrt{(3n+1,3)}$, $f_n = n^3$, where $(\cdot,\cdot)$ denotes the Pochhammer symbol, cf. \cref{eq:def-pochhammer}.
As $\lambda \to 0$, the points spiral around the limit circle of the limiting operator, shown in black.
}
\label{fig:eigenvalue-limit-b}
\end{figure}

The result yields an explicit construction through which the full family of self-adjoint extensions of a symmetric operator emerges in a small-coupling limit. This raises natural questions about whether similar mechanisms exist for more general classes of operators. A key technical ingredient of our analysis is the study of resolvent convergence in the limit. We first establish weak convergence of generalized eigenvectors for suitable sequences of the coupling parameter. The main difficulty then lies in upgrading this weak convergence to strong resolvent convergence; this step relies crucially on the specific structure of Jacobi operators. The argument is highly nontrivial and does not appear to extend directly beyond the Jacobi setting. One may expect that analogous phenomena could arise for continuous models, such as Sturm--Liouville operators, where many of the underlying concepts admit direct counterparts. However, even for symmetric operators with deficiency indices $(1,1)$, it is currently unclear whether an analogous result should hold in greater generality.

\subsection{Previous results}\label{sec:literature}
Jacobi operators and their generalizations are ubiquitous in the literature, appearing in areas such as solid-state physics~\cite{zhang-andersonlocalizationblock-2024,thouless-bandwidthsquasiperiodictightbinding-1983,jitomirskaya-metalinsulatortransitionalmost-1999,bellisard-spectralproperties-1990}, light--matter interaction~\cite{tur-jaynescummingsmodel-2000,tur-jaynescummingsmodelrotating-2002,monvel-oscillatorybehaviorlarge-2018,charif-perturbationseriesjacobi-2021,nagel-higherpowersqueezed-1997}, and the moment problem~\cite{schmudgen-momentproblem-2017,akhiezer-classicalmomentproblem-2020,simon-classicalmomentproblem-1998}, among many others~\cite{izaac-computationalquantummechanics-2018,simon-operatorssingularcontinuous-1995,eisert-supersonicquantumcommunication-2009}. They can be viewed as discrete counterparts of Sturm--Liouville operators, and many results admit direct analogues in that setting; see, for instance, \cite{schmudgen-unboundedselfadjointoperators-2012,behrndt-boundaryvalueproblems-2020,everitt-weylsworksingular-2008,reed-mmmp2-fourier-1975,zettl-sturmliouvilletheory-2005,levitan-sturmliouvilledirac-1991}. General introductions to Jacobi operators can be found in \cite{teschl-jacobioperatorscompletely-1999,schmudgen-unboundedselfadjointoperators-2012,berezanskii-expansioneigenfunctionsselfadjoint-1968}, while non-symmetric generalizations are discussed in \cite{eichinger-weylmatrixperspective-2025,bernhardbeckermann-complexjacobimatrices-2001,allahverdiev-extensionsdilationsfunctional-2005}.

Necessary and sufficient conditions for essential self-adjointness of Jacobi and block Jacobi operators have been extensively studied; see, e.g., \cite{gesztesy-jacobioperator$$11$$-2024,yafaev-analyticscatteringtheory-2018,yafaev-asymptoticbehaviororthogonal-2020,yafaev-asymptoticbehaviororthogonal-2021,yafaev-selfadjointjacobioperators-2021,yafaev-spectralanalysisjacobi-2022,yafaev-spectraltheoryjacobi-2024,swiderski-spectralpropertiesunbounded-2016,swiderski-periodicperturbationsunbounded-2017,berg-indeterminatejacobioperators-2025,welstead-boundaryconditionsinfinity-1982,malcolmbrown-kreinfriedrichsextensions-2005,budyka-deficiencyindicesblock-2024,swiderski-spectralpropertiesblock-2018,braeutigam-deficiencynumbersoperators-2016}. A large body of work is devoted to the spectral and resolvent analysis of both bounded \cite{janas-jacobimatricespowerlike-1999,janas-decayboundseigenfunctions-2009,dombrowski-resultsabsolutecontinuity-2017,gesztesy-mfunctionsinversespectral-1997,damanik-uniformspectralproperties-2000,last-eigenfunctionstransfermatrices-1999} and unbounded \cite{hinton-spectralanalysissecond-1978,pruckner-densityspectrumjacobi-2018,stovicek-infinitejacobimatrices-2019,eckhardt-singularweyltitchmarshkodairatheory-2013,teschl-traceformulasinverse-1998,gesztesy-commutationmethodsjacobi-1996,clark-spectralanalysisselfadjoint-1996,geronimo-spectrainfinitedimensionaljacobi-1988,swiderski-asymptoticzerosdistribution-2025,swiderski-periodicperturbationsunbounded-2017,swiderski-periodicperturbationsunbounded-2017a,swiderski-periodicperturbationsunbounded-2018,ammann-relativeoscillationtheory-2012,teschl-oscillationtheoryrenormalized-1996,aptekarev-measuresorthogonalpolynomials-2016} Jacobi operators. 
Perturbation theory for essentially self-adjoint Jacobi operators is well understood; see, for example, \cite{remling-absolutelycontinuousspectrum-2011,janas-spectralpropertiesjacobi-2003,damanik-analytictheorymatrix-2008,damanik-perturbationsorthogonalpolynomials-2010,kaluzhny-preservationabsolutelycontinuous-2010,judge-eigenvaluesperturbedperiodic-2016,judge-spectralresultsperturbed-2018,lakaev-thresholdvirtualstates-2026,lukic-generalizedprufervariables-2016,simon-szegostheoremits-2011,harrat-asymptoticexpansionlarge-2020}. For the particular case of compact perturbations, see \cite{rio-inverseproblemsjacobi-2013,webb-spectrajacobioperators-2020,iantchenko-periodicjacobioperator-2012,yafaev-spectralanalysisjacobi-2022}. However, none of the works mentioned above considers a limiting operator that fails to be self-adjoint.

Similarly, the general theory of analytic perturbations (see, e.g., \cite{reed-mmmp4-operators-2005,kato-perturbationtheorylinear-1995} for an introduction) deals with small perturbations of a self-adjoint operator. The situation is less well understood when the limiting operator is not self-adjoint. In the semibounded case, perturbations typically converge to a unique self-adjoint extension of the limiting operator; in this context, the Kato--Robinson theorem~\cite{kato-perturbationtheorylinear-1995,oliveira-intermediatespectraltheory-2009} plays a central role. Self-adjoint extensions can also be parametrized via unbounded finite-rank perturbations, which are usually treated within the framework of singular perturbation theory~\cite{albeverio-singularperturbationsdifferential-2000,bush-singularboundaryconditions-2023,frymark-spectralpropertiessingular-2023}. By contrast, in our setting the limiting operator is not semibounded, and the diagonal perturbation has full (infinite) rank.

Related phenomena appear in the renormalization of singular potentials~\cite{case-singularpotentials-1950,beane-singularpotentialslimit-2001,bawin-singularinversesquare-2003a}; see also \cite{borg-pauliapproximationsselfadjoint-2003,tamura-resolventconvergencenorm-2003} for the Aharonov--Bohm effect. In these cases, singularities at the origin lead to Schrödinger or Dirac operators that are not self-adjoint. Self-adjointness can be restored by introducing a cutoff at a finite radius, and the self-adjoint extensions of the original operator are then obtained by removing the cutoff. Crucially, additional parameters of the regularization must be tuned in order to recover different extensions.

The result obtained in this work is qualitatively reminiscent of Picard's Great Theorem in complex analysis~\cite[Theorem~11.3.2]{simon-basiccomplexanalysis-2015}, a refinement of the Casorati--Weierstrass theorem. There, an analytic function $f(\lambda)$ with an essential singularity at $\lambda=0$ attains, in any neighborhood of the origin, all complex values with at most one exception; in particular, for every $z\in\cnum$ there exists a sequence $\lambda_j\to0$ such that $f(\lambda_j)\to z$. In a similar spirit, our main result shows that, as $\lambda\to0$, the family of operators $J(\lambda)$ approaches all self-adjoint extensions of the limiting operator along suitable sequences of the parameter. In this sense, the limit $\lambda\to0$ behaves analogously to an essential singularity, in that it allows one to approximate arbitrarily prescribed self-adjoint extensions along suitable sequences.

A key technical ingredient of our proof is the derivation of bounds on generalized eigenvectors of the Jacobi operators under consideration. A variety of such bounds for self-adjoint operators can be found in the literature on the spectral theory of Jacobi operators cited above; for works specifically devoted to this topic, see \cite{janas-asymptoticbehaviorgeneralized-2009,janas-decayboundseigenfunctions-2009,miguel-matrixvaluedschrodinger-2026}. Closely related is the study of asymptotics for linear recurrence relations and difference equations. The classical theory, initiated by Birkhoff~\cite{birkhoff-analytictheorysingular-1933,birkhoff-formaltheoryirregular-1930,wimp-resurrectingasymptoticslinear-1985}, was later refined and systematically developed in a series of works by Wong and Wang~\cite{wong-asymptoticexpansionssecondorder-1992,wong-asymptoticexpansionssecondorder-1992a,wong-recentadvancesasymptotic-2022,wang-asymptoticexpansionssecondorder-2003,wang-lineardifferenceequations-,wang-uniformasymptoticexpansion-2002}. Finally, of particular relevance to the present work are results on discrete WKB methods and turning point theory~\cite{fedotov-complexwkbmethod-2019,geronimo-wkbliouvillegreenanalysis-1992,geronimo-wkbturningpoint-2009}, which we adapt and extend to our setting.

\subsection{Connection with higher-order squeezing}
\label{sec:intro_squeezing}
The original motivation for our study comes from quantum optics. One of its central concepts is the squeezing of light, as introduced by Walls~\cite{walls-squeezedstateslight-1983}. Under the dynamics generated by the second-order Hamiltonian
\begin{equation}
    A_2 = (a^\dagger)^2 + a^2 \, ,
\end{equation}
coherent states evolve into squeezed states, in which the uncertainty of one field quadrature is reduced at the expense of increased uncertainty in the conjugate one.
On the experimental side, such squeezed states are routinely realized in optical parametric oscillators~\cite{wu-generationsqueezedstates-1986,vahlbruch-observationsqueezedlight-2008} and superconducting circuits~\cite{eichler-observationtwomodesqueezing-2011,vaartjes-strongmicrowavesqueezing-2024}, and they constitute a key resource for surpassing the standard quantum limit. In particular, they enable enhanced measurement precision in gravitational-wave detection~\cite{eberle-quantumenhancementzeroarea-2010,tse-quantumenhancedadvancedligo-2019,mcculler-frequencydependentsqueezingadvanced-2020,ganapathy-broadbandquantumenhancement-2023}, as well as the generation of nonclassical states relevant for bosonic quantum error correction and quantum information processing~\cite{milburn-quantumteleportationsqueezed-1999,braunstein-squeezingirreducibleresource-2005,braunstein-quantuminformationcontinuous-2005,gottesman-encodingqubitoscillator-2001,andersen-continuousvariablequantuminformation-2010,schlegel-quantumerrorcorrection-2022,lu-recentprogresscoherent-2023,zeng-quantumerrorcorrection-2025,delgrosso-controlledsqueezegatesuperconducting-2025}.

While conventional (quadrature) squeezing is by now a well-established and controllable resource, there is growing interest in moving beyond Gaussian operations toward genuinely nonlinear, non-Gaussian regimes~\cite{cai-bosonicquantumerror-2021,puri-engineeringquantumstates-2017,mirrahimi-dynamicallyprotectedcatqubits-2014,gordillo-hachuel-quantummetrologicaladvantage-2026}. In particular, non-Gaussian gates are necessary for universal bosonic quantum computation~\cite{lloyd-quantumcomputationcontinuous-1999}. A natural class of candidates for such interactions is given by
\begin{equation}
\label{eq:higher-order}
A_k = (a^\dagger)^k + a^k ,
\end{equation}
with $k\geq3$. These operators were introduced by Braunstein and McLachlan in the 1980s~\cite{braunstein-generalizedsqueezing-1987} as higher-order generalizations of the squeezing interaction.
Beyond their intrinsic conceptual interest, these operators have been the subject of a longstanding debate concerning the existence and uniqueness of the dynamics they generate; see, e.g.,~\cite{fisher-impossibilitynaivelygeneralizing-1984,lo-multiquantumjaynescummingsmodel-1998,nagel-higherpowersqueezed-1997,lee-exactsolutionsfamily-2011,gardas-multiphotonrabimodel-2013,lo-commentmultiphotonrabi-2014,gardas-replycommentmultiphoton-2014,gardas-initialstatesqubit-2013,lo-commentinitialstates-2014,zhang-solvingtwomodesqueezed-2013,lo-commentsolvingtwomode-2014,gorska-squeezingarbitraryorder-2014,zhang-2modekphotonquantum-2017,braak-$k$photonquantumrabi-2025,ayyash-dispersiveregimemultiphoton-2025,fischer-selfadjointrealizationshigherorder-2025,ashhab-propertiesdynamicsgeneralized-2025,gordillo-hachuel-commentpropertiesdynamics-2026,ashhab-finitedimensionalapproximationsgeneralized-2026,ashhab-replycommentproperties-2026,ashhab-fractionalsqueezingspectra-2026}. 
In parallel, there has been recent experimental progress in the realization of higher-order photon processes~\cite{corona-thirdorderspontaneousparametric-2011,chang-observationthreephotonspontaneous-2020,menard-emissionphotonmultiplets-2022,boutin-effecthigherordernonlinearities-2017,gregory-foursixphotonstimulated-2026,bencheikh-demonstratingquantumproperties-2022,eriksson-universalcontrolbosonic-2024,bazavan-squeezingtrisqueezingquadsqueezing-2026}. 
Of particular relevance to the present work is the recent experimental realization of third- and fourth-order squeezing reported in \cite{bazavan-squeezingtrisqueezingquadsqueezing-2026}. 
There, trapped ions subject to spin-dependent forces in a magnetic trap are used to engineer effective higher-order bosonic interactions. 
By applying suitable pulse sequences, the resulting dynamics are argued, via Floquet--Magnus techniques~\cite{blanes-magnusexpansionits-2009,bukov-universalhighfrequencybehavior-2015}, to approximate Hamiltonians of the form~\eqref{eq:higher-order}. 
In this context, a rigorous operator-theoretic understanding of higher-order squeezing Hamiltonians becomes essential.

In previous work~\cite{ashhab-finitedimensionalapproximationsgeneralized-2026,fischer-selfadjointrealizationshigherorder-2025}, we analyzed these operators, which are naturally defined on the finite linear span of Fock states. Using the Birkhoff--Trjitzinsky theory of asymptotic expansions for recurrence relations, we showed that they admit families of self-adjoint extensions whose dimension depends on the order of squeezing. In particular, a choice of self-adjoint extension of $A_k$ is required in order to obtain a well-defined unitary dynamics.
We further showed that finite-dimensional truncations of the corresponding Fock space dynamics (i.e., restrictions to states with finitely many photon excitations) reproduce the dynamics associated with two distinguished self-adjoint extensions, depending on the truncation scheme. Moreover, adding a free-field term which is polynomial in the number operator $a^\dagger a$ of degree $h>k$, for example
\begin{equation}
\label{eq:general-higher-order}
A_{k,h}(K) = (a^\dagger)^k + K (a^\dagger a)^h + a^k,
\end{equation}
with $K>0$, yields an essentially self-adjoint operator for all values of $K$, thereby acting as a regularization mechanism.
This naturally raises the question of whether a physically distinguished self-adjoint realization of higher-order squeezing emerges in the regime where the strength of the free-field term becomes small compared to the squeezing interaction, that is, in the limit $K\to0$. 

In the present work, we address this question by expressing higher-order squeezing operators as direct sums of Jacobi operators. This connection not only allows us to revisit the parametrization of their self-adjoint extensions, yielding a new and more direct proof that avoids the use of Birkhoff--Trjitzinsky theory, but also enables a detailed analysis of the small-coupling limit $K \to 0$. We show that, in this limit, different self-adjoint extensions of the squeezing operator are selected along different sequences of the coupling parameter (\cref{prop:squeezing-main}). In particular, the limiting procedure selects only those extensions that are compatible with the natural decomposition of the operator into independent sectors. From a physical perspective, these are exactly the extensions that preserve the $2\pi/k$-rotational symmetry of the unperturbed operator $A_{k,h}(0)$. In this sense, the small-coupling limit does not merely select a self-adjoint realization, but rather selects a symmetry-preserving subclass of all possible extensions.

\subsection{Outline}

The paper is organized as follows. In \cref{sec:main}, we introduce the notation used throughout the paper and state our main results: \cref{thm:main-result}, concerning the limiting behavior of the family of Jacobi operators as $\lambda \to 0$, and \cref{prop:squeezing-main}, which applies this result to higher-order squeezing operators. In \cref{sec:proof}, we recall relevant concepts from the general theory of Jacobi operators, establish spectral properties of the operator family under consideration, and prove the main result, \cref{thm:main-result}. The proof of \cref{thm:main-result} relies on a key technical input concerning the asymptotics of generalized eigenvectors of Jacobi operators, \cref{thm:main-asymptotics}. Since this constitutes the most technically involved part of the analysis, it is developed independently in \cref{sec:asymptotics-of-eigenvectors}. Finally, in \cref{sec:proof-squeezing}, we prove \cref{prop:squeezing-main}. In \cref{app:airy,app:turning-point}, we recall and prove several properties of Airy functions and notions from turning point theory that are needed in the analysis of generalized eigenvector asymptotics.

\section{Setting and main results}\label{sec:main}

\subsection{Jacobi operators with vanishing diagonal perturbations}\label{sec:main_2}

Throughout this work, we consider operators acting on sequences of complex numbers $c = (c_n)_{n \in \nnum}$, and we denote by $\ell(\nnum)$ the vector space of all such sequences. Furthermore, we introduce the space of square-summable sequences
\begin{equation}
    \ell^2(\nnum) = \left\{ u \in \ell(\nnum) \, : \,  \sum_{n=0}^\infty |u_n|^2 < \infty\,  \right\},
\end{equation}
which is a Hilbert space with scalar product and associated norm
\begin{equation}
    \braket{u,v}=\sum_{n\in\mathbb{N}}u_n^*v_n,\qquad \|u\|=\sqrt{\braket{u,u}},
\end{equation}
and the space of sequences with compact support
\begin{equation}
    \ell_0(\nnum) = \{ u \in \ell(\nnum) \, : \, \exists N \in \nnum \text{ s.t. } u_n = 0 \, \forall n \geq N\,  \}\, ,
\end{equation}
which is a dense subspace of $\ell^2(\nnum)$. Equivalently, $\ell_0(\nnum)=\operatorname{Span}\left((e_n)_{n\in\mathbb{N}}\right)$, where the vectors $(e_n)_{n\in\mathbb{N}}$ defined by $(e_n)_m = \delta_{nm}$ form a complete orthonormal set of $\ell^2(\nnum)$.

In the present paper we consider a one-parameter family $(\Jop(\lambda))_{\lambda \geq 0}$ of Jacobi operators on $\ell^2(\nnum)$, cf.~\cref{def:jacobi-lambda}. The action of Jacobi operators can be represented with infinite tridiagonal matrices. In our case, we set
\begin{equation}
    \Jop(\lambda) = \begin{pmatrix}
        \lambda f_0 & a_0 & 0 & 0 & 0 & \dots \\
        a_0 & \lambda f_1 & a_1 & 0 & 0 &  \dots \\
        0 & a_1 & \lambda f_2 & a_2 & 0 & \dots \\
        0 & 0 & a_2 & \lambda f_3 & a_3 & \dots \\
        \vdots & \vdots & \vdots & \vdots & \vdots & \ddots
    \end{pmatrix},
\end{equation}
where $(a_n)_{n\in\nnum}$ and $(f_n)_{n\in\nnum}$ are suitable real-valued sequences, cf.~\cref{hyp:conditions-an-fn}, and $\lambda \geq 0$.

Our goal is to consider situations in which the limiting behavior of this family as $\lambda \to 0$ is singular in the sense that $\Jop(\lambda)$ is essentially self-adjoint for $\lambda>0$, but the limiting operator at $\lambda=0$ admits multiple self-adjoint extensions. To this end, we impose suitable conditions on the sequences $(a_n)_{n\in\nnum}$ and $(f_n)_{n\in\nnum}$.
We start by recalling the total variation of a sequence:
\begin{definition}
    \label{def:total-variation}
    Let $x=(x_n)_{n\in \nnum} \in \ell(\nnum)$ be a sequence and let $m,n \in \nnum$ with $m\le n$.
    The \textit{total variation} of $x$ on the interval $[m,n]$ is defined by
    \begin{equation}
        \mathcal{V}_{m,n}(x) = \sum_{j = m}^{n} |x_j-x_{j-1}| \, .
    \end{equation}
    For convenience we write $\mathcal{V}_{n}(x) = \mathcal{V}_{1,n}(x)$ and $\mathcal{V}(x) = \mathcal{V}_{1,\infty}(x)$.
\end{definition}
We adopt in this paper the usual big-$O$ and small-$o$ notation; see, e.g., \cite[§4]{olver-asymptoticsspecialfunctions-2010}. 
Given two real-valued sequences $(x_n)_{n\in\nnum}$ and $(y_n)_{n\in\nnum}$, we write $x_n=O(y_n)$ if there exist constants $N\in\nnum$ and $C\ge 0$ such that $|x_n|\le C\,|y_n|$ for all $n\ge N$, and $x_n=o(y_n)$ if $\lim_{n\to\infty}|x_n|/|y_n|=0$.

Throughout this paper, we impose the following conditions on the sequences $(a_n)_{n \in \nnum}$ and $(f_n)_{n \in \nnum}$:
\begin{hypothesis}
\label{hyp:conditions-an-fn}
The sequences $(a_n)_{n \in \nnum}$ and $(f_n)_{n \in \nnum}$ satisfy the following conditions:
\begin{enumerate}[(i)]
    \item $a_n > 0$ and $f_n \geq 0$ for all $n \in \nnum$;
    \item $a_n = n^\alpha\!\left(1+\frac{c_a}{n}+O(1/n^2)\right)$ as $n\to\infty$ for some $\alpha > \frac{4}{3}$ and $c_a \in \rnum$;
    \item $f_n = n^\beta\!\left(1+\frac{c_f}{n}+O(1/n^2)\right)$ as $n\to\infty$ for some $\beta > \alpha$ and $c_f \in \rnum$;
    \item \label{item:condition-monotonic}
    There exists $C \in \rnum$ such that
    \begin{equation}
        \mathcal{V}_n\!\bigl((f_j/a_j)_{j \in \nnum}\bigr)
        \leq
        \frac{f_n}{a_n} + C
        \qquad \text{for all } n \in \nnum ,
    \end{equation}
    where $\mathcal{V}_n(\cdot)$ denotes the total variation (cf.~\cref{def:total-variation}).
\end{enumerate}
\end{hypothesis}
\begin{remark}
    Condition \labelcref{item:condition-monotonic} is fulfilled if $f_n/a_n$ is monotonically increasing for all $n$.
\end{remark}
We now introduce the family of Jacobi operators under analysis in a precise form.
\begin{definition}
    \label{def:jacobi-lambda}
    Let $(a_n)_{n \in \nnum},(f_n)_{n \in \nnum} \in \ell(\nnum)$ be two sequences fulfilling \cref{hyp:conditions-an-fn}.
    For all $\lambda \geq 0$, $\Jop(\lambda):\ell(\nnum)\rightarrow\ell(\nnum)$ is the operator defined by
   \begin{align}
    (\Jop(\lambda) u)_n
    & =
    a_n u_{n+1}
    +
    \lambda f_n u_n
    +
    a_{n-1}u_{n-1},
    \qquad n \geq 1\, ,
    \\
    (\Jop(\lambda) u)_0
    & =
    a_0 u_1
    +
    \lambda f_0 u_0 \, .
\end{align}
    Furthermore, we define the operators $\Jmin(\lambda),J(\lambda)$, and $\Jmax(\lambda)$ on $\ell^2(\nnum)$: $\Jmin(\lambda)$ and $\Jmax(\lambda)$ are the restrictions of $\Jop(\lambda)$ respectively to
    \begin{equation}
        \domain(\Jmin(\lambda))=\ell_0(\nnum),\qquad\domain(\Jmax(\lambda))=\{u\in\ell^2(\nnum):\Jop(\lambda) u\in\ell^2(\nnum)\},
    \end{equation}
    and $J(\lambda)=\overline{\Jmin(\lambda)}$ is the closure of $\Jmin(\lambda)$.
\end{definition}
Note that, here and in the following, we use calligraphic letters such as $\Jop(\lambda)$ to denote the Jacobi operator acting on the full space of sequences $\ell(\nnum)$, i.e. defined purely by its formal action. By contrast, straight letters such as $J(\lambda)$ denote its realization as a (generally unbounded) operator on a dense subspace of the Hilbert space $\ell^2(\nnum)$, obtained by restricting $\Jop(\lambda)$ to an appropriate domain. Also note that the closure $J(\lambda)$ of $\Jmin(\lambda)$ always exists as $\Jmin(\lambda)$ is symmetric.

We begin with the following observation:
\begin{proposition}[cf.~\cref{prop:j0-limit-circle,prop:spectrum-jlambda}]
    \label{prop:self-adjointness}
    $J(\lambda)$ is self-adjoint if and only if $\lambda > 0$.
    In particular, $J(0)$ admits a one-parameter family of self-adjoint extensions $(J_t)_{t \in \brnum}$.
\end{proposition}
Above and in the following, $\brnum=\rnum\cup\{\infty\}$. We refer to \cref{prop:jacobi-boundary-triplets} for the explicit parametrization of all self-adjoint extensions $J_t$ of $J(0)$. Further properties of the operators $J(\lambda)$ and $J_t$ are introduced and discussed in \cref{sec:proof}.
\begin{remark}
    While our assumptions include $\alpha > \nicefrac{4}{3}$, \cref{prop:self-adjointness} (in particular, the fact that $J(0)$ admits a one-parameter family of self-adjoint extensions) remains valid under the weaker condition $\alpha>1$. The stronger assumption $\alpha > \nicefrac{4}{3}$ is, however, essential for the proof of the main result, \cref{thm:main-result}; see \cref{rem:asymptotics,sec:asymptotics-discussion} for further discussion.
\end{remark}
Therefore, $J(\lambda)$ is self-adjoint for $\lambda>0$ and generates the unitary dynamics $U_\lambda(t)=\e^{-\iu J(\lambda)t}$, whereas the limit operator $J(0)$ admits multiple self-adjoint extensions. This raises the question of the behavior of $J(\lambda)$, its resolvents, and its dynamics as $\lambda\to0$. The limit is singular in several respects: $J(0)$ admits multiple self-adjoint extensions, their domains $\domain(J_t)$ strictly contain $\domain(J(\lambda))$ (cf.~\cref{sec:proof}), and the resulting operators $J_t$ are unbounded from below. 

These features prevent the direct application of standard perturbative techniques. In particular, analytic perturbation theory does not apply due to the discontinuous change of domains (see, e.g., \cite[Chapter~XII]{reed-mmmp4-operators-2005}), while the lack of semiboundedness excludes the use of the Kato--Robinson theorem and related results (see, e.g., \cite[Theorem~10.4.2]{oliveira-intermediatespectraltheory-2009}). The relation to other perturbative results is discussed in \cref{sec:literature}. To the best of our knowledge, the limit $\lambda \to 0$ therefore lies beyond the reach of standard techniques.

We now present the main result of the paper:

\pagebreak
\begin{theorem}
    \label{thm:main-result}
    The following statements hold:
    \begin{enumerate}[(i)]
        \item \label{item:main-result-some-extension}
        Let $(\lambda_j)_{j \in \nnum}$ be a sequence with $\lambda_j>0$ for all $j \in \nnum$ and $\lambda_j \to 0$. 
        Then there exist a subsequence $(\lambda_{j_k})_{k \in \nnum}$ and a parameter $t \in \brnum$ such that $J(\lambda_{j_k})$ converges to $J_t$ in the strong resolvent sense.        
        \item \label{item:main-result-all-extensions}
        Let $t \in \brnum$. 
        Then there exists a sequence $(\lambda_j)_{j \in \nnum}$ with $\lambda_j>0$ for all $j \in \nnum$ and $\lambda_j \to 0$ such that $J(\lambda_j)$ converges to $J_t$ in the strong resolvent sense.
    \end{enumerate}
    Moreover, in both cases the eigenvalues of the corresponding operators converge.
\end{theorem}

We refer to \cref{sec:proof} for the proof of \cref{thm:main-result}.
A key technical result needed for the proof, which may also be of independent interest, is the following bound on square-summable solutions of linear recurrence relations:
\begin{theorem}
    \label{thm:main-asymptotics}
    Suppose that \cref{hyp:conditions-an-fn} holds.
    Let $\Omega \subset \cnum$ be a compact set and, for every $z \in \Omega$, let $u^{\lambda,z}$ be a square-summable solution of the linear recurrence relation 
    \begin{equation}
        a_n u_{n+1} + (\lambda f_n -z)u_n + a_{n-1} u_{n-1}= 0 \quad \forall n \geq 1 \,.
    \end{equation}
    Then, there exist $\Lambda > 0$ and $C > 0$ such that, for all $0 \leq \lambda < \Lambda$ and $z \in \Omega$,
    \begin{equation}
        |u^{\lambda,z}_n| \leq C \frac{|u^{\lambda,z}_0|+|u^{\lambda,z}_1|}{n^{\alpha/2-1/6}} \quad \forall n \geq 1\, .
    \end{equation}
\end{theorem}
The theorem is restated in \cref{sec:proof} (cf. \cref{prop:asymptotics}) in the context of Jacobi operators, and proven in \cref{sec:asymptotics-of-eigenvectors}.

\subsection{Application to higher-order squeezing operators}
\label{sec:main_squeezing}
As outlined in \cref{sec:intro_squeezing}, our original motivation for proving \cref{thm:main-result} is the study of higher-order squeezing operators. In this context, we consider a single bosonic mode, represented by the Hilbert space $L^2(\rnum)$ of square-integrable functions on the real line, together with an orthonormal basis $(\phi_n)_{n \in \nnum}$. We interpret the states $\phi_n$ as Fock states, corresponding to exactly $n$ excitations of the bosonic field, and define the standard annihilation and creation operators $a$ and $a^\dagger$ on $\mspan(\phi_n)_{n \in \nnum}$ by
\begin{align}
    a \phi_0 & = 0\, , \\
    a \phi_n & = \sqrt{n}\,\phi_{n-1} \quad n \geq 1 \, , \\
    a^\dagger \phi_n & = \sqrt{n+1}\,\phi_{n+1} \quad n \in \nnum \, .
\end{align}
In particular, $(\phi_n)_{n \in \nnum}$ forms an eigenbasis of the number operator $a^\dagger a$, with $a^\dagger a \phi_n = n \phi_n$ for all $n \in \nnum$. It follows that, for every $k,h\in\nnum$ and $K\geq0$, the operators
\begin{equation}
    \label{eq:def-higher-order}
     A_k = (a^\dagger)^k + a^k \,, \qquad A_{k,h}(K)=A_k+K(a^\dagger a)^h,
\end{equation}
cf.~\cref{eq:higher-order,eq:general-higher-order}, are well-defined on $\mspan(\phi_n)_{n \in \nnum}$.
In the following, we write $A_k$ and $A_{k,h}(K)$ for the closure of the operators in \cref{eq:def-higher-order} defined on  $\mspan(\phi_n)_{n \in \nnum}$.

In previous work \cite{fischer-selfadjointrealizationshigherorder-2025,ashhab-finitedimensionalapproximationsgeneralized-2026} it was shown that $A_k$ is not essentially self-adjoint for $k \geq 3$. On the other hand, the regularized Hamiltonian $A_{k,h}(K)$, which includes the additional free-field term $K (a^\dagger a)^h$, becomes self-adjoint provided that $2h > k$ and $K > 0$. This naturally leads to the question of whether an arbitrarily weak free-field contribution selects a distinguished self-adjoint extension of $A_k$ in the limit $K \to 0$.

In the following, we address this question by applying \cref{thm:main-result} to the family of higher-order squeezing operators, showing how different self-adjoint extensions arise as limiting regimes of vanishing coupling. We begin by expressing $A_{k,h}(K)$ as a direct sum of Jacobi operators and identifying the conditions on the parameters $k,h,K$ under which the resulting Jacobi components satisfy \cref{hyp:conditions-an-fn}.
To this end, we recall the definition of the Pochhammer symbol $(x,s)$ for $x,s \in \nnum$:
\begin{equation}
    \label{eq:def-pochhammer}
    (x,s) = \prod_{i =  1}^s (x+i-1) = x (x+1) \dots (x+s-1) \, .
\end{equation}
\begin{proposition}
    \label{lem:squeezing-decomposition}
    Let $k\geq 3$, $2h>k$ and $K \geq 0$.
    Then, 
    \begin{equation}
        A_{k,h}(K) = k^{k/2} \bigoplus_{m=0}^{k-1} A^{(m)}_{k,h}(k^{h-k/2} K) \, , 
    \end{equation}
    where each $A^{(m)}_{k,h}(\lambda)$ is a closed Jacobi operator defined by the two sequences
    \begin{align}
        a^{(k,m)}_n & = k^{-k/2} \sqrt{(m+nk+1,k)} \, , \\
        f^{(k,h,m)}_n  & =k^{-h}(m+nk)^h
    \end{align}
    as in \cref{def:jacobi-lambda}.
    Moreover, the sequences $(a^{(k,m)}_n)_{n \in \nnum}$ and $(f^{(k,h,m)}_n)_{n \in \nnum}$ obey \cref{hyp:conditions-an-fn}.
\end{proposition}
The statement is proven in \cref{sec:proof-squeezing}. 
The above decomposition reduces the study of higher-order squeezing operators to the Jacobi setting developed in \cref{sec:main_2}. 
In particular, combining \cref{lem:squeezing-decomposition} with \cref{prop:self-adjointness}, we obtain the following immediate consequence:
\begin{corollary}
    Let $A_{k,h}(K)$ be defined as in \cref{eq:general-higher-order} for $K \geq 0$.
    Suppose that $k \geq 3$ and $2h > k$.
    Then $A_{k,h}(K)$ is self-adjoint if and only if $K > 0$.
    In particular, for $K = 0$, $A_{k,h}(0) = A_k$ admits a family of self-adjoint extensions, parametrized by $V \in U(k)$.
\end{corollary}
This result was already established in \cite{fischer-selfadjointrealizationshigherorder-2025}, where it was obtained via a detailed asymptotic analysis based on the classical Birkhoff–Trjitzinsky theory. The approach developed there provides a complete and precise description of the deficiency structure in this setting.
In contrast, the Jacobi operator point of view adopted in the present work leads to a more direct argument.  Beyond its conceptual simplicity, this formulation is also more flexible, especially in view of possible extensions to related bosonic models, where similar ideas can be applied.
    
Of particular interest are self-adjoint extensions of $A_{k,h}(0)$ which decompose into a direct sum of self-adjoint extensions of the Jacobi operators $A^{(m)}_{k,h}(0)$.
We write
\begin{equation}
    \label{eq:diagonal-extensions}
    A_{k}^{(t_0,\dots,t_{k-1})} = k^{k/2} \bigoplus_{m = 0}^{k-1} A^{(m,t_m)}_k \, , 
\end{equation}
where $A^{(m,t_m)}_k$ is a self-adjoint extension of $A^{(m)}_{k,h}(0)$, parametrized by $t_m \in \brnum$ as in \cref{prop:jacobi-boundary-triplets}.
Interestingly, self-adjoint extensions of the form \cref{eq:diagonal-extensions} are exactly those which respect the $\nicefrac{2\pi}{k}$-rotational symmetry of the original operator $A_{k,h}(0)$, cf.~\cite[Theorem~3.5]{ibort-selfadjointextensionssymmetries-2015}.

For $2h > k$ and $K > 0$, the operator $A_{k,h}(K)$ can be viewed as a regularization of the higher-order squeezing operator $A_k$. As discussed above, it is not a priori clear how this regularization behaves in the limit $K \to 0$. Using \cref{thm:main-result}, we show that all possible limiting operators of $A_{k,h}(K)$ as $K \to 0$ are precisely of the form in \cref{eq:diagonal-extensions}:

\begin{theorem}
\label{prop:squeezing-main}
Suppose that $k \geq 3$ and $2h > k$, and let $A_{k,h}(K)$ be defined as in \cref{eq:general-higher-order}. Furthermore, let $A_{k}^{(t_0,\dots,t_{k-1})}$ denote the self-adjoint extensions of $A_{k,h}(0)$, parametrized by $(t_0,\dots,t_{k-1}) \in \brnum^k$ as in \cref{eq:diagonal-extensions}. Then the following holds:
\begin{enumerate}[(i)]
    \item \label{item:squeezing-main-i}
    Let $(K_j)_{j \in \nnum}$ be a sequence with $K_j > 0$ for all $j \in \nnum$ and $K_j \to 0$. Then there exist a subsequence $(K_{j_\ell})_{\ell \in \nnum}$ and parameters $(t_0,\dots,t_{k-1}) \in \brnum^k$ such that $A_{k,h}(K_{j_\ell})$ converges to $A_{k}^{(t_0,\dots,t_{k-1})}$ in the strong resolvent sense.

    \item \label{item:squeezing-main-ii}
    Let $m \in \{0,\dots,k-1\}$ and $t_m \in \brnum$. Then there exist parameters $(t_0,\dots,\cancel{t_m},\dots,t_{k-1}) \in \brnum^{k-1}$ and a sequence $(K_j)_{j \in \nnum}$ with $K_j > 0$ and $K_j \to 0$ such that $A_{k,h}(K_j)$ converges to $A_{k}^{(t_0,\dots,t_{k-1})}$ in the strong resolvent sense.
\end{enumerate}
\end{theorem}

The proof is given in \cref{sec:proof-squeezing}. 
This result clarifies the notion of physical regularization for higher-order squeezing operators by showing that the self-adjoint operators obtained in the limit $K \to 0$ form a genuinely non-unique family. In particular, the small-coupling limit does not select a canonical extension of the unregularized operator, but rather parametrizes a family of extensions of the form~\eqref{eq:diagonal-extensions}, which are thus compatible with the $\nicefrac{2\pi}{k}$-rotational symmetry. This answers a question first raised in \cite{fischer-selfadjointrealizationshigherorder-2025,ashhab-finitedimensionalapproximationsgeneralized-2026} on whether physically motivated regularizations might single out a preferred realization, and shows that, within this class of perturbations, no such selection principle emerges.

We illustrate our results with numerical simulations. In quantum optics, squeezing of the vacuum state $\phi_0$ is of particular interest. For $K > 0$, it follows from the decomposition in \cref{lem:squeezing-decomposition} that only the operator $A^{(0)}_{k,h}(K)$ governs the dynamics of the vacuum state. We therefore write, with slight abuse of notation,
\begin{equation}
    \label{eq:squeezed-plus-kerr-vacuum}
    \Psi_{k,h}(K,T) = \e^{-\iu k^{k/2} A^{(0)}_{k,h}(k^{h-k/2}K)T} \phi_0
\end{equation}
for $K,T>0$.

We are interested in the limiting dynamics of the vacuum state as $K \to 0$. Since all possible limiting extensions of $A_{k,h}(0)$ are of the form in \cref{eq:diagonal-extensions}, this leads naturally to the family of higher-order squeezed vacuum states
\begin{equation}
    \label{eq:sqeezed-vaccum}
    \Phi_k(t,T) = \e^{-\iu k^{k/2} A^{(0,t)}_k T} \phi_0 \, .
\end{equation}
Our main result on higher-order squeezing operators, \cref{prop:squeezing-main}, implies that for every $t \in \brnum$ there exists a sequence $K_j \to 0$, with $K_j > 0$, such that
\begin{equation}
    \lim_{j \to \infty} \Psi_{k,h}(K_j,T) = \Phi_k(t,T)
\end{equation}
for all $T>0$.

We illustrate this result numerically for the representative choice $k=h=3$ and $T=1$. 
To this end, we truncate the involved Hamiltonians to finite dimension $n$ and evaluate the corresponding time-evolution operators numerically. 
For $\lambda>0$, these finite-dimensional approximations converge strongly to the exact dynamics as $n\to\infty$, and therefore provide a faithful approximation scheme~\cite{fischer2025wrong,ashhab-finitedimensionalapproximationsgeneralized-2026}. 
In the singular case $\lambda=0$, the limiting dynamics depend on the parity of the truncation dimension: for even truncations $n=2r$, the dynamics converge to those generated by the extension corresponding to $t=0$, whereas odd truncations $n=2r+1$ converge to the extension corresponding to $t=\infty$, cf.~\cite[Theorem~1]{ashhab-finitedimensionalapproximationsgeneralized-2026}.

\begin{figure}[ht]
    \centering
    \includegraphics[width=0.95\linewidth]{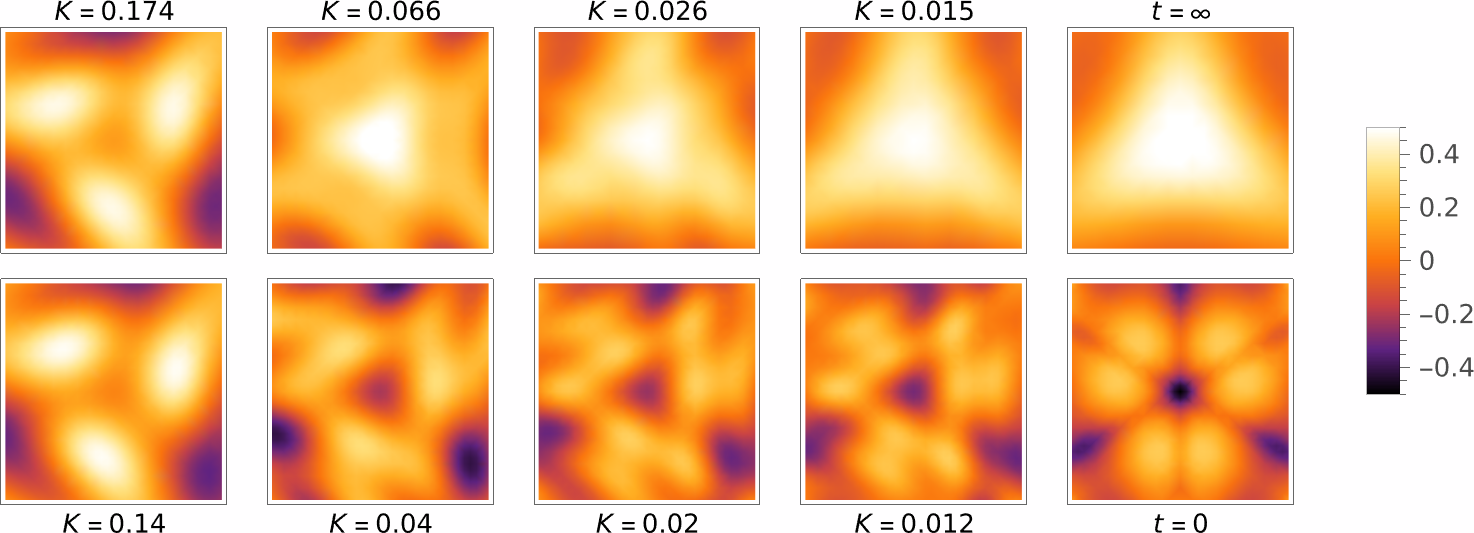}
    \caption{
Wigner functions of $\Psi_{3,3}(K,T)$ (cf.~\cref{eq:squeezed-plus-kerr-vacuum}) for different values of $K$ and fixed time $T=1$, together with the limiting states $\Phi_3(0,T)$ and $\Phi_3(\infty,T)$ (cf.~\cref{eq:sqeezed-vaccum}) corresponding to the two self-adjoint extensions with parameters $t=0$ and $t=\infty$, respectively. Different subsequences converge to different limiting states.
    }
    \label{fig:wigner-conv}
\end{figure}

Following the proof of \cref{thm:main-result}, we numerically construct a sequence $(K_j)_{j\in\nnum}\subset\rnum_+$ such that the eigenvalue $0\in\sigma(A^{(0,\infty)}_3)$ is approximately contained in the spectrum of $A^{(0)}_{3,3}(k^{h-k/2} K_j)$ for each $j$. This yields strong dynamical convergence of $A^{(0)}_{3,3}(k^{h-k/2} K_j)$ towards the extension $A^{(0,\infty)}_3$, and in particular
\begin{equation}
    \lim_{j\to\infty}\Psi_{3,3}(K_j,1)
    =
    \Phi_3(\infty,1)\, .
\end{equation}
We repeat the same construction for the extension $A^{(0,0)}_3$, obtaining a sequence $(K_j)_{j\in\nnum}\subset\rnum_+$ such that
\begin{equation}
    \lim_{j\to\infty}\Psi_{3,3}(K_j,1)
    =
    \Phi_3(0,1)\, .
\end{equation}
The resulting convergence is illustrated in \cref{fig:wigner-conv}, where the states are represented by their Wigner functions (see e.g.~\cite[11.8]{mandel-opticalcoherencequantum-1995} for an introduction). The figure shows that different subsequences converge to different limiting states. Moreover, the states $\Psi_{3,3}(K_j,1)$ become increasingly sensitive to small variations of $K_j$ as $K_j\to0$.

\section{Proof of \texorpdfstring{\cref{thm:main-result}}{the main theorem}}\label{sec:proof}

In this section we prove the main result of this work, \cref{thm:main-result}. The argument relies crucially on spectral properties of Jacobi operators, which we collect here. In addition to \cref{prop:self-adjointness}, we will in particular use the following facts:
\begin{itemize}
\item[(a)] The spectra of distinct self-adjoint extensions $J_t$ of $J(0)$ are discrete, mutually disjoint, and together cover the whole real line (\cref{prop:jacobi-limit-circle-spectrum});
\item[(b)] The eigenvalues of the operators $J(\lambda)$, $\lambda>0$, depend analytically on $\lambda$ (\cref{cor:eigen-analytic}) and diverge to $-\infty$ as $\lambda \to 0$ (\cref{prop:limit-eigenvalues});
\item[(c)] The square-summable generalized eigenvectors of the operators $J(\lambda)$ are uniformly bounded in $\lambda$ by a square-summable sequence (\cref{prop:asymptotics}).
\end{itemize}

The proof proceeds in two steps, corresponding to the two parts of \cref{thm:main-result}.
We first establish \cref{thm:main-result}\labelcref{item:main-result-some-extension}. 
To this end, we show that for every sequence of positive real numbers $(\lambda_j)_{j\in\nnum}$ one can extract a subsequence along which the generalized eigenvectors of $J(\lambda_j)$ converge weakly to those of a self-adjoint extension of $J(0)$. 
The key step is then to upgrade this weak convergence to strong resolvent convergence. Property (c) plays a crucial role in this step, which constitutes the most technically demanding part of the argument.

We then prove \cref{thm:main-result}\labelcref{item:main-result-all-extensions} using the following strategy:
\begin{enumerate}
\item Given $t \in \brnum$ and $E \in \sigma(J_t)$, we construct a sequence $(\lambda_j)_{j \in \nnum}$ with $\lambda_j>0$ and $\lambda_j \to 0$ such that $E \in \sigma(J(\lambda_j))$ for all $j \in \nnum$;
\item By \cref{thm:main-result}\labelcref{item:main-result-some-extension}, there exist a subsequence $(\lambda_{j_k})_{k \in \nnum}$ and a parameter $\tilde t \in \overline{\rnum}$ such that $J(\lambda_{j_k})$ converges to $J_{\tilde t}$ in the strong resolvent sense. Consequently, the spectra of $J(\lambda_{j_k})$ converge to the spectrum of $J_{\tilde t}$, implying $E \in \sigma(J_{\tilde t})$;
\item Since the spectra of the self-adjoint extensions of $J(0)$ are disjoint, this implies $t=\tilde t$, and the subsequence $(\lambda_{j_k})_{k \in \nnum}$ is the desired sequence.
\end{enumerate}

The remainder of this section is organized as follows. In \cref{sec:preliminaries}, we recall the necessary properties of Jacobi operators. In \cref{sec:spectral-properties}, we analyze the spectral properties of the operators $J(\lambda)$ and $J_t$. Finally, \cref{sec:main-theorem-proof} contains the proof of \cref{thm:main-result} following the strategy outlined above. The proof relies on a uniform bound on the square-summable generalized eigenvectors of $J(\lambda)$ (\cref{prop:asymptotics}), whose quite technical proof is postponed to  \cref{sec:asymptotics-of-eigenvectors}.

\subsection{Preliminaries}
\label{sec:preliminaries}

In this section, we collect notation and recall general properties of Jacobi operators that will be used throughout the remainder of the analysis. At this stage, we do not restrict to the specific class of operators introduced in \cref{sec:main_2}, nor do we assume that \cref{hyp:conditions-an-fn} holds.
\begin{definition}
    \label{def:jacobi}
    Let $(a_n)_{n \in \nnum},(b_n)_{n \in \nnum} \in \ell(\nnum)$ be two sequences with $a_n > 0$ and $b_n \in \rnum$ for all $n\in \nnum$.
    The Jacobi operator $\Top:\ell(\nnum)\rightarrow\ell(\nnum)$ is the operator defined by
    \begin{align}
    (\Top u)_n
    & =
    a_n u_{n+1}
    +
    b_n u_n
    +
    a_{n-1}u_{n-1},
    \qquad n \geq 1\, ,
    \\
    (\Top u)_0
    & =
    a_0 u_1
    +
    b_0 u_0 \, .
\end{align}
    Furthermore, we define the operators $\Tmin,T$, and $\Tmax$ on $\ell^2(\nnum)$: $\Tmin$ and $\Tmax$ are the restrictions of $\Top$ respectively to
    \begin{equation}
        \domain(\Tmin)=\ell_0(\nnum),\qquad\domain(\Tmax)=\{u\in\ell^2(\nnum):\Top u\in\ell^2(\nnum)\},
    \end{equation}
    and $T=\overline{\Tmin}$ is the closure of $\Tmin$.
\end{definition}
As before, we use calligraphic letters to denote Jacobi operators acting on the space of all sequences $\ell(\nnum)$, and straight letters for their realizations as (possibly unbounded) operators on the Hilbert space $\ell^2(\nnum)$.

Given $z\in\cnum$, we seek sequences $u \in \ell(\nnum)$ satisfying
\begin{equation}\label{eq:generalized_eigenvalue_0}
    \Top u = z u ,
\end{equation}
that is,
\begin{align}
    \label{eq:generalized-eigenvalue}
    a_n u_{n+1} + b_n u_n + a_{n-1} u_{n-1} &= z u_n , 
    \qquad n \geq 1,\\
\label{eq:generalized-eigenvalue_initial}
    a_0 u_1 + b_0 u_0 &= z u_0 .
\end{align}
In the literature, solutions of \cref{eq:generalized-eigenvalue} for some $z\in\cnum$ are commonly referred to as \textit{generalized eigenvectors} of $\Top$, with $z$ the corresponding \textit{generalized eigenvalue}. We adopt this convention in the following. Note, however, that a solution of \cref{eq:generalized-eigenvalue,eq:generalized-eigenvalue_initial} need not belong to $\ell^2(\nnum)$, and defines an actual eigenvector of a Hilbert space realization of $\Top$ only if $u\in\ell^2(\nnum)$.

By the theory of linear recurrence relations, \cref{eq:generalized-eigenvalue} admits a two-dimensional space of solutions in $\ell(\nnum)$. Equivalently, for any prescribed initial values $u_0,u_1 \in \cnum$, there exists a unique solution $u \in \ell(\nnum)$ satisfying these initial conditions \cite[Theorem~2.7]{elaydi-introductiondifferenceequations-2005}. We now introduce a pair of linearly independent solutions of \cref{eq:generalized-eigenvalue}, obtained from specific choices of initial data, which play a central role in the theory of Jacobi operators \cite[Chapter~16]{schmudgen-unboundedselfadjointoperators-2012}:
\begin{definition}
    \label{def:orthogonal-polynomials}
    Let $z\in\cnum$. The \textit{orthogonal polynomials of first and second kind} are the unique solutions $P(z)=(P_n(z))_{n\in\nnum},Q(z)=(Q_n(z))_{n\in\nnum}\in\ell(\nnum)$ of \cref{eq:generalized-eigenvalue} respectively corresponding to the initial conditions
    \begin{alignat}{2}
        P_0(z) & = 1\, , & \quad  P_1(z)& = (z-b_0)/a_0\, ,  \\
        Q_0(z) & = 0 & \quad  Q_1(z) & = 1/a_0 \, .
    \end{alignat} 
\end{definition}
From the recurrence relation~\eqref{eq:generalized-eigenvalue} it follows that $P_n(z)$ and $Q_{n+1}(z)$ are polynomials of degree $n$ in $z$ \cite[Section~16.1-2]{schmudgen-unboundedselfadjointoperators-2012}.
In particular, as $P(z)$ also satisfies the initial condition~\eqref{eq:generalized-eigenvalue_initial}, we get
\begin{equation}
    \Top P(z)=zP(z),
\end{equation}
i.e. $P(z)$ is the (up to scalar multiples) unique solution of \cref{eq:generalized_eigenvalue_0} with eigenvalue $z \in \cnum$.
On the other hand, $Q(z)$ does not satisfy \cref{eq:generalized-eigenvalue_initial}, and
\begin{equation}
    \Top Q(z)=zQ(z)+e_0,
\end{equation}
where $(e_n)_m=\delta_{nm}$, cf.~\cref{sec:main_2}. Moreover, since $\Top$ has real coefficients,
\begin{equation}
    \label{eq:orthogonal-poly-conjugate}
    P_n(z^\ast) = P_n(z)^\ast\, , \quad Q_n(z^\ast)=Q_n(z)^\ast \quad \forall n \in \nnum.
\end{equation}
Another useful tool is the \textit{Wronskian} (also called the \textit{Casoratian}, see \cite[Section~2.2]{elaydi-introductiondifferenceequations-2005}), which has a direct analogue in the theory of Sturm--Liouville operators \cite[Chapter~15]{schmudgen-unboundedselfadjointoperators-2012}. 
We follow the definition given in \cite[Eq.~1.21]{teschl-jacobioperatorscompletely-1999}.
\begin{definition}
    \label{def:wronksian}
    Let $\Top$ be a Jacobi operator defined as in \cref{def:jacobi}.
    For $u,v \in \ell(\nnum)$ the (modified) \textit{Wronskian} is given by
    \begin{equation}
        W_n[u,v] = a_n\left( u_n v_{n+1} - u_{n+1} v_n \right) \,.
    \end{equation}
\end{definition}
It has the following important properties:
\begin{lemma}
    \label{lem:wronskian}
    Let $W_n[\cdot,\cdot]$ be the Wronskian as defined in \cref{def:wronksian}, and let $u,v,w$ be three solutions of \cref{eq:generalized-eigenvalue}.
    Then the following holds:
    \begin{enumerate}[(i)]
        \item $W_n[u,v]= W_0[u,v]$ for all $n \in \nnum$, and we write $W[u,v] = W_n[u,v]$;
        \item $W[u,v] = 0$ if and only if $u$ and $v$ are linearly dependent;
        \item Suppose that $u,v$ are linearly independent. Then
        \begin{equation}
            w = \frac{W[w,v]}{W[u,v]} u - \frac{W[w,u]}{W[u,v]} v \, ;
        \end{equation}
        \item $W[P(z),Q(z)] = 1$;
    \end{enumerate}
\end{lemma}
\begin{proof}
    The first three statements follow from \cite[7]{teschl-jacobioperatorscompletely-1999}.
    The last statement follows directly from the definition of the orthogonal polynomials (\cref{def:orthogonal-polynomials}):
    \begin{equation}
        W[P(z),Q(z)] = W_0[P(z),Q(z)] = a_0 \left(P_0(z) Q_1(z) - P_1(z) Q_0(z) \right) = 1.\qedhere
    \end{equation}
\end{proof}
We return to the operator-theoretic properties of $T$ and recall a fundamental criterion for the essential self-adjointness of the operator $T:\domain(T)\subset\ell^2(\nnum)\to\ell^2(\nnum)$ associated with $\Top$; see, e.g., \cite[Proposition~6.5, Corollary~6.7, Theorem~6.16]{schmudgen-momentproblem-2017} (cf.\ also \cref{prop:weyl-alternative}):
\begin{proposition}
    \label{prop:jacobi-esssa}
    Let $\Top$ be a Jacobi operator, and $T$ the corresponding operator on $\ell^2(\nnum)$ from \cref{def:jacobi}.
    Then $T^\ast= \Tmax$, and the deficiency indices of $T$ are either $(0,0)$ or $(1,1)$.
    In particular, $T$ is essentially self-adjoint if and only if $P(z) \notin \ell^2(\nnum)$ for some, and then all, $z \in \cnum \setminus \rnum$, and is otherwise not essentially self-adjoint.
\end{proposition}
In the context of Jacobi operators, the first case is called the \emph{limit point} case, while the second is referred to as the \emph{limit circle} case; we refer to \cref{rem:terminology} for an explanation of this terminology.
A well-known sufficient (but not necessary) condition for essential self-adjointness is the following one, see \cite[24]{akhiezer-classicalmomentproblem-2020}:

\begin{proposition}[Carleman condition]\label{prop:carleman}
$T$ is essentially self-adjoint whenever the following condition holds:
\begin{equation}
    \label{eq:carleman}
    \sum_{n=0}^\infty \frac{1}{a_n} = \infty. \,
\end{equation}
\end{proposition}
The following result provides a characterization of the deficiency index dichotomy (\cref{prop:jacobi-esssa}) in terms of square-summable solutions of the generalized eigenvalue equation, and is known as the Weyl alternative \cite[Lemmas~2.15--2.16]{teschl-jacobioperatorscompletely-1999}:
\begin{proposition}[Weyl alternative]
    \label{prop:weyl-alternative}
    Let $\Top$ be a Jacobi operator. 
    Then exactly one of the following holds:
    \begin{enumerate}[(i)]
        \item For one (and hence all) $z \in \cnum$, all solutions of \cref{eq:generalized-eigenvalue} are square-summable. In this case $\Top$ is limit circle, i.e. $T$ is not essentially self-adjoint.
        \item \label{item:weyl-alternative-limit-point}
        For one (and hence all) $z \in \cnum$, there exists, up to scalar multiples, a unique square-summable solution of \cref{eq:generalized-eigenvalue}. In this case $\Top$ is limit point, i.e. $T$ is essentially self-adjoint.
    \end{enumerate}
\end{proposition}
Finally, we recall a result by Świderski~\cite{swiderski-spectralpropertiesblock-2018}, which will be used throughout this paper. The original statement is formulated in the more general setting of block Jacobi matrices; here we present the version adapted to our framework.
\begin{theorem}[{\cite[Theorem~2, Theorem~3]{swiderski-spectralpropertiesblock-2018}}]
    \label{thm:swiderski}
    Let $(a_n)_{n \in \nnum}$ and $(b_n)_{n \in \nnum}$ be two sequences with $a_n > 0$ and $b_n \in \rnum$ for all $n \in \nnum$.
    Furthermore, suppose the following:
    \begin{enumerate}[(i)]
        \item $\mathcal{V}(a_n^{-1}) + \mathcal{V}(a_n^{-1}b_n) + \mathcal{V}(a_n^{-1}a_{n-1})< \infty$, where $\mathcal{V}$ is the total variation (cf.~\cref{def:total-variation});
        \item $\sum_{n=0}^\infty a_n^{-1} < \infty$, i.e. the Carleman condition \cref{eq:carleman} does not hold;
        \item The limits 
        \begin{equation}
        \tau = \lim_{n \to \infty} a_n^{-1},\quad q=\lim_{n \to \infty}a_n^{-1}b_n,\quad r = \lim_{n \to \infty}a_n^{-1} a_{n-1}, \quad c = \lim_{n \to \infty} a_n/|a_n|
        \end{equation}
        exist and are finite;
        \item There exists $z \in \rnum$ such that the quadratic form associated with the matrix
        \begin{equation}
            \mathcal{F}(z) = \Real \left( \begin{pmatrix}
                0 & -c \\ c & 0
            \end{pmatrix}
            \begin{pmatrix}
                0 & 1 \\ r & z\tau-q
            \end{pmatrix}\right)\, , 
        \end{equation}
        where $\Real A = \frac{1}{2}(A+A^\ast)$, is either strictly positive or strictly negative definite.
    \end{enumerate}
    Then, the deficiency indices of the Jacobi operator $T$ associated with $(a_n)_{n \in \nnum}$ and $(b_n)_{n \in \nnum}$,  are $(1,1)$.
\end{theorem}
\begin{remark}
\cite[Theorem~25]{swiderski-spectralpropertiesblock-2018} additionally provides bounds on the generalized eigenvectors of $T$.
    In \cref{sec:turan}, we closely follow the proof of Świderski to prove an auxiliary lemma.
\end{remark}
\noindent Further sufficient or necessary conditions for essential self-adjointness (or the lack thereof) can be found in the literature, see e.g. \cite{schmudgen-momentproblem-2017}.

In the limit circle case, all self-adjoint extensions of $T$ can be parametrized explicitly. 
Recalling from \cref{prop:jacobi-esssa} that the deficiency indices of $T$ are either $(0,0)$ or $(1,1)$, it follows that in the limit circle case the deficiency spaces of $T$ are given by \cite[Lemma~16.8]{schmudgen-unboundedselfadjointoperators-2012}
\begin{equation}
    \mathcal{N}_{\pm \iu}(T) = \ker(\Tmax \mp \iu) = \cnum P(\pm \iu)\, .
\end{equation}
By von Neumann's extension theory, this yields a one-parameter family of self-adjoint extensions of $T$ parametrized by $U(1)$ \cite[Theorem~13.10]{schmudgen-unboundedselfadjointoperators-2012}. 

For Jacobi operators, however, it is more convenient to work with the equivalent parametrization in terms of boundary triplets; see \cite[Chapter~14]{schmudgen-unboundedselfadjointoperators-2012} for an overview. We shall therefore adopt this framework throughout the paper and recall the relevant construction next. Starting from the identity
\begin{equation}\label{eq:domain_tmax}
    \domain(\Tmax) = \domain(T) + \cnum Q(0) + \cnum P(0)\, ,
\end{equation}
we now introduce a boundary triplet for $\Tmax$:
\begin{proposition}[{\cite[Corollary~16.27]{schmudgen-unboundedselfadjointoperators-2012} and \cite{allahverdiev-extensionsdilationsfunctional-2005}}]
    \label{prop:jacobi-boundary-triplets}
    Let $\Top$ be in the limit circle case, and define the maps 
    $\Gamma_0,\Gamma_1:\domain(\Tmax)\to\cnum$ by
    \begin{equation}
        \Gamma_0(\psi)=c_0, 
        \qquad 
        \Gamma_1(\psi)=-c_1,
    \end{equation}
    where $\psi\in\domain(\Tmax)$ is written (cf.~\cref{eq:domain_tmax}) as
    \begin{equation}
        \psi=\psi_0+c_0Q(0)+c_1P(0)
    \end{equation}
    with $\psi_0\in\domain(T)$ and $c_0,c_1\in\cnum$. 
    Then $(\cnum,\Gamma_0,\Gamma_1)$ is a boundary triplet for $\Tmax$. Moreover, all maximally accretive or dissipative extensions of $T$ are parametrized by $t\in\bcnum=\cnum\cup\{\infty\}$ via
    \begin{equation}
        \domain(T_t)
        =
        \domain(T)
        +\cnum\bigl(Q(0)+tP(0)\bigr),
        \qquad
        T_t\psi=\Tmax\psi
        \quad
        \forall\,\psi\in\domain(T_t).
    \end{equation}
    The self-adjoint extensions correspond precisely to $t\in\overline{\rnum}$. Finally, these extensions are \emph{disjoint} in the sense that if $t_1\neq t_2$, then
    \begin{equation}
        T_{t_1}\cap T_{t_2}=T.
    \end{equation}
\end{proposition}
\begin{remark}
    Self-adjoint extensions are discussed in \cite[Corollary~16.27]{schmudgen-unboundedselfadjointoperators-2012}, while maximally accretive or dissipative extensions of Jacobi operators are treated in \cite{allahverdiev-extensionsdilationsfunctional-2005}. For a general reference on boundary triplets, see \cite[Proposition~14.7]{schmudgen-unboundedselfadjointoperators-2012}. Note that, in the latter reference, the extension $T_t$ corresponds to the linear relation $B=-t$.

Furthermore, the maps $\Gamma_0$ and $\Gamma_1$ can also be defined in terms of the Wronskian of $\Top$, cf.\ \cite{allahverdiev-extensionsdilationsfunctional-2005}.
\end{remark}

We now turn to resolvents of Jacobi operators. Throughout the paper, we denote the spectrum and resolvent of an operator $A$ on a Hilbert space by $\sigma(A)$ and $\rho(A)$, respectively. 
To treat the limit point and limit circle cases in a unified way, we fix an operator $\tilde{T}$ as follows: in the limit point case we set $\tilde{T}=T$, while in the limit circle case we let $\tilde{T}=T_t$ be an arbitrary maximally accretive or dissipative extension of $T$, corresponding to some fixed parameter $t\in\bcnum$. From now on, this choice of $\tilde{T}$ remains fixed throughout the discussion.

\begin{definition}
    \label{def:green}
    Let $\tilde{T}$ be a maximally accretive or dissipative extension of $T$, and $n,m\in\nnum$.
    The \textit{Green function} of $\tilde{T}$ is the function $G_{nm}:\rho(\tilde{T})\rightarrow\cnum$ given by
    \begin{equation}
        G_{nm}(z) = \braket{e_n,(\tilde{T}-z)^{-1}e_m} \, ,
    \end{equation}
    and $m(z) = G_{00}(z)$ is the \textit{Weyl $m$-function} of $\tilde{T}$.
\end{definition}
\begin{remark}
    The Weyl $m$-function $m(z)$ is closely related to the Weyl function appearing in the general theory of boundary triplets; more precisely, it can be identified with a particular realization of the abstract Weyl function corresponding to a specific choice of boundary triplet, see \cite[Definition~14.4]{schmudgen-unboundedselfadjointoperators-2012}. It also plays an important role in the theory of Sturm--Liouville operators; see, e.g., \cite[p.~353]{schmudgen-unboundedselfadjointoperators-2012}.
\end{remark}
Due to the special structure of Jacobi operators, the Green function can be explicitly written in terms of the Weyl $m$-function and the orthogonal polynomials of first and second kind, see \cref{def:orthogonal-polynomials}:
\begin{lemma}
    \label{lem:green}
    Let $\tilde{T}$ be a maximally accretive or dissipative extension of $T$. For every $z \in \rho(\tilde{T})$,
        \begin{equation}
            G_{nm}(z) = \begin{cases}
            (m(z)P_n(z)+Q_n(z)) P_m(z) & n \geq m \\
            (m(z)P_m(z)+Q_m(z)) P_n(z) & n \leq m 
        \end{cases}\, .
    \end{equation}
\end{lemma}
\begin{proof}
    Let $z \in \rho(\tilde{T})$ and $v \in \ell^2(\nnum)$.
    By \cite[Eq.~2.27]{yafaev-selfadjointjacobioperators-2021}, 
    \begin{equation}
        ((\tilde{T}-z)^{-1} v)_n = \frac{1}{W_0[P(z),u(z)]} \left( u_n(z) \sum_{m = 0}^n P_m(z) v_m + P_n(z) \sum_{m = n+1}^\infty u_m(z) v_m \right)\, ,
    \end{equation}
    where $u(z) = m(z)P(z) + Q(z)$.
    Using \cref{lem:wronskian}, $W_0[P(z),u(z)] = W_0[P(z),Q(z)] = 1$, and the above expression is equivalent to
    \begin{equation}
    ((\tilde{T}-z)^{-1} v)_n = \sum_{m = 0}^\infty G_{nm}(z) v_m \, \quad 
        G_{nm}(z) = \begin{cases}
            u_n(z) P_m(z) & n \geq m \\
            u_m(z) P_n(z) & n \leq m 
        \end{cases}\, .
    \end{equation}
    The claim follows using $u(z) = m(z)P(z) + Q(z)$.
\end{proof}
The sequence $m(z)P(z) + Q(z)$ has some additional useful properties:
\begin{lemma}
    \label{lem:weyl-solution}
    Let $\tilde{T}$ be a maximally accretive or dissipative extension of $T$, and $z \in \rho(\tilde{T})$.
    Then, $u(z) = m(z) P(z) + Q(z)$ is a solution of \cref{eq:generalized-eigenvalue}.
    Moreover, $u(z) \in \domain(\tilde{T})$, and in particular $u(z) \in \ell^2(\nnum)$.
\end{lemma}
\begin{proof}
    $P(z)$ and $Q(z)$ are solutions of \cref{eq:generalized-eigenvalue}, and thus by linearity the same follows for $m(z) P(z) + Q(z)$.
    Furthermore, by \cref{lem:green},
    \begin{equation}
        (\tilde{T}-z)^{-1} e_0 = m(z) P(z) + Q(z)\, , 
    \end{equation}
    and therefore $m(z) P(z) + Q(z) \in \domain(\tilde{T})$.
\end{proof}
\begin{remark}
    If $\Top$ is limit point, $T$ is self-adjoint, and thus has no other maximally accretive or dissipative extensions.
    Hence, $u(z)$ is the unique square-summable solution of \cref{eq:generalized-eigenvalue} described in \cref{prop:weyl-alternative}.
    It is also sometimes called the Weyl solution~\cite{eckhardt-singularweyltitchmarshkodairatheory-2013}.
\end{remark}
The Weyl $m$-function $m(z)$ plays a crucial role in the spectral analysis of Jacobi operators.
For self-adjoint extensions $T$, $m(z)$ is a Nevanlinna (or Herglotz) function \cite[Proposition~16.18]{schmudgen-unboundedselfadjointoperators-2012} and coincides with the Stieltjes transform of the spectral measure of $T$. Hence the entire spectral information of $T$ is encoded in $m(z)$.

In the non-essentially self-adjoint case, we denote by $m(z,t)$ the Weyl $m$-function of the maximally accretive or dissipative extension $T_t$, and by $G_{nm}(z,t)$ its Green function.
The function $m(z,t)$ can be expressed explicitly in terms of the orthogonal polynomials:
\begin{lemma}
    \label{lem:weyl-m-flt}
    Let $\Top$ be limit circle, $P(z),Q(z)$ the associated orthogonal polynomials of first and second kind, and $m(z,t)$ be the Weyl $m$-function of the maximally accretive or dissipative extension $T_t$ of $T$, parametrized by $t \in \bcnum$ as in \cref{prop:jacobi-boundary-triplets}.
    Then, for fixed $z \in \cnum$, $m(z,t)$ is a fractional linear (Möbius) transform in $t \in \bcnum$ explicitly given by
    \begin{equation}\label{eq:weyl_fractional}
        m(z,t) = -\frac{A(z,0) + C(z,0) t}{B(z,0) + D(z,0) t} \, , 
    \end{equation}
    where $A(z,w),B(z,w),C(z,w),D(z,w)$ are entire Nevanlinna functions in $z,w \in \cnum$ defined by
    \begin{alignat}{2}
        A(z,w) & = (z-w)\sum_{n=0}^\infty Q_n(z)Q_n(w)\, , \quad & B(z,w) & = -1 + (z-w) \sum_{n=0}^\infty P_n(z) Q_n(w) \, , \\
        C(z,w) & = -1 + (z-w) \sum_{n=0}^\infty Q_n(z) P_n(w) \, , \quad & D(z,w) & = (z-w) \sum_{n=0}^\infty P_n(z) P_n(w) \, .
    \end{alignat}
    Moreover, $m(z^\ast,t^\ast) = m(z,t)$ holds for all $z \in \cnum$ and $t \in \bcnum$.
    
    Conversely, the set of points where the right-hand side of \cref{eq:weyl_fractional} is well-defined coincides with the resolvent set of $T_t$, that is, $z\in\sigma(T_t)$ if and only if $B(z,0)+D(z,0)t=0$, with $B(z,0),D(z,0)$ as defined above.
\end{lemma}
\begin{proof}
    The properties of the Nevanlinna functions follow from \cite[Lemma~16.23]{schmudgen-unboundedselfadjointoperators-2012}, and the expression of the Weyl $m$-function from \cite[Eq.~16.36]{schmudgen-unboundedselfadjointoperators-2012}.
    Using \cref{eq:orthogonal-poly-conjugate}, we obtain $A(z^\ast,0) = A(z,0)^\ast$, $B(z^\ast,0) = B(z,0)^\ast$, $C(z^\ast,0) = C(z,0)^\ast$, $D(z^\ast,0) = D(z,0)^\ast$ for all $z \in \cnum$, and $m(z^\ast,t^\ast) = m(z,t)$ follows.

    For all $t \in \bcnum$ the resolvent of $T_t$ is Hilbert--Schmidt~\cite[Corollary~3.1]{allahverdiev-spectralproblemsjacobi-}, and hence the spectrum of $T_t$ is purely discrete.
    Since the coefficients $A(z,0),B(z,0),C(z,0),D(z,0)$ are entire functions, and $A(z,0) + C(z,0) t$ and $B(z,0) + D(z,0) t$ have no common zeros \cite[383]{schmudgen-unboundedselfadjointoperators-2012}, the poles of $m(z,t)$ are exactly the solutions $z \in \cnum$ of the equation
    \begin{equation}
        B(z,0) + D(z,0) t = 0 \,.
    \end{equation}
    But by \cite[Theorem~2.2]{allakhverdiev-spectraltheorydissipative-1990}, these $z \in \cnum$ are exactly the eigenvalues of $T_t$.
\end{proof}
\begin{remark}\label{rem:terminology}
For any fixed $z\in\cnum$ with $\Imag z>0$, the values of $m(z,t)$ as $t$ ranges over $\overline{\rnum}$ lie on a circle in $\cnum_+$, which motivates the terminology ``limit circle''. In contrast, in the essentially self-adjoint case there is only one closed extension of $\Tmin$, and hence only a single Weyl $m$-function. In this sense the circle degenerates to a single point, which is consistent with the complementary terminology ``limit point''.
\end{remark}

\subsection{Spectral properties of the operators \texorpdfstring{$J(\lambda)$}{ }}
\label{sec:spectral-properties}

We now turn to the family of Jacobi operators introduced in \cref{sec:main_2}, see \cref{def:jacobi-lambda}. Throughout this section we assume \cref{hyp:conditions-an-fn}. We analyze separately the operators $J(0)$ and $J(\lambda)$ for $\lambda>0$, establishing \cref{prop:self-adjointness} together with additional properties needed for the proof of the main theorem.

We start by analyzing the operator $J(0)$, and show that it is limit circle using \cite[Theorem~3]{swiderski-periodicperturbationsunbounded-2018}, cf.~\cref{thm:swiderski}.
The following property will be useful throughout this paper:
\begin{lemma}
    \label{lem:an-variation}
    The following properties hold:
    \begin{equation}
        \mathcal{V}((a_n^{-1})_{n\in\nnum}) < \infty \, , \quad \mathcal{V}((a_{n-1}/a_n)_{n\in\nnum}) < \infty \, ,
    \end{equation}
    with $\mathcal{V}(\cdot)$ being the total variation (\cref{def:total-variation}). Furthermore, the Carleman condition \cref{eq:carleman} is not satisfied.
\end{lemma}
\begin{proof}
   A simple calculation yields, for $n \geq 1$,
    \begin{align}
        \frac{1}{a_{n-1}} - \frac{1}{a_n} & = \frac{a_n-a_{n-1}}{a_n a_{n-1}} = \frac{n^\alpha(1+O(1/n))-(n-1)^\alpha(1+O(1/n))}{a_n a_{n-1}} \\
        & = \frac{n^\alpha}{a_n a_{n-1}} \left(1+O(1/n)-(1-1/n)^\alpha(1+O(1/n))\right) \\
        & = \frac{n^\alpha}{a_n a_{n-1}} \left(1+O(1/n)-(1-O(1/n))(1+O(1/n))\right) \\
        & = \frac{n^\alpha }{a_n a_{n-1}}O(1/n)\, ,
    \end{align}
    where we used \cref{hyp:conditions-an-fn} and $(1+x)^\alpha = 1+\alpha x + o(x)$ for $x\to0$. But, again by \cref{hyp:conditions-an-fn},
    \begin{equation}
        a_na_{n-1}=n^{2\alpha}\left(1+O(1/n)\right),
    \end{equation}
    whence
    \begin{equation}
         \frac{1}{a_{n-1}} - \frac{1}{a_n}=\frac{1}{n^\alpha}O(1/n),
    \end{equation}
    that is, there exists $C>0$ such that 
    \begin{equation}
        \left|\frac{1}{a_{n-1}} - \frac{1}{a_n}\right| \leq \frac{C}{n^\alpha n} \,, 
    \end{equation}
    and as $\alpha > 0$, $\mathcal{V}((a_n^{-1})_{n\in\nnum}) < \infty$.
    
    For the second statement, let $n \geq 2$ and consider
    \begin{align}
        \frac{a_{n-1}}{a_n} & = \frac{(n-1)^\alpha}{n^\alpha}\frac{1+c_a/(n-1)+O(1/n^2)}{1+c_a/n+O(1/n^2)} \\
        & = \left(1-\frac{1}{n}\right)^\alpha \left(1+c_a/(n-1)+O(1/n^2)\right)\left(1-c_a/n+O(1/n^2)\right) \\
        & = \left(1-\frac{1}{n}\right)^\alpha \left(1+\frac{c_a}{n(n-1)}+O(1/n^2) \right) \\
        & = \left(1-\frac{1}{n}\right)^\alpha \left(1+O(1/n^2) \right) \, .
    \end{align}
    Thus we obtain
    \begin{align}
        \frac{a_{n-1}}{a_n} - \frac{a_{n-2}}{a_{n-1}}&  = \left(1-\frac{1}{n}\right)^\alpha (1+O(1/n^2)) - \left(1-\frac{1}{n-1}\right)^\alpha (1+O(1/n^2)) \\
        & = \left(1-\frac{\alpha}{n}+O(1/n^2)\right)(1+O(1/n^2)) - \left(1-\frac{\alpha}{n-1}+O(1/n^2)\right) (1+O(1/n^2)) \\
        & = \frac{\alpha}{n(n-1)} + O(1/n^2) \, , 
    \end{align}
    and $\mathcal{V}(a_{n-1}/a_n) < \infty$ follows.
    Finally, as $a_n = n^\alpha(1+O(1/n))$, 
    \begin{equation}
        \sum_{n=0}^\infty \frac{1}{a_n} < \infty
    \end{equation}
    if and only if $\alpha > 1$; by \cref{hyp:conditions-an-fn}, $\alpha>\nicefrac{4}{3}>1$ and the proof is complete.
\end{proof}
\begin{proposition}
    \label{prop:j0-limit-circle}
 The operator $J(0)$ is not self-adjoint and has deficiency indices $(1,1)$; therefore, it admits a one-parameter family of maximally accretive (respectively, dissipative) extensions $(J_t)_{t \in \bcnum}$. In particular, the self-adjoint extensions are parametrized by $t \in \brnum$.
\end{proposition}
\begin{proof}
    We show that the conditions of \cref{thm:swiderski} hold.
    By \cref{lem:an-variation}, $\mathcal{V}((a_n^{-1})_{n\in\nnum}) < \infty$ and  $\mathcal{V}((a_{n-1}/a_n)_{n\in\nnum}) < \infty$.
    Furthermore, $\mathcal{V}((\lambda f_n/a_n)_{n \in \nnum}) = 0$ as $\lambda = 0$. 
    Finally, 
    \begin{alignat}{2}
        \tau & = \lim_{n \to \infty} a_n^{-1}  = 0 \, & \quad  r   & = \lim_{n \to \infty} a_{n-1}a_n^{-1} = 1 \\
        q & = \lim_{n \to \infty} \lambda f_n a_n^{-1}  = 0 &\quad  c & = \lim_{n \to \infty} |a_n| a_n^{-1} = 1 \, ,
    \end{alignat}
    and the matrix $\mathcal{F}(z)$ in \cref{thm:swiderski} simplifies to
    \begin{equation}
        \mathcal{F}(z) = \begin{pmatrix}
            1 & 0 \\ 0 & 1
        \end{pmatrix}\, ,
    \end{equation}
    which is positive definite for all $z \in \cnum$.
    Furthermore, the Carleman condition does not hold as $\alpha > 1$, cf.~\cref{lem:an-variation}, and all assumptions of \cref{thm:swiderski} are satisfied, which proves that the operator is not essentially self-adjoint. \cref{prop:jacobi-boundary-triplets} concludes the proof.
\end{proof}
The importance of the matrix $\mathcal{F}(z)$ will become apparent in \cref{sec:turan}, where we will obtain explicit asymptotics of generalized eigenvectors.

Hereafter we focus on the self-adjoint extensions of $J(0)$. We collect here some properties of the spectra of the operators $J_t$.
\begin{proposition}
    \label{lem:jt-unbounded-from-below}
   All self-adjoint extensions $J_t$ of $J(0)$ are unbounded from below.
\end{proposition}
\begin{proof}
   It suffices to prove that $J(0)$ itself is unbounded from below.
    To this end, let $N \in \nnum$ and define $\psi^N \in \ell_0(\nnum)\subset \domain(J(0))$ as
    \begin{equation}
        \psi^N_n = \begin{cases}
            \frac{(-1)^n}{\sqrt{N}} & n \leq N -1 \\
            0 & n \geq N 
        \end{cases}\, .
    \end{equation}
    Clearly,
    \begin{equation}
        \norm{\psi^N}^2 = \sum_{n = 0}^{N-1} \frac{1}{N} = 1 \, .
    \end{equation}
    Furthermore,
    \begin{equation}  
        (J(0) \psi^N)_n =a_n \psi^N_{n+1} + a_{n-1} \psi^N_{n-1} = \frac{1}{\sqrt{N}} \begin{cases}     
            -a_0 & n = 0 \\
            (-1)^{n+1}a_n + (-1)^{n-1}a_{n-1} & 1 \leq n \leq N-2 \\
            (-1)^{n-1}a_{n-1} & N-1 \leq n \leq N \\
            0 & n > N
        \end{cases} \, , 
    \end{equation}
    and hence
    \begin{equation}
        \frac{\braket{\psi^N,J(0)\psi^N}}{\norm{\psi^N}^2}  = -\frac{1}{N} \left(a_0  +\sum_{n=1}^{N-2} (a_n+a_{n-1}) + a_{N-2}\right) \leq -\frac{a_{N-2}}{N} \, . 
    \end{equation}
    Using \cref{hyp:conditions-an-fn}, $\lim_{N \to \infty} -\frac{a_{N-2}}{N} = -\infty$, and the statement follows.
\end{proof}

We will denote by $m(z,t)$ the Weyl $m$-function of the operator $J_t$, see \cref{def:green}. 
\begin{proposition}
    \label{prop:jacobi-limit-circle-spectrum}
    All self-adjoint extensions $J_t$ of $J(0)$ have purely discrete spectrum with isolated eigenvalues of multiplicity one.
    The eigenvalues are exactly the poles of $m(z,t)$.
    Furthermore, for each $E \in \rnum$ there exists exactly one $t \in \brnum$ such that $E \in \sigma(J_t)$.
\end{proposition}
\begin{proof}
    This is \cite[Coroll.~16.30]{schmudgen-unboundedselfadjointoperators-2012}.\footnote{
The last statement does not appear explicitly in the formulation of \cite[Corollary~16.30]{schmudgen-unboundedselfadjointoperators-2012}, but follows directly from its proof.}
\end{proof}
This concludes the analysis of $J(0)$ and its extensions.\medskip

We now turn to the operators $J(\lambda)$ for $\lambda>0$. Our strategy is to view the minimal operator $\Jmin(\lambda)$, defined on $\ell_0(\nnum)$ (cf.~\cref{def:jacobi-lambda}), as a perturbation of $\Jmin(0)$ by decomposing it as $\Jmin(\lambda)=\Jmin(0)+\lambda F$, where $F$ is the Jacobi operator determined by the diagonal entries of $\Jmin(\lambda)$. We will exploit this decomposition to analyze the spectral properties of $J(\lambda)$ in terms of those of $F$. The rationale behind this approach lies in the scaling assumptions in \cref{hyp:conditions-an-fn}: since $a_n\sim n^\alpha$ and $f_n\sim n^\beta$ with $\beta>\alpha$, the diagonal part grows asymptotically faster than the off-diagonal coefficients, suggesting that $\Jmin(0)$ behaves as a lower-order perturbation of $F$ as long as $\lambda>0$. This heuristic will be made precise below.

\begin{definition}\label{def:diagonal-term}
    $F:\domain(F)\subset\ell^2(\nnum)\rightarrow\ell^2(\nnum)$ is the unique closed operator satisfying $Fe_n=f_ne_n$.
\end{definition}

We gather here some properties of this operator:

\begin{proposition}\label{prop:diagonal-term}
    $F$ is a nonnegative self-adjoint operator and has a core in $\ell_0(\nnum)$. Besides, its spectrum $\sigma(F)$ satisfies
    \begin{equation}\label{eq:spectrum_F}
        \sigma(F)=\{f_n\}_{n\in\mathbb{N}}
    \end{equation}
    and is purely discrete, that is, its essential spectrum is empty: $\sigma_{\rm ess}(F)=\emptyset$.
\end{proposition}
\begin{proof}
    By definition, $(e_n)_{n\in\nnum}$ is a complete eigenbasis of $\ell^2(\nnum)$, the corresponding eigenvalues $\{f_n\}_{n\in\nnum}$ being nonnegative. The only remaining statement to prove is $\sigma_{\rm ess}(F)=\emptyset$. By \cref{hyp:conditions-an-fn}, $f_n = n^\beta(1+O(1/n))$, and therefore $\lim_{n \to \infty} f_n = \infty$.
   Hence, the eigenvalues $f_n$ diverge to $+\infty$ and therefore have no finite accumulation point. Since $(e_n)_{n\in\nnum}$ is a complete eigenbasis, it follows that the spectrum is purely discrete and $\sigma_{\rm ess}(F)=\emptyset$.
\end{proof}

\begin{lemma}
    \label{lem:relatively-compact}
    $J(0)$ is infinitesimally relatively bounded with respect to $F$, and relatively compact with respect to $F^k$ for some $k \in \nnum$.
\end{lemma}
\begin{proof}
    We begin by proving infinitesimal relative boundedness of $J(0)$ with respect to $F$. Let $\epsilon > 0$.
    By \cref{hyp:conditions-an-fn}, and as $\beta>\alpha$, there exists $N_0 \in \nnum$ and $0<C<\epsilon^2/2$ such that $a_n \leq C f_n$ for all $n \geq N_0$.
    Let 
    \begin{equation}
        b^2 = \max_{n \leq N_0} (a_n^2+a_{n-1}^2 ) \, .
    \end{equation}
    Take $\psi\in\ell_0(\nnum)$, that is, $\psi = \sum_{n=0}^N c_n e_n$ for some $N \in \nnum$. We can assume $N \geq N_0$ without loss of generality.  Then,
    \begin{align}
        \norm{J(0)\psi}^2 & \leq \sum_{n=0}^N |c_n|^2 (a_n^2+a_{n-1}^2 ) \\
        & = \sum_{n=0}^{N_0} |c_n|^2 (a_n^2+a_{n-1}^2 ) + \sum_{n=N_0+1}^N |c_n|^2 (a_n^2+a_{n-1}^2 ) \\
        & \leq b^2 \sum_{n=0}^N |c_n|^2 + 2 C \sum_{n = 0}^N f_n^2 |c_n|^2 \\
        & = b^2 \norm{\psi}^2 + 2 C \norm{F \psi}^2 \, \\
        & < b^2 \norm{\psi}^2 + \epsilon^2 \norm{F \psi}^2. \,
    \end{align}
    As $\ell_0(\nnum)$ is a core for both $J(0)$ and $F$ (\cref{prop:diagonal-term}), this inequality proves $\domain(J(0))\supset\domain(F)$ and infinitesimal $F$-boundedness of $J(0)$.

    We turn to relative compactness.
    Let $k \in \nnum$ such that $k(\beta - \alpha) > 1$.
    We show that the operator $T=J(0) (F^k+\iu)^{-1}$ is compact, i.e. $J(0)$ is relatively compact with respect to $F^k$, see e.g. \cite[Theorem~VI.12]{reed-mmmp1-funkana-1980}. $T$ is bounded since $(F^k+\iu)^{-1}$ is well-defined and bounded (since $F^k$ is self-adjoint) and has values in $\domain(F^k)\subset\domain(F)\subset\domain(J(0))$. To show that it is compact, we construct the following approximating sequence of finite-rank operators: 
    \begin{equation}
        T_j = \sum_{i = 0}^j \braket{e_i, T \cdot}e_i \, .
    \end{equation}
    Given $\psi = \sum_{n = 0}^\infty c_n e_n$, we obtain
    \begin{align}
        T\psi - T_j \psi & = \sum_{n=0}^\infty (a_n e_{n+1} + a_{n-1}e_{n-1}) \frac{c_n}{f_n^k + \iu} - \sum_{n=0}^{j-1}(a_n e_{n+1} + a_{n-1}e_{n-1}) \frac{c_n}{f_n^k + \iu}\\
        & - a_{j-1}\frac{c_j}{f_j^k+\iu} e_{j-1} - a_j \frac{c_{j+1}}{f_{j+1}^k+\iu}e_j \\
        & = \sum_{n=j}^\infty (a_n e_{n+1} + a_{n-1}e_{n-1}) \frac{c_n}{f_n^k + \iu} - a_{j-1}\frac{c_j}{f_j^k+\iu} e_{j-1} - a_j \frac{c_{j+1}}{f_{j+1}^k+\iu}e_j \, .
    \end{align}
    As $f_n = n^\beta(1+O(1/n))$ and $a_n = n^\alpha(1+O(1/n))$, there exists a constant $C \in \rnum$ such that 
    \begin{equation}
        \left|\frac{a_n}{f_n^k+\iu}\right| \leq \frac{C}{n^{k(\beta-\alpha)}} \, .
    \end{equation}
    Hence, 
    \begin{align}
        \norm{T \psi-T_j \psi}^2 & \leq \sum_{n = j}^\infty \frac{a_n^2+a_{n-1}^2}{|f_n^k+\iu|^2}|c_n|^2 + \frac{a_{j-1}^2}{|f_j^k+\iu|^2}|c_j|^2 + \frac{a_j^2}{|f_{j+1}^k+\iu|^2}|c_{j+1}|^2 \\
        & \leq \sum_{n = j}^\infty \frac{2 C^2}{n^{2k(\beta-\alpha)}}|c_n|^2 + \frac{2C^2}{j^{2k(\beta-\alpha)}}|c_j|^2\,\\
        & \leq \left( \sum_{n = j}^\infty \frac{2 C^2}{n^{2k(\beta-\alpha)}} + \frac{2C^2}{j^{2k(\beta-\alpha)}}\right)\|\psi\|^2,\,
    \end{align}
    where we used $|c_n|^2\leq\|\psi\|^2=\sum_{n'}|c_{n'}|^2$ for all $n\in\nnum$. As $k(\beta-\alpha)>1$, we also have $2k(\beta-\alpha)>1$, therefore the series converges and vanishes as $j\to\infty$, whence
    \begin{equation}
        \lim_{j \to \infty} \norm{T-T_j}^2 = 0\,,
    \end{equation}
    which proves compactness and finishes the proof.
\end{proof}
We will hereafter denote by $P^\lambda(z),Q^\lambda(z)$ the orthogonal polynomials of first and second kind associated with the Jacobi operator $J(\lambda)$, see \cref{def:jacobi}.
\begin{proposition}
    \label{prop:spectrum-jlambda}
    For every $\lambda>0$, $J(\lambda)$ is an unbounded self-adjoint operator with domain $\domain(J(\lambda)) = \domain(F)$, is bounded from below, and has purely discrete spectrum consisting of eigenvalues with multiplicity one.
    Besides, given any eigenvalue $E$ of $J(\lambda)$, the corresponding eigenvector is given by the orthogonal polynomial of first kind $P^\lambda(E)$.
\end{proposition}

\begin{proof}
    By \cref{lem:relatively-compact}, $J(0)$ is infinitesimally relatively bounded with respect to $F$, and therefore with respect to $\lambda F$ for every $\lambda$. As $\lambda F$ is self-adjoint and nonnegative (\cref{prop:diagonal-term}), the Kato--Rellich theorem \cite[Theorem~X.12]{reed-mmmp2-fourier-1975} implies that $J(\lambda)$ is self-adjoint on $\domain(J(\lambda)) = \domain(\lambda F) = \domain(F)$ and bounded from below.
    
    Furthermore, again by \cref{lem:relatively-compact}, $J(0)$ is relatively compact with respect to $F^k$ for some $k \in \nnum$, and hence with respect to $(\lambda F)^k$.
    Using Weyl's essential spectrum theorem \cite[Theorem~XIII.14, Corollary 3]{reed-mmmp4-operators-2005}, this implies that the essential spectra of $J(\lambda)$ and $\lambda F$ are equal; by \cref{prop:diagonal-term},
    \begin{equation}
        \sigma_{ess}(J(\lambda)) = \sigma_{ess}(\lambda F) =  \sigma_{ess}( F)=\emptyset \, .
    \end{equation}
    Finally, let $E \in \sigma(J(\lambda))$. 
    Then, $E$ is an isolated eigenvalue of $J(\lambda)$, and the associated eigenvector $\psi^E$ satisfies
    \begin{equation}
        J(\lambda) \psi^E = \Jop(\lambda) \psi^E = E \psi^E \, , 
    \end{equation}
    which is the generalized eigenvalue equation, \cref{eq:generalized-eigenvalue}, for $\Jop(\lambda)$ with $z = E$ and boundary condition $a_0 \psi^E_1 + \lambda f_0 \psi^E_0=E \psi^E_0$.
    For a given boundary condition, there exists a unique solution (up to multiplicative factors) of the recurrence relation~\eqref{eq:generalized-eigenvalue}, which is precisely the orthogonal polynomial of first kind $P^\lambda(E)$.
\end{proof}
\begin{remark}
    The fact that all eigenvalues have multiplicity one could also be proven with a more abstract argument: $e_0$ is a cyclic vector for $J(\lambda)$ \cite[366]{schmudgen-unboundedselfadjointoperators-2012}, and hence the spectrum of $J(\lambda)$ is simple, cf.~\cite[Corollary~5.19]{schmudgen-unboundedselfadjointoperators-2012}.
\end{remark}

\begin{remark}\label{rem:lambda_can_be_complex}
    The definition $J(\lambda)=J(0)+\lambda F$ extends naturally to complex values of the parameter $\lambda$. In this case, $\lambda F$ is understood as a (generally non-symmetric) operator obtained by scalar multiplication of $F$ by $\lambda \in \mathbb{C}$. 
    Arguing as in the proof of \cref{prop:spectrum-jlambda}, and applying the Kato--Rellich theorem to the perturbation $\lambda F$, one finds that for every $\lambda \in \mathbb{C}\setminus\{0\}$ the operator $J(\lambda)$ is closed and satisfies $\domain(J(\lambda)) = \domain(F)$.
    In particular, the family $(J(\lambda))_{\lambda\in\mathbb{C}\setminus\{0\}}$ is well-defined as a family of closed operators with parameter-independent domain.
\end{remark}
Having established that the spectrum of $J(\lambda)$ is bounded from below and consists of simple eigenvalues, we proceed with determining their properties.
\begin{definition}
    \label{def:eigenvalues-ej}
    For $\lambda> 0$, we denote by $E^{(j)}(\lambda)$ the $j$th eigenvalue of $J(\lambda)$, starting from the bottom of the spectrum.
\end{definition}
We begin by showing that the eigenvalues $E_j(\lambda)$ are analytic in $\lambda$ and do not cross (\cref{cor:eigen-analytic}), and are monotonically decreasing as $\lambda\to0$ (\cref{lem:eigenvalues_monotonous}). To prove these results, we will use analytic perturbation theory; see, for instance, \cite[Chapter~XII.2]{reed-mmmp4-operators-2005} for an overview.
\begin{definition}
    Let $R \subset \cnum$ be a connected domain in the complex plane, and $(T(r))_{r\in R}$ a family of closed operators with non-empty resolvent set.
    $(T(r))_{r \in R}$ is 
    \begin{enumerate}[(i)]
        \item an \textit{analytic family in the sense of Kato} if, for all $r_0 \in R$, there exists $z_0 \in \rho(T(r_0))$ and a neighborhood of $r_0$ such that $z_0 \in \rho(T(r))$ for all $r$ near $r_0$, and that $r\mapsto(T(r)-z_0)^{-1}$ is an analytic operator-valued function of $r$ near $r_0$;
        \item an \textit{analytic family of type (A)} if $\domain(T(r)) = \domain$ is independent of $r \in R$, and if, for all $\psi \in \domain$, $r\mapsto T(r)\psi$ is an analytic vector-valued function of $r$.
    \end{enumerate}
\end{definition}
If $(T(r))_{r\in R}$ is an analytic family of type (A), it is analytic in the sense of Kato~\cite[16]{reed-mmmp4-operators-2005}.
The following lemma shows that $(J(\lambda))_{\lambda > 0}$ can be embedded in such a family. To this end, note (cf.~\cref{rem:lambda_can_be_complex}) that $J(\lambda)$ admits a natural extension to complex nonzero values of $\lambda$, and is a closed operator with $\domain(J(\lambda))=\domain(F)$.
\begin{lemma}
    \label{lem:analytic-lambda}
    There exists an open set $R \supset \rnum_+$ such that $(J(\lambda))_{\lambda \in R}$ is an analytic family of type A.
    In particular, it is an analytic family in the sense of Kato.
\end{lemma}
\begin{proof}
   Fix $\lambda_0>0$. For $\lambda$ in a (complex) neighborhood of $\lambda_0$ we can write
    \begin{equation}
        J(\lambda)=J(\lambda_0)+(\lambda-\lambda_0)F
        \quad \text{on } \domain(F).
    \end{equation}
    As $\domain(F) = \domain(J(\lambda_0))$, $F$ is $J(\lambda_0)$-bounded~\cite[Lemma~6.2]{teschl-mathematicalmethodsquantum-2009}, and it follows from standard perturbation results (see, e.g., \cite[Section~XII.2]{reed-mmmp4-operators-2005}) that there exists a neighborhood $O_{\lambda_0}\subset\cnum$ of $\lambda_0$ such that $(J(\lambda))_{\lambda\in O_{\lambda_0}}$ is an analytic family of type (A).
    Since this holds for every $\lambda_0>0$, the union
    \begin{equation}
        R=\bigcup_{\lambda_0>0} O_{\lambda_0}
    \end{equation}
    is an open set containing $\rnum_+$, and $(J(\lambda))_{\lambda\in R}$ is an analytic family of type (A). By \cite[Theorem~XII.9]{reed-mmmp4-operators-2005}, it is also an analytic family in the sense of Kato.
\end{proof}

\begin{corollary}
    \label{cor:eigen-analytic}
    The functions $\mathbb{R}_+\ni\lambda\mapsto E^{(j)}(\lambda)\in\mathbb{R}$ are real-analytic and do not cross, that is, $E^{(j)}(\lambda)=E^{(k)}(\lambda)$ for some $\lambda>0$ implies $j=k$. Moreover, the associated eigenvectors can be chosen to depend analytically on $\lambda$.
\end{corollary}

\begin{proof}
    By \cref{prop:spectrum-jlambda}, all eigenvalues $E^{(j)}(\lambda)$ are simple for every $\lambda>0$. Since $(J(\lambda))_{\lambda\in R}$ is an analytic family in the sense of Kato (\cref{lem:analytic-lambda}), it follows from \cite[Theorem~XII.8]{reed-mmmp4-operators-2005} that the eigenvalues and the corresponding spectral projections depend analytically on $\lambda$. In particular, the eigenvalues are real-analytic on $\mathbb{R}_+$. Since the eigenvalues are simple for every $\lambda>0$, analytic perturbation theory implies that the analytic eigenvalue branches cannot intersect; hence crossings are absent. Finally, analyticity of the spectral projections implies that the corresponding eigenvectors can be chosen to depend analytically on $\lambda$.
\end{proof}
Additionally, we obtain that the eigenvalues are monotonically decreasing as $\lambda \to 0$:
\begin{lemma}\label{lem:eigenvalues_monotonous}
    For all $\lambda > 0$, 
    \begin{equation}\label{eq:hf}
        \diff{E^{(j)}(\lambda)}{\lambda} = \braket{\psi_j(\lambda), F \psi_j(\lambda)} \geq 0 \, , 
    \end{equation}
    where $\psi_j(\lambda)$ is the eigenvector associated with $E^{(j)}(\lambda)$.
\end{lemma}
\begin{proof}
    By \cref{lem:analytic-lambda}, both the eigenvalues $E^{(j)}(\lambda)$ and the eigenvectors $\psi_j(\lambda)$ are analytic in $\lambda$.
    Furthermore, for all $\lambda > 0$, $\domain(J(\lambda)) = \domain(F)$, and 
    \begin{equation}
        \diff{J(\lambda)}{\lambda} = F 
    \end{equation}
    in the strong sense on  $\domain(F)$.
    Hence, we can apply the Hellmann--Feynman theorem \cite{esteve-generalizationhellmannfeynman-2010} to get the equality in~\eqref{eq:hf}. As $F$ is nonnegative, the claim follows.
\end{proof}

Let us quickly summarize the findings of this section.
For $\lambda = 0$, the self-adjoint extensions $J_t$ of the limit circle operator $\Jmin(0)$ have purely discrete spectra consisting of eigenvalues with multiplicity one, which are distinct for different $t$, cf.~\cref{prop:jacobi-limit-circle-spectrum}.
For $\lambda > 0$, the spectrum of each operator $J(\lambda)$ is also purely discrete, and its eigenvalues $E^{(j)}(\lambda)$ defined as in \cref{def:eigenvalues-ej} satisfy the following:
\begin{enumerate}[(i)]
    \item $E^{(j)}(\lambda)$ is of multiplicity $1$, cf.~\cref{prop:spectrum-jlambda};
    \item $E^{(j)}(\lambda)$ is analytic in $\lambda$, and for different $j$, the $E^{(j)}(\lambda)$ do not cross, cf.~\cref{cor:eigen-analytic}; 
    \item For all $j \in \nnum$, $E^{(j)}(\lambda)$ is monotonically decreasing as $\lambda\downarrow0$, cf.~\cref{lem:eigenvalues_monotonous}.
\end{enumerate}
In \cref{prop:limit-eigenvalues}, we will additionally prove that for all $j \in \nnum$ $\lim_{\lambda \to 0}E^{(j)}(\lambda) = -\infty$.
These properties will be essential in the proof of \cref{thm:main-result}.

\subsection{Proof of \texorpdfstring{\cref{thm:main-result}}{the main theorem}}
\label{sec:main-theorem-proof}

In this section, we prove \cref{thm:main-result}. Throughout the section, we assume \cref{hyp:conditions-an-fn}. To this end, we require bounds on solutions of the generalized eigenvalue equation, \cref{eq:generalized-eigenvalue}, that are uniform in $\lambda \geq 0$.
For $\Jop(\lambda)$, \cref{eq:generalized-eigenvalue} reads
\begin{equation}
    \label{eq:recall-generalized-eigenvalue-lambda}
    a_n u_{n+1} + (\lambda f_n -z)u_n + a_{n-1} u_{n-1}= 0 \quad \forall n \geq 1 \,.
\end{equation}
A crucial ingredient is \cref{thm:main-asymptotics}, a technical result concerning square-summable solutions of \cref{eq:recall-generalized-eigenvalue-lambda}, which is restated here. Due to its technical nature (see \cref{rem:asymptotics}), the proof is postponed to \cref{sec:asymptotics-of-eigenvectors}.
\begin{theorem}
    \label{prop:asymptotics}
    Let $\Omega \subset \cnum$ be a compact set and, for every $z \in \Omega$, let $u^{\lambda,z}$ be a square-summable solution of \cref{eq:recall-generalized-eigenvalue-lambda}.
    Then, there exist $\Lambda > 0$ and $C > 0$ such that, for all $0 \leq \lambda < \Lambda$ and $z \in \Omega$,
    \begin{equation}
        |u^{\lambda,z}_n| \leq C \frac{|u^{\lambda,z}_0|+|u^{\lambda,z}_1|}{n^{\alpha/2-1/6}} \quad \forall n \geq 1\, .
    \end{equation}
\end{theorem}
\begin{remark}
    \label{rem:asymptotics}
   For $\lambda = 0$, the result follows directly from \cite[Theorem~2]{swiderski-spectralpropertiesblock-2018}; see also \cref{thm:swiderski}. For $\lambda>0$, the proof is more involved and is deferred to \cref{sec:asymptotics-of-eigenvectors}. The origin and necessity of the term $\nicefrac{1}{6}$ are discussed in \cref{sec:asymptotics-discussion}, where we show that, within our proof strategy, the assumption $\alpha \geq \nicefrac{4}{3}$ is in fact necessary.
\end{remark}
Equipped with \cref{prop:asymptotics} and the results of the previous sections, we now turn to the proof of \cref{thm:main-result}. 
Recall that the operator $J(\lambda)$ is self-adjoint for $\lambda>0$, while $J(0)$ is not self-adjoint and admits a family $(J_t)_{t\in\bcnum}$ of maximally accretive or dissipative extensions, with $J_t$ being self-adjoint if and only if $t\in\brnum$. We adopt the following notation:
\begin{itemize}
    \item  For $\lambda>0$, we denote by $M(z,\lambda)$ and $G_{nm}(z,\lambda)$ the Weyl $m$-function and the Green function associated with the operator $J(\lambda)$ (see \cref{def:green}), that is,
    \begin{equation}\label{green_lambda}
        G_{nm}(z,\lambda) = \braket{e_n,(J(\lambda)-z)^{-1}e_m} \, ,\qquad M(z,\lambda)=G_{00}(z,\lambda).
    \end{equation}
    \item For $t\in\bcnum$, we denote by $m(z,t)$ and $g_{nm}(z,t)$ the Weyl $m$-function and the Green function associated with the operator $J_t$.
    We use lowercase letters to avoid confusion with $M(z,\lambda)$ and $G_{nm}(z,\lambda)$, that is,
    \begin{equation}\label{green_t}
        g_{nm}(z,t) = \braket{e_n,(J_t-z)^{-1}e_m} \, ,\qquad m(z,t)=g_{00}(z,t).
    \end{equation}
    \item For $\lambda\ge 0$, we denote by $P^\lambda(z)$ and $Q^\lambda(z)$ the orthogonal polynomials (see \cref{def:orthogonal-polynomials}) associated with $\Jop(\lambda)$, that is, the solutions of \cref{eq:recall-generalized-eigenvalue-lambda} with initial conditions
    \begin{alignat}{2}\label{eq:plambda-qlambda}
        P_0^{\lambda}(z) & = 1\, , & \quad  P_1^{\lambda}(z)& = (z-\lambda f_0)/a_0\, ,  \\
        Q_0^{\lambda}(z) & = 0 & \quad  Q_1^{\lambda}(z) & = 1/a_0 \, .
    \end{alignat} 
\end{itemize}
The following two results will be used to prove \cref{thm:main-result}~\ref{item:main-result-some-extension}.
\begin{lemma}
    \label{lem:pointwise-convergence}
    Let $(\lambda_j)_{j \in \nnum}\subset\rnum_+$ with $\lim_{j \to \infty} \lambda_j = 0$, and $(z_j)_{j \in \nnum}\subset\cnum$ with $\lim_{j \to \infty} z_j = z \in \cnum$.
    Then, for every $n \in \nnum$,
    \begin{align}
        \lim_{j \to \infty} P^{\lambda_j}_n(z_j) & = P^0_n(z) \,, \\
        \lim_{j \to \infty} Q^{\lambda_j}_n(z_j) & = Q^0_n(z) \,.
    \end{align}
\end{lemma}
\begin{proof}
   The orthogonal polynomials are, by definition, polynomial in $z$, and hence continuous in $z$. We now show by induction that they are also polynomial in $\lambda$. 
By definition, $P^\lambda_0$, $P^\lambda_1$, $Q^\lambda_0$, and $Q^\lambda_1$ are polynomial in $(z,\lambda)$, see \cref{eq:plambda-qlambda}.

Now let $n \geq 1$ and assume that, for all $k \leq n$, the functions $P^\lambda_k(z)$ and $Q^\lambda_k(z)$ are polynomial in $(z,\lambda)$. Using \cref{eq:recall-generalized-eigenvalue-lambda}, we obtain
\begin{align}
    P_{n+1}^\lambda(z) & = \frac{z-\lambda f_n}{a_n} P_n^\lambda(z) - \frac{a_{n-1}}{a_n} P_{n-1}^\lambda(z)\, , \\
    Q_{n+1}^\lambda(z) & = \frac{z-\lambda f_n}{a_n} Q_n^\lambda(z) - \frac{a_{n-1}}{a_n} Q_{n-1}^\lambda(z)\, .
\end{align}
By the induction hypothesis, it follows that $P^\lambda_{n+1}(z)$ and $Q^\lambda_{n+1}(z)$ are again polynomial in $(z,\lambda)$, which completes the induction.
\end{proof}
\begin{proposition}
    \label{prop:strong-conv}
    Let $z \in \cnum$ with $\Imag z \neq 0$, and $(\lambda_j)_{j \in \nnum}\subset\rnum_+$ such that $\lim_{j \to \infty}\lambda_j = 0$.
    Suppose that the Weyl $m$-functions of $J(\lambda_j)$ evaluated at $z$ converge, i.e. $\lim_{j \to \infty}M(z,\lambda_j) = p \in \cnum$. 
    Then there exists $t\in \brnum$ such that $p = m(z,t)$, and $J(\lambda_j)$ converges to the self-adjoint operator $J_t$ in the strong resolvent sense.
\end{proposition}
\begin{proof}
    By \cref{lem:weyl-m-flt}, for each $t \in \bcnum$ the function
\begin{equation}
z' \mapsto -\frac{A(z',0) + C(z',0)t}{B(z',0) + D(z',0)t},
\end{equation}
with $z'\mapsto A(z',0),B(z',0),C(z',0),D(z',0)$ being entire functions determined by the orthogonal polynomials $P^0_n(z),Q^0_n(z)$, defines a meromorphic function of $z'$ which is well-defined at $z'$ if and only if $z' \in \rho(J_t)$, in which case it coincides with the Weyl $m$-function $m(z,t)$. Since the above expression is a fractional linear transform in $t$, there exists a unique $t \in \bcnum$ such that
\begin{equation}
-\frac{A(z,0) + C(z,0)t}{B(z,0) + D(z,0)t} = p.
\end{equation}
By the characterization of the resolvent set, this implies $z \in \rho(J_t)$, and hence $m(z,t)=p$.

    We claim that $J(\lambda_j)$ converges, in the strong resolvent sense, to the operator $J_t$, and that, in particular, $t\in\brnum$. We begin by showing that the Green functions of $J(\lambda_j)$, evaluated at $z$, converge to the one of $J_t$. By \cref{lem:green},
    \begin{align}
        g_{nm}(z,t) & = \begin{cases}
            (m(z,t)P^0_n(z)+Q^0_n(z)) P^0_m(z) & n \geq m \\
            (m(z,t)P^0_m(z)+Q^0_m(z)) P^0_n(z) & n \leq m 
        \end{cases}\, , \\\label{eq:gnmzl}
        G_{nm}(z,\lambda) & = \begin{cases}
            (M(z,\lambda)P^\lambda_n(z)+Q^\lambda_n(z)) P^\lambda_m(z) & n \geq m \\
            (M(z,\lambda)P^\lambda_m(z)+Q^\lambda_m(z)) P^\lambda_n(z) & n \leq m 
        \end{cases}\, ,
    \end{align}
    whence, by \cref{lem:pointwise-convergence} and the assumption $\lim_{j \to \infty}M(z,\lambda_j) = p \in \cnum$,
    \begin{equation}
        \label{proofeq:pointwise}
        \lim_{j \to \infty} G_{nm}(z,\lambda_j) = g_{nm}(z,t) \quad \forall n,m \in \nnum \, .
    \end{equation}
    Inferring strong resolvent convergence from pointwise convergence of the Green functions is highly nontrivial, as we do not a priori know that the limiting operator $J_t$ is self-adjoint, cf. \cref{rem:bootstrap}.
    We will use \cref{proofeq:pointwise}, together with \cref{prop:asymptotics}, to show that, for all $\psi\in\ell^2(\nnum)$,
    \begin{equation}
        \lim_{j\to\infty}(J(\lambda_j)-z)^{-1}\psi=(J_t-z)^{-1}\psi\quad\text{and}\quad\lim_{j\to\infty}(J(\lambda_j)-z^*)^{-1}\psi=(J_{t^*}-z^*)^{-1}\psi.
    \end{equation}
    To this end, it suffices to prove that, for every $m\in\nnum$,
    \begin{equation}\label{eq:limits}
        \lim_{j\to\infty}(J(\lambda_j)-z)^{-1}e_m=(J_t-z)^{-1}e_m\quad\text{and}\quad\lim_{j\to\infty}(J(\lambda_j)-z^*)^{-1}e_m=(J_{t^*}-z^*)^{-1}e_m,
    \end{equation}
    i.e. $(J(\lambda_j)-z)^{-1}\to(J_t-z)^{-1}$ and $(J(\lambda_j)-z^*)^{-1}\to(J_t-z^*)^{-1}$ strongly on the dense subspace $\ell_0(\nnum)$. Indeed, since all operators $J(\lambda_j)$ are self-adjoint, their resolvents computed at $z$ are uniformly bounded:    \begin{equation}\label{eq:resolvents_bounded}
        \norm{(J(\lambda_j)-z)^{-1}}\leq\frac{1}{\Imag z},
    \end{equation}
    whence strong convergence on the dense subspace $\ell_0(\nnum)$ is sufficient for strong convergence on the whole $\ell^2(\nnum)$.

    To prove \cref{eq:limits}, note that
    \begin{align}
        \norm{(J(\lambda_j)-z)^{-1}e_m - (J_t-z)^{-1}e_m}^2 &=\sum_{n=0}^\infty\left|\Braket{e_n,(J(\lambda_j)-z)^{-1}e_m - (J_t-z)^{-1}e_m}\right|^2\nonumber\\
        & = \sum_{n = 0}^\infty |G_{nm}(z,\lambda_j) - g_{nm}(z,t)|^2.
    \end{align}
    By \cref{proofeq:pointwise}, each term inside the sum converges to zero as $\lambda_j\to0$.
    In order to infer $(J(\lambda_j)-z)^{-1}e_m\to (J_t-z)^{-1}e_m$, we will apply Lebesgue's dominated convergence theorem \cite[Theorem~1.16]{reed-mmmp1-funkana-1980}.
    To this purpose, we need to show that the quantity $|G_{nm}(z,\lambda_j) - g_{nm}(z,t)|^2$ is bounded by some summable sequence that does not depend on $j$. As
    \begin{equation}
        |G_{nm}(z,\lambda_j) - g_{nm}(z,t)|^2\leq2|G_{nm}(z,\lambda_j)|^2+2|g_{nm}(z,t)|^2,
    \end{equation}
    the second term being independent of $j$ and summable, it suffices to find such a $j$-independent bound for $|G_{nm}(z,\lambda_j)|^2$.
    
    To this end, fix $m \in \nnum$.
    Then, by \cref{eq:gnmzl}, for all $n > m$ we have
    \begin{align}
        \label{eq:gnm-uform}
        G_{nm}(z,\lambda_j) &=(M(z,\lambda_j)P^{\lambda_j}_n(z)+Q^{\lambda_j}_n(z)) P^{\lambda_j}_m(z)\nonumber\\&\equiv P^{\lambda_j}_m(z) u^{\lambda_j,z}_n \, , 
    \end{align}
    where we set $u^{\lambda_j,z} := M(z,\lambda_j) P^{\lambda_j}(z) + Q^{\lambda_j}(z)$. This is a square-summable solution of \cref{eq:recall-generalized-eigenvalue-lambda}, cf.~\cref{lem:weyl-solution}; therefore, by \cref{prop:asymptotics}, there exist $\Lambda > 0$ and $C_2>0$ such that, for all $\lambda_j < \Lambda$ and $n \geq 1$,
    \begin{equation}\label{eq:ineq1}
        |u^{\lambda_j,z}_n| \leq \frac{C_2}{n^{\alpha/2-1/6}}  \left(|u^{\lambda_j,z}_0|+|u^{\lambda_j,z}_1|\right). 
    \end{equation}
    As $P^{\lambda}_0(z)=1$ (\cref{eq:plambda-qlambda}), recalling \cref{eq:gnm-uform} and the definition of $G_{nm}(z,\lambda)$ (cf.~\cref{green_lambda}), we have
    \begin{align}
        u^{\lambda_j,z}_0 & = G_{00}(z,\lambda_j) =\braket{e_0,(J(\lambda_j)-z)^{-1}e_0} \\
        u^{\lambda_j,z}_1 & = G_{10}(z,\lambda_j) =\braket{e_1,(J(\lambda_j)-z)^{-1}e_0}
    \end{align}
    and therefore, by \cref{eq:resolvents_bounded},
    \begin{align}
    \label{eq:ineq2}
        \bigl|u^{\lambda_j,z}_0\bigr|& =\left|\braket{e_0,(J(\lambda_j)-z)^{-1}e_0}\right|\leq \norm{(J(\lambda_j)-z)^{-1}}\leq\frac{1}{\Imag z}\, , \\
        \bigl|u^{\lambda_j,z}_1\bigr|& =\left|\braket{e_1,(J(\lambda_j)-z)^{-1}e_0}\right|\leq \norm{(J(\lambda_j)-z)^{-1}}\leq\frac{1}{\Imag z} \, .
    \end{align} 
    Moreover, as $P^{\lambda_j}_n(z)$ is continuous in $\lambda_j$ and $\lim_{j \to \infty} \lambda_j = 0$, there exists $C_1 > 0$ such that, for all $j \in \nnum$,
    \begin{equation}
    |P^{\lambda_j}_m(z)| \leq C_1.
    \end{equation}    
    By \cref{eq:ineq1,eq:ineq2}, setting $C = \frac{2C_1C_2}{\Imag z}$, we therefore obtain
    \begin{equation}
        |G_{nm}(z,\lambda_j)| \leq \frac{C}{n^{\alpha/2-1/6}}
    \end{equation}
    for all $n > m$ and $\lambda < \Lambda$.

    Let us now consider $n \leq m$. Due to \cref{proofeq:pointwise}, there exists $C_3> 0$ such that
    \begin{equation}
        |G_{nm}(z,\lambda_j)| \leq C_3 \quad \forall n \leq m\, .
    \end{equation}
    Setting
    \begin{equation}
        x_n = \begin{cases}
            C_3 & n \leq m \\
            \frac{C}{n^{\alpha/2-1/6}} & n > m
        \end{cases}
    \end{equation}
    we thus obtain the following upper bound:
    \begin{equation}
        |G_{nm}(z,\lambda_j)| \leq x_n 
    \end{equation}
    for all $n \in \nnum$, which is uniform in $\lambda_j < \Lambda$.
    Furthermore, $(x_n)_{n \in \nnum} \in \ell^2(\nnum)$ since
    \begin{equation}
        \label{proofeq:strong-conv-dominating}
        \sum_{n = 0}^\infty x_n^2\leq \sum_{n = 0}^m C_3^2 + \sum_{n=m+1}^\infty \frac{C^2}{n^{\alpha-1/3}} < \infty \, ,
    \end{equation}
    as $\alpha-1/3 >1$ by \cref{hyp:conditions-an-fn}. Therefore, Lebesgue's dominated convergence theorem applies, and
    \begin{equation}
       \lim_{j \to \infty} \norm{(J(\lambda_j)-z)^{-1}e_m - (J_t-z)^{-1}e_m}^2 = \lim_{j \to \infty} \sum_{n = 0}^\infty |G_{nm}(z,\lambda_j) - g_{nm}(z,t)|^2 = 0 \,.
    \end{equation}

    We proved that $(J(\lambda_j)-z)^{-1}$ converges strongly to $(J_t-z)^{-1}$ on all vectors $(e_n)_{n\in\nnum}$, and therefore
    \begin{equation}\label{eq:sr1}
        \lim_{j \to \infty} (J(\lambda_j)-z)^{-1} \psi = (J_t-z)^{-1}\psi \quad \forall \psi \in \ell^2(\nnum) \,.
    \end{equation}   
    We have thus proven strong convergence of $(J(\lambda_j)-z)^{-1}$ to $(J_t - z)^{-1}\psi$ for fixed $\Imag z > 0$, i.e. in the upper half plane.
    As all $J(\lambda_j)$ are self-adjoint, we also obtain
    \begin{equation}
        \lim_{j \to \infty}M(z^\ast,\lambda_j) = \lim_{j \to \infty}M(z,\lambda_j)^\ast = p^\ast \,,
    \end{equation}
    and $p^\ast = m(z^\ast,t^\ast)$, cf.~\cref{lem:weyl-m-flt}.
    Repeating the above analysis for $z^\ast$ and $J_{t^\ast}$, we get
    \begin{equation}\label{eq:sr2}
        \lim_{j \to \infty} (J(\lambda_j)-z^\ast)^{-1} \psi = (J_{t^\ast}-z^\ast)^{-1} \psi \quad \forall \psi \in \ell^2(\nnum) \,.
    \end{equation}
    Hence, we have a sequence of resolvents of self-adjoint operators $J(\lambda_j)$ converging strongly at both $z$ and $z^*$. Moreover, the limits $(J_t-z)^{-1}$ and $(J_{t^*}-z^*)^{-1}$ are resolvents of closed operators and therefore have dense range. Thus, the assumptions of the Trotter--Kato theorem \cite[Theorem~VIII.22]{reed-mmmp1-funkana-1980} are satisfied.
    By the Trotter--Kato theorem \cite[Theorem~VIII.22]{reed-mmmp1-funkana-1980}, there exists a self-adjoint operator $A$ to which $J(\lambda_j)$ converges in the strong resolvent sense. In particular,
    \begin{equation}
        (J_t-z)^{-1}=(A-z)^{-1},\quad (J_{t^*}-z^\ast)^{-1}=(A-z^*)^{-1},
    \end{equation}
    implying $J_t=J_{t^*}=A$ and thus completing the proof.   
\end{proof}
    \begin{remark}
    \label{rem:bootstrap}
    \cref{prop:asymptotics} is essential to establish strong resolvent convergence, since the self-adjointness of the limiting operator $J_t$ is not known a priori. Indeed, the pointwise convergence of the Green functions $G_{nm}(z,\lambda_j) \to g_{nm}(z,t)$ yields convergence of the matrix elements of the resolvent, i.e. weak resolvent convergence. If one knew in advance that $J_t$ is self-adjoint (equivalently, that $t = t^*$), this would already suffice: for sequences of self-adjoint operators, weak and strong resolvent convergence coincide (see \cite[284]{reed-mmmp1-funkana-1980}).

    However, in our setting the self-adjointness of $J_t$ is \textit{itself} part of the conclusion. Thus, weak resolvent convergence alone is not sufficient, and one cannot upgrade it to strong resolvent convergence without additional input. This creates a bootstrapping problem: once strong resolvent convergence is established, the limit must be self-adjoint, but proving strong resolvent convergence appears to require this information beforehand. This issue is resolved by the explicit estimates provided in \cref{prop:asymptotics}, which allow us to prove strong resolvent convergence directly, without assuming self-adjointness of the limit.
\end{remark}

\begin{corollary}
    \label{cor:strongly-convergent-subsequence}
    Let $(\lambda_j)_{j \in \nnum}\subset\rnum_+$ such that $\lim_{j \to \infty}\lambda_j = 0$.
    Then there exists a subsequence $(\lambda_{j_k})_{k \in \nnum}$ and $t\in \brnum$ such that $J(\lambda_{j_k})$ converges to the self-adjoint operator $J_t$ in the strong resolvent sense.
\end{corollary}
\begin{proof}
    As already shown in the proof of \cref{prop:strong-conv}, $|M(z,\lambda_j)| \leq \frac{1}{\Imag z}$, and hence the set $(M(z,\lambda_j))_{j \in \nnum}$ is compact.
    Thus, there exists a convergent subsequence $(M(z,\lambda_{j_k}))_{k \in \nnum}$.
    The statement follows by applying \cref{prop:strong-conv} to the sequence $(\lambda_{j_k})_{k \in \nnum}$.
\end{proof}

Additionally, we show that, if $J(\lambda_j)\to J_t$ in the strong resolvent sense, a real number is in the spectrum of $J_t$ if and only if it is the limit of a sequence of eigenvalues of the approximating operators. While one implication is always true for general self-adjoint operators, the converse will again hold as a direct consequence of \cref{prop:asymptotics}.
\begin{proposition}
    \label{prop:eigenvalue-convergence}
    Let $(\lambda_j)_{j \in \nnum}\subset\rnum_+$ such that $\lim_{j \to \infty}\lambda_j = 0$.
    Suppose further that there exists $t \in \brnum$ such that $J(\lambda_j) \to J_t$ in the strong resolvent sense.
    Then $E \in \sigma(J_t)$ if and only if there exists a sequence $(E_j)_{j\in\nnum} \subset \sigma(J(\lambda_j))$ with $\lim_{j \to \infty}E_j = E$.
\end{proposition}
\begin{proof}
    Let $E \in \sigma(J_t)$.
    Then, by strong resolvent convergence and \cite[Theorem~VIII.24]{reed-mmmp1-funkana-1980}, there exists a sequence $(E_j)_{j\in\nnum} \subset \sigma(J(\lambda_j))$ such that $\lim_{j \to \infty}E_j = E$. 
        
    Conversely, let $(E_j)_{j\in\nnum} \subset \sigma(J(\lambda_j))$ and assume there exists $E \in \rnum$ such that $\lim_{j \to \infty}E_j = E$.
    We claim $E \in \sigma(J_t)$.
    Given $E_j \in \sigma(J(\lambda_j))$, the associated eigenvector is given by $P^{\lambda_j}(E_j)$, cf.~\cref{prop:spectrum-jlambda}, and by \cref{lem:pointwise-convergence} and the assumptions, 
    \begin{equation}
        \label{proofeq:eigenvector-convergence-pointwise}
        \lim_{j \to \infty}P^{\lambda_j}_n(E_j) = P^0_n(E) \quad \forall n \in \nnum \,,
    \end{equation}
    that is, $P^{\lambda_j}(E_j)\to P^0(E)$ pointwise in $n$. We claim $P^{\lambda_j}(E_j)\to P^0(E)$ in the $\ell^2$ sense. To this end, note that the set $\Omega = (E_j)_{j \in \nnum}$ is compact as $\lim_{j \to \infty} E_j = E$.
    Each $P^{\lambda_j}(E_j)$ is a square-summable generalized eigenvector of $\Jop(\lambda_j)$, hence by \cref{prop:asymptotics} there exist $\Lambda > 0$ and $C_1>0$ such that, for all $\lambda_j < \Lambda$ and $n \in \nnum$,
    \begin{equation}
        \label{proofeq:spectral-conv-eigenvector-bound}
        |P_n^{\lambda_j}(E_j)| \leq \frac{C_1}{n^{\alpha/2-1/6}}\left(|P_0^{\lambda_j}(E_j)|+|P_1^{\lambda_j}(E_j)|\right) \, .
    \end{equation}
    By the definition of $P^\lambda(z)$ (cf.~\cref{eq:plambda-qlambda}), we obtain
    \begin{equation}
        P_0^{\lambda_j}(E_j) = 1 \, , \quad  P_1^{\lambda_j}(E_j) = \frac{E_j-\lambda f_0}{a_0},
    \end{equation}
    and, as the sets $\Omega = (E_j)_{j \in \nnum}$ and $(\lambda_j)_{j \in \nnum}$ are compact, there exists $C_2>0$ such that
    \begin{equation}
        \left(|P_0^{\lambda_j}(E_j)|+|P_1^{\lambda_j}(E_j)|\right) \leq C_2 
    \end{equation}
    for all $j \in \nnum$.
    Setting $C = C_1 C_2$ and combining this with \cref{proofeq:spectral-conv-eigenvector-bound}, we get
    \begin{equation}
         |P_n^{\lambda_j}(E_j)| \leq \frac{C}{n^{\alpha/2-1/6}} \, \quad \forall n \geq 1 \, .
    \end{equation}
    Using \cref{hyp:conditions-an-fn}, we have
    \begin{equation}
        \|P^{\lambda_j}(E_j)\|^2\leq 1+ \sum_{n = 1}^\infty \frac{C^2}{n^{\alpha-1/3}} < \infty\, ,
    \end{equation}
    i.e. the $\ell^2$ norms of $P^{\lambda_j}(E_j)$ are \emph{uniformly} bounded in $j$. Furthermore, by \cref{prop:jacobi-esssa}, 
    \begin{equation}
        \sum_{n = 0}^\infty |P^{0}_n(E)|^2 < \infty\, .
    \end{equation}
    Thus, we get 
    \begin{align}
        \norm{P^{\lambda_j}(E_j)-P^0(E)}^2 & = \sum_{n = 0}^\infty |P^{\lambda_j}_n(E_j)-P^{0}_n(E)|^2 \\
        & \leq \sum_{n = 0}^\infty 2|P^{\lambda_j}_n(E_j)|^2 +\sum_{n = 0}^\infty 2|P^{0}_n(E)|^2\\
        & \leq  2+\sum_{n = 1}^\infty \frac{2C^2}{n^{\alpha-1/3}} +\sum_{n = 0}^\infty 2|P^{0}_n(E)|^2.
    \end{align}
    As the last term is independent of $j$, Lebesgue's dominated convergence theorem applies. Therefore, by \cref{proofeq:eigenvector-convergence-pointwise},
    \begin{equation}        \label{proofeq:eigenvector-strong-conv}
        \lim_{j \to \infty} P^{\lambda_j}(E_j)=P^0(E)
    \end{equation}
in $\ell^2(\nnum)$, as claimed.

    We now consider the spectral projections of the involved operators.
    For $a<b$, let $\mathcal{P}_{a,b}$ be the spectral projection of $J_t$ onto the interval $(a,b)$, and $\mathcal{P}^\lambda_{a,b}$ the corresponding spectral projection of $J(\lambda)$.
    Assume $E \notin \sigma(J_t)$.
    Since the spectrum is closed, there exists $\epsilon > 0$ such that $(E-\epsilon,E+\epsilon) \cap \sigma(J_t) = \emptyset$, and therefore $\mathcal{P}_{E-\epsilon,E+\epsilon}=0$.
    As $J(\lambda_j) \to J_t$ in the strong resolvent sense, the corresponding spectral projections converge strongly \cite[Theorem~VIII.24]{reed-mmmp1-funkana-1980}, that is,
    \begin{equation}
        \lim_{j \to \infty}\mathcal{P}^{\lambda_j}_{E-\epsilon,E+\epsilon} \psi = \mathcal{P}_{E-\epsilon,E+\epsilon} \psi \quad \forall \psi \in \hilbert \, .
    \end{equation}
    Together with  \cref{proofeq:eigenvector-strong-conv}, this implies
    \begin{equation}\label{eq:stronglimits1}
        \lim_{j \to \infty} \mathcal{P}^{\lambda_j}_{E-\epsilon,E+\epsilon}P^{\lambda_j}(E_j) = \mathcal{P}_{E-\epsilon,E+\epsilon}P^0(E)\, .
    \end{equation}
    On the other hand, $P^{\lambda_j}(E_j)$ is an eigenvector of $J(\lambda_j)$, implying that $\mathcal{P}^{\lambda_j}_{a,b}P^{\lambda_j}(E_j)$ is either $P^{\lambda_j}(E_j)$ if $E_j\in(a,b)$, or $0$ otherwise. Since $E_j\to E$, we have $E_j\in(E-\epsilon,E+\epsilon)$ for sufficiently large $j$. Therefore,
    \begin{equation}\label{eq:stronglimits2}
        \lim_{j \to \infty} \mathcal{P}^{\lambda_j}_{E-\epsilon,E+\epsilon}P^{\lambda_j}(E_j) = \lim_{j\to\infty}P^{\lambda_j}(E_j)=P^0(E)\, .
    \end{equation}
    By \cref{eq:stronglimits1,eq:stronglimits2}, $\mathcal{P}_{E-\epsilon,E+\epsilon}P^0(E) = P^0(E) \neq 0$. This shows $\mathcal{P}_{E-\epsilon,E+\epsilon}\neq0$, leading to a contradiction. Therefore, $E \in \sigma(J_t)$.
\end{proof}
Using this characterization of the spectrum, we can now prove that the eigenvalues of $J(\lambda)$ diverge to $-\infty$ as $\lambda \to 0$ (cf.~\cref{fig:eigenvalue-limit}).
\begin{proposition}
    \label{prop:limit-eigenvalues}
    For $\lambda >0$ and $j \in \nnum$, let $E^{(j)}(\lambda)$ be the $j$th lowest eigenvalue of $J(\lambda)$ defined as in \cref{def:eigenvalues-ej}.
    Then, $\lim_{\lambda \to 0}E^{(j)}(\lambda) = -\infty$.
\end{proposition}
\begin{proof}
    By \cref{lem:eigenvalues_monotonous,cor:eigen-analytic}, each $E^{(j)}(\lambda)$ is an analytic functions in $\lambda>0$, monotonically decreasing with $\lambda \to 0$.
    Hence, for all $j \in \nnum$,
    \begin{equation}
        \lim_{\lambda \to 0}E^{(j)}(\lambda) \in \rnum \cup \{-\infty\} \, , 
    \end{equation}
    i.e. either $E^{(j)}(\lambda)$ converges to a real number or diverges strictly to $-\infty$ as $\lambda \to 0$.
    
    Assume that, for some $j \in\nnum$, $\lim_{\lambda \to 0} E^{(j)}(\lambda) \neq -\infty$.
    Then, as eigenvalues do not cross (\cref{cor:eigen-analytic}), the same holds true for all $E^{(l)}(\lambda)$ with $l \geq j$.
    Hence, we can assume without loss of generality that $j$ is the smallest integer such that $\lim_{\lambda \to 0} E^{(j)}(\lambda)=:E^{(j)} \neq -\infty$.
    In particular,
    \begin{equation}
        \label{proofeq:el-ej-spectral-bound}
         E^{(l)}(\lambda) \geq E^{(j)} \quad  \forall l \geq j\, , \quad \forall \lambda > 0     
    \end{equation}
    because of the monotonicity of $E^{(l)}(\lambda)$ and the non-crossing property.
    We show that the existence of such a $j$ leads to a contradiction.

    By \cref{cor:strongly-convergent-subsequence}, there exists a sequence $(\lambda_k)_{k \in \nnum}\subset\rnum_+$ with $\lim_{k \to \infty} \lambda_k= 0$, and $t \in \brnum$, such that $J(\lambda_k)$ converges to $J_t$ in the strong resolvent sense.
    As $J_t$ is unbounded from below, cf.~\cref{prop:spectrum-jlambda}, there exists $E \in \sigma(J_t)$ such that
    \begin{equation}\label{proofeq:el-ej-spectral-bound-2}
        E \leq E^{(j)} -1.
    \end{equation}
    By \cref{prop:eigenvalue-convergence}, there exists a sequence $(E_k)_{k\in\nnum} \subset\rnum$, with $E_k\in \sigma(J(\lambda_k))$, such that $\lim_{k \to \infty} E_k = E$. As each $E_k$ is an eigenvalue of $J(\lambda_k)$, we have
    \begin{equation}
        E_k = E^{(j_k)}(\lambda_k) \quad \forall k \in \nnum
    \end{equation}
    for some sequence $(j_k)_{k\in\nnum} \in \nnum$.

    By definition, for $l \geq j$ the sequences $E^{(l)}(\lambda_k)$ are bounded from below by $E^{(j)} \geq E+1$, and hence cannot be used to obtain the limit $E$.
    On the other hand, there only exists finitely many sequences $E^{(l)}(\lambda_k)$ with $0 \leq l \leq j-1$, and each of them diverges to $-\infty$.
    Thus, the limit $E$ can also not be obtained from these sequences, which leads to a contradiction.
    The following steps formalize this claim:

    We treat the cases $j=0$ and $j\geq1$ separately.
    If $j = 0$, $j_k \geq j$ for all $k \in \nnum$, and by \cref{proofeq:el-ej-spectral-bound,proofeq:el-ej-spectral-bound-2} we get
    \begin{equation}
        E^{(j)} -1 \geq E = \lim_{k \to \infty}E_k = \lim_{k \to \infty}E^{(j_k)}(\lambda_k) \geq E^{(j)}\, , 
    \end{equation}
    which is a contradiction.
    
    If $j\geq1$, using \cref{proofeq:el-ej-spectral-bound} and again $\lim_{k\to \infty}E^{(j_k)}(\lambda_k) = E \leq E^{(j)} -1<E^{(j)}$, it follows that, for large enough $k$, $E^{(j_k)}(\lambda_k) < E^{(j)}$.
    In the following, let $K \in \nnum$ such that $E^{(j_k)}(\lambda_k) < E^{(j)}$ holds for all $k \geq K$.
    By monotonicity of $E^{(j)}(\lambda)$, it follows that $E^{(j)} \leq E^{(j)}(\lambda)$ for all $\lambda > 0$; together with the previous inequality, this yields
    \begin{equation}
        E^{(j_k)}(\lambda_k) < E^{(j)}(\lambda_k) \quad \forall k \geq K\, .
    \end{equation}
    For all $k \geq K$ we hence get $j_k < j$ by definition, and, as both values are integers, $j_k \leq j-1$ follows for all $k \geq K$.
    Again, by definition, this implies
    \begin{equation}
    E^{(j_k)}(\lambda_k) \leq E^{(j-1)}(\lambda_k) \quad \forall k \geq K\, ,    
    \end{equation}
    and hence
    \begin{equation}
        \lim_{k \to \infty}E_k = \lim_{k \to \infty}E^{(j_k)}(\lambda_k) \leq \lim_{k \to \infty}E^{(j-1)}(\lambda_k) = -\infty \, , 
    \end{equation}
    follows by definition of $j$.
    But we assumed $\lim_{k \to \infty}E_k = E \in \rnum$, and thus obtain a contradiction.

    Therefore, the only possibility to obtain $E$ as the limit of some sequence $E^{(j_k)}(\lambda_k)$ with $\lambda_k \to 0$ is if there exists no $j \in \nnum$ such that $\lim_{\lambda \to 0}E^{(j)}(\lambda) = E^{(j)} \in \rnum$, and we get the desired claim.
\end{proof}
\begin{figure}[ht]
    \centering
    \includegraphics[width=0.6\linewidth]{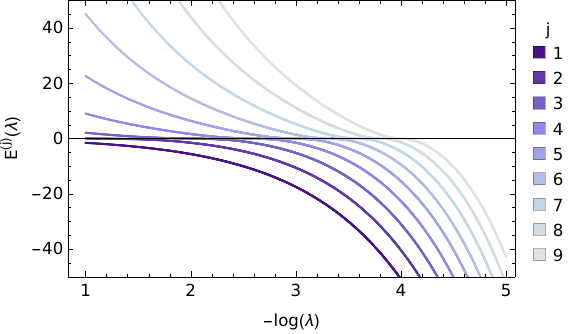}
    \caption{
    Logarithmic plot of the eigenvalues $E^{(j)}(\lambda)$ of the Jacobi operator $J(\lambda)$, defined as in \cref{def:jacobi-lambda}, with coefficients $a_n = 3^{-3/2}\sqrt{(3n+1,3)}$ and $f_n = n^3$, where $(\cdot,\cdot)$ denotes the Pochhammer symbol, cf. \cref{eq:def-pochhammer}.
    Note that this operator coincides with $A^{(0)}_{3,3}(\lambda)$ from \cref{lem:squeezing-decomposition}.
    Monotonicity, continuity and $\lim_{\lambda \to 0} E^{(j)}(\lambda) = -\infty$ are clearly visible.
    }
    \label{fig:eigenvalue-limit}
\end{figure}

Equipped with these results, we can finally prove the main result of this paper.
\begin{proof}[Proof of \cref{thm:main-result}]
    The first statement is \cref{cor:strongly-convergent-subsequence}.

    For the second statement, let $t\in\brnum$ and $E \in \sigma(J_t)$.
    For $l \in \nnum$ and $\lambda > 0$, let $E^{(l)}(\lambda)$ be the $l$th lowest eigenvalue of $J(\lambda)$ as in \cref{def:eigenvalues-ej}.
    By \cref{cor:eigen-analytic,prop:limit-eigenvalues}, $E^{(l)}(\lambda)$ is analytic in $\lambda>0$ and 
    \begin{equation}
    \label{proofeq:sequence-el-toneginf}
    \lim_{\lambda \to 0} E^{(l)}(\lambda) = -\infty\, .
    \end{equation}
    Moreover, for fixed $\lambda > 0$, $J(\lambda)$ is unbounded from above and hence there exists $l \in \nnum$ such that  $E^{(l)}(\lambda)$ is arbitrarily big.
    We proceed to construct a sequence $(\lambda_j)_{j\in\nnum}$ with $\lim_{j \to \infty} \lambda_j = 0$ such that $E \in \sigma(J(\lambda_j))$ for all $j \in \nnum$:
    To begin with, there exists $l_0 \in \nnum$ such that $E^{(l_0)}(1) > E$, and by \cref{proofeq:sequence-el-toneginf} there exists $0 < \lambda_0 < 1$ such that $E^{(l_0)}(\lambda_0) = E$.
    Next, there exists $l_1 \in \nnum$ such that $E^{(l_1)}(\lambda_0/2) > E$, and consequently there also exists $0 < \lambda_1 <\lambda_0/2$ such that 
    $E^{(l_1)}(\lambda_1) = E$.
    Iterating this process we construct a sequence $(\lambda_j)_{j \in \nnum}$ with $0 < \lambda_j \leq 2^{-j}$ and $E \in \sigma(J(\lambda_j))$ for all $j \in \nnum$.
    
    By \cref{thm:main-result}(i), there exists a subsequence $(\lambda_{j_k})_{k\in\nnum}$ and $\tilde{t} \in \brnum$ such that $J(\lambda_{j_k})$ converges to $J_{\tilde{t}}$ in the strong resolvent sense.
    In particular, $E\in \sigma(J(\lambda_{j_k}))$ for all $k \in \nnum$, and by \cref{prop:eigenvalue-convergence} $E \in \sigma(J_{\tilde{t}})$.
    By \cref{prop:jacobi-limit-circle-spectrum}, eigenvalues of different self-adjoint extensions are distinct, whence $\tilde{t} = t$. 
\end{proof}
The statement of the main theorem is visualized in \cref{fig:eigenvalue-limit-b} in \cref{sec:intro} (also cf.~\cref{rem:terminology}).
The Weyl $m$-functions $M(z,\lambda)$ uniquely represent the operator $J(\lambda)$.
As $\lambda \to 0$, the values of $M(z,\lambda)$ for fixed $z$ form a spiral, with the circle $m(z,t)_{t \in \brnum}$ of the limiting operator $J(0)$ being their accumulation set.
Visually, it is clear that, for all $t \in \brnum$, there exists a sequence $\lambda_j \to 0$ such that $M(z,\lambda_j) \to m(z,t)$.

\section{Asymptotics of generalized eigenvectors}
\label{sec:asymptotics-of-eigenvectors}

This section is devoted to the proof of \cref{prop:asymptotics} (cf. \cref{thm:main-asymptotics}), which we utilize in the proof of \cref{prop:strong-conv} to show that weak convergence of solutions of \cref{eq:recall-generalized-eigenvalue-lambda} as $\lambda\to0$ implies strong resolvent convergence. This entails proving the following bound for square-summable solutions $u^{\lambda,z}$ of \cref{eq:recall-generalized-eigenvalue-lambda}, with $z$ ranging in some compact set $\Omega$:
\begin{equation}\label{eq:thebound}
    |u^{\lambda,z}_n| \leq \frac{C}{n^{\alpha/2-1/6}}\left(|u^{\lambda,z}_0| +|u^{\lambda,z}_1|\right) \quad n \geq 1  \, , 
\end{equation}
where $C$ does not depend on $\lambda$ and $z$. \Cref{hyp:conditions-an-fn} will continue to be assumed throughout this section, and the compact set $\Omega$ will remain fixed.

\subsection{Heuristics and structure of the proof}\label{sec:asymptotics_overview}

We begin by providing an overview of the proof, focusing on the most technically demanding points.
Recall the generalized eigenvalue equation for $\Jop(\lambda)$, cf.~\cref{eq:generalized-eigenvalue}:
\begin{equation}
    \label{eq:generalized-eigenvalue-lambda}
    a_n u_{n+1} + (\lambda f_n -z)u_n +a_{n-1} u_{n-1}= 0 \quad \forall n \geq 1 \,,
\end{equation}
where the initial conditions $u_0,u_1$ remain arbitrary.
Dividing by $a_n$, we obtain
\begin{equation}
    u_{n+1} + \frac{\lambda f_n}{a_n}u_n + \frac{a_{n-1}}{a_n} u_{n-1} - \frac{z}{a_n} u_n = 0 \quad \forall n \geq 1 \, .
\end{equation}
For large enough $n$, we get $\frac{f_n}{a_n} \sim n^\delta$, where $\delta = \beta-\alpha$, $\frac{a_{n-1}}{a_n} \sim 1$, and $\frac{z}{a_n} \sim 0$.
Hence, the above equation is formally approximated by the following equation:
\begin{equation}\label{eq:difference_equation_asymptotic}
    u_{n+1} + \lambda n^\delta u_n + u_{n-1} = 0 \, .
\end{equation}
It will be useful to recast this equation in terms of a discrete variable $x$. The relevant scale of this problem is given by $x = nh$, where $h=\lambda^{1/\delta}$.
Setting $u_n = (-1)^n\psi(nh) = (-1)^n\psi(x)$ for some function $\psi$, \cref{eq:difference_equation_asymptotic} becomes
\begin{equation}
    \label{eq:difference-turning-point}
    \psi(x+h) -x^\delta \psi(x) + \psi(x-h) = 0 \, 
\end{equation}
for all $x = nh$ and $n \geq 1$.
Recalling the definition of the discrete Laplacian, we can identify \cref{eq:difference-turning-point} as the discrete analogue of the following differential equation:
\begin{equation}
    \label{eq:differential-turning-point}
    \Delta \psi(x) = (x^\delta-2)\psi(x) \, . 
\end{equation}
This is a second-order differential equation with a so-called \textit{turning point} at $x^\delta = 2$.
In the region $x^\delta<2$, all solutions of \cref{eq:differential-turning-point} are oscillating and hence bounded.
Instead, in the region $x^\delta>2$, the so-called \textit{recessive solution} of \cref{eq:differential-turning-point} decays exponentially, whereas all other solutions are \textit{dominant} and grow exponentially.
This is indeed reminiscent of the behavior of solutions of \cref{eq:generalized-eigenvalue-lambda}: from the Weyl alternative theorem (\cref{prop:weyl-alternative}) we know that one solution of \cref{eq:generalized-eigenvalue-lambda}, namely $u^{\lambda,z}$ is square-summable, whereas all others are not.

\begin{figure}[t]
    \centering
    \includegraphics[width=0.9\linewidth]{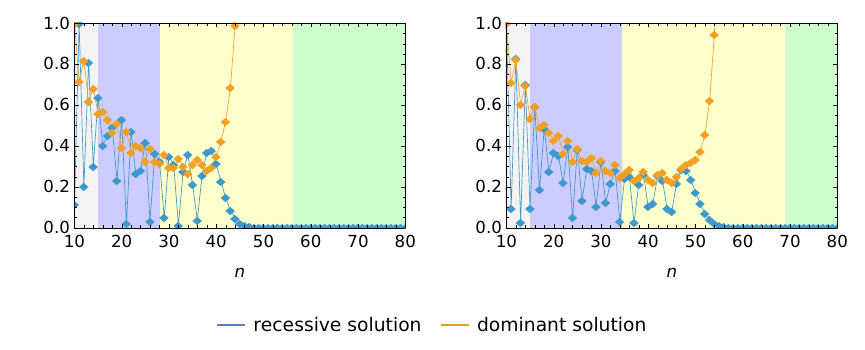}
    \caption{
        Numerical evaluation of the dominant and the recessive solution of \cref{eq:generalized-eigenvalue} for the family of Jacobi operators $J(\lambda)$, defined as in \cref{def:jacobi-lambda}, with coefficients $a_n = 4^{-1/2} \sqrt{(4n+1,4)}$ and $f_n = n^4$, where $(\cdot,\cdot)$ denotes the Pochhammer symbol, cf. \cref{eq:def-pochhammer}, for two different values of $\lambda$.
        The left plot corresponds to $\lambda = \nicefrac{1}{791}$, while the right plot corresponds to $\lambda=\nicefrac{1}{1191}$.
        The different background colors indicate the different regions described in the text and in \cref{def:asymptotics-regions}, i.e. to $n \leq N_0$ (gray), $n \in [N_0,N_2(\lambda)]$ (blue), $n \in [N_2(\lambda),N_3(\lambda)]$ (yellow), $n \geq N_3(\lambda)$ (green).
        In particular, the turning point moves clearly to higher values of $n$ as $\lambda$ gets smaller.
    }
    \label{fig:asymptotics-schematics}
\end{figure}
For $\delta = 1$, \cref{eq:differential-turning-point} is equivalent to the Airy equation, cf.~\cref{app:airy}, and admits two linearly independent solutions given by $\Ai(x+2)$ and $\Bi(x+2)$, where $\Ai$ and $\Bi$ are the Airy functions of first and second kind \cite[392]{olver-asymptoticsspecialfunctions-2010}.
For $\delta \neq 1$, solutions of \cref{eq:differential-turning-point} can be approximated by Airy functions via a suitable variable transformation $x\rightarrow\xi(x)$ corresponding to the so-called \textit{Langer transform}.
Standard references on rigorous error analysis for such equations are \cite[Chapter~11]{olver-asymptoticsspecialfunctions-2010} and \cite[Chapter~VIII]{wasow-asymptoticexpansionsordinary-2002}.

We expect analogous considerations to apply to the discrete case~\eqref{eq:difference-turning-point}. Fewer results are available in this setting. An explicit solution for $\delta =1$ is given in \cite{ehrhardt-solutionsdiscreteairy-2004}.
Wang and Wong \cite{wang-asymptoticexpansionssecondorder-2003,wang-lineardifferenceequations-} study the asymptotics of solutions of difference equations with a turning point; however, their work does not provide bounds on these solutions near the turning point. 
Geronimo et al.~\cite{geronimo-wkbturningpoint-2004,geronimo-wkbturningpoint-2009} and Fedotov and Klopp~\cite{fedotov-complexwkbmethod-2019} transfer results from the continuous case (with $x$ ranging in some compact set $K$) to the discrete case while keeping track of errors.\medskip

We will employ and adapt these results to our purposes. As we expect solutions of \cref{eq:difference-turning-point} to exhibit a turning point at $x^\delta=2$, our strategy is to divide the analysis into three overlapping regions $[0,x_2],[x_1,x_4],[x_3,\infty)$ (see \cref{fig:asymptotics-schematics}), with $x_2^\delta<2<x_4^\delta$, and diversify our bounding technique as follows:
\begin{enumerate}[(i)]
    \item For $x\in[0,x_2]$, we employ standard methods for difference equations, involving the use of \textit{Tur\'an determinants} (cf.~\cite{swiderski-spectralpropertiesblock-2018}), to bound all solutions of \cref{eq:difference-turning-point};
    \item For $x \in [x_1,x_4]$, we adapt the turning-point analysis by Geronimo et al. and Fedotov--Klopp to our case;
    \item For $x\in [x_3,\infty)$, we again use standard methods from the theory of difference equations, namely Poincaré's theorem, to obtain explicit asymptotics of the decaying solution, cf.~\cite[Theorem~8.35]{elaydi-introductiondifferenceequations-2005}.
\end{enumerate}
The overlap between the three regions will allow us to connect the bounds obtained in each regime.

We carry out the analysis both in terms of the rescaled variable $x$ and the original integer variable $n$. In particular, the boundary points $x_j$ correspond to suitable $\lambda$-dependent integers $N_j(\lambda)$.
The strategy of the proof is organized as follows:
\begin{itemize}
    \item In \cref{sec:asympotics-preliminaries}, we define the boundary points $N_j(\lambda)$ rigorously and establish the regularity properties required for the analysis, using the asymptotics of the coefficients $a_n$ and $f_n$. We also show that the regime of small $n$, where the asymptotic expansions do not accurately describe the coefficients, can be absorbed into a constant prefactor in the final bound; see \cref{lem:before-N0}.

    \item The oscillatory region $n \leq N_2(\lambda)$ is treated in \cref{sec:turan} using Tur\'an determinants. This yields a uniform bound for all solutions of \cref{eq:generalized-eigenvalue-lambda}; see \cref{prop:turan}.

    \item In \cref{sec:turning-point}, we analyze the intermediate region
    $N_1(\lambda) \leq n \leq N_4(\lambda)$
    by constructing approximate solutions of \cref{eq:generalized-eigenvalue-lambda} in terms of Airy functions. Explicit error bounds are derived in \cref{prop:turning-point}. These estimates, together with the properties of the Airy functions, are then used in \cref{prop:properties-turningpoint-solutions} to obtain bounds for the actual solutions.

    \item Finally, in \cref{sec:poincare}, we derive asymptotics for the recessive solutions of \cref{eq:generalized-eigenvalue-lambda} using Poincaré's theorem, and obtain explicit error bounds for $n \geq N_4(\lambda)$; see \cref{prop:asymptotics-tail}.
\end{itemize}
In \cref{sec:proof-asymptotics}, all of these ingredients are combined to prove the final bound~\eqref{eq:thebound}.

\subsection{Preliminary definitions and results}
\label{sec:asympotics-preliminaries}

We begin by making a concrete choice of the points $x_j$, and then of the integers $N_j(\lambda)$, defining the three regions. 
We remark, however, that our strategy is independent of the specific choice of $x_j$ (or, equivalently, $N_j(\lambda)$), as long as they are chosen in a suitable way with respect to the turning point $x^\delta=2$, that is, $x_2^\delta<2<x_4^\delta$.
\begin{definition}
    \label{def:asymptotics-regions}
    Set $\delta = \beta-\alpha$ and, for $\lambda>0$, $h = \lambda^{1/\delta}$.
    We define
    \begin{equation}
        x_1^\delta = \nicefrac{1}{2}+\nicefrac{1}{8} \, , \quad x_2^\delta = 1- \nicefrac{1}{8} \, , \quad x_3^\delta = 4+\nicefrac{1}{8} \, , \quad x_4^\delta = 5 - \nicefrac{1}{8} \, . 
    \end{equation}
    Moreover, we define the integers $N_j(\lambda)$ as follows: for $\lambda>0$,
\begin{equation}
    N_j(\lambda)=\begin{dcases}
        \left\lceil \frac{x_j}{h} \right\rceil,& j=1,3;\\ \left\lfloor \frac{x_j}{h} \right\rfloor,& j=2,4,
    \end{dcases}
\end{equation}
and $N_j(0) = \infty$.

    Finally, we set $K_x = [x_1,x_4]$ and $K_N(\lambda) = [N_1(\lambda),N_4(\lambda)]$.
\end{definition}
The offset $\nicefrac{1}{8}$ is necessary as $f_n/a_n \sim n^\delta$ only holds asymptotically, cf. \cref{lem:regularity}; any other sufficiently small value would suffice. 
The choice of the floor and ceiling functions is made in such a way that $\lambda^{1/\delta} K_N(\lambda) \subset K_x$ for all $\lambda > 0$.

\begin{lemma}\label{lem:regularity}
    Set $g(\lambda,n) = \lambda \frac{f_n}{a_n}$.
    Then there exist $N_0$ and $\Lambda > 0$ such that, for all $0\leq \lambda \leq \Lambda$, the following statements hold:
    \begin{enumerate}[(i)]
        \item For all $n \geq N_0$, $|1-\mu_n| \leq \nicefrac{1}{8}$, where $\mu_n = a_{n-1}/a_n$;
        \item For all $n \geq N_0$ and $z \in \Omega$, $|z/a_n| \leq \nicefrac{1}{8}$;
    \end{enumerate}
    additionally, for $0<\lambda<\Lambda$,
    \begin{enumerate}[(i),resume]
        \item $N_0 < N_1(\lambda)<N_2(\lambda) < N_3(\lambda) < N_4(\lambda)$;
        \item $g(\lambda,n) \geq \nicefrac{1}{2}$ for all $n \geq N_1(\lambda)$;
        \item $g(\lambda,n) \leq 1$ for all $N_0 \leq n \leq N_2(\lambda)$;
        \item $g(\lambda,n) \geq 4$ for all $n \geq N_3(\lambda)$;
        \item $g(\lambda,n) \leq 5$ for all $N_0 \leq n \leq N_4(\lambda)$.
    \end{enumerate}
\end{lemma}
\begin{remark}
    \label{rem:lambda-decrease}
    Throughout \cref{sec:asymptotics-of-eigenvectors}, the symbol $\Lambda>0$ denotes a sufficiently small constant for which the corresponding estimates hold. 
    In particular, throughout the oscillatory regime (cf.~\cref{sec:turan}) and the tail region (cf.~\cref{sec:poincare}), it suffices to choose $\Lambda>0$ such that the above lemma holds. In \cref{sec:turning-point}, which deals with the turning-point analysis of the difference equation, the precise value of $\Lambda$ may vary and is implicitly allowed to decrease from one statement to the next.

    Consequently, in the proof of \cref{prop:asymptotics}, $\Lambda>0$ is chosen sufficiently small so that both \cref{lem:regularity} and the main result of the turning-point analysis, \cref{prop:properties-turningpoint-solutions}, hold.
\end{remark}
\begin{proof}
    By $\lim_{n \to \infty}a_{n-1}/a_n = 1$, the compactness of $\Omega$, and $\lim_{n \to \infty}a_n^{-1} = 0$, there exists $\tilde N_0$ such that, for all $n\geq\tilde N_0$, the first two statements hold.
    
    We now treat the remaining cases for $\lambda > 0$.
    Let $\epsilon>0$ be small enough so that
    \begin{align}
        \label{proofeq:regularity-epsilon}
        x_1^\delta(1-\epsilon) \geq \nicefrac{1}{2} \, , \quad x_2^\delta(1+\epsilon) \leq 1\, , \quad x_3^\delta(1-\epsilon) \geq 4 \,, \quad x_4^\delta(1+\epsilon) \leq 5.
    \end{align}
    Per assumption (\cref{hyp:conditions-an-fn}), $f_n/a_n \sim n^\delta$, and hence there exists $N_0 \geq \tilde N_{0}(\Omega)$ such that
    \begin{equation}
        \left|\frac{f_n}{a_n}-n^\delta\right| \leq \epsilon n^\delta
    \end{equation}
    for all $n \geq N_0$.
    Therefore, for all $\lambda > 0$ and $x = \lambda^{1/\delta} n$, we have
    \begin{equation}
        \label{proofeq:regularity-expansion}
        \left|g(\lambda,n) - x^\delta\right| = \lambda \left|\frac{f_n}{a_n}-n^\delta\right| \leq \lambda  \epsilon n^\delta = \epsilon x^\delta
    \end{equation}
    for all $n \geq N_0$.
    
    By definition, $N_1(\lambda) = \left\lceil(x_1/ \lambda)^{1/\delta}\right\rceil$, and hence there exists $\Lambda > 0$ such that $N_0<N_1(\lambda)$ holds for all $\lambda < \Lambda$.
    The inequality for the $N_j(\lambda)$ follows per definition.

    For $0 < \lambda < \Lambda$ and $n \geq N_1(\lambda)$ we obtain, using \cref{proofeq:regularity-expansion,proofeq:regularity-epsilon} and $x = \lambda^{1/\delta}n \geq x_1$,
    \begin{equation}
        g(\lambda,n) \geq x^\delta -  \left|g(\lambda,n) - x^\delta\right| \geq x^\delta(1-\epsilon) \geq x_1^\delta(1-\epsilon) \geq \nicefrac{1}{2} \, .
    \end{equation}
    Similarly, for $N_0 \leq n \leq N_2(\lambda)$ and $x = \lambda^{1/\delta}n \leq x_2$
    \begin{equation}
        g(\lambda,n) \leq x^\delta + \left|g(\lambda,n) - x^\delta\right| \leq x^\delta(1+\epsilon) \leq x_2^\delta(1+\epsilon)  \leq 1 \, .
    \end{equation}
    The last two statements follow analogously.
\end{proof}
We state two further auxiliary results needed throughout the following analysis.
We begin by noting that the second-order recurrence relation~\eqref{eq:generalized-eigenvalue-lambda} can be rewritten as a first-order matrix recurrence equation:
\begin{lemma}
    \label{lem:transfer-matrix}
    A sequence $u \in \ell(\nnum)$ obeys \cref{eq:generalized-eigenvalue-lambda} if and only if it obeys
    \begin{equation}
        \begin{pmatrix}
        u_{n+1} \\ u_n 
    \end{pmatrix} = T^{\lambda,z}_n \begin{pmatrix}
        u_n \\ u_{n-1}
    \end{pmatrix} \, , \quad T^{\lambda,z}_n = \begin{pmatrix}
        -\lambda f_n/a_n+\frac{z}{a_n} & -\frac{a_{n-1}}{a_n} \\ 1 & 0 
    \end{pmatrix} \quad \forall n \geq 1\, , 
    \end{equation}
    where $T^{\lambda,z}_n$ is the so-called transfer matrix.
\end{lemma}
Within this representation, we can easily obtain the following:
\begin{lemma}
    \label{lem:before-N0}
    Let $\Lambda$, $N_0$ be as given in \cref{lem:regularity,def:asymptotics-regions}.
    Then there exists $C_0>0$ such that, for all $0\leq \lambda \leq \Lambda$, $z \in \Omega$, all solutions $u$ of \cref{eq:generalized-eigenvalue-lambda}, and $n \leq N_0$, the following estimate holds:
    \begin{equation}
        |u_{n+1}|^2 + |u_n|^2 \leq C_0 (|u_0|^2+|u_1|^2)
    \end{equation}
\end{lemma}
\begin{proof}
 By \cref{lem:transfer-matrix}, for all $n\geq1$ we have
    \begin{equation}
        \begin{pmatrix}
        u_{n+1} \\ u_n 
    \end{pmatrix} = T^{\lambda,z}_n \begin{pmatrix}
        u_n \\ u_{n-1}
    \end{pmatrix} \, , \quad T^{\lambda,z}_n = \begin{pmatrix}
        -\lambda f_n/a_n+\frac{z}{a_n} & -\frac{a_{n-1}}{a_n} \\ 1 & 0 
    \end{pmatrix}\, .
    \end{equation}
    Since $\Omega$, $[0,\Lambda]$ and $[0,N_0]$ are compact sets, there exists $c_0 >0$ such that 
    \begin{equation}
        \norm{T_n} \leq c_0
    \end{equation}
    for all $z \in \Omega$, $\lambda \in [0,\Lambda]$ and $n \leq N_0$.
    The claim follows with $C_0 = c_0^{2N_0}$.
\end{proof}

\subsection{First region: Tur\'an determinants}
\label{sec:turan}
Here we focus on the region $n \in [N_0,N_2(\lambda)]$, cf.~\cref{def:asymptotics-regions}, i.e. the region in which $\lambda f_n/a_n\sim x^\delta$ is sufficiently small. In particular, \cref{lem:regularity} applies in this region.

We aim to derive $\lambda$-independent estimates for all solutions of the generalized eigenvalue equation~\eqref{eq:generalized-eigenvalue-lambda} in this region using Tur\'an determinants, cf.~\cref{def:turan}. To this end, we follow the approach of~\cite{swiderski-spectralpropertiesblock-2018}, keeping track of all constants.

Throughout the section, we will denote by $E\in\cnum^{2\times 2}$ the matrix
\begin{equation}
    E = \begin{pmatrix}
        0 &  1\\ -1 & 0
    \end{pmatrix}\, ,
\end{equation}
and write $\Real A = \frac{1}{2}(A+A^\dagger)$ for a matrix $A \in \cnum^{2\times 2}$.
\begin{definition}
    \label{def:turan}
    Given $n \in \nnum$ and $\lambda \geq 0$, let $Q^{\lambda,z}_n$ be the quadratic form on $\cnum^2$ given by
    \begin{equation}
        Q^{\lambda,z}_n(v) = \braket{\Real(E T^{\lambda,z}_n)v,v}\,, \quad v \in \cnum^2\,, 
    \end{equation}
    where $T^{\lambda,z}_n$ is the transfer matrix from \cref{lem:transfer-matrix}.
Besides, given a solution $u=(u_n)_{n\in\nnum}$ of \cref{eq:generalized-eigenvalue-lambda}, the $n$th \textit{Tur\'an determinant} $S^{\lambda,z}_n(\alpha)$ is given by
    \begin{equation}
        S^{\lambda,z}_n(\alpha) = a_{n}Q^{\lambda,z}_n((u_n,u_{n-1})^T)\, , 
    \end{equation}
    where $\alpha=(u_0,u_1)^\intercal$.
\end{definition}
Using the explicit expression of the transfer matrix, we have
\begin{equation}
    Q^{\lambda,z}_n(v) = \braket{A^{\lambda,z}_n v,v}\, , \quad A^{\lambda,z}_n = \begin{pmatrix}
        1 & \frac{\lambda f_n}{2 a_n} - \frac{z^\ast}{2 a_n} \\ \frac{\lambda f_n}{2 a_n} - \frac{z}{2 a_n} & \frac{a_{n-1}}{a_n} \, .
    \end{pmatrix}
\end{equation}
Hence, $Q^{\lambda,z}_n$ is positive definite if and only if $\det A^{\lambda,z}_n > 0$, which holds provided that $\frac{\lambda f_n}{2a_n}$ is sufficiently small.
The following lemma quantifies this:
\begin{lemma}
    \label{lem:uniformly-positive}
    For all $n\in[N_0,N_2(\lambda)]$ and $v\in\cnum^2$,
    \begin{equation}
        \frac{1}{4} \norm{v}^2 \leq Q^{\lambda,z}_n(v) \leq 2 \norm{v}^2\, .
    \end{equation}
\end{lemma}
\begin{proof}
    We write $A^{\lambda,z}_n$ as 
    \begin{equation}
        A^{\lambda,z}_n = \begin{pmatrix}
            1 & \sigma_n(\lambda,z)^\ast \\ \sigma_n(\lambda,z) & \mu_n
        \end{pmatrix}\, ,
    \end{equation}
    where $\sigma_n(\lambda,z) = \frac{\lambda f_n}{2 a_n}-\frac{z}{2 a_n}$ and $\mu_n = \frac{a_{n-1}}{a_n}$.
    As $A^{\lambda,z}_n$ is symmetric, the lower and upper bound of the quadratic form $Q^{\lambda,z}_n(v) = \braket{A^{\lambda,z}_n v, v}$ are given by the two eigenvalues $\rho^{\lambda,z}_\pm$ of $A^{\lambda,z}_n$:
    \begin{equation}
        \rho^{\lambda,z}_\pm = \frac{1+\mu_n}{2} \pm \sqrt{\frac{(1-\mu_n)^2}{4}+|\sigma_n(\lambda,z)|^2} \, .
    \end{equation}
    For all $N_0 \leq n \leq N_2(\lambda)$, \cref{lem:regularity} implies
    \begin{align}
        |1-\mu_n| &  \leq \frac{1}{8}, \\
        |\sigma_n(\lambda,z)| & \leq \frac{\lambda f_n}{2 a_n} + \frac{|z|}{2 a_n} \leq \frac{1}{2} + \frac{1}{16} \, .
    \end{align}
    Hence, 
    \begin{align}
        \rho^{\lambda,z}_+ & \leq 1+\frac{1}{16} + \sqrt{\nicefrac{1}{16^2}+(\nicefrac{1}{2}+\nicefrac{1}{16})^2} \leq 2 \, , \\
        \rho^{\lambda,z}_- & \geq 1-\frac{1}{16}-\sqrt{\nicefrac{1}{16^2}+(\nicefrac{1}{2}+\nicefrac{1}{16})^2} \geq \frac{1}{4} \, ,
    \end{align}
    whence the claim follows.
\end{proof}
The next estimate gives bounds on the increment of $S^{\lambda,z}_n(\alpha)$ in terms of the variation of the entries of the transfer matrix $T^{\lambda,z}_n$:
\begin{lemma}
    \label{lem:turan-diff}
    Let $u$ be a solution of \cref{eq:generalized-eigenvalue-lambda}, and $S^{\lambda,z}_n(\alpha)$ be the corresponding Tur\'an determinant (\cref{def:turan}). Then 
    \begin{equation}
        |S^{\lambda,z}_{n+1}(\alpha)-S^{\lambda,z}_n(\alpha)|  \leq \norm{T^{\lambda,z}_n} a_{n+1} \Delta_n(\lambda,z)(|u_n|^2+|u_{n-1}|^2)\, ,
    \end{equation}
    where 
    \begin{equation}\label{eq:deltanlz}
        \Delta_n(\lambda,z) = \Delta^{a,1}_n+|z| \Delta^{a,2}_n+ |z-z^\ast| a_n^{-1} + \lambda \Delta^{f}_n 
    \end{equation}
    and
    \begin{align}
        \Delta^{a,1}_n & = |a_{n+1}^{-1}a_n-a_n^{-1} a_{n-1}|\, , \\
        \Delta^{a,2}_n & = |a_{n+1}^{-1}-a_{n}^{-1}| \, , \\
        \Delta^{f}_n & = |a_{n+1}^{-1} f_{n+1}-a_n^{-1} f_n| \, .
    \end{align}
\end{lemma}
\begin{proof}
    The proof closely follows \cite[Lemma~3]{swiderski-spectralpropertiesblock-2018}.
    By \cref{def:turan,lem:transfer-matrix}, we get
    \begin{align}
        \label{proofeq:turan-difference}
        S^{\lambda,z}_{n+1}(\alpha)-S^{\lambda,z}_n(\alpha) & = a_{n+1} Q^{\lambda,z}_{n+1}((u_{n+1},u_{n})^T) - a_{n}Q^{\lambda,z}_n((u_{n},u_{n-1})^T) \\
        & = a_{n+1} Q^{\lambda,z}_{n+1}(T^{\lambda,z}_n(u_{n},u_{n-1})^T) - a_{n}Q^{\lambda,z}_n((u_{n},u_{n-1})^T) \\
        & = \braket{\Real C^{\lambda,z}_n (u_{n},u_{n-1})^T,(u_{n},u_{n-1})^T}\, ,
    \end{align}
    where $C^{\lambda,z}_n$ is given by
    \begin{equation}
        C^{\lambda,z}_n = (T^{\lambda,z}_n)^\dagger a_{n+1} E T^{\lambda,z}_{n+1} T^{\lambda,z}_n - a_n E T^{\lambda,z}_n \, . 
    \end{equation}
    As $E^2 = -\id$, we get
    \begin{equation}
        C^{\lambda,z}_n = -\left(a_{n+1} (T^{\lambda,z}_n)^\dagger E T^{\lambda,z}_{n+1}E +a_n\right) E T^{\lambda,z}_n = -\tilde{C}^{\lambda,z}_n E T^{\lambda,z}_n \, . 
    \end{equation}
    A direct calculation yields
    \begin{equation}
        \tilde{C}^{\lambda,z}_n = a_{n+1}\begin{pmatrix}
            0 & \frac{\lambda f_{n+1}-z}{a_{n+1}} + \frac{z^\ast -\lambda f_n}{a_n} \\
            0 & \frac{a_n}{a_{n+1}} - \frac{a_{n-1}}{a_{n}}
        \end{pmatrix} = a_{n+1}\begin{pmatrix}
            0 & \frac{\lambda f_{n+1}-z}{a_{n+1}} - \frac{\lambda f_n-z}{a_n}+\frac{z^\ast-z}{a_n} \\
            0 & \frac{a_n}{a_{n+1}} - \frac{a_{n-1}}{a_{n}}
        \end{pmatrix}\, , 
    \end{equation}
    and hence
    \begin{equation}
        \norm{\tilde{C}^{\lambda,z}_n} \leq a_{n+1} \left(\Delta^{a,1}_n + |z|\Delta^{a,2}_n+\lambda \Delta^{f}_n+|z-z^\ast| a_n^{-1} \right) = a_{n+1}\Delta_n(\lambda,z)\,. 
    \end{equation}
    Thus,
    \begin{equation}
        \norm{C^{\lambda,z}_n}  \leq \norm{\tilde{C}^{\lambda,z}_n}\norm{E}\norm{T^{\lambda,z}_n} \leq \norm{T^{\lambda,z}_n} a_{n+1}\Delta_n(\lambda,z)
    \end{equation}
    Using \cref{proofeq:turan-difference}, $\norm{\Real C_n} \leq \norm{C_n}$, and the Cauchy--Schwartz inequality, we obtain
    \begin{equation}
        |S^{\lambda,z}_{n+1}(\alpha)-S^{\lambda,z}_n(\alpha)| = \left|\braket{\Real C^{\lambda,z}_n (u_{n},u_{n-1})^T,(u_{n},u_{n-1})^T}\right| \leq \norm{C^{\lambda,z}_n} (|u_{n-1}|^2+|u_n|^2)\,, 
    \end{equation}
    and the claim follows.
\end{proof}
\begin{lemma}
    \label{lem:turan-upperbound}
   Let $u$ be a solution of \cref{eq:generalized-eigenvalue-lambda}.
    Then, for all $N_0 \leq n \leq N_2(\lambda)$, 
    \begin{equation}
         a_n(|u_{n-1}|^2+|u_n|^2)\leq 8 a_{N_0} \exp\left(\frac{128}{7}\sum_{k = N_0}^n \Delta_k(\lambda,z)\right) (|u_{N_0}|^2+|u_{N_0-1}|^2)\, ,
    \end{equation}
    where  $\Delta_n(\lambda,z)$ is given by \cref{eq:deltanlz}.
\end{lemma}
\begin{proof}
    We begin by proving two estimates for the Tur\'an determinants $S^{\lambda,z}_n(\alpha)$.
    Let 
    \begin{equation}
        F^{\lambda,z}_k = \frac{S^{\lambda,z}_{k+1}(\alpha)-S^{\lambda,z}_k(\alpha)}{S^{\lambda,z}_k(\alpha)} \, .
    \end{equation}
    By \cref{lem:uniformly-positive}, we get 
    \begin{equation}
        \label{proofeq:turan-upperbound-denominator}
        |S^{\lambda,z}_k(\alpha)| = a_k |Q^{\lambda,z}_k((u_k,u_{k-1})^T)| \geq \frac{a_k}{4} (|u_k|^2+|u_{k-1}|^2) \, .
    \end{equation}
    Furthermore, using \cref{lem:transfer-matrix,lem:regularity,def:asymptotics-regions} for all $N_0 \leq k \leq N_2(\lambda)$,
    \begin{equation}
        \norm{T^{\lambda,z}_k} \leq 1+ \frac{\lambda f_k}{a_k} + \frac{|z|}{a_k} + \frac{a_{k-1}}{a_k} \leq 1+1+\nicefrac{1}{8} +1+\nicefrac{1}{8} \leq 4 \, .
    \end{equation}
    Hence, using \cref{lem:turan-diff,proofeq:turan-upperbound-denominator}, we obtain
    \begin{align}
        |F^{\lambda,z}_k| & \leq a_{k+1}\norm{T^{\lambda,z}_k} \Delta_k(\lambda,z) (|u_{k-1}|^2+|u_k|^2)  \times \left(\frac{a_k}{4} (|u_{k-1}|^2+|u_k|^2)\right)^{-1} \\
        & \leq 16\frac{a_{k+1}}{a_k}\Delta_k(\lambda,z) \leq \frac{128}{7} \Delta_k(\lambda,z),
    \end{align}
    where we used $a_{k+1}/a_k \leq \frac{8}{7}$ (\cref{lem:regularity}) for all $k \geq N_0$.
    We thus get an estimate on the Tur\'an determinant $S^{\lambda,z}_n(\alpha)$ in terms of $S^{\lambda,z}_{N_0}(\alpha)$:
    \begin{align}
        |S^{\lambda,z}_n(\alpha)| & = \left| S^{\lambda,z}_{N_0}(\alpha)\prod_{k=N_0}^{n-1}(1+F^{\lambda,z}_k)\right| \leq|S^{\lambda,z}_{N_0}(\alpha)| \prod_{k=N_0}^{n}(1+|F^{\lambda,z}_k|) \leq |S^{\lambda,z}_{N_0}(\alpha)|\exp\left(\sum_{k=N_0}^{n-1}|F^{\lambda,z}_k|\right) \\
        & = |S^{\lambda,z}_{N_0}(\alpha)|\exp\left(\frac{128}{7}\sum_{k = N_0}^n\Delta_k(\lambda,z)\right)\, .
    \end{align}
    By \cref{lem:uniformly-positive},
    \begin{align}
        a_n(|u_n|^2+|u_{n-1}|^2) & \leq 4 |S^{\lambda,z}_n(\alpha)|\, ,  \\
        |S^{\lambda,z}_{N_0}(\alpha)| & \leq 2 a_{N_0} (|u_{N_0}|^2+|u_{N_0-1}|^2) \, , 
    \end{align}
    and the proof is complete.
    \end{proof}
We can finally obtain an estimate for the entries of a solution of the generalized eigenvalue equation with $N_0 \leq n \leq N_2(\lambda)$:
\begin{proposition}
    \label{prop:turan}
  Let $N_0$, $N_2(\lambda)$, and $\Lambda$ be as given in \cref{lem:regularity,def:asymptotics-regions}.
    Let $u$ be a solution of \cref{eq:generalized-eigenvalue-lambda}.
    Then, there exists a constant $C_1 >0$ such that, for all $z \in \Omega$, $0 \leq \lambda \leq \Lambda$ and $N_0 \leq n \leq N_2(\lambda)$, the following holds:
    \begin{equation}
        a_n(|u_n|^2+|u_{n-1}|^2) \leq C_1 (|u_{N_0}|^2+|u_{N_0-1}|^2) \, .
    \end{equation}
\end{proposition}
\begin{proof}
    Using \cref{lem:turan-upperbound}, it is sufficient to show that there exists $\tilde{C}$, independent of $\lambda \geq 0$ and $z \in \Omega$, such that for all $N_0 \leq n \leq N_2(\lambda)$
    \begin{equation}
        \sum_{k = N_0}^n \Delta_k(\lambda,z) \leq \tilde{C}
    \end{equation}
    holds, where $\Delta_k(\lambda,z)$ is defined as in \cref{eq:deltanlz}.
    Using \cref{def:total-variation,lem:an-variation}, we obtain
    \begin{align}
        \sum_{k = N_0}^n \Delta^{a,1}_k & = \mathcal{V}_{N_0,n}(a_{k-1} a_{k}^{-1}) \leq \mathcal{V}(a_{k-1} a_{k}^{-1}) < \infty\, ,  \\
        \sum_{k = N_0}^n \Delta^{a,2}_k & = \mathcal{V}_{N_0,n}(a_k^{-1}) \leq \mathcal{V}(a_k^{-1}) < \infty\, ,  \\
        \sum_{k = N_0}^n \frac{1}{a_k} & \leq \sum_{k = 0}^{\infty}\frac{1}{a_k} < \infty \, .
    \end{align}
    Moreover, by \cref{hyp:conditions-an-fn},
    \begin{equation}
        \sum_{k = N_0}^n \Delta_k^{f} = \mathcal{V}_{N_0,n}(f_k a_k^{-1}) \leq \mathcal{V}_{n}(f_k a_k^{-1}) \leq \frac{f_n}{a_n} + C
    \end{equation}
    for some $C \in \rnum$.
    In summary, we obtain
    \begin{align}
        \sum_{k = N_0}^n \Delta_k(\lambda,z) & =   \sum_{k = N_0}^n\Delta^{a,1}_k+|z| \Delta^{a,2}_k+ |z-z^\ast| a_k^{-1} + \lambda \Delta^{f}_k \\ 
        & \leq \mathcal{V}(a_{k-1} a_{k}^{-1})  + |z| \mathcal{V}(a_k^{-1}) + |z-z^\ast|  \sum_{k = 0}^{\infty}\frac{1}{a_k} + \lambda \frac{f_n}{a_n} + \lambda C \, .
    \end{align}
    By \cref{lem:regularity} we have $\lambda \frac{f_n}{a_n} \leq 1$ for all $N_0 \leq n \leq N_2(\lambda)$, and as $\Omega$ and $[0,\Lambda]$ are compact, there exists $\tilde{C} > 0$ such that
    \begin{equation}
        \sum_{k = N_0}^n \Delta_k(\lambda,z) \leq \tilde{C}
    \end{equation}
    for all $z \in \Omega$, $0 \leq \lambda \leq \Lambda$ and $N_0 \leq n \leq N_2(\lambda)$.
    Using \cref{lem:turan-upperbound}, we hence obtain
    \begin{equation}
        a_n(|u_{n-1}|^2+|u_n|^2)\leq 8 a_{N_0} \exp\left(\frac{128}{7}\tilde{C}\right) (|u_{N_0}|^2+|u_{N_0-1}|^2)
    \end{equation}
    and the claim follows.
\end{proof}

\subsection{Second region: turning point analysis}
\label{sec:turning-point}

We now turn to the second region $n\in K_N(\lambda)=[N_1(\lambda),N_4(\lambda)]$, corresponding to $x\in K_x=[x_1,x_4]$.
Recalling \cref{eq:differential-turning-point}, we expect \cref{eq:generalized-eigenvalue-lambda} to admit one solution that decays for $x^\delta>2$, the so-called \textit{recessive} solution $u^{\recessSol,h,z}$, and one linearly independent solution that grows for $x^\delta>2$, the so-called \textit{dominant} solution $u^{\domSol,h,z}$.
For $x^\delta<2$, both solutions are expected to oscillate and remain bounded.
Throughout this section, objects associated with the recessive and dominant solutions will be labeled by the indices $\recessSol$ and $\domSol$, respectively.

Recalling the discussion in \cref{sec:asymptotics_overview}, we aim to approximate both the recessive and dominant solutions with solutions of the Airy equation. An overview of relevant properties needed in this work are discussed in \cref{app:airy}.
To this end, we apply results of Geronimo et al.~\cite{geronimo-wkbturningpoint-2004} to obtain explicit error bounds on the approximate solutions.
Our strategy for the analysis in the turning point region is organized as follows.

\begin{enumerate}
    \item We begin by multiplying \cref{eq:generalized-eigenvalue-lambda} by $h^{\alpha}$, which leads to a rescaled difference operator $L^{h,z}$ with coefficient functions $a_1(x,h)$ and $b_1(x,h)$ (\cref{lem:geronimo}). By construction, these coefficients are bounded on $K_x \times [0,h]$, allowing us to apply the results of Geronimo et al.
    
\item Next, we replace $a_1(x,h)$ and $b_1(x,h)$ with simplified coefficients $\tilde a(x,h)$ and $\tilde b(x,h)$, obtained from the asymptotics of $a_n$ and $f_n$. The assumptions in \cref{hyp:conditions-an-fn} imply that the original operator $L^{h,z}$ (\cref{lem:turning-point-op}) can be treated as a perturbation of the simpler operator $\tilde L^h$ (cf.~\cref{def:turning-point-easy,lem:turning-point-diff-bounds}). In particular, the contribution of the parameter $z$ will also be perturbative in the relevant regime.

\item Using the coefficients $\tilde a(x,h)$ and $\tilde b(x,h)$, we then construct the Langer transform (\cref{def:langer-transform}),
\begin{equation}
    x \mapsto \xi(x,h),
\end{equation}
which maps the turning point of the difference equation to $\xi(x,h)=0$. The transformed equation is oscillatory for $\xi(x,h)<0$, while for $\xi(x,h)>0$ it admits a recessive and a dominant solution. This matches the qualitative behavior of the Airy equation for negative and positive arguments, respectively.

\item Motivated by this correspondence, we construct approximate solutions (cf.~\cref{def:approximate-solutions})
\begin{equation}
    \psi^{\recessSol}(x,h),
    \qquad
    \psi^{\domSol}(x,h),
\end{equation}
in terms of Airy functions. These functions solve the approximate equation
\begin{equation}
    (\tilde L^h \psi)(x)=0
\end{equation}
up to an error of order $O(h^2)$ on $K_x$.

\item The main theorem of Geronimo et al.\ then implies that these approximate solutions faithfully describe actual solutions $u^{\recessSol,h,z}$ and $u^{\domSol,h,z}$ of the original difference equation. More precisely,
\begin{align}
    u^{\recessSol,h,z}_n
    &=
    (-1)^n \psi^{\recessSol}(nh,h)
    +
    w^{\recessSol}(nh,h)\sigma^{\recessSol,h,z}_n,
    \\
    u^{\domSol,h,z}_n
    &=
    (-1)^n \psi^{\domSol}(nh,h)
    +
    w^{\domSol}(nh,h)\sigma^{\domSol,h,z}_n,
\end{align}
where $w^{\recessSol}(x,h)$ and $w^{\domSol}(x,h)$ (\cref{def:approximate-solutions}) serve as scaling functions for the relative errors
$\sigma^{\recessSol,h,z}_n$ and $\sigma^{\domSol,h,z}_n$.
Heuristically, the functions $w^{\recessSol}(x,h)$ and $w^{\domSol}(x,h)$ may be viewed as envelopes of the corresponding approximate solutions.

\item The main estimate of \cref{prop:turning-point} shows that the scaled errors satisfy bounds of order $h^{1/3}$. Consequently, the approximations improve as
\begin{equation}
    h=\lambda^{1/\delta}\to0.
\end{equation}
This allows us to derive quantitative bounds on the actual solutions $u^{\recessSol,h,z}$ and $u^{\domSol,h,z}$ in \cref{prop:properties-turningpoint-solutions}.
\end{enumerate}
Throughout these steps, the precise value of $\Lambda>0$ may vary and is implicitly allowed to decrease from one statement to the next, cf. \cref{rem:lambda-decrease}.

For the convenience of the reader, \cref{tab:turning-point-symbols} collects the main symbols used throughout this section. A summary of the results of Geronimo et al.\ relevant to the present setting is given in \cref{app:turning-point}. Related results in a similar setting were obtained by Fedotov and Klopp in \cite{fedotov-complexwkbmethod-2019}.

\begin{table}[t]
    \centering
    \begin{tabular}{l|l}
        $a_n,f_n,\lambda,z$
        & original parameters of \cref{eq:generalized-eigenvalue-lambda}
        \\
        $h=\lambda^{1/\delta}$, $a_1(x,h)$, $b_1(x,h)$
        & rescaled coefficients, cf.\ \cref{lem:geronimo}
        \\
        $L^{h,z}$
        & rescaled difference operator, cf.\ \cref{lem:turning-point-op}
        \\
        $\tilde a(x,h)$, $\tilde b(x,h)$
        & asymptotic approximations of $a_1(x,h)$, $b_1(x,h)$, cf.\ \cref{def:turning-point-easy}
        \\
        $\tilde L^h$
        & approximation of $L^{h,z}$, cf.\ \cref{def:turning-point-easy}
        \\
        $\xi(x,h)$
        & Langer transform, cf.\ \cref{def:langer-transform}
        \\
        $\Ai_j(x)$, $\tilde w_j(x)$
        & Airy and modified Hankel functions, cf.\ \cref{app:airy}
        \\
        $\psi^{\recessSol}(x,h)$, $\psi^{\domSol}(x,h)$
        & approximate recessive and dominant solutions of $\tilde L^h$
        \\
        $w^{\recessSol}(x,h)$, $w^{\domSol}(x,h)$
        & scaling functions associated with the approximation errors
        \\
        $u^{\recessSol,h,z}$, $u^{\domSol,h,z}$
        & actual recessive and dominant solutions of \cref{eq:generalized-eigenvalue-lambda}
        \\
        $\sigma^{\recessSol,h,z}$, $\sigma^{\domSol,h,z}$
        & scaled approximation errors associated with $u^{\recessSol,h,z}$, $u^{\domSol,h,z}$
    \end{tabular}
    \caption{
        List of symbols used in the turning point analysis of
        \cref{eq:generalized-eigenvalue-lambda}.
        The indices $\recessSol$ and $\domSol$ indicate quantities associated with the recessive and dominant solutions, respectively.
    }
    \label{tab:turning-point-symbols}
\end{table}

We follow the strategy outlined above and begin with a preparatory result that will enable us to convert the generalized eigenvalue equation~\eqref{eq:generalized-eigenvalue-lambda} in a form compatible with the setting of~\cite{geronimo-wkbturningpoint-2004}, cf.~\cref{eq:app-difference-equation}:
\begin{lemma}
    \label{lem:geronimo}
    Suppose that \cref{hyp:conditions-an-fn} holds with coefficients $\alpha,\beta,c_a,c_f$, and let $K_x,K_N(\lambda)$ be given as in \cref{def:asymptotics-regions} for $\lambda > 0$.
    Set $\delta = \beta-\alpha$ and $h = \lambda^{1/\delta}$, and let $a,f:\rnum_+ \to \rnum_+$ be the linear interpolations of the sequences $(a_n)_{n\in\nnum},(f_n)_{n\in\nnum}$.
    For $h>0$ and $x\in K_x$, define 
    \begin{equation}
        a_1(x,h) = h^\alpha a(x/h) \, , \quad b_1(x,h) = h^\beta f(x/h) \, .
    \end{equation}
    Then there exists $\Lambda > 0$ such that, for all $x\in K_x$, the limits
    \begin{equation}
        \lim_{h \to 0} a_1(x,h), \quad \lim_{h \to 0}b_1(x,h)
    \end{equation}
    exist and are finite, so that the functions $a_1,b_1$ can be continuously extended  to $K_x\times[0,\Lambda^{1/\delta}]$. Moreover, $a_1(x,h)$ is strictly positive on $K_x \times [0,\Lambda^{1/\delta}]$, and 
    \begin{equation}
        a_1(x,h) = x^\alpha \left(1+\frac{c_a h}{x} +O(h^2)\right) \, , \quad b_1(x,h) = x^\beta \left(1+\frac{c_f h}{x} +O(h^2)\right)
    \end{equation}
    uniformly in $x \in K_x$, i.e. there exists $C>0$ such that 
    \begin{equation}
        \left|x^{-\alpha} a_1(x,h) - 1-\frac{c_a h}{x}\right| \leq C h^2\, , \quad  \left|x^{-\beta} b_1(x,h) - 1-\frac{c_f h}{x}\right| \leq C h^2
    \end{equation}
    for all $(x,h) \in K_x \times [0,\Lambda^{1/\delta}]$.
\end{lemma}
\begin{proof}
    By construction, the functions $a$ and $f$ are continuous, and thus so are $K_x\times(0,\Lambda^{1/\delta})\ni(x,h)\mapsto a_1(x,h),b_1(x,h)$ for arbitrary $\Lambda > 0$.
    In order to consider the limit $h \to 0$, we discuss the asymptotics first.
    By definition, for $h > 0$ and $n \in \nnum$,
    \begin{align}
        a_1(nh,h) & = h^{\alpha} n^\alpha \left( 1+ \frac{c_a}{n} + O(1/n^2) \right) \, , \\
        b_1(nh,h) & = h^{\beta} n^\beta \left( 1+ \frac{c_f}{n} + O(1/n^2) \right) \, ,
    \end{align}
    and hence there exist $N \in \nnum$ and $C>0$ independent of $h$ such that 
    \begin{align}
        \left|(nh)^{-\alpha}a_1(nh,h) - 1- \frac{c_a}{n} \right| \leq \frac{C}{n^2}\, ,  \\
        \left|(nh)^{-\beta}b_1(nh,h) - 1- \frac{c_f}{n} \right| \leq \frac{C}{n^2} 
    \end{align}
    for all $n \geq N$.
    Setting $x = nh$, we thus get
    \begin{align}
        \left|x^{-\alpha}a_1(nh,h) - 1- \frac{c_a h}{x} \right| \leq \frac{Ch^2}{x^2}\, ,  \\
        \left|x^{-\beta}b_1(nh,h) - 1- \frac{c_f h}{x} \right| \leq \frac{Ch^2}{x^2} 
    \end{align}
    for all $x/h \in \nnum$ and $x/h \geq N$.
    As the functions $a,b$ have been defined via linear interpolation, the same asymptotics hold for all other values of $x \in K_x$,
    and we obtain 
        \begin{align}
        \left|x^{-\alpha}a_1(x,h) - 1- \frac{c_a h}{x} \right| \leq \frac{C h^2}{x^2} \leq \frac{C}{x_1^2} h^2 \, , \\
        \left|x^{-\beta}b_1(x,h) - 1- \frac{c_f h}{x} \right| \leq \frac{C h^2}{x^2} \leq \frac{C}{x_1^2} h^2 \, , 
    \end{align}
    for all $h\leq x_1/N \eqqcolon\Lambda^{1/\delta}$ and $x \in K_x$, yielding the claimed asymptotics.
    Finally, we obtain 
    \begin{equation}
        \lim_{h \to 0} a_1(x,h) = x^\alpha \,  , \quad \lim_{h \to 0}b_1(x,h) = x^\beta
    \end{equation}
    uniformly in $x \in K_x$, and hence the two functions admit continuous extensions to $h = 0$.
\end{proof}
From now on, we choose $\Lambda>0$ so that, on top of the estimates in \cref{sec:turan}, \cref{lem:geronimo} holds as well.

Using the previous lemma, we now transform the generalized eigenvalue equation~\eqref{eq:generalized-eigenvalue-lambda} into a difference equation of the form~\eqref{eq:app-difference-equation}:
\begin{lemma}
    \label{lem:turning-point-op}
    Let $0\leq \lambda < \Lambda$, $h = \lambda^{1/\delta}$, and  $a_1,b_1,K_N(\lambda)$ be given as above.
    For $z \in \Omega$, define the difference operator $L^{h,z}$ via
    \begin{equation}
        (L^{h,z} v)_n = a_1(nh,h) v_{n+1} - (b_1(nh,h)-h^\alpha z)v_n + a_1(nh-h,h) v_{n-1}\, .
    \end{equation}
    Suppose that $(L^{h,z}v)_n = 0$ holds for all $n \in K_N(\lambda)$.
    Let the vector $u\in\ell(\nnum)$ be given by $u_n = (-1)^n v_n$. Then  \cref{eq:generalized-eigenvalue-lambda} holds for all $z\in\Omega$ and $n \in K_N(\lambda)$.
\end{lemma}
\begin{proof}
    For $n \in K_N(\lambda)$, we insert the ansatz $u_n = (-1)^n v_n$ into \cref{eq:generalized-eigenvalue-lambda} and obtain
    \begin{equation}
        (\Jop(\lambda)u)_n = a_n (-1)^{n+1}v_{n+1} + \lambda f_n (-1)^n v_n -z (-1)^n v_n + a_{n-1}(-1)^{n-1}v_{n-1} \, .
    \end{equation}
    Multiplying by $(-1)^{n+1} h^\alpha$, this becomes
    \begin{equation}
        (-1)^{n+1} h^\alpha (\Jop(\lambda) u)_n = h^\alpha a_n v_{n+1} - (\lambda h^\alpha f_n - h^\alpha z)v_n + h^\alpha a_{n-1} v_{n-1} \, , 
    \end{equation}
    and using the definitions of $a_1,b_1$ in \cref{lem:geronimo} and $h^\alpha \lambda = h^\alpha h^\delta = h^\beta$ we get
    \begin{equation}
        (-1)^{n+1} h^\alpha (\Jop(\lambda) u)_n = a_1(nh,h) v_{n+1} - (b_1(nh,h) - h^\alpha z)v_n + a_1(nh-h,h) v_{n-1},  
    \end{equation}
    which is zero by assumption.
\end{proof}
The difference operator $L^{h,z}$ is exactly of the form treated in \cite{geronimo-wkbturningpoint-2004}, see the discussion in \cref{app:turning-point}. As previously discussed, following \cite{geronimo-wkbturningpoint-2004}, we treat the difference equation with coefficients $a_1,b_1$ as a perturbation of a simpler difference equation with coefficients $\tilde{a},\tilde{b}$. Additionally, the term $h^\alpha z$ in $L^{h,z}$ will also be treated as a perturbation in the same manner, significantly simplifying our analysis.
\begin{definition}
    \label{def:turning-point-easy}
    For $0\leq h<\Lambda^{1/\delta}$ define the difference operator $\tilde{L}^{h}$ 
    by 
    \begin{equation}
        (\tilde{L}^{h} \psi)(x) = \tilde{a}(x,h) \psi(x+h) -\tilde{b}(x,h)\psi(x) + \tilde{a}(x-h,h) \psi(x-h) \, , 
    \end{equation}
    where 
    \begin{equation}
        \tilde{a}(x,h) = x^\alpha\left(1+\frac{c_ah}{x}\right) \, , \quad \tilde{b}(x,h) = x^\beta\left(1+\frac{c_fh}{x}\right) \, .
    \end{equation}
\end{definition}
The idea is to treat $L^{h,z}$ as a perturbation of the simpler operator $\tilde{L}^h$ introduced above, and establish error bounds for the difference of the corresponding solutions. The following lemma will be needed:
\begin{lemma}
    \label{lem:turning-point-diff-bounds}
    For $0<h<\Lambda^{1/\delta}$ and $z \in \Omega$, let $a_1,\tilde{a},b_1,\tilde{b}$ be given as in \cref{lem:geronimo,def:turning-point-easy}.
    Moreover, define 
    \begin{align}
        D^{h,z,a}_n & = \frac{\tilde{a}(nh,h)}{\tilde{a}(nh-h,h)}-\frac{a_1(nh,h)}{a_1(nh-h,h)} \, , \\
        D^{h,z,b}_n & = \frac{\tilde{b}(nh,h)}{\tilde{a}(nh-h,h)}-\frac{b_1(nh,h)-h^\alpha z}{a_1(nh-h,h)} \, .    
    \end{align}
    Then there exists $C>0$ such that, for all $z\in \Omega$, $n \in K_N(\lambda)$ and $0<\lambda < \Lambda$, 
    \begin{equation}
        |D^{h,z,a}_n| \leq C h^{4/3} \, , \quad |D^{h,z,b}_n| \leq C h^{4/3} \, , 
    \end{equation}
    where $h = \lambda^{1/\delta}$.
\end{lemma}
\begin{proof}
    Recall that $n \in K_N(\lambda)$ implies $n \geq N_0$ and $x=nh \in K_x =[x_1,x_4]$.
    For the first statement, from \cref{lem:geronimo} and using \cref{hyp:conditions-an-fn} we get the following for all $x\in K_x$:
    \begin{align}
        D^{h,z,a}_n & = \frac{\tilde{a}(x,h)}{\tilde{a}(x-h,h)} - \frac{a_1(nh,h)}{a_1(nh-h,h)} \\
        & = \frac{1+\frac{c_a h}{x}}{1+\frac{c_ah}{x-h}}-\frac{1+\frac{c_ah}{x}+O(h^2)}{1+\frac{c_ah}{x-h}+O(h^2)} = O(h^2)
    \end{align}
    uniformly in $x \in K_x$, whence the first inequality follows as $x=nh \in K_x$ for $ n \in K_N(\lambda)$.
    
    For the second inequality, we start with $(x-h)^\alpha = x^\alpha(1+O(h))$ and
    \begin{equation}
        \left|\frac{h^\alpha z}{a_1(x-h,h)}\right|  = |h^\alpha z|\frac{1}{(x-h)^\alpha\left(1+\frac{c_ah}{x} + O(h^2)\right)} = \frac{h^\alpha |z|}{x^\alpha} O(1) \, .
    \end{equation}
    As $z\in \Omega$ and $x \in K_x$ range in compact sets, there exist $C_2$ such that 
    \begin{equation}
         \left|\frac{h^\alpha z}{a_1(x,h-h)}\right| \leq \frac{C_2}{x^\alpha}h^\alpha \leq  \frac{C_2}{x_1^\alpha}h^\alpha \, 
    \end{equation}
    for all $x \in K_x$ and hence for all $n \in K_N(\lambda)$.
    Since $\alpha > \nicefrac{4}{3}$ per \cref{hyp:conditions-an-fn}, this quantity satisfies the required estimate.
    Finally, consider
    \begin{equation}
        \frac{\tilde{b}(x,h)}{\tilde{a}(x-h,h)} - \frac{b_1(x,h)}{a_1(x-h,h)} = \frac{x^\beta\left(1+\frac{c_fh}{x}\right)}{(x-h)^\alpha\left(1+\frac{c_ah}{x} \right)}-\frac{x^\beta\left(1+\frac{c_fh}{x} + O(h^2)\right)}{(x-h)^\alpha\left(1+\frac{c_ah}{x} + O(h^2)\right)} = x^\delta O(h^2) \, .
    \end{equation}
    Thus, there is $C_3 > 0$ such that 
    \begin{equation}
        \left|\frac{\tilde{b}(x,h)}{\tilde{a}(x-h,h)} - \frac{b_1(x,h)}{a_1(x-h,h)}\right| \leq C_3 x^\delta h^2 \leq C_3 x_4^\delta h^2 \, , 
    \end{equation}
    and the claim follows as $x=nh \in K_x$ for $n \in K_N(\lambda)$.
\end{proof}

We now introduce a variant of the \textit{Langer transform}, which plays a central role in the asymptotic analysis of recurrence relations, cf.~\cite{geronimo-wkbturningpoint-2004,fedotov-complexwkbmethod-2019}. Although the terminology originates from the differential equation literature, where the Langer transform is more commonly used and takes a slightly different form, we adopt here the convention standard in the literature on recurrence relations. We again refer to \cref{app:turning-point} for the related results needed in the remainder of this section.

To this end, we first introduce an auxiliary function, cf.~\cref{eq:app-qfun}:
\begin{lemma}
    \label{lem:qfun}
    There exists $\Lambda>0$ so that the function $q:K_x\times[0,\Lambda^{1/\delta}]\rightarrow\rnum$ defined by 
    \begin{equation}
        q(x,h) = \frac{\tilde{b}(x,h)}{2 \tilde{a}(x-h/2,h)} \, ,
    \end{equation}
    with $\tilde{a},\tilde{b}$ from \cref{def:turning-point-easy}, 
    is analytic and strictly positive on $K_x\times[0,\Lambda^{1/\delta}]$, and the equation $q(x,h)=1$ has a unique solution $x_0(h) \in (x_2,x_3) \subset K_x$ for every $h\leq \Lambda^{1/\delta}$.
    Moreover, denoting by $q'(x,h)$ its partial derivative with respect to $x$, $q'(x,h)$ is strictly positive and uniformly bounded from above on $K_x\times[0,\Lambda^{1/\delta}]$.
\end{lemma}

\begin{proof}
    Let $\Lambda_1 > 0$ such that, for all $h \in [0,\Lambda_1^{1/\delta}]$ and $x \in K_x = [x_1,x_4]$, 
    \begin{equation} 
         x-h/2 \geq \nicefrac{1}{2} \,, \quad h x^{-1} |c_f| \leq \nicefrac{1}{2}\, , \quad h (x-h/2)^{-1}|c_a| \leq \nicefrac{1}{2} \, .
    \end{equation}    
    It follows that $\tilde{a}(x-h/2,h)$ and $\tilde{b}(x,h)$ are strictly positive on $K_x \times [0,\Lambda_1^{1/\delta}]$, and as both $\tilde{a},\tilde{b}$ are analytic, the first claim follows.
    A direct computation yields
    \begin{align} 
        2q'(x,h)& =-\frac{h c_f x^{\beta -2} \left(x-\frac{h}{2}\right)^{-\alpha }}{\frac{h c_a}{x-\frac{h}{2}}+1}+\frac{\beta  x^{\beta -1} \left(x-\frac{h}{2}\right)^{-\alpha } \left(\frac{h c_f}{x}+1\right)}{\frac{h
   c_a}{x-\frac{h}{2}}+1}\\
   & \quad+\frac{h c_a x^{\beta } \left(x-\frac{h}{2}\right)^{-\alpha -2} \left(\frac{h c_f}{x}+1\right)}{\left(\frac{h c_a}{x-\frac{h}{2}}+1\right)^2}-\frac{\alpha  x^{\beta }
   \left(x-\frac{h}{2}\right)^{-\alpha -1} \left(\frac{h c_f}{x}+1\right)}{\frac{h c_a}{x-\frac{h}{2}}+1} \\
   \label{proofeq:qfun-bounded-below}& = x^{\delta-1} \left(1-\frac{h}{2x}\right)^{-\alpha} \left(\beta - \alpha\left(1-\frac{h}{2x}\right)^{-1}\right)\frac{
   1+\frac{h c_f}{x}}{1+\frac{h c_a}{x-\frac{h}{2}}} \\
   & \quad + h \left( \frac{c_a x^{\beta } \left(x-\frac{h}{2}\right)^{-\alpha -2} \left(\frac{h c_f}{x}+1\right)}{\left(\frac{h c_a}{x-\frac{h}{2}}+1\right)^2} - \frac{c_f x^{\beta -2} \left(x-\frac{h}{2}\right)^{-\alpha }}{\frac{h c_a}{x-\frac{h}{2}}+1}\right) \, .
    \end{align}
    By \cref{hyp:conditions-an-fn}, $\beta > \alpha$ and hence there exists $c>0$ and $0< \Lambda_2 \leq \Lambda_1$ such that 
    \begin{equation}
        1-\frac{h}{2x} > c\, , \quad \beta - \alpha\left(1-\frac{h}{2x}\right)^{-1} > c
    \end{equation}
    for all $(x,h) \in K_x \times [0,\Lambda_2^{1/\delta}]$.
    Hence the first term in \cref{proofeq:qfun-bounded-below} is strictly positive on $K_x \times [0,\Lambda_2^{1/\delta}]$.
    The second term is proportional to $h$, and thus there exists $0 < \Lambda_3 \leq \Lambda_2$ such that $q'(x,h)$ is strictly positive on $K_x \times [0,\Lambda_3^{1/\delta}]$.    
    As $q'(x,h)$ is continuous in both $x$ and $h$, it is uniformly bounded from above on $K_x \times [0,\Lambda^{1/\delta}]$.

    Finally, we write
    \begin{equation}
        q(x,h) = \frac{\tilde{b}(x,h)}{2 \tilde{a}(x-h/2,h)} = x^\delta \frac{1+\frac{c_fh}{x}}{\left(1-\frac{h}{2x}\right)^{\alpha} \left(1+\frac{c_ah}{x-h/2}\right)}\, , 
    \end{equation}
    and thus there exists $0< \Lambda \leq \Lambda_3$ such that $q(x_2,h) < 1$ and $q(x_3,h) > 1$ for all $0 \leq h \leq \Lambda^{1/\delta}$, cf.~\cref{def:asymptotics-regions}.
    Together with the monotonicity of $q(x,h)$ in $x$, this also implies that there exists a unique solution $x_0(h) \in [x_2,x_3] \subset K_x$ of $q(x_0(h),h) = 1$.
\end{proof}
Consequently, the technical assumptions necessary to apply \cite{geronimo-wkbturningpoint-2004}, cf.~\cref{hyp:app-qfun}, are satisfied:
\begin{lemma}
    \label{lem:qfun-hypothesis}
    For $(x,h) \in K_x \times [0,\Lambda^{1/\delta}]$, let $q(x,h)$ be defined as in \cref{lem:qfun}.
    Then $q(x,h)$, $\tilde{a}(x,h)$ and $\tilde{b}(x,h)$ satisfy \cref{hyp:app-qfun}.
\end{lemma}
\begin{proof}
    By \cref{lem:qfun}, conditions (i)--(iii) are fulfilled.
    By construction, $\tilde{a}(x,h) > 0$ and $\tilde{b}(x,h) \in \rnum$.
    Finally, as $q'(x,h)$ is strictly positive and $q(x_0(h),h) = 1$, it follows that $q(x,h) < 1$ for $x < x_0(h)$ and $q(x,h) > 1$ for $x>x_0(h)$.
\end{proof}
\Cref{lem:app-langer-fedotov,lem:app-hypothesis} allow us to apply the results of~\cite{geronimo-wkbturningpoint-2004,fedotov-complexwkbmethod-2019}. We can now introduce the version of the Langer transform used in the recurrence relation literature; following the conventions and notation of Geronimo et al.~\cite{geronimo-wkbturningpoint-2004}, cf.~\cref{def:app-langer}:
\begin{definition}[Langer transform]
    \label{def:langer-transform}
    Let $q(x,h)$ and $x_0(h)$ be given as in \cref{lem:qfun}.
    Then the \emph{Langer transform} $\xi:K_x\times[0,\Lambda^{1/\delta}]\rightarrow\cnum$ is the function defined by
     \begin{align}
        \xi(x,h) & = -\left(\frac{3}{2} \int_{x}^{x_0(h)} \arccos(q(y,h)) \dl y \right)^{2/3} \quad x<x_0(h) \,, \\
        \xi(x,h) & = \left(\frac{3}{2} \int_{x_0(h)}^{x} \arccosh(q(y,h)) \dl y \right)^{2/3} \quad x\geq x_0(h) \,,
    \end{align}
where $\arccos$ is the principal branch of the arccosine function.
\end{definition}

\begin{remark}
    Differently from~\cite{geronimo-wkbturningpoint-2004}, our definition of $\xi(x,h)$ is independent of $z$ as we are treating the additional $z$-dependent term $h^\alpha z$ as a perturbation of $\tilde{L}^{h}$ instead of including it in the difference equation.
    This will enable us to apply some of the results in~\cite{fedotov-complexwkbmethod-2019}, thereby simplifying the proof.
\end{remark}
We refer to \cref{app:turning-point}, cf.~\cref{lem:app-hypothesis}, for some properties of the Langer transform.
Next, following \cite{geronimo-wkbturningpoint-2004,fedotov-complexwkbmethod-2019} (cf.~\cref{eq:app-weightfunction}), we introduce another auxiliary function $g(x,h)$ which will allow us to construct approximate solutions of the recurrence relation of the form $g(x,h) \Ai_j(h^{-2/3}\xi(x,h))$ 
\begin{definition}
    \label{def:afunc}
    Let $\xi$ be the Langer transform from \cref{def:langer-transform}. We define the functions $A,g:K_x\times[0,\Lambda^{1/\delta}]\to\cnum$ by
    \begin{equation}
        A(x,h)
        =
        \xi(x,h)^{1/4}
        \sinh\Bigl(\sqrt{\xi(x,h)}\,\xi'(x,h)\Bigr)^{-1/2},
    \end{equation}
    \begin{equation}
        g(x,h)
        =
        A(x,h)\,\tilde{a}(x-h/2,h)^{-1/2}.
    \end{equation}
    Here, all fractional powers are taken with respect to the principal branch.
\end{definition}
Thus, the function $g(x,h)$ carries the asymptotic factor $x^{-\alpha/2}$ required for the construction of approximate solutions to \cref{eq:generalized-eigenvalue-lambda}. Following Geronimo et al.~\cite[Eqs.~4.8--4.9]{geronimo-wkbturningpoint-2004}, cf.~\cref{eq:app-def-approx-solution}, we now construct approximations to the recessive and dominant solutions of \cref{eq:generalized-eigenvalue-lambda}. As anticipated, these approximations are built from functions of the form
\begin{equation}
g(x,h)\Ai_j(h^{-2/3}\xi(x,h)).
\end{equation}
More precisely, we seek solutions of \cref{eq:generalized-eigenvalue-lambda} of the form
\begin{equation}
    \label{eq:turning-point-solution}
    u^{\recessSol,h,z}_n
    =
    (-1)^n\psi^{\recessSol}(nh,h)
    +
    w^{\recessSol}(nh,h)\sigma^{\recessSol,h,z}_n,
\quad
    u^{\domSol,h,z}_n
    =
    (-1)^n\psi^{\domSol}(nh,h)
    +
    w^{\domSol}(nh,h)\sigma^{\domSol,h,z}_n,
\end{equation}
where $\psi^{\recessSol}$ and $\psi^{\domSol}$ are explicit Airy-type approximations, $\sigma^{\recessSol,h,z}_n$, $\sigma^{\domSol,h,z}_n$ are error terms, and the functions $w^{\recessSol},w^{\domSol}$ provide suitable weights for controlling the error terms. Their precise definition is given below:
\begin{definition}
    \label{def:approximate-solutions}
    Let $\Ai_j$, $j\in\{0,1,2\}$, be the functions defined by
    \begin{equation}
    \label{eq:def-complex-airy0}
    \Ai_j(x) = \Ai\left(\omega^j x\right)\, , \quad \omega = \e^{-\iu 2 \pi/3}\, , \quad j \in \znum_3 \, .
\end{equation}
where $\Ai$ is the Airy function of the first kind (\cref{eq:airy}), $\tilde{w}_1$ be the helper function associated with the Hankel function of the first kind (\cref{def:hankel-helper}), and $g$ the function from \cref{def:afunc}. We define
    \begin{alignat}{2}
        \psi^{\recessSol}(x,h) & = g(x,h) \Ai_0(h^{-2/3}\xi(x,h)) \, , \quad & \psi^{\domSol}(x,h) & = g(x,h) \Ai_1(h^{-2/3}\xi(x,h))\,, \\ 
        w^{\recessSol}(x,h) & = g(x,h) \tilde{w}_1(h^{-2/3}\xi(x,h)) \, , \quad & w^{\domSol}(x,h) & = \psi^{\domSol}(x,h)\,. 
    \end{alignat}
\end{definition}
As we will see in the next proposition, the function $w^{\recessSol}$ is introduced to control the deviation between the approximate solution $\psi^{\recessSol}$ and the corresponding exact solution $u^{\recessSol,h,z}$. In the recessive case, this is necessary because the Airy function $\Ai_0(x)$ vanishes for some $x<0$, preventing one from directly estimating relative errors in terms of $\Ai_0$ itself. The following proposition makes this statement precise by quantifying the deviation between the approximate and exact solutions.
\begin{proposition}
    \label{prop:turning-point}
    Let $z \in \Omega$ and $h>0$. There exists two linearly independent solutions $u^{\recessSol,h,z}$ and $u^{\domSol,h,z}$ of \cref{eq:generalized-eigenvalue-lambda} in the form~\eqref{eq:turning-point-solution},
    with $\psi^{\recessSol},\psi^{\domSol}$ and $w^{\recessSol}$ and $w^{\domSol}$ as from \cref{def:approximate-solutions}, and $\sigma_n^{r,h,z},\sigma_n^{d,h,z}$ satisfy the following: there exist $C_2 > 0$ and $\Lambda > 0$ such that, for all $0 < \lambda \leq \Lambda$, $z \in \Omega$ and $N_1(\lambda) \leq n \leq N_4(\lambda)$ (cf.~\cref{def:asymptotics-regions}) the following estimates hold uniformly:
    \begin{equation}
        |\sigma^{\recessSol,h,z}_n| \leq C_2 h^{1/3} ,\quad |\sigma^{\domSol,h,z}_n| \leq C_2 h^{1/3}.
    \end{equation}
\end{proposition}
\begin{proof}
    Throughout this proof, let $\qSol \in \{\recessSol,\domSol\}$, and let $\Lambda > 0$ be small enough such that the estimates in \cref{lem:geronimo,lem:qfun-hypothesis} hold.
    As in \cref{lem:turning-point-op}, we make the ansatz $u^{\qSol,h,z}_n = (-1)^n v^{\qSol,h,z}_n$.
    Then, $u^{\qSol,h,z}$ is a solution of \cref{eq:generalized-eigenvalue-lambda} if and only if $(L^{h,z}v^{\qSol,h,z})_n = 0$ for all $n \geq 1$.

    We prove the statement by applying the turning-point theorem \cite[Theorem~4.4]{geronimo-wkbturningpoint-2004}, recalled in \cref{thm:app-geronimo-2}, to the operator $L^{h,z}$. Roughly speaking, this result asserts that, under suitable regularity assumptions, Airy-type approximate solutions give rise to genuine solutions of the associated difference equation with explicitly controlled error terms. In our setting, the resulting solutions correspond to the recessive and dominant solutions $v^{\recessSol,h,z}$ and $v^{\domSol,h,z}$ introduced above. The term $h^\alpha z$ is treated as a perturbative contribution in the sense of \cref{thm:app-geronimo-2}. More precisely, the solutions denoted by $U^{(1,h)}$ and $U^{(2,h)}$ in \cref{thm:app-geronimo-2} are identified with $v^{\recessSol,h,z}$ and $v^{\domSol,h,z}$, respectively.

    We therefore consider a family of difference equations depending on the parameter $z$, with coefficients in \cref{eq:app-difference-equation} given by
\begin{equation}
    A_1^z(x,h)=a_1(x,h),
    \qquad
    B_1^z(x,h)=-b_1(x,h)+h^\alpha z,
\end{equation}
where $a_1,b_1$ are defined in \cref{lem:geronimo}. In this formulation, the term $h^\alpha z$ appears as a perturbative contribution to the unperturbed equation associated with the coefficients $\tilde a(x,h)$ and $\tilde b(x,h)$ from \cref{def:turning-point-easy}. 
By \cref{lem:geronimo}, the functions $a_1(x,h)$ and $b_1(x,h)$ are continuous on $K_x\times[0,\Lambda^{1/\delta}]$, and hence so are $A_1^z$ and $B_1^z$. Furthermore, defining $q(x,h)$ as in \cref{lem:qfun}, \cref{lem:qfun-hypothesis,lem:app-hypothesis} ensure that the hypotheses of the turning-point theorem \cref{thm:app-geronimo-2} are satisfied.

We now verify the quantitative assumptions required to apply the turning-point theorem, \cref{thm:app-geronimo-2}.
To this end, we use the auxiliary quantities introduced in \cref{app:turning-point}, namely the positive propagator functions $G^{\qSol}(i,h)$ and the perturbative error terms $\hat\beta^{\qSol}(i,h)$ and $\hat\beta_1^{\qSol}(i,h)$, with $q=\recessSol,\domSol$.
Their precise definition is not important here; what matters is that they satisfy explicit uniform bounds. More precisely, by \cref{eq:app-beta-errors-estimate}, there exists $c_1>0$ such that
\begin{equation}
    \label{proofeq:turning-point-beta}
    |\hat{\beta}^{\qSol}(i,h)| \leq c_1 h^2,
    \qquad
    |\hat{\beta}_1^{\qSol}(i,h)| \leq c_1 h^2
\end{equation}
for all $i \in K_N(\lambda)$, $h>0$, and $\qSol=\recessSol,\domSol$. Furthermore, the propagator functions satisfy
\begin{equation}
    \label{proofeq:turning-point-green}
    \sum_{i=N_1(\lambda)}^{N_2(\lambda)} G^{\qSol}(i,h)
    \leq
    c_1 h^{-1},
    \qquad \qSol=\recessSol,\domSol,
\end{equation}
by \cref{lem:app-error-propagator}.

Following \cref{thm:app-geronimo-2}, we then define the cumulative error terms
\begin{equation}
    \label{proofeq:turningpoint-kbound}
    K^{\qSol}(i,h,z)
    =
    c_2 G^{\qSol}(i,h)\bigl(D^{h,z,a}_i+D^{h,z,b}_i\bigr),
\end{equation}
where $c_2$ is independent of $h,i,z$. Combining \cref{lem:turning-point-diff-bounds} with \cref{proofeq:turning-point-green}, we obtain
\begin{equation}
    \left|
    \sum_{i=N_1(\lambda)}^{N_2(\lambda)}
    K^{\qSol}(i,h,z)
    \right|
    \leq
    2Ch^{4/3}
    \sum_{i=N_1(\lambda)}^{N_2(\lambda)}
    G^{\qSol}(i,h)
    \leq
    2Cc_1 h^{1/3},
    \qquad \qSol=\recessSol,\domSol.
\end{equation}
    Finally, by \cref{proofeq:turning-point-beta,proofeq:turning-point-green}
    \begin{align}
        \label{proofeq:turningpoint-betabound}
        \sum_{i=N_1(\lambda)}^{N_2(\lambda)} G^{\qSol}(i,h) |\hat{\beta}^{\qSol}(i,h)| & \leq c_1^2 h \,  ,\\
        \label{proofeq:turningpoint-beta1bound}
        \sum_{i=N_1(\lambda)}^{N_2(\lambda)} G^{\qSol}(i,h) |\hat{\beta}_1^{\qSol}(i,h)| & \leq c_1^2 h \,  .  
    \end{align}
    
    We are now in a position to apply \cref{thm:app-geronimo-2}. This yields two linearly independent solutions $v^{\qSol,h,z}$, $\qSol=\recessSol,\domSol$, obtained as perturbations of the Airy-type approximations $\psi^{\qSol}$. More precisely,
\begin{equation}
    v^{\qSol,h,z}_n
    =
    \psi^{\qSol}(nh,h)
    +
    w^{\qSol}(nh,h)\tilde{\sigma}^{\qSol,h,z}_n
\end{equation}
for some error sequence $\tilde{\sigma}^{\qSol,h,z}_n$.

Furthermore, by \cref{thm:app-geronimo-2}, there exists $\tilde C>0$, independent of $z\in\Omega$ and $h>0$, such that
\begin{multline}
    |\tilde{\sigma}^{\qSol,h,z}_n|
    \leq
    \tilde C
    \sum_{i=N_1(\lambda)}^{N_2(\lambda)}
    \left(
        K^{\qSol}(i,h,z)
        +
        G^{\qSol}(i,h)|\hat\beta(i,h)|
    \right)
    \\
    \times
    \exp\left(
        \tilde C
        \sum_{i=N_1(\lambda)}^{N_2(\lambda)}
        \left(
            K^{\qSol}(i,h,z)
            +
            G^{\qSol}(i,h)|\hat\beta_1(i,h)|
        \right)
    \right).
\end{multline}
Combining this estimate with
\cref{proofeq:turningpoint-kbound,proofeq:turningpoint-betabound,proofeq:turningpoint-beta1bound}, we obtain
\begin{equation}
    |\tilde{\sigma}^{\qSol,h,z}_n|
    \leq
    \tilde C
    \bigl(
        2Cc_1 h^{1/3}
        +
        c_1^2 h
    \bigr)
    \exp\!\Bigl(
        \tilde C
        \bigl(
            2Cc_1 h^{1/3}
            +
            c_1^2 h
        \bigr)
    \Bigr),
\end{equation}
and therefore
\begin{equation}
    |\tilde{\sigma}^{\qSol,h,z}_n|
    \leq
    C_2 h^{1/3}
\end{equation}
for some constant $C_2>0$ and sufficiently small $h$. 
Finally, using the ansatz $u^{\qSol,h,z}_n = (-1)^n v^{\qSol,h,z}_n$, we obtain the claim.
\end{proof}

This completes the construction of the recessive and dominant solutions $u^{\recessSol,h,z},u^{\domSol,h,z}$ in the turning-point regime. We now establish several properties of these solutions that will be important later. To this end, we first derive a few auxiliary properties of the Langer transform.

\begin{lemma}
    \label{lem:langer-transform-bounds}
     Let $\xi(x,h)$ be the Langer transform introduced in \cref{def:langer-transform}, and let $N_1(\lambda)$, $N_2(\lambda)$ and $N_4(\lambda)$ be as in \cref{def:asymptotics-regions}.
    Then there exists $\Lambda>0$ and $c>0$ such that the following statements hold for all $0\leq h\leq \Lambda^{1/\delta}$:
    \begin{enumerate}[(i)]
        \item \label{item:langer-transform-bounds-xi2} There exists $\xi_2 < 0$ such that $\xi(nh,h) \leq \xi_2$ for all $n \leq N_2(\lambda)$ and $0 \leq h \leq \Lambda^{1/\delta}$;
        \item \label{item:langer-transform-bounds-xi4} There exists $\xi_4 > 0$ such that $\xi(N_4(\lambda) h,h)$, $\xi((N_4(\lambda)-1)h,h) \geq \xi_4$ for all $h \leq \Lambda^{1/\delta}$;
        \item \label{item:langer-transform-bounds-xi1} There exists $\xi_1 < 0$ such that $\xi(nh,h) \geq \xi_1$ for all $n \geq N_1(\lambda)$ and  $h \leq \Lambda^{1/\delta}$;
        \item \label{item:langer-transform-bounds-xi12tilde} There exists $\xi_1 \leq \tilde{\xi}_1 < \tilde{\xi}_2  \leq \xi_2$ such that, for all $0 \leq h \leq \Lambda^{1/\delta}$ and all intervals $I \subset [\tilde{\xi}_1,\tilde{\xi}_2]$ with $|I| \geq ch$, there exists $M \in \nnum$ with $N_1(\lambda) \leq M \leq N_2(\lambda)$ and $\xi(Mh,h) \in I$.
    \end{enumerate}
\end{lemma}
\begin{proof}
    By \cref{lem:qfun}, the turning point satisfies
    \begin{equation}
        x_2 < x_0(h) < x_4
    \end{equation}
    for all $h \in [0,\Lambda^{1/\delta}]$. Furthermore, \cref{lem:qfun-hypothesis,lem:app-hypothesis} imply that the Langer transform $\xi(x,h)$ is smooth in $(x,h)$, strictly negative for $x<x_0(h)$, strictly positive for $x>x_0(h)$, and monotone in $x$.

    Since $x_2<x_0(0)$ and $x_4>x_0(0)$, we obtain
    \begin{equation}
        \xi(x_2,0)<0,
        \qquad
        \xi(x_4,0)>0.
    \end{equation}
    By continuity of $\xi(x,h)$ in both variables, these inequalities remain uniformly valid for sufficiently small $h$. Recalling the definition of $N_2(\lambda)$ and $N_4(\lambda)$ from \cref{def:asymptotics-regions}, this proves \labelcref{item:langer-transform-bounds-xi2,item:langer-transform-bounds-xi4}. 
    Similarly, continuity of $\xi(x,h)$ together with the definition of $N_1(\lambda)$ yields \labelcref{item:langer-transform-bounds-xi1}.

    It remains to prove \labelcref{item:langer-transform-bounds-xi12tilde}.
    In the following, we write $\xi_h(x) = \xi(x,h)$ for fixed $h \in [0,\Lambda^{1/\delta}]$.
    Because of the strict monotonicity of $\xi(x,0)$ in $x$, every value
    \begin{equation}
        \xi \in [\xi(x_1,0),\xi(x_2,0)]
    \end{equation}
    is attained for some exactly one $x\in[x_1,x_2]$. Using continuity in $(x,h)$, we may therefore choose $\Lambda>0$ sufficiently small and constants
    \begin{equation}
        \xi_1 \leq \tilde\xi_1 < \tilde\xi_2 \leq \xi_2
    \end{equation}
    such that, for every $0\leq h\leq \Lambda^{1/\delta}$,
    \begin{equation}
        (\xi_h)^{-1}(\xi) \in [x_1,x_2] \quad \forall \xi \in [\tilde\xi_1,\tilde\xi_2] \, .
    \end{equation}
    Moreover, by \cref{lem:app-hypothesis}\labelcref{item:app-langer-analytic}, $\xi(x,h) \in C^\infty(K_x \times [0,\Lambda^{1/\delta}])$, and hence $\xi_h'(x)$ is uniformly bounded from above in $K_x \times [0,\Lambda^{1/\delta}]$.
    It follows that 
    \begin{equation}
        (\xi_h^{-1})'(\xi) = \left(\xi_h'(\xi_h^{-1}(\xi))\right)^{-1}
    \end{equation}
    is uniformly bounded from below by some $\tilde{c}>0$ on $[\tilde{\xi}_1,\tilde{\xi}_2]\times [0,\Lambda^{1/\delta}]$.
    Hereafter we set $c = 4 \tilde{c}^{-1}$.
    Then, given an interval $I \subset [\tilde{\xi}_1,\tilde{\xi}_2]$ with $|I| \geq ch$, it follows from the above discussion that $ \xi_h^{-1}(I) \subset [x_1,x_2]$.
    Using the lower bound of $(\xi_h^{-1})'(\xi)$ and the assumption $|I| \geq c h$, we obtain
    \begin{equation}
        \left|\xi_h^{-1}(I)\right| \geq \frac{1}{2} \tilde{c} |I| \geq 2 h 
    \end{equation}
    for small enough $h$.
    Thus there exists $M \in \nnum$ such that $Mh \in \xi_h^{-1}(I)$, which proves the claim.
\end{proof}

The solutions constructed in \cref{prop:turning-point} inherit the qualitative behavior of the corresponding Airy functions. In particular, the recessive solution exhibits the expected decay, while the dominant solution remains sufficiently large throughout the turning-point region. The following proposition quantifies these properties uniformly in $\lambda$, and will play a crucial role in matching the turning-point asymptotics with the oscillatory and exponential regimes considered later on.

\begin{proposition}
    \label{prop:properties-turningpoint-solutions}    
    Let $z \in \Omega$, $\lambda>0$, and set $h = \lambda^{1/\delta}$. Let $u^{\recessSol,h,z}$ and $u^{\domSol,h,z}$ be the recessive and dominant solutions of \cref{eq:generalized-eigenvalue-lambda} constructed in \cref{prop:turning-point}. Then there exist constants $\Lambda>0$ and $c_1,c_2>0$ such that the following statements hold for all $0<\lambda\leq \Lambda$, all $z\in\Omega$, and all $N_1(\lambda)\leq n\leq m:=N_4(\lambda)-1$:
    \begin{enumerate}[(i)]
        \item \label{item:pts-udratio-m}
        The dominant solution satisfies
        \begin{equation}
            \left|\frac{u^{\domSol,h,z}_{m+1}}{u^{\domSol,h,z}_m}\right| \geq \frac{7}{8}.
        \end{equation}

        \item \label{item:pts-urratio-n}
        The recessive solution obeys the uniform bound
        \begin{equation}
            |u^{\recessSol,h,z}_n|
            \leq c_1 (nh)^{-\alpha/2}.
        \end{equation}

        \item \label{item:pts-udratio-nm}
        The dominant solution satisfies
        \begin{equation}
            \left|
                \frac{u^{\domSol,h,z}_n}{u^{\domSol,h,z}_m}
            \right|
            \leq
            c_1 \left(\frac{m}{n}\right)^{\alpha/2}.
        \end{equation}

        \item \label{item:pts-udratio-nm-epsilon}
        For every $\epsilon>0$, there exists $\Lambda(\epsilon)<\Lambda$ such that
        \begin{equation}
            \left|
                \frac{u^{\domSol,h,z}_n}{u^{\domSol,h,z}_m}
            \right|
            \leq
            \epsilon
            \left(\frac{m}{n}\right)^{\alpha/2}
            h^{1/6}
        \end{equation}
        for all $n\leq N_2(\lambda)$ and all $0<\lambda<\Lambda(\epsilon)$.

        \item \label{item:pts-urratio-nm}
        For every $\lambda>0$, there exists $M\in\nnum$ with $N_1(\lambda)\leq M\leq N_2(\lambda)$ such that
        \begin{equation}
            |u^{\recessSol,h,z}_M|
            \geq
            c_2 (Mh)^{-\alpha/2} h^{1/6}\, .
        \end{equation}
    \end{enumerate}
\end{proposition}

\begin{proof}
    We set $h=\lambda^{1/\delta}$ and use the properties of the Langer transform collected in \cref{lem:app-hypothesis}. In particular, $\xi(x,h)$ is real-valued, smooth in $(x,h)$, and monotonically increasing in $x$, with $\xi(x,h)<0$ for $x<x_0(h)$ and $\xi(x,h)>0$ for $x>x_0(h)$.

    We begin with \labelcref{item:pts-udratio-m}.  
    Set
    \begin{equation}
        x_4(h)=N_4(\lambda)h,
        \qquad
        m=N_4(\lambda)-1.
    \end{equation}
    By definition, $x_4(h)\to x_4$ as $h\to0$. Using the representation from \cref{prop:turning-point}, we obtain
    \begin{equation}
        \left|\frac{u^{d,h,z}_{m+1}}{u^{\domSol,h,z}_m}\right|
        =
        \left|
        \frac{
            g(x_4(h),h)\Ai_1(h^{-2/3}\xi(x_4(h),h))
        }{
            g(x_4(h)-h,h)\Ai_1(h^{-2/3}\xi(x_4(h)-h,h))
        }
        \frac{
            (-1)^{m+1}+\sigma^{\domSol,h,z}_{m+1}
        }{
            (-1)^m+\sigma^{\domSol,h,z}_m
        }
        \right|.
    \end{equation}
    Since $g(x,h)$ is continuous and nonvanishing on $K_x\times[0,\Lambda^{1/\delta}]$,
    \begin{equation}
        \lim_{h\to0}
        \frac{g(x_4(h),h)}{g(x_4(h)-h,h)}
        =1.
    \end{equation}
    Moreover, $\xi(x,h)$ is real and increasing in $x$, and $\Ai_1$ is monotonically increasing on the positive real axis, cf.~\cref{lem:airy-properties}\labelcref{item:airy-monotone}. Hence, for every fixed $\varepsilon>0$,
    \begin{equation}
        \left|
        \frac{
            g(x_4(h),h)\Ai_1(h^{-2/3}\xi(x_4(h),h))
        }{
            g(x_4(h)-h,h)\Ai_1(h^{-2/3}\xi(x_4(h)-h,h))
        }
        \right|
        \geq 1-\varepsilon
    \end{equation}
    for $h$ sufficiently small.
    Furthermore, by \cref{prop:turning-point},
    \begin{equation}
        \left|
        \frac{
            (-1)^{m+1}+\sigma^{\domSol,h,z}_{m+1}
        }{
            (-1)^m+\sigma^{\domSol,h,z}_m
        }
        \right|
        \geq
        \frac{1-C_2 h^{1/3}}{1+C_2 h^{1/3}}.
    \end{equation}
    Combining the previous estimates and taking $h$ small enough yields
    \begin{equation}
        \left|u^{\domSol,h,z}_{m+1}\big/u^{\domSol,h,z}_m\right|
        \geq \frac78,
    \end{equation}
    proving \labelcref{item:pts-udratio-m}.

    We now prove \labelcref{item:pts-urratio-n}.  
    The prefactor $A(x,h)=\tilde{a}(x-h/2,h)^{1/2}g(x,h)$ from \cref{def:afunc} is uniformly bounded on $K_x\times[0,\Lambda^{1/\delta}]$ (cf.~\cref{lem:app-hypothesis}\labelcref{item:app-gfunc-bound}), while $\xi(x,h)$ remains real-valued.
    Consequently, the Airy function $\Ai_0(h^{-2/3}\xi(x,h))$ and the auxiliary function $\tilde w_1(h^{-2/3}\xi(x,h))$ remain uniformly bounded on this region, cf.~\cref{lem:airy-properties}\labelcref{item:airy-bounded} and \cref{lem:hankel-properties}\labelcref{item:hankel-bounded}. Hence there exists $\tilde c_0>0$ such that
    \begin{align}
        |\psi^{\recessSol,h,z}(x,h)|
        &\leq
        \frac{|A(x,h)|}{\tilde a(x-h/2,h)^{1/2}}
        \left|\Ai_0(h^{-2/3}\xi(x,h))\right|
        \leq
        \frac{\tilde c_0}{x^{\alpha/2}},
        \\
        |w^{\recessSol,h,z}(x,h)|
        &\leq
        \frac{|A(x,h)|}{\tilde a(x-h/2,h)^{1/2}}
        \left|\tilde w_1(h^{-2/3}\xi(x,h))\right|
        \leq
        \frac{\tilde c_0}{x^{\alpha/2}},
    \end{align}
    where we used the asymptotic behavior of $\tilde a(x,h)$ from \cref{def:turning-point-easy}.
    Using again \cref{prop:turning-point}, we obtain
    \begin{equation}
        |u^{\recessSol,h,z}_n|
        \leq
        |\psi^{\recessSol,h,z}(nh,h)|
        +
        C h^{1/3}|w^{\recessSol,h,z}(nh,h)|
        \leq
        \frac{\tilde c_0(1+Ch^{1/3})}{(nh)^{\alpha/2}},
    \end{equation}
    proving the claim.

    We next establish \labelcref{item:pts-udratio-nm,item:pts-udratio-nm-epsilon}.  
    As before, set $x_4(h)=N_4(\lambda)h$. Since $A(x,h)$ is uniformly bounded above and below on $K_x\times[0,\Lambda^{1/\delta}]$ (cf.~\cref{lem:app-hypothesis}\labelcref{item:app-gfunc-bound}), there exists $\tilde c_1>0$ such that
    \begin{equation}
        \label{proofeq:turning-point-bounds-aratio-polished}
        \left|
        \frac{A(x,h)}{A(y,h)}
        \right|
        \leq
        \tilde c_1
    \end{equation}
    for all $x,y\in K_x$.
    Moreover, the asymptotics of $\tilde a(x,h)$ imply the existence of $\tilde d_0>0$ such that
    \begin{equation}
        \label{proofeq:turningpoint-bounds-gratio-polished}
        \frac{n^{\alpha/2}}{m^{\alpha/2}}
        \frac{
            \tilde a(mh-h/2,h)^{1/2}
        }{
            \tilde a(nh-h/2,h)^{1/2}
        }
        \leq
        \tilde d_0
    \end{equation}
    for all $N_1(\lambda) \leq n \leq m$ and $h \in [0,\Lambda^{1/\delta}]$.
    Combining these bounds with the representation from \cref{prop:turning-point}, we obtain
    \begin{align}
        \frac{n^{\alpha/2}}{m^{\alpha/2}}
        \left|
        \frac{u^{\domSol,h,z}_n}{u^{\domSol,h,z}_m}
        \right|
        &\leq
        \tilde d_0\tilde c_1
        \left|
        \frac{
            \Ai_1(h^{-2/3}\xi(nh,h))
        }{
            \Ai_1(h^{-2/3}\xi(mh,h))
        }
        \right|
        \frac{1+Ch^{1/3}}{1-Ch^{1/3}}
        \\
        \label{proofeq:turning-point-bounds-dratio-polished} & \leq
        \tilde c_2
        \left|
        \frac{
            \Ai_1(h^{-2/3}\xi(nh,h))
        }{
            \Ai_1(h^{-2/3}\xi(mh,h))
        }
        \right|
    \end{align}
    for some $\tilde c_2>0$ and $h$ sufficiently small.

    Since $\xi(nh,h)\leq \xi(mh,h)$ for $n\leq m$, monotonicity of $\Ai_1$ on the real axis (cf.~\cref{lem:airy-properties}\labelcref{item:airy-monotone}) implies that the quotient of Airy functions is bounded by $1$. Therefore,
    \begin{equation}
        \left|
        \frac{u^{\domSol,h,z}_n}{u^{\domSol,h,z}_m}
        \right|
        \leq
        c_1\left(\frac{m}{n}\right)^{\alpha/2},
    \end{equation}
    establishing \labelcref{item:pts-udratio-nm}.

    Assume now that $n\leq N_2(\lambda)$, and let $c_3 > 0$ be arbitrary. By \cref{lem:langer-transform-bounds},
    \begin{equation}
        \xi(nh,h)\leq \xi_2<0,
        \qquad
        \xi(mh,h)\geq \xi_4>0.
    \end{equation}
    Using the standard asymptotics of $\Ai_1$ on the positive and negative real axes (cf.~\cref{lem:airy-properties}), we obtain
    \begin{align}
        |\Ai_1(h^{-2/3}\xi(nh,h))|
        &\leq
        \frac{\tilde c_6}{(h^{-2/3}|\xi(nh,h)|)^{1/4}}
        \leq
        \frac{\tilde c_6}{|\xi_2|^{1/4}}h^{1/6},
        \\
        |\Ai_1(h^{-2/3}\xi(mh,h))|
        &\geq c_3
    \end{align}
    for some $\tilde{c}_6>0$ and all sufficiently small $h$. 
    Combining this with \cref{proofeq:turning-point-bounds-dratio-polished} and choosing $c_3$ large enough proves \labelcref{item:pts-udratio-nm-epsilon}.

    Finally, we prove \labelcref{item:pts-urratio-nm}.  
    Let $M\in[N_1(\lambda),N_2(\lambda)]$. Using the representation from \cref{prop:turning-point} and the asymptotics of $\tilde{a}(x,h)$ (cf. \cref{def:turning-point-easy}),
    \begin{equation}
        (Mh)^{\alpha/2}|u^{\recessSol,h,z}_M|
        \geq
        \tilde{c}_8|A(Mh,h)|
        \left|
        \tilde w_1(h^{-2/3}\xi(Mh,h))
        \left(
            (-1)^M\chi(h^{-2/3}\xi(Mh,h))
            +
            \sigma^{\recessSol,h,z}_M
        \right)
        \right|,
    \end{equation}
    for some $\tilde{c}_8>0$, where
    \begin{equation}
        \chi(x)=\frac{\Ai_0(x)}{\tilde w_1(x)}.
    \end{equation}

    Let $\tilde\xi_1<\tilde\xi_2<0$ and $\tilde{c}_7>0$ be as in \cref{lem:langer-transform-bounds}. By the oscillatory behavior of $\chi$ on the negative real axis (cf.~\cref{lem:airy-hankel-ratio}), for every sufficiently small $h>0$ there exists an interval $I_{\xi,h} \subset [\tilde\xi_1,\tilde\xi_2]$ with $|I_{\xi,h}| \geq \tilde{c}_7 h$ such that
    \begin{equation}
        |\chi(h^{-2/3}\xi)|
        \geq
        \frac12 \quad \forall \xi \in I_{\xi,h} \, .
    \end{equation}

    By \labelcref{item:langer-transform-bounds-xi12tilde} of \cref{lem:langer-transform-bounds}, there exists $M\in[N_1(\lambda),N_2(\lambda)]$ such that $\xi(Mh,h)\in I_{\xi,h}$. Since $A(x,h)$ is uniformly bounded away from zero, \cref{prop:turning-point} yields
    \begin{equation}
        (Mh)^{\alpha/2}|u^{\recessSol,h,z}_M|
        \geq
        \frac{\tilde c_3}{2}
        |\tilde w_1(h^{-2/3}\xi(Mh,h))|
        (1-Ch^{1/3})
    \end{equation}
    for $h$ sufficiently small, and therefore
    \begin{equation}
        (Mh)^{\alpha/2}|u^{\recessSol,h,z}_M|
        \geq
        \tilde c_4
        |\tilde w_1(h^{-2/3}\xi(Mh,h))|.
    \end{equation}

    Finally, since $M\geq N_1(\lambda)$, \cref{lem:langer-transform-bounds} gives $\xi(Mh,h)\geq \xi_1$. Using the asymptotics of $\tilde w_1$ on the negative real axis, we conclude that
    \begin{equation}
        |\tilde w_1(h^{-2/3}\xi(Mh,h))|
        \geq
        \frac{\tilde c_5}{|h^{-2/3}\xi(Mh,h)|^{1/4}}
        \geq
        \frac{\tilde c_5}{|\xi_1|^{1/4}}h^{1/6},
    \end{equation}
    proving the claim.
\end{proof}

\subsection{Third region: Poincaré's theorem}
\label{sec:poincare}
We now turn to the tail region of the recessive solution, namely $n \geq N_4(\lambda)$.  
In this regime, it is convenient to renormalize the generalized eigenvalue equation so as to isolate its asymptotically constant part. To this end, we introduce the following change of variables.

\begin{lemma}
    \label{lem:asymptotics-an-transform}
    Let
    \begin{equation}
        u_n = (-1)^n a_n^{-1/2} v_n
    \end{equation}
    for all $n \in \nnum$.
    Then $u$ fulfills \cref{eq:generalized-eigenvalue-lambda} if and only if $v$ fulfills
    \begin{equation}
        \label{eq:asymptotics-an-transform-poincare}
        v_{n+1}
        -
        \left(\lambda f_n - z\right)
        \frac{a_{n+1}^{1/2}}{a_n^{3/2}}
        v_n
        +
        \frac{a_{n-1}^{1/2}a_{n+1}^{1/2}}{a_n}
        v_{n-1}
        =
        0
        \qquad \forall n \geq 1 \, .
    \end{equation}
\end{lemma}
\begin{proof}
    Plugging the ansatz $u_n = (-1)^n a_n^{-1/2} v_n$ into \cref{eq:generalized-eigenvalue-lambda}, we obtain
    \begin{equation}
        a_n a_{n+1}^{-1/2}(-1)^{n+1}v_{n+1} + (\lambda f_n-z) a_n^{-1/2} (-1)^nv_n + a_{n-1}^{1/2} (-1)^{n-1} v_{n-1} = 0 \quad \forall n \geq n \, .
    \end{equation}
    Multiplying by $(-1)^{n+1} a_{n+1}^{1/2} a_n^{-1}$, we obtain the result.
\end{proof}
We now exploit the asymptotic properties stated in \cref{hyp:conditions-an-fn} to derive asymptotic solutions of the transformed equation~\eqref{eq:asymptotics-an-transform-poincare}.
\begin{lemma}
    \label{lem:asymptotic-solutions}
    Set $\delta = \beta-\alpha$.
    There exists two linearly independent solutions $\tilde{v}^{\recessSol}$ and $\tilde{v}^{\domSol}$ of \cref{eq:asymptotics-an-transform-poincare} such that
    \begin{equation}
        \frac{\tilde{v}^{\recessSol}_{n+1}}{\tilde{v}^{\recessSol}_n} \sim \frac{1}{\lambda} n^{-\delta} \, , \qquad  \frac{\tilde{v}^{\domSol}_{n+1}}{\tilde{v}^{\domSol}_n} \sim \lambda n^{\delta} \, . 
    \end{equation}
\end{lemma}
\begin{proof}
We rewrite \cref{eq:asymptotics-an-transform-poincare} as a standard second-order linear difference equation of the form
\begin{equation}
    v_{n+2} + p_1(n)\, v_{n+1} + p_2(n)\, v_n = 0,
\end{equation}
where, by direct inspection of the coefficients in \cref{eq:asymptotics-an-transform-poincare}, we have the asymptotic behavior
\begin{equation}
    p_1(n) \sim -\lambda n^{\delta},
    \qquad
    p_2(n) \sim 1.
\end{equation}
The existence of two linearly independent solutions of \cref{eq:asymptotics-an-transform-poincare} with the claimed asymptotics follows directly from a generalization of the Poincaré--Perron theorem, cf.~\cite[Theorem~8.35]{elaydi-introductiondifferenceequations-2005}.
\end{proof}
The following lemma shows that among all solutions of \cref{eq:asymptotics-an-transform-poincare}, only the recessive solution $\tilde{v}^{\recessSol}$ is compatible with square-summability, whereas every other linearly independent solution grows too fast as $n \to \infty$.
\begin{lemma}
    \label{lem:which-asymptotic-solution}
    Let $u^{\lambda,z}$ be a square-summable solution of \cref{eq:generalized-eigenvalue-lambda}, and let $\tilde{v}^{\recessSol},\tilde{v}^{\domSol}$ be the two linearly independent solutions from \cref{lem:asymptotic-solutions}.
    Then there exists $\tilde{A} \in \cnum$ such that
    \begin{equation}
        u^{\lambda,z}_n = \tilde{A}(-1)^n a_n^{-1/2} \tilde{v}^{\recessSol}_n \quad \forall n \in \nnum \, .
    \end{equation}
\end{lemma}
\begin{proof}
We start from the ansatz
\begin{equation}
    u^{\lambda,z}_n = (-1)^n a_n^{-1/2} v_n,
\end{equation}
for some sequence $v$. By \cref{lem:asymptotics-an-transform}, $v$ satisfies \cref{eq:asymptotics-an-transform-poincare}. Hence, by \cref{lem:asymptotic-solutions}, there exist two linearly independent solutions $\tilde{v}^{\recessSol},\tilde{v}^{\domSol}$ such that
\begin{equation}
    \frac{\tilde{v}^{\recessSol}_{n+1}}{\tilde{v}^{\recessSol}_n} \sim \frac{1}{\lambda} n^{-\delta},
    \qquad
    \frac{\tilde{v}^{\domSol}_{n+1}}{\tilde{v}^{\domSol}_n} \sim \lambda n^{\delta}.
\end{equation}
In particular, since $\delta>0$, there exists $N\in\nnum$ such that for all $n\ge N$,
\begin{equation}
    \left|\frac{\tilde{v}^{\recessSol}_{n+1}}{\tilde{v}^{\recessSol}_n}\right| \le \frac{1}{2},
    \qquad
    \left|\frac{\tilde{v}^{\domSol}_{n+1}}{\tilde{v}^{\domSol}_n}\right| \ge 2.
\end{equation}
By iteration, this implies
\begin{equation}
    |\tilde{v}^{\recessSol}_n| \le |\tilde{v}^{\recessSol}_N|\,2^{-(n-N)},
    \qquad
    |\tilde{v}^{\domSol}_n| \ge |\tilde{v}^{\domSol}_N|\,2^{n-N}.
\end{equation}

Since $v_n = \tilde{A}\tilde{v}^{\recessSol}_n + \tilde{B}\tilde{v}^{\domSol}_n$, we obtain
\begin{equation}
    |v_n| \ge |\tilde{B}|\,|\tilde{v}^{\domSol}_N|\,2^{n-N} - |\tilde{A}|\,|\tilde{v}^{\recessSol}_N|\,2^{-(n-N)}.
\end{equation}
Multiplying by $a_n^{-1/2} \sim n^{-\alpha/2}$ shows that if $\tilde{B}\neq 0$, then $|u^{\lambda,z}_n|$ grows at least exponentially in $n$, modulo polynomial prefactors, and in particular is not square-summable.

Since $u^{\lambda,z} \in \ell^2(\nnum)$ by assumption, we must have $\tilde{B}=0$, which proves the claim.
\end{proof}

To obtain explicit bounds on the square-summable solution $u^{\lambda,z}_n$ for all $n \geq N_3(\lambda)$, we now perform a refined error analysis.

The solution $\tilde{v}^{\recessSol}$ of \cref{eq:asymptotics-an-transform-poincare} is recessive as $n \to \infty$, and therefore the relevant asymptotic information is naturally encoded at infinity. In order to propagate this information back to the finite scale $n = N_3(\lambda)$ in a controlled way, we reformulate the second-order difference equation in terms of a nonlinear first-order recursion for the ratio of successive values. This leads to a discrete Riccati-type equation, cf.~also \cite{teschl-jacobioperatorscompletely-1999}:
\begin{lemma}
    \label{lem:riccati-forward}
    Let $v$ be a solution of the recurrence relation~\eqref{eq:asymptotics-an-transform-poincare}, and define
    \begin{equation}
        r_n = \frac{v_n}{v_{n-1}}.
    \end{equation}
     Then the sequence $(r_n)_{n\in\nnum}$ satisfies a recursion in the form $r_n = \mathcal{R}_n(r_{n+1})$, where the map $\mathcal{R}_n$ is defined by
    \begin{equation}
     \mathcal{R}_n(x) = \frac{a_{n-1}^{1/2}a_{n+1}^{1/2}}{a_n} \frac{1}{\left(\lambda \frac{f_n}{a_n}-\frac{z}{a_n}\right)a_{n+1}^{1/2}a_n^{-1/2}-x} \, .
    \end{equation}
\end{lemma}
\begin{proof}
    Dividing~\eqref{eq:asymptotics-an-transform-poincare} by $v_n$ (which is nonzero since $v$ is a non-vanishing solution), we obtain
\begin{equation}
    \frac{v_{n+1}}{v_n}
    =
    \left(\frac{\lambda f_n - z}{a_n}\right)\frac{a_{n+1}^{1/2}}{a_n^{1/2}}
    - \frac{a_{n-1}^{1/2} a_{n+1}^{1/2}}{a_n} \frac{v_{n-1}}{v_n} \, .
\end{equation}
In terms of $r_n = v_n / v_{n-1}$ this becomes
\begin{equation}
    r_{n+1}
    =
    \left(\frac{\lambda f_n - z}{a_n}\right)\frac{a_{n+1}^{1/2}}{a_n^{1/2}}
    - \frac{a_{n-1}^{1/2} a_{n+1}^{1/2}}{a_n} \frac{1}{r_n} \, .
\end{equation}
    Solving for $r_n$, we get the result.
\end{proof}
By \cref{lem:asymptotic-solutions}, we know the leading-order asymptotic behavior of $r_{n+1}$ in the regime $n \to \infty$. To utilize this information, we reformulate the Riccati recursion in terms of deviations from this asymptotic profile and use it to propagate errors backwards in $n$.

More precisely, we factor out the dominant behavior and write
\begin{equation}
    r_n = \frac{a_n}{\lambda f_n}(1+\epsilon_n), \qquad \epsilon_n \in \cnum,
\end{equation}
where $\epsilon_n$ measures the relative deviation from the leading-order asymptotics.
Inserting this ansatz into the Riccati equation $\mathcal{R}_n$ yields an induced recursion for the error terms of the form
\begin{equation}
    \label{eq:def-riccati-error}
    \epsilon_n = \tilde{\mathcal{R}}_n(\epsilon_{n+1}),
\end{equation}
where
\begin{equation}
    \tilde{\mathcal{R}}_n(\epsilon)
    =
    \frac{\lambda f_n}{a_n}
    \mathcal{R}_n\!\left(\frac{a_{n+1}}{\lambda f_{n+1}}(1+\epsilon)\right)
    - 1.
\end{equation}

The key point is that this reformulation turns the asymptotic Riccati dynamics into a perturbative fixed-point problem for the error sequence $(\epsilon_n)_{n\in\nnum}$. In particular, for $n \geq N_3(\lambda)$, we show that the error does not grow under backward iteration:
\begin{lemma}
    Let $\Lambda$ and $N_3(\lambda)$ be given as in \cref{def:asymptotics-regions}, and let $\mathbb{D} \subset \cnum$ denote the open disk of radius $\nicefrac{1}{3}$ centered at $0$. Then, for all $\lambda < \Lambda$, $z \in \Omega$, and all $n \geq N_3(\lambda)$, one has
    \begin{equation}
        \tilde{\mathcal{R}}_n(\mathbb{D}) \subset \mathbb{D}.
    \end{equation}
\end{lemma}
\begin{proof}
    We rewrite the error map by applying \cref{lem:riccati-forward} to \cref{eq:def-riccati-error}. This yields
    \begin{align}
        \tilde{\mathcal{R}}_n(\epsilon)
        &=
        \frac{a_{n-1}^{1/2}}{a_n^{1/2}}
        \left(
        1 - \frac{z}{\lambda f_n}
        - \frac{a_n a_{n+1}}{\lambda^2 f_n f_{n+1}}
        \frac{a_n^{1/2}}{a_{n+1}^{1/2}} (1+\epsilon)
        \right)^{-1}
        - 1.
    \end{align}
    By \cref{def:asymptotics-regions,lem:regularity}, there exists $\Lambda > 0$ such that for all $n \geq N_3(\lambda)$, $\lambda < \Lambda$, and $z \in \Omega$,
    \begin{align}
        \frac{a_n a_{n+1}}{\lambda^2 f_n f_{n+1}} & \leq \nicefrac{1}{16}, \\
        \left| \frac{a_{n-1}}{a_n} - 1 \right| & \leq \nicefrac{1}{8}, \\
        \left| \frac{z}{\lambda f_n} \right| & \leq \nicefrac{1}{32}.
    \end{align}
    From the second estimate we obtain
    \begin{equation}
        \frac{7}{8} \leq \frac{a_{n-1}}{a_n} \leq \frac{9}{8},
        \qquad \text{hence} \qquad
        \frac{7}{8} < \frac{a_{n-1}^{1/2}}{a_n^{1/2}} < \frac{9}{8}.
    \end{equation}
    For $|\epsilon| \leq \nicefrac{1}{3}$, the denominator in the definition of $\tilde{\mathcal{R}}_n(\epsilon)$ is uniformly bounded away from zero, and we obtain
    \begin{align}
        |\tilde{\mathcal{R}}_n(\epsilon)|
        &\leq
        \frac{a_{n-1}^{1/2}}{a_n^{1/2}}
        \left|
        \left(
        1 - \frac{z}{\lambda f_n}
        - \frac{a_n a_{n+1}}{\lambda^2 f_n f_{n+1}}
        \frac{a_n^{1/2}}{a_{n+1}^{1/2}} (1+\epsilon)
        \right)^{-1}
        - 1
        \right|
        +
        \left|
        \frac{a_{n-1}^{1/2}}{a_n^{1/2}} - 1
        \right| \\
        &\leq
        \frac{9}{8}
        \left(
        \frac{1}{1 - \nicefrac{1}{32} - \nicefrac{9}{128}(1+|\epsilon|)} - 1
        \right)
        + \nicefrac{1}{8}.
    \end{align}
    The right-hand side is strictly smaller than $\nicefrac{1}{3}$ for all $|\epsilon| \leq \nicefrac{1}{3}$, which proves $        \tilde{\mathcal{R}}_n(\mathbb{D}) \subset \mathbb{D}$.
\end{proof}

As a consequence, we obtain a controlled asymptotic description of the recessive solution.

\begin{proposition}
    \label{prop:asymptotics-tail}
    Let $\Lambda$ and $N_3(\lambda)$ be as in \cref{def:asymptotics-regions}, and let $\tilde{v}^{\recessSol}$ denote the recessive solution of \cref{eq:asymptotics-an-transform-poincare} from \cref{lem:asymptotic-solutions}. Then, for all $\lambda < \Lambda$, $z \in \Omega$, and $n \geq N_3(\lambda)$, there exists $|\epsilon_n| \leq \nicefrac{1}{3}$ such that
    \begin{equation}
        \frac{\tilde{v}^{\recessSol}_{n+1}}{\tilde{v}^{\recessSol}_n}
        =
        \frac{a_n}{\lambda f_n}(1+\epsilon_n).
    \end{equation}
\end{proposition}

\begin{proof}
    By \cref{lem:asymptotic-solutions}, there exists $N \geq N_3(\lambda)$ such that the stated representation holds for all $n \geq N$, with errors lying in $\mathbb{D}$.

    The previous lemma shows that the Riccati map $\tilde{\mathcal{R}}_n$ preserves $\mathbb{D}$ for all $n \geq N_3(\lambda)$. Hence, starting from index $N$ and iterating the recursion backward, the error remains uniformly controlled down to $N_3(\lambda)$, which yields the claim for all $n \geq N_3(\lambda)$.
\end{proof}

\subsection{Proof of \texorpdfstring{\cref{prop:asymptotics}}{the proposition}}
\label{sec:proof-asymptotics}
We are finally in a position to prove \cref{prop:asymptotics}. The proof proceeds by combining the estimates obtained in the different asymptotic regimes. In particular, we use the overlap regions defined in \cref{def:asymptotics-regions} to consistently match the corresponding approximations.
In the following, we assume $\Lambda > 0$ to be chosen such that both \cref{lem:regularity,prop:properties-turningpoint-solutions} hold, cf. \cref{rem:lambda-decrease}.

The following lemma is the main matching step between the tail and the turning-point regime:
\begin{lemma}
    \label{lem:connection-turningpoint}
    Let $u$ be a square-summable solution of \cref{eq:generalized-eigenvalue-lambda}, let $h = \lambda^{1/\delta}$, and let $u^{\recessSol,h,z}, u^{\domSol,h,z}$ be the recessive and dominant solutions given in \cref{prop:turning-point}. 
    Then there exist constants $A, B \in \cnum$ such that
    \begin{equation}
        u_n = A u^{\recessSol,h,z}_n + B\, u^{\domSol,h,z}_n \qquad \forall n \in \nnum.
    \end{equation}
    Moreover, there exist $\Lambda > 0$ and $C_3 > 0$ such that, for all $0 < \lambda < \Lambda$ and all $N_1(\lambda)+1 \leq n \leq N_4(\lambda)-1$, one has
    \begin{equation}
        \left|\frac{B}{A} u^{\domSol,h,z}_n\right| \leq C_3 (nh)^{-\alpha/2}.
    \end{equation}
    In particular, for every $\epsilon > 0$ there exists $\Lambda(\epsilon) > 0$ such that, for all $0 < \lambda < \Lambda(\epsilon)$ and all $N_1(\lambda)+1 \leq n \leq N_2(\lambda)$,
    \begin{equation}
        \left|\frac{B}{A} u^{\domSol,h,z}_n\right| \leq \epsilon (nh)^{-\alpha/2} h^{1/6}.
    \end{equation}
\end{lemma}

\begin{proof}
    We choose $\Lambda > 0$ such that \cref{lem:regularity,prop:properties-turningpoint-solutions} hold throughout, and write $u^{\recessSol} = u^{\recessSol,h,z}$ and $u^{\domSol} = u^{\domSol,h,z}$ for convenience, leaving the dependence on $h,z$ implicit.

    Let $u$ be a solution of \cref{eq:generalized-eigenvalue-lambda}. By construction (cf.~\cref{prop:turning-point}), $u^{\recessSol}$ and $u^{\domSol}$ are linearly independent solutions, hence there exist unique constants $A,B \in \cnum$ such that
    \begin{equation}
        u_n = A u^{\recessSol}_n + B u^{\domSol}_n \qquad \forall n \in \nnum.
    \end{equation}

    We determine $A$ and $B$ using the discrete Wronskian (\cref{def:wronksian}). Fix $m = N_4(\lambda)-1$. By \cref{lem:wronskian},
    \begin{equation}
        A =  \frac{W_m[u,u^{\domSol}]}{W_m[u^{\recessSol},u^{\domSol}]}, 
        \qquad 
        B = -\frac{W_m[u,u^{\recessSol}]}{W_m[u^{\recessSol},u^{\domSol}]}.
    \end{equation}
    Hence,
    \begin{align}
        \frac{B}{A}
        &=-\frac{W_m[u,u^{\recessSol}]}{W_m[u,u^{\domSol}]}
        = -\frac{W_m[u^{\recessSol},u]}{W_m[u^{\domSol},u]} \\
        &=
        -\frac{u^{\recessSol}_m u_{m+1}-u^{\recessSol}_{m+1}u_m}{u^{\domSol}_m u_{m+1}-u^{\domSol}_{m+1}u_m} \\
        &=
        -\frac{\frac{u^{\recessSol}_m}{u^{\domSol}_{m+1}}\frac{u_{m+1}}{u_m}-\frac{u^{\recessSol}_{m+1}}{u^{\domSol}_{m+1}}}
        {\frac{u^{\domSol}_m}{u^{\domSol}_{m+1}}\frac{u_{m+1}}{u_m}-1}.
    \end{align}
    Multiplying by $u^{\domSol}_n$ yields
    \begin{equation}
        \label{proofeq:ratio-ab-v2}
        \frac{B}{A}u^{\domSol}_n
        =
        -\frac{u^{\domSol}_n}{u^{\domSol}_{m+1}}
        \frac{
            u^{\recessSol}_m \frac{u_{m+1}}{u_m} - u^{\recessSol}_{m+1}
        }
        {
            \frac{u^{\domSol}_m}{u^{\domSol}_{m+1}}\frac{u_{m+1}}{u_m}-1
        }.
    \end{equation}
    We now estimate the individual factors.
    Since $m \geq N_3(\lambda)$, we can use \cref{lem:regularity,lem:which-asymptotic-solution,prop:asymptotics-tail} to obtain
    \begin{equation}
        \left|\frac{u_{m+1}}{u_m}\right|
        =
        \frac{a_{m+1}^{-1/2}}{a_m^{-1/2}}
        \left|\frac{\tilde{v}^{\recessSol}_{m+1}}{\tilde{v}^{\recessSol}_m}\right|
        \leq
        \left(1+\nicefrac{1}{8}\right)\frac{a_m}{\lambda f_m}(1+\nicefrac{1}{3})
        \leq \frac{3}{8}.
    \end{equation}
    Moreover, by \cref{prop:properties-turningpoint-solutions}~\labelcref{item:pts-udratio-m,item:pts-urratio-n},
    \begin{align}
        \label{proofeq:umump1bound}
        \left|\frac{u^{\domSol}_{m+1}}{u^{\domSol}_m}\right| &\geq \nicefrac{7}{8}
        \qquad \Rightarrow \qquad
        \left|\frac{u^{\domSol}_m}{u^{\domSol}_{m+1}}\right| \leq \nicefrac{8}{7}, \\
        |u^{\recessSol}_m|, |u^{\recessSol}_{m+1}|
        &\leq c_1 (mh)^{-\alpha/2}
        \leq c_1 x_3^{-\alpha/2}.
    \end{align}
    Finally, by \cref{prop:properties-turningpoint-solutions}~\labelcref{item:pts-udratio-nm} and \cref{proofeq:umump1bound},
    \begin{equation}
        \left|\frac{u^{\domSol}_n}{u^{\domSol}_{m+1}}\right|
        \leq
        \left| \frac{u^{\domSol}_n}{u^{\domSol}_{m}}\right|\left|\frac{u^{\domSol}_m}{u^{\domSol}_{m+1}}\right|
        \leq \frac{8}{7} c_1 \frac{m^{\alpha/2}}{n^{\alpha/2}} 
        \leq \tilde{c}_1 \frac{m^{\alpha/2}}{n^{\alpha/2}}.
    \end{equation}
    Combining these estimates in \cref{proofeq:ratio-ab-v2}, we obtain
    \begin{align}
        \left|\frac{B}{A}u^{\domSol}_n\right|
        &\leq
        \left|\frac{u^{\domSol}_n}{u^{\domSol}_{m+1}}\right|
        \frac{
            |u^{\recessSol}_m|\left|\frac{u_{m+1}}{u_m}\right|
            +
            |u^{\recessSol}_{m+1}|
        }
        {
            \left|\frac{u^{\domSol}_m}{u^{\domSol}_{m+1}}\frac{u_{m+1}}{u_m}-1\right|
        } \\
        &\leq
        \tilde{c}_1 c_1 x_3^{-\alpha/2}
        \frac{\nicefrac{3}{8}+1}{1-\nicefrac{3}{7}}
        \frac{x_4^{\alpha/2}}{(nh)^{\alpha/2}}
        \eqqcolon C_3 (nh)^{-\alpha/2}.
    \end{align}
    This proves the first bound.

    Finally, let $\epsilon > 0$. By \cref{prop:properties-turningpoint-solutions}~\labelcref{item:pts-udratio-nm-epsilon}, there exists $\Lambda(\epsilon) > 0$ such that for all $0 < \lambda < \Lambda(\epsilon)$ and $N_1(\lambda)+1 \leq n \leq N_2(\lambda)$,
    \begin{equation}
        \left|\frac{u^{\domSol}_n}{u^{\domSol}_m}\right|
        \leq
        \epsilon (m/n)^{\alpha/2} h^{1/6}.
    \end{equation}
    Using again $m \leq x_4/h$, the same argument as above yields
    \begin{equation}
        \left|\frac{B}{A}u^{\domSol}_n\right|
        \leq
        \epsilon c_1 x_3^{-\alpha/2}
        \frac{1+\nicefrac{3}{8}}{\nicefrac{7}{3}-1}
        \frac{x_4^{\alpha/2}}{(nh)^{\alpha/2}}
        h^{1/6},
    \end{equation}
    which concludes the proof.
\end{proof}

We are now ready to prove \cref{prop:asymptotics}. The argument is a matching procedure across the three regimes defined in \cref{def:asymptotics-regions}:
\begin{itemize}
    \item We first analyze the tail region $n \geq N_3(\lambda)$, where \cref{lem:asymptotics-an-transform,lem:asymptotic-solutions} imply that any square-summable solution coincides with the recessive one, and the Riccati formulation yields uniform decay estimates. 
\item We then propagate this information into the turning-point region $N_1(\lambda) \leq n \leq N_4(\lambda)$, where \cref{prop:turning-point,lem:connection-turningpoint} allow us to express the solution in terms of recessive and dominant components and to control the contribution of the latter in the overlap. 
\item Finally, we extend the resulting bounds to the interior region $n \leq N_0$ using standard discrete comparison estimates. Combining the three regimes yields a uniform polynomial bound for all $n \in \nnum$, completing the proof.
\end{itemize}

\begin{proof}[{Proof of \cref{prop:asymptotics}}]
Let $N_j(\lambda)$, $j \in \{0,1,2,3,4\}$, and $\Lambda > 0$ be as in \cref{lem:regularity,def:asymptotics-regions}. Moreover, let $\tilde{v}^{\recessSol}$, $u^{\recessSol}$, and $u^{\domSol}$ be the solutions constructed in \cref{prop:turning-point,lem:asymptotic-solutions}, depending implicitly on $h$ and $z$.
We choose $\Lambda > 0$ sufficiently small such that \cref{lem:regularity,prop:properties-turningpoint-solutions} hold, and such that the second part of \cref{lem:connection-turningpoint} holds with $\epsilon = c_2/2$, where $c_2$ is the constant from \cref{prop:properties-turningpoint-solutions}. All constants below can be chosen uniformly in the overlapping regimes after possibly further shrinking $\Lambda$.

Fix $0 < \lambda < \Lambda$, and let $u$ be a square-summable solution of \cref{eq:generalized-eigenvalue-lambda}.
We start with the analysis in the tail region. We use the ansatz $u_n = (-1)^n a_n^{-1/2} v_n$. By \cref{lem:asymptotics-an-transform}, the sequence $v_n$ satisfies \cref{eq:asymptotics-an-transform-poincare}. Since $u \in \ell^2(\mathbb{N})$, \cref{lem:which-asymptotic-solution} implies that
\begin{equation}
v_n = \tilde{A}\,\tilde{v}^{\recessSol}_n
\end{equation}
for some $\tilde{A} \in \mathbb{C}$, where $\tilde{v}^{\recessSol}$ is the recessive solution from \cref{lem:asymptotic-solutions}.

Hence, for $n \geq N_3(\lambda)$,
\begin{equation}
\left|\frac{v_{n+1}}{v_n}\right|
= \left|\frac{\tilde{v}^{\recessSol}_{n+1}}{\tilde{v}^{\recessSol}_n}\right|
\leq \frac{a_n}{\lambda f_n}\left(1+\tfrac{1}{3}\right)
\leq \tfrac{1}{3}.
\end{equation}
Iterating this estimate and using \cref{hyp:conditions-an-fn}, we obtain
\begin{equation}
\label{proofeq:asymptotics-ratio-tail}
\left|\frac{u_n}{u_{N_3(\lambda)}}\right|
\leq \frac{a_{N_3(\lambda)}^{1/2}}{a_n^{1/2}}
\leq \tilde{C}_3 \left(\frac{N_3(\lambda)}{n}\right)^{\alpha/2}.
\end{equation}

We next pass to the turning-point regime. We write
\begin{equation}
u_n = A u^{\recessSol,h,z}_n + B u^{\domSol,h,z}_n,
\end{equation}
where $u^{\recessSol,h,z}$ and $u^{\domSol,h,z}$ are the solutions from \cref{prop:turning-point}.
By \cref{prop:properties-turningpoint-solutions}, there exists $M \in [N_1(\lambda), N_2(\lambda)]$ such that
\begin{equation}
|u^{\recessSol}_M| \geq c_2 (Mh)^{-\alpha/2} h^{1/6}.
\end{equation}
Moreover, by construction of $\Lambda$ and \cref{lem:connection-turningpoint},
\begin{equation}
\left|\frac{B}{A}u^{\domSol}_M\right|
\leq \frac{c_2}{2} (Mh)^{-\alpha/2} h^{1/6}.
\end{equation}
Using \cref{prop:properties-turningpoint-solutions} and \cref{lem:connection-turningpoint}, we obtain for $M \leq n \leq N_4(\lambda)$:
\begin{align}
\left|\frac{u_n}{u_M}\right|
&= \left|\frac{u^{\recessSol}_n + \frac{B}{A}u^{\domSol}_n}{u^{\recessSol}_M + \frac{B}{A}u^{\domSol}_M}\right| \\
&\leq \frac{|u^{\recessSol}_n| + \left|\frac{B}{A}u^{\domSol}_n\right|}{|u^{\recessSol}_M| - \left|\frac{B}{A}u^{\domSol}_M\right|} \\
&\leq \frac{(c_1 + C_3)(nh)^{-\alpha/2}}{(c_2/2)(Mh)^{-\alpha/2}h^{1/6}}
= \tilde{C}_2 \left(\frac{M}{n}\right)^{\alpha/2} h^{-1/6}.
\end{align}
For $n \leq N_0$, \cref{lem:before-N0} yields
\begin{equation}
    \label{proofeq:bound-before-n0}
|u_{n}|^2 + |u_{n-1}|^2 \leq C_0(|u_0|^2 + |u_1|^2),
\end{equation}
and as $a_n$ is strictly positive and $N_0$ finite, there exists $\tilde{C}_0$ such that
\begin{equation}
\label{proofeq:bound-before-n0-2}
|u_{n}|^2 + |u_{n-1}|^2 \leq \tilde{C}_0 a_n^{-1}(|u_0|^2 + |u_1|^2) \qquad (1\leq n \leq N_0)\, .
\end{equation}
By \cref{prop:turan,proofeq:bound-before-n0}, we get
\begin{align}
\label{proofeq:bound-before-n2}
    |u_{n}|^2 + |u_{n-1}|^2 & \leq C_1 a_n^{-1} \left(|u_{N_0}|^2 + |u_{N_0-1}|^2\right) \\
    & \leq C_0 C_1 a_n^{-1} (|u_0|^2 + |u_1|^2) \qquad (N_0 \leq n \leq N_2(\lambda))\, ;
\end{align}
using $a_n^{-1/2} \leq c_4 n^{-\alpha/2}$ and taking square roots, together with \cref{proofeq:bound-before-n0-2}, we obtain
\begin{equation}
|u_n| \leq \tilde{C}_1 n^{-\alpha/2}(|u_0| + |u_1|)
\qquad (1 \leq n \leq N_2(\lambda)).
\end{equation}
Combining the turning-point estimate with the previous bounds and $h = x/n$ gives
\begin{align}
|u_n| & \leq \tilde{C}_1 \tilde{C}_2 n^{-\alpha/2} h^{-1/6}(|u_0| + |u_1|) \\
& \leq \tilde{C}_1 \tilde{C}_2 n^{-\alpha/2} x_1^{-1/6} n^{+1/6} (|u_0| + |u_1|) \\
& = C_4 n^{-\alpha/2+1/6}  (|u_0| + |u_1|)
\qquad (M \leq n \leq N_4(\lambda)).
\end{align}
Since $n^{-\alpha/2} \leq n^{-\alpha/2 + 1/6}$ for $n \geq 1$, by combining this with~\eqref{proofeq:asymptotics-ratio-tail}, we obtain
\begin{equation}
|u_n| \leq C n^{-\alpha/2 + 1/6}(|u_0| + |u_1|) \, .
\qquad (n \geq N_3(\lambda)).
\end{equation}
Together with \cref{proofeq:bound-before-n2}, this yields the claimed bound and completes the proof.
\end{proof}

\subsection{Discussion of the bound}
\label{sec:asymptotics-discussion}

We briefly comment on the tightness of the bound provided by \cref{prop:asymptotics}. At first sight, the additional loss of $\nicefrac{1}{6}$ in the exponent may appear surprising, since in the unperturbed case $\lambda = 0$ one obtains strictly sharper decay estimates for solutions of \cref{eq:generalized-eigenvalue-lambda}. More precisely, while \cref{prop:asymptotics} yields a bound with an additional factor $n^{1/6}$ which is uniform in $\lambda >0$, no such loss occurs for $\lambda =0$. The following proposition shows that in the uncoupled case the optimal exponent $\alpha/2$ is recovered exactly. This is a direct consequence of \cite[Theorem~2]{swiderski-spectralpropertiesblock-2018}; for the convenience of the reader, we nevertheless include a short proof.

\begin{proposition}
    \label{prop:limit-circle-asymptotics}
    Let $J(0)$ be the Jacobi operator with $\lambda=0$ from \cref{def:jacobi-lambda}, let $\Omega \subset \cnum$ be compact, and assume that \cref{hyp:conditions-an-fn} holds with $\alpha > 1$.
    Then there exists a constant $C>0$, independent of $z \in \Omega$, such that every solution $u^z$ of \cref{eq:generalized-eigenvalue-lambda} with $\lambda =0$ satisfies
    \begin{equation}
        |u^z_n| \leq C \frac{|u^z_0|+|u^z_1|}{n^{\alpha/2}}
    \end{equation}
    for all $n \in \nnum$.
\end{proposition}
\begin{proof}
    Let $\Omega \subset \cnum$ be compact, and let $N_0$ be given as in \cref{lem:regularity}.
    Furthermore, let $u^z$ be a solution of \cref{eq:generalized-eigenvalue-lambda} with $\lambda =0$. 
    By \cref{lem:before-N0}, there exists a constant $C_0>0$, independent of $z \in \Omega$, such that
    \begin{equation}
        \label{proofeq:limit-circle-asymptotics-n0}
        |u^z_{N_0-1}|^2 + |u^z_{N_0}|^2
        \leq
        C_0 \left(|u^z_0|^2+|u^z_1|^2\right) \, .
    \end{equation}
    Moreover, for $\lambda =0$ we have $N_2(\lambda)=\infty$ by \cref{def:asymptotics-regions}. Hence, applying \cref{prop:turan} for all $n \geq N_0$, we obtain
    \begin{equation}
        a_n\left(|u^z_n|^2+|u^z_{n-1}|^2\right)
        \leq
        C_1 \left(|u^z_{N_0-1}|^2 + |u^z_{N_0}|^2\right)
    \end{equation}
    with some constant $C_1>0$ independent of $z \in \Omega$.
    Since $a_n \sim n^\alpha$ by \cref{hyp:conditions-an-fn}, there exists $C_2>0$ such that
    \begin{equation}
        |u^z_n|
        \leq
        \frac{C_2}{n^{\alpha/2}}
        \left(|u^z_{N_0-1}|^2 + |u^z_{N_0}|^2\right)^{1/2}
    \end{equation}
    for all $n \geq N_0$.
    Combining this with \cref{proofeq:limit-circle-asymptotics-n0}, and using the sub-linearity of the square root the claim follows after adjusting the constant.
\end{proof}
As shown in \cref{lem:pointwise-convergence}, solutions of \cref{eq:generalized-eigenvalue-lambda} with $\lambda >0$ converge pointwise to solutions of the limiting equation corresponding to $\lambda =0$ as $\lambda \to 0$. 
At first sight, one might therefore expect that the uniform bound from \cref{prop:asymptotics} should also exhibit the same decay rate $n^{-\alpha/2}$ as in the uncoupled case, and that the additional factor $n^{1/6}$ is merely an artifact of the proof technique.
However, the additional factor arises precisely in the turning point region in the proof of \cref{prop:asymptotics}. In particular, by \cref{prop:turan}, all solutions of \cref{eq:generalized-eigenvalue-lambda} are bounded by $C n^{-\alpha/2}$ for all $n \leq N_2(\lambda)$ and $\lambda \geq 0$. As $\lambda \to 0$, $N_2(\lambda)$ diverges to $\infty$, cf.~\cref{def:asymptotics-regions}, so that the turning point region disappears in the limit circle case.

The following numerical analysis indicates that the additional factor of $\nicefrac{1}{6}$ is strictly necessary in the turning point region.
More precisely, the solution develops a pronounced transition region near $(nh)^\delta =2$, where a localized ``hump'' forms whose relative size increases as $\lambda \to 0$. The growth of this turning-point contribution is consistent with the additional factor appearing in \cref{prop:asymptotics}, and therefore supports the sharpness of the obtained bound.
To this end, we evaluate the solution
\begin{equation}
    u^{\lambda,z} = (J(\lambda)-z)^{-1} e_0
\end{equation}
of \cref{eq:generalized-eigenvalue-lambda} by computing the resolvent $(J(\lambda)-z)^{-1}$ numerically. We normalize the solutions by their initial values $u^{\lambda,z}_0$, $u^{\lambda,z}_1$ and introduce the rescaled quantity
\begin{equation}
    \label{eq:ratio-weyl-solution-bound}
    r^{\lambda,z}_n
    =
    n^{\alpha/2-1/6}
        \frac{|u^{\lambda,z}_n|}{|u^{\lambda,z}_0|+|u^{\lambda,z}_1|} \, .
\end{equation}
By \cref{prop:asymptotics}, for every fixed $z \in \cnum$ there exists a constant $C_1>0$, independent of $\lambda$, such that
\begin{equation}
    \sup_{n \geq 1} r^{\lambda,z}_n \leq C_1 \, .
\end{equation}
Conversely, if the exponent $-\alpha/2+1/6$ is optimal, one expects that for every $n_0 \in \nnum$ there exists $\Lambda>0$ such that
\begin{equation}
    \sup_{n \geq n_0} r^{\lambda,z}_n
\end{equation}
remains uniformly bounded from below by some constant $C_2>0$, independent of $0<\lambda<\Lambda$.

The following example indicates that this scenario indeed occurs. We consider the family of Jacobi operators $J(\lambda)$, defined as in \cref{def:jacobi-lambda}, with parameters
\begin{equation}
    \label{eq:parameters-asymptotics-plot}
    a_n = 4^{-1/2}\sqrt{(4n+1,4)}\,, \quad f_n = n^4\, , 
\end{equation}
where $(\cdot,\cdot)$ denotes the Pochhammer symbol, cf. \cref{eq:def-pochhammer}.
\cref{hyp:conditions-an-fn} is thus satisfied with $\alpha = 2$ and $\beta = 4$.
We note that these Jacobi operators are equivalent to the operators $A_{4,4}^{(0)}(\lambda)$ introduced in \cref{lem:squeezing-decomposition} in the context of higher-order squeezing operators. 
For fixed spectral parameter $z=\iu$, the resulting quantities $r_n^{\lambda,z}$ are displayed in \cref{fig:ratioplot}. 
As predicted by \cref{prop:asymptotics}, the ratios remain uniformly bounded from above as $\lambda \to 0$. At the same time, the numerical data indicate that the ratios also remain uniformly bounded from below for sufficiently small $\lambda$. Moreover, the maximal values are attained close to the turning point $(nh)^\delta =2$, precisely where the turning-point analysis predicts the additional growth factor. 

This indicates that the correction factor $\nicefrac{1}{6}$, and consequently the assumption $\alpha > \nicefrac{4}{3}$ (cf.~\cref{hyp:conditions-an-fn}), is intrinsic to our proof strategy. It remains open whether \cref{thm:main-result} can be obtained under the weaker condition $\alpha > 1$ by alternative methods.
\begin{figure}[ht]
    \centering
    \includegraphics[width=0.7\linewidth]{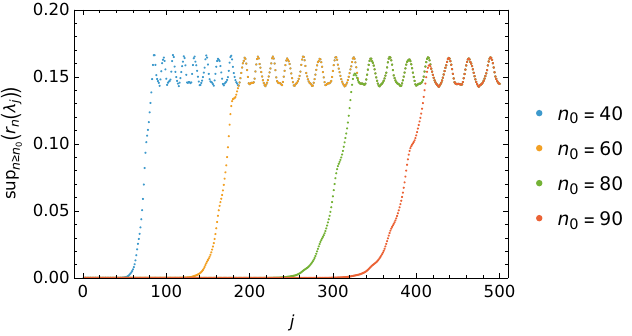}
    \caption{
    Plot of $\sup_{n \geq n_0} r^{\lambda_j,\iu}_n$, with $r^{\lambda,z}_n$ defined in \cref{eq:ratio-weyl-solution-bound}, for the sequence $\lambda_j = (10j)^{-1}$ and different values of $n_0$. The underlying Jacobi operator $J(\lambda)$ is parametrized as in \cref{def:jacobi-lambda} with the sequences $(a_n)_{n\in\nnum}$ and $(f_n)_{n\in\nnum}$ from \cref{eq:parameters-asymptotics-plot}. As predicted by \cref{prop:asymptotics}, the ratios remain uniformly bounded from above by a constant independent of $\lambda_j$. At the same time, for the displayed values of $n_0$, the ratios remain uniformly bounded from below for sufficiently small $\lambda_j$, indicating that the bound from \cref{prop:asymptotics} is saturated and that the correction factor $\nicefrac{1}{6}$ is genuinely necessary.
}
    \label{fig:ratioplot}
\end{figure}

\section{Proofs regarding higher-order squeezing}\label{sec:proof-squeezing}
In this section, we provide the proofs of the results stated in \cref{sec:main_squeezing}.
We start by proving \cref{lem:squeezing-decomposition}, which establishes that higher-order squeezing operators can be realized as a direct sum of Jacobi operators defined by sequences satisfying \cref{hyp:conditions-an-fn}.
\begin{proof}[Proof of \cref{lem:squeezing-decomposition}]
We have
\begin{equation}
    (a^\dagger a)^h \phi_n = n^h \phi_n
\end{equation}
for all $n \in \nnum$, where $\phi_n$ is the $n$th eigenvector of the number operator. 
Via a direct computation (cf.~\cite[Definition~2.4, Lemmas~2.5--2.6]{fischer-selfadjointrealizationshigherorder-2025}, we obtain the action
\begin{equation}
    \label{proofeq:higher-order-action}
    A_{k,h}(K)\phi_n
    = \beta^k_n \phi_{n+k}
    + K \tilde{f}_n \phi_n
    + \beta^k_{n-k} \phi_{n-k},
\end{equation}
for all $n \in \nnum$, where
\begin{align}
    \beta^k_n &= \sqrt{(n+1,k)} =  \sqrt{(n+1)\dots(n+k)} \, , \\
    \tilde{f}_n &= n^h \, ,
\end{align}
and $(\cdot,\cdot)$ denotes the Pochhammer symbol, cf. \cref{eq:def-pochhammer}.

The structure of \cref{proofeq:higher-order-action} suggests decomposing the Hilbert space into $k$ invariant subspaces corresponding to residue classes modulo $k$. Indeed, for $0 \leq m \leq k-1$ and $r \in \nnum$, restricting to vectors of the form $\phi_{m+rk}$ yields
\begin{equation}
     A_{k,h}(K)\phi_{m+rk}
    = \gamma^{(m)}_r \phi_{m+(r+1)k}
    + K \tilde{f}^{(m)}_r \phi_{m+rk}
    + \gamma^{(m)}_{r-1} \phi_{m+(r-1)k},
\end{equation}
where
\begin{equation}
    \gamma^{(m)}_r = \beta^k_{m+rk}, \qquad
    \tilde{f}^{(m)}_r = \tilde{f}_{m+rk}.
\end{equation}
Rescaling by $k^{-k/2}$ and setting
\begin{equation}
    a^{(k,m)}_r = k^{-k/2}\gamma^{(m)}_r,
    \qquad
    f^{(k,h,m)}_r = k^{-h}\tilde{f}^{(m)}_r, \qquad \lambda = k^{h-k/2} K
\end{equation}
yields a Jacobi operator on each invariant subspace, and hence the decomposition on $\mspan(\phi_n)_{n\in\nnum}$.
Since the closure of a direct sum of operators equals the direct sum of the closures of the components \cite[79]{teschl-mathematicalmethodsquantum-2009}, the decomposition extends to the closed operators.

We now verify \cref{hyp:conditions-an-fn}. Condition (i) is immediate since $a^{(k,m)}_r>0$ and $f^{(k,h,m)}_r\geq 0$.
For condition (ii), we expand
\begin{align}
    a^{(k,m)}_r
    &= k^{-k/2} \sqrt{(m+rk+1)\cdots(m+rk+k)} \\
    &= k^{-k/2} (rk)^{k/2}
       \left(\prod_{l=1}^k \left(1 + \frac{m+l}{rk}\right)\right)^{1/2} \\
    &= r^{k/2}\left(1 + \frac{c_a}{r} +O(r^{-2})\right).
\end{align}
Since $k/2 > 4/3$ for $k\geq 3$, condition (ii) follows.
Similarly,
\begin{equation}
    f^{(k,h,m)}_r
    = k^{-h}(m+rk)^h
    = r^h\left(1 + \frac{c_f}{r} +O(r^{-2})\right),
\end{equation}
and since $2h>k$, condition (iii) holds.

Finally, we verify that the quantity
\begin{equation}
    d^{(k,h,m)}_r := \frac{f^{(k,h,m)}_r}{a^{(k,m)}_r}
\end{equation}
is strictly increasing in $r$. Writing the ratio explicitly,
\begin{equation}
    \label{proofeq:squeezing-monotonic-ratio}
    \frac{d^{(k,h,m)}_{r+1}}{d^{(k,h,m)}_r}
    = \left(1+\frac{1}{r}\right)^{h-k/2}
      \left(\frac{1+\frac{m}{(r+1)k}}{1+\frac{m}{rk}}\right)^h
      \prod_{l=1}^k
      \left(\frac{1+\frac{m+l}{(r+1)k}}{1+\frac{m+l}{rk}}\right)^{-1/2}.
\end{equation}
A direct comparison of the factors shows that each term in the product is compensated by the monotonicity of the prefactor, and since $h>k/2$, we conclude
\begin{equation}
    \frac{d^{(k,h,m)}_{r+1}}{d^{(k,h,m)}_r} > 1,
\end{equation}
for all $r \in \nnum$. Hence $d^{(k,h,m)}_r$ is strictly increasing, which establishes condition (iv).
\end{proof}

We finally present the proof of the main result for higher-order squeezing operators, \cref{prop:squeezing-main}. The proof is based on combining the Jacobi decomposition from \cref{lem:squeezing-decomposition} with the general convergence theory developed for Jacobi operators in \cref{thm:main-result}. 
The key point is that each invariant subspace yields a family of Jacobi operators depending on the coupling parameter $\lambda$: the global limit is obtained by constructing compatible subsequences across all blocks and then assembling the resulting strong resolvent limits via the direct-sum decomposition. 
\begin{proof}[Proof of \cref{prop:squeezing-main}]
We begin by proving \labelcref{item:squeezing-main-i}, which establishes the existence of a common subsequence yielding simultaneous strong resolvent convergence of all Jacobi components.
By \cref{lem:squeezing-decomposition}, the higher-order squeezing operator admits the orthogonal decomposition
\begin{equation}
    A_{k,h}(K)
    = k^{k/2} \bigoplus_{m=0}^{k-1} A^{(m)}_{k,h}(\lambda),
    \qquad \lambda = k^{h-k/2}K,
\end{equation}
where each Jacobi operator $A^{(m)}_{k,h}(\lambda)$ satisfies \cref{hyp:conditions-an-fn}.
Let $(K_j)_{j\in\nnum}$ be a sequence with $K_j>0$ and $K_j \to 0$, and define $\lambda_j = k^{h-k/2}K_j$. Then $\lambda_j \to 0$ as well.

We construct a common subsequence by a finite diagonal extraction over the blocks. For $m=0$, applying \cref{thm:main-result} yields a subsequence $(\lambda_{j^{(0)}_l})_l$ and a parameter $t_0 \in \brnum$ such that
\begin{equation}
    A^{(0)}_{k,h}(\lambda_{j^{(0)}_l})
    \xrightarrow[l\to\infty]{\text{s.r.}}
    A^{(0,t_0)}_{k}.
\end{equation}

Next, we apply \cref{thm:main-result} to $A^{(1)}_{k,h}(\lambda_{j^{(0)}_l})$, obtaining a further subsequence $(\lambda_{j^{(1)}_l})_l \subset (\lambda_{j^{(0)}_l})_l$ and $t_1 \in \brnum$ such that
\begin{equation}
    A^{(1)}_{k,h}(\lambda_{j^{(1)}_l})
    \xrightarrow[l\to\infty]{\text{s.r.}}
    A^{(1,t_1)}_{k}.
\end{equation}
Iterating this procedure for $m=0,\dots,k-1$, we obtain a final subsequence $(\lambda_{j_l})_l$ and parameters $(t_0,\dots,t_{k-1})$ such that, for every $m$,
\begin{equation}
    A^{(m)}_{k,h}(\lambda_{j_l})
    \xrightarrow[l\to\infty]{\text{s.r.}}
    A^{(m,t_m)}_{k}.
\end{equation}
We now pass to resolvents. For $\Imag z \neq 0$, the decomposition in \cref{lem:squeezing-decomposition} yields
\begin{equation}
    \left(A_{k,h}(K) - z\right)^{-1}
    = k^{-k/2} \bigoplus_{m=0}^{k-1}
    \left(A^{(m)}_{k,h}(\lambda) - y\right)^{-1},
    \qquad y = k^{-k/2}z.
\end{equation}
Setting $K_{j_l} = k^{k/2-h}\lambda_{j_l}$, we obtain
\begin{equation}
\label{proofeq:resolvent-decomposition-limit}
\lim_{l\to\infty} (A_{k,h}(K_{j_l}) - z)^{-1}
=
k^{-k/2} \bigoplus_{m=0}^{k-1}
\left(A_k^{(m,t_m)} - y\right)^{-1},
\end{equation}
in the strong sense, which proves \labelcref{item:squeezing-main-i}.

We now turn to \labelcref{item:squeezing-main-ii}. Fix $m \in \{0,\dots,k-1\}$ and $t_m \in \brnum$. Applying \cref{thm:main-result} to the $m$th component yields a subsequence $(\lambda_{j^{(m)}_l})_l$ such that
\begin{equation}
    A^{(m)}_{k,h}(\lambda_{j^{(m)}_l})
    \xrightarrow[l\to\infty]{\text{s.r.}}
    A_k^{(m,t_m)}.
\end{equation}
Repeating the previous diagonal extraction on the remaining components $p \neq m$, we obtain a common subsequence $(\lambda_{j_l})_l$ and parameters $(t_p)_{p\neq m}$ such that all limits
\begin{equation}
    A^{(p)}_{k,h}(\lambda_{j_l})
    \xrightarrow[l\to\infty]{\text{s.r.}}
    A_k^{(p,t_p)}
\end{equation}
hold simultaneously.

The same resolvent decomposition argument as above then yields convergence of $A_{k,h}(K_{j_l})$ to $A_k^{(t_0,\dots,t_{k-1})}$ in the strong resolvent sense, completing the proof.
\end{proof}

\appendix
\section{Airy and related functions}
\label{app:airy}

In this appendix, we recall basic facts about the Airy functions and related special functions.
For a comprehensive overview, see \cite{olver-asymptoticsspecialfunctions-2010}, in particular Sections 8.8, 11.1, and 11.8.

The Airy differential equation is given by \cite[392]{olver-asymptoticsspecialfunctions-2010}
\begin{equation}
    \label{eq:airy}
    \diff[2]{y(x)}{x} = x\,y(x)\, .
\end{equation}
A distinguished real solution is the Airy function of the first kind $\Ai(x)$, which admits the integral representation
\begin{equation}
    \Ai(x) = \frac{1}{\pi} \int_0^\infty \cos\!\left(\frac{t^3}{3} + xt\right)\, \dl t\, .
\end{equation}
Its characteristic property is that $\Ai(x) \to 0$ as $x \to \infty$, and it is therefore referred to as the \emph{recessive} solution in this regime.
A second linearly independent solution is the Airy function of the second kind $\Bi(x)$, which satisfies $\Bi(x) \to \infty$ as $x \to \infty$ and is correspondingly called the \emph{dominant} solution. 
Both $\Ai(x)$ and $\Bi(x)$ are oscillatory for $x < 0$.

In this work, we also require solutions of \cref{eq:airy} on the complex plane.
A convenient family of three linearly independent solutions is given by
\begin{equation}
    \label{eq:def-complex-airy}
    \Ai_j(x) = \Ai\!\left(\omega^j x\right)\, , \quad 
    \omega = \mathrm{e}^{-\iu 2\pi/3}\, , \quad j \in \mathbb{Z}_3 \, .
\end{equation}
Here, $\mathbb{Z}_3$ denotes the cyclic group of integers modulo $3$, which we identify with the set $\{0,1,2\}$ equipped with addition modulo $3$.
In particular, we freely identify indices modulo $3$, so that $\Ai_{j+1}$ and $\Ai_{j-1}$ are well-defined for all $j \in \mathbb{Z}_3$ (e.g.\ $\Ai_{2+1} = \Ai_0$).
For convenience of notation, we also occasionally use the equivalent labeling $j \in \{0,\pm 1\}$, where the identification is understood modulo $3$ (for instance, $-1 \equiv 2 \mod 3$). This alternative labeling is often more transparent when discussing the Stokes sectors below.

The corresponding closed sectors $S_j \subset \mathbb{C}$, $j \in \mathbb{Z}_3$, are defined by \cite[Section~11.8]{olver-asymptoticsspecialfunctions-2010}
\begin{align}
    \label{eq:def-complex-airy-regions}
    S_0 & = \left\{x \in \mathbb{C} \,:\, |\arg x| \leq \tfrac{\pi}{3} \right\}\, , \\
    S_{\pm 1} & = \mathrm{e}^{\pm \iu 2\pi/3} S_0 \, .
\end{align}
These sectors are chosen so that the asymptotic behavior of the Airy function is aligned with the Stokes geometry.
Using \cref{eq:def-complex-airy,eq:def-complex-airy-regions}, one verifies that $\Ai_j$ is recessive in $S_j$ and dominant in $S_{j\pm 1}$.
Consequently, the pair $(\Ai_j,\Ai_{j+1})$ forms a numerically satisfactory basis of solutions of \cref{eq:airy} in the sector $S_j \cup S_{j+1}$, but not in $S_{j-1}$.
A geometric illustration of the sectors $S_j$ is provided in \cref{fig:airy-regions}.
\begin{figure}[ht]
    \centering
    \begin{tikzpicture}[scale=2]
    \draw[->] (-1.2,0) -- (1.2,0) node[right]{$\Real(x)$};
    \draw[->] (0,-1.2) -- (0,1.2) node[above]{$\Imag(x)$};
    \foreach \k in {0,...,2} {
        \pgfmathsetmacro{\angle}{\k*120+60}
        \draw[thick, blue] (0,0) -- (\angle:1.2);
    }
    \foreach \k in {-1,...,1} {
        \pgfmathsetmacro{\angle}{\k*120+(1+(-1)^\k)*10}
        \node at (\angle:0.6) {$S_{\k}$};
    }
    \filldraw (0,0) circle (0.5pt) node[below left]{$0$};
    \end{tikzpicture}
    \caption{
    Plot of the three regions $S_{-1}, S_0, S_1 \subset \mathbb{C}$ as defined in \cref{eq:def-complex-airy-regions}, cf.\ also \cite[413]{olver-asymptoticsspecialfunctions-2010}.
    The Stokes rays $\omega^j \mathbb{R}_-$ are indicated in blue.
    The Airy functions $\Ai_j$ from \cref{eq:def-complex-airy} are recessive in $S_j$ and dominant in $S_{j\pm 1}$, in accordance with the Stokes phenomenon for \cref{eq:airy}.
}
    \label{fig:airy-regions}
\end{figure}
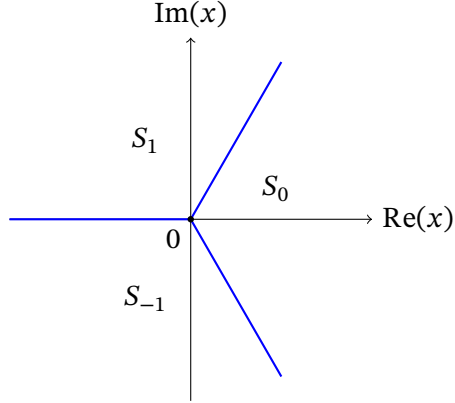

Of particular importance are the boundaries of the sectors $S_j$, given by $\omega^{j\pm 1}\mathbb{R}_-$. Along these rays, the asymptotic behavior of $\Ai_j$ changes discontinuously in the sense described above. This is a manifestation of the \emph{Stokes phenomenon}, and the rays $\omega^{j}\mathbb{R}_-$ are therefore called the \emph{Stokes lines} of the Airy equation.

We now collect some well-known identities and asymptotic expansions for the Airy functions that are used in the main text.
First, the Airy function of the second kind $\Bi(x)$ can be expressed in terms of rotated Airy functions by
\begin{equation}
    \label{eq:airy-ai-identity}
    \Ai_{\mp 1}(x) = \frac{1}{2}\mathrm{e}^{\iu \pi/3} \left(\Ai(x) + \iu \Bi(x) \right)\, ,
\end{equation}
cf.\ \cite[Section~11.8, Eq.~8.04]{olver-asymptoticsspecialfunctions-2010}.

Next, from \cite[Section~11.1, Eq.~1.07–1.16]{olver-asymptoticsspecialfunctions-2010}, we obtain the following standard asymptotic expansions for $x>0$, where $\zeta = \frac{2}{3}x^{3/2}$:
\begin{alignat}{2}
    \label{eq:airy-expansion}
    \Ai(-x) & \sim \frac{1}{\sqrt{\pi}\,x^{1/4}} \cos(\zeta-\pi/4)\, , \quad &
    \Bi(-x) & \sim -\frac{1}{\sqrt{\pi}\,x^{1/4}} \sin(\zeta-\pi/4)\, , \\
    \Ai(x) & \sim \frac{1}{2\sqrt{\pi}\,x^{1/4}} e^{-\zeta}\, , \quad &
    \Bi(x) & \sim \frac{1}{\sqrt{\pi}\,x^{1/4}} e^{\zeta}\, .
\end{alignat}
These properties will now be used to derive the bounds required in \cref{prop:properties-turningpoint-solutions}.
\begin{lemma}
    \label{lem:airy-properties}
    Let $\Ai_j(x)$ be defined as in \cref{eq:def-complex-airy}, and let $\Bi(x)$ denote the Airy function of the second kind.
    Then there exist constants $c_0,c_1,c_2>0$ and $x_0>0$ such that the following holds:
    \begin{enumerate}[(i)]
        \item \label{item:airy-bounded}
        $|\Ai_0(x)| \leq c_0$ for all $x \in \mathbb{R}$;

        \item \label{item:airy-monotone}
        $|\Ai_{\pm 1}(x)|$ is strictly increasing for $x \geq 0$;

        \item \label{item:airy-asymptotics-neg}
        $|\Ai_j(x)|, |\Bi(x)| \leq \frac{c_1}{|x|^{1/4}}$ for all $x < -x_0$;

        \item \label{item:airy-asymptotics-pos}
        $|\Ai_{\pm 1}(x)|, |\Bi(x)| \geq c_2 x^{-1/4} \exp\!\left(\frac{2}{3}x^{3/2}\right)$ for all $x > x_0$.
    \end{enumerate}
\end{lemma}
\begin{proof}
We prove each item separately.

\emph{(i)} The boundedness of $\Ai(x)$ on $\mathbb{R}$ is classical and follows from its oscillatory asymptotics for $x<0$ and exponential decay for $x>0$, cf.\ \cref{eq:airy-expansion}.

\emph{(ii)} Using \cref{eq:airy-ai-identity}, we compute
\begin{equation}
    \label{proofeq:airy-ai-identity}
    |\Ai_{\pm 1}(x)|
    = \frac{1}{2}\sqrt{\Ai(x)^2 + \Bi(x)^2}\, ,
\end{equation}
which is strictly increasing for $x>0$.

\emph{(iii)} For $x>0$, the asymptotic expansions in \cref{eq:airy-expansion} imply that $\Ai(-x)$ and $\Bi(-x)$ are both of order $x^{-1/4}$ up to oscillatory factors. Hence there exist $c_1>0$ and $x_0>0$ such that
\begin{equation}
    |\Ai(-x)| \leq \frac{c_1}{x^{1/4}}, \qquad
    |\Bi(-x)| \leq \frac{c_1}{x^{1/4}}
\end{equation}
for all $x>x_0$, which yields the claim for $j=0$. For $j=\pm 1$, we use \cref{proofeq:airy-ai-identity} to obtain
\begin{equation}
    |\Ai_{\pm 1}(-x)|
    \leq \frac{1}{2}\big(|\Ai(-x)| + |\Bi(-x)|\big)
    \leq \frac{\sqrt{2}c_1}{x^{1/4}}.
\end{equation}

\emph{(iv)} For $x>0$, \cref{eq:airy-expansion} gives $\Bi(x) \sim \pi^{-1/2}x^{-1/4}e^{\frac{2}{3}x^{3/2}}$, hence there exist $c_2>0$ and $x_0>0$ such that
\begin{equation}
    \Bi(x) \geq c_2 x^{-1/4} \exp\!\left(\frac{2}{3}x^{3/2}\right)
\end{equation}
for all $x>x_0$. Since $\Ai_{\pm 1}(x)$ is a linear combination of $\Ai(x)$ and $\Bi(x)$ with nonzero $\Bi$-component (cf.\ \cref{eq:airy-ai-identity}), the same lower bound holds for $|\Ai_{\pm 1}(x)|$ up to adjusting $c_2$.
\end{proof}

We now introduce a Hankel-type auxiliary function $\tilde{w}_1$, following \cite[Eq.~4.4]{geronimo-wkbturningpoint-2004}, up to a harmless normalization factor.
\begin{definition}
    \label{def:hankel-helper}
    Let $H^{(1)}_\nu(x)$ and $H^{(2)}_\nu(x)$ denote the Hankel functions of first and second kind, respectively, cf.\ \cite[Section~7.4]{olver-asymptoticsspecialfunctions-2010}.
    We define the auxiliary functions
    \begin{equation}
        \tilde{w}_j(x)
        = \left(\frac{x}{3}\right)^{1/2}
        \exp\!\left(( -1)^{j+1} \iu \frac{\pi}{6}\right)
        H^{(j)}_{1/3}\!\left(\mathrm{e}^{\iu \pi /2}\zeta\right),
    \end{equation}
    where $\zeta = \frac{2}{3}x^{3/2}$.
\end{definition}
The advantage of $\tilde{w}_1(x)$ compared to $\Ai_0(x)$ is that it shares the same leading asymptotic behavior in the relevant Stokes sector, but does not have real zeros.
For this reason, it is commonly used as an error-weighting function in WKB-type turning point analyses, cf.\ \cite{geronimo-wkbturningpoint-2004}.

The functions $\tilde{w}_j$ are closely related to the Airy functions.
Using Equations 1.05 and 1.14 in \cite[Section~11.1]{olver-asymptoticsspecialfunctions-2010}, together with the identity $(-1)^{3/2} = \mathrm{e}^{\iu \pi/2}$, one obtains for $x<0$:
\begin{align}
    \label{eq:hankel-id-neg}
    \Real \tilde{w}_1(x) &= \Ai(x)\, , \\
    \Imag \tilde{w}_1(x) &= -\Bi(x)\, .
\end{align}
See also \cite[111]{geronimo-wkbturningpoint-2004}.
Furthermore, combining Equation 8.02 in \cite[Section~7.8]{olver-asymptoticsspecialfunctions-2010} with Equation 1.04 in \cite[Section~11.1]{olver-asymptoticsspecialfunctions-2010}, we obtain for $x>0$:
\begin{equation}
    \label{eq:hankel-id-pos}
    \Ai(x) = \frac{\iu}{2}\,\tilde{w}_1(x)\, .
\end{equation}

These identities allow us prove some useful bounds for $\tilde{w}_1(x)$:
\begin{lemma}
    \label{lem:hankel-properties}
    Let $\tilde{w}_1(x)$ be the function from \cref{def:hankel-helper}.
    Then there exist constants $c_3,c_4,c_5>0$ and $x_1<0$ such that the following holds:
    \begin{enumerate}[(i)]
        \item \label{item:hankel-bounded}
        $|\tilde{w}_1(x)| \leq c_3$ for all $x \in \rnum$;
        \item \label{item:hankel-asympotics-neg}
        $\dfrac{c_4}{|x|^{1/4}} \leq |\tilde{w}_1(x)| \leq \dfrac{c_5}{|x|^{1/4}}$ for all $x < x_1$.
    \end{enumerate}
\end{lemma}
\begin{proof}
For $x<0$, combining \cref{eq:hankel-id-neg} with \cref{lem:airy-properties} yields
\begin{equation}
    |\tilde{w}_j(x)| = \sqrt{\Ai(x)^2 + \Bi(x)^2} \leq \frac{\sqrt{2}c_1}{|x|^{1/4}} ,
\end{equation}
which proves the upper bound in \labelcref{item:hankel-asympotics-neg}.
Moreover, using the standard asymptotic expansions in \cref{eq:airy-expansion}, we obtain
\begin{equation}
    |\tilde{w}_j(x)|^2 = \Ai(x)^2 + \Bi(x)^2 \sim \frac{1}{\pi |x|^{1/2}} \quad (x \to -\infty),
\end{equation}
hence there exist constants $c_4>0$ and $x_1<0$ such that
\begin{equation}
    |\tilde{w}_j(x)| \geq \frac{c_4}{|x|^{1/4}}
\end{equation}
for all $x < x_1$.

Finally, \labelcref{item:hankel-bounded} follows from \cref{eq:hankel-id-pos} together with \cref{lem:airy-properties}\labelcref{item:airy-bounded}, and the fact that boundedness on $x<0$ already follows from \labelcref{item:hankel-asympotics-neg}.
\end{proof}
Finally, we provide a bound on the ratio between the Airy function $\Ai_0(x)$ and the helper function $\tilde{w}_1(x)$:
\begin{lemma}
    \label{lem:airy-hankel-ratio}
    Let $\Ai_0(x)$ be the Airy function of the first kind and $\tilde{w}_1(x)$ as in \cref{def:hankel-helper}.
    For $x<0$, define $\chi(x) = \Ai_0(x)/\tilde{w}_1(x)$.
    Let $I \subset \rnum_-$ be a nonempty interval, and $c>0$.
    Then there exists $H>0$ such that, for all $0<h<H$, there exists a subinterval $X_h \subset I$ with length $|X_h| \geq c h$ such that
    \begin{equation}
        |\chi(h^{-2/3} x)| \geq \nicefrac{1}{2} \quad \forall x \in X_h \, .
    \end{equation}
\end{lemma}
\begin{proof}
For $x>0$, combining \cref{eq:hankel-id-neg} with \cref{eq:airy-expansion} yields
\begin{align}
    \Ai(-x) & \sim \frac{1}{\pi^{1/2}x^{1/4}} \cos(\zeta - \pi/4), \\
    |\tilde{w}_j(-x)|^2 &= \Ai(-x)^2 + \Bi(-x)^2 \sim \frac{1}{\pi x^{1/2}},
\end{align}
where $\zeta = \frac{2}{3}x^{3/2}$.
Hence,
\begin{equation}
    \chi(-x) \sim \cos(\zeta - \pi/4) \quad (x \to \infty).
\end{equation}

Let $h>0$ and define
\begin{equation}
    \zeta_I(h) \coloneqq \frac{2}{3}\left(-h^{-2/3} I\right)^{3/2}.
\end{equation}
Since $I \subset \rnum_-$, the set $-h^{-2/3}I \subset \rnum_+$ is an interval whose length diverges as $h \to 0$. Consequently,
\begin{equation}
    |\zeta_I(h)| \to \infty \quad \text{as } h \to 0.
\end{equation}
Therefore, for sufficiently small $h>0$, there exists $\zeta_1 \in \zeta_I(h)$ such that
\begin{equation}
    |\cos(\zeta_2 - \pi/4)| \geq \frac{3}{4} \quad \forall \zeta_2 \in [\zeta_1,\zeta_1+\nicefrac{1}{2}]
\end{equation}
and $\zeta_1 + \nicefrac{1}{2} \in \zeta_I(h)$.
We define the interval $X_h$ as
\begin{equation}
    X_h = \left[-\left(\frac{3}{2}(\zeta_1+\nicefrac{1}{2})\right)^{2/3} h^{2/3}, -\left(\frac{3}{2}\zeta_1\right)^{2/3} h^{2/3}\right] \subset I \, ,
\end{equation}
and obtain for small enough $h$
\begin{equation}
    |X_h| = \left( \left(\frac{3}{2}(\zeta_1+\nicefrac{1}{2})\right)^{2/3} -\left(\frac{3}{2}\zeta_1\right)^{2/3}\right) h^{2/3} \geq c h \, .
\end{equation}
Using the asymptotic expansion above yields
\begin{equation}
    |\chi(h^{-2/3}x)| \geq \frac{1}{2} \quad \forall x \in X_h
\end{equation}
for all sufficiently small $h$, completing the proof.
\end{proof}

\section{Discrete turning point theory}
\label{app:turning-point}
In this appendix, we summarize the discrete WKB and turning point theory developed by Geronimo, Bruno, and Van Assche~\cite{geronimo-wkbturningpoint-2004}; see also the related work of Fedotov and Klopp~\cite{fedotov-complexwkbmethod-2019}. Since throughout this work we only require the case $y=0$ in the notation of~\cite{geronimo-wkbturningpoint-2004}, we suppress the dependence on the spectral parameter $y$ for simplicity.

To distinguish the notation in this appendix from that used in the main text, we employ capital letters and consider two continuous coefficient functions
$
A_1, B_1 \in C([x_i,x_e]\times[0,H]) \, ,
$
where $0 \leq x_i < x_e < \infty$ and $H>0$. We assume $A_1(x,h)>0$
for all $(x,h)\in [x_i,x_e]\times[0,H]$, while $B_1(x,h)$ may be complex-valued.

Given $U=(U(n))_{n\in\nnum} \in \ell(\nnum)$, we consider the second-order difference equation
\begin{equation}
    \label{eq:app-difference-equation}
    A_1(n h,h) U(n+1)
    + B_1(n h,h) U(n)
    + A_1((n-1)h,h) U(n-1)
    = 0
\end{equation}
for all $0<h\leq H$ and $n\in\nnum$ satisfying $x_i \leq nh \leq x_e$. 
Compared with \cite[Eq.~3.37]{geronimo-wkbturningpoint-2004}, we shift the argument in the second off-diagonal coefficient by $-h$ in order to match the Jacobi operator convention used in \cref{eq:generalized-eigenvalue}.

Our goal is to construct two linearly independent solutions $U^{(1)},U^{(2)}$ of \cref{eq:app-difference-equation} and approximate them in terms of Airy functions near a turning point. Following~\cite{geronimo-wkbturningpoint-2004}, we regard \cref{eq:app-difference-equation} as a perturbation of a reference equation with coefficients $A(x,h)$ and $B(x,h)$. The deviation between the two systems plays a central role in the error analysis, and we therefore introduce
\begin{align}
    \label{eq:app-def-perturbation}
    \Delta\!\left[\frac{A(n h,h)}{A(n h-h,h)}\right]
    &=
    \frac{A(n h,h)}{A(n h-h,h)}
    -
    \frac{A_1(n h,h)}{A_1(n h-h,h)} \, , \\
    \Delta\!\left[\frac{B(n h,h)}{A(n h-h,h)}\right]
    &=
    \frac{B(n h,h)}{A(n h-h,h)}
    -
    \frac{B_1(n h,h)}{A_1(n h-h,h)}
\end{align}
for $0<h\leq H$ and $nh\in[x_i,x_e]$.

Following \cite[Eqs.~2.7--2.8]{geronimo-wkbturningpoint-2004}, we next introduce the associated Langer transform. For convenience, we define
\begin{equation}
    \label{eq:app-qfun}
    q(x,h)
    =
    -\frac{B(x,h)}{2A(x-h/2,h)} \, .
\end{equation}
This quantity plays the role of $v(z)/2$ in the notation of~\cite{fedotov-complexwkbmethod-2019}.

We make the following assumptions:
\begin{hypothesis}
    \label{hyp:app-qfun}
    Let $q(x,h)$ be given as in \cref{eq:app-qfun}. Suppose that there exists an open region $O \subset \cnum$ with $[x_i,x_e] \subset O$ such that the following holds:
    \begin{enumerate}[(i)]
        \item For fixed $h \in [0,H]$, the function $x\mapsto q(x,h)$ admits an analytic extension in $x \in O$;
        \item For every $h \in [0,H]$, there exists a unique solution $x_0(h)$ such that $q(x_0(h),h)=1$;
        \item The turning point is simple, i.e. $q'(x_0(h),h) \neq 0$ for all $h \in [0,H]$.
    \end{enumerate}
    Moreover, for $x \in \rnum$ and $h \in [0,H]$, we additionally assume:
    \begin{enumerate}[(i),resume]
        \item $q(x,h) \in C^\infty([x_i,x_e]\times [0,H])$;
        \item $A(x,h) > 0$ and $B(x,h) \in \rnum$;
        \item $-1 < q(x,h)<1$ for $x< x_0(h)$, while $q(x,h) > 1$ for $x > x_0(h)$.
    \end{enumerate}
\end{hypothesis}
These conditions guarantee the existence of a simple turning point at $x=x_0(h)$, separating an oscillatory regime from an exponentially decaying or growing one. Under these assumptions, one can introduce the associated Langer transform.
\begin{definition}
    \label{def:app-langer}
    Let $q(x,h)$ be defined as in \cref{eq:app-qfun}, and suppose that \cref{hyp:app-qfun} holds.
    The \textit{Langer transform} $\rho(x,h)$ is defined by
    \begin{align}
        \label{eq:app-langer}
        \rho(x,h) & = - \left(\frac{3}{2} \int_x^{x_0(h)} \arccos q(u,h) \dl u \right)^{2/3} \quad x \leq x_0(h) \, , \\
        \rho(x,h) & = \left(\frac{3}{2} \int_{x_0(h)}^x \arccosh q(u,h) \dl u \right)^{2/3} \quad x \geq x_0(h) \, ,
    \end{align}
    where $x_0(h)$ denotes the unique solution of $q(x,h)=1$.
\end{definition}
The following lemma shows that the above definition of the Langer transform agrees, for fixed $h$, with the formulation used by Fedotov and Klopp~\cite[Eqs.~2.1--2.4]{fedotov-complexwkbmethod-2019}.
\begin{lemma}
    \label{lem:app-langer-fedotov}
    Let $q(x,h)$ be given as in \cref{eq:app-qfun}, suppose that \cref{hyp:app-qfun} holds, and let $\rho(x,h)$ be defined as in \cref{def:app-langer}. Then
    \begin{equation}
        \rho(x,h)^{3/2}
        =
        \frac{3}{2\iu}
        \int_{x_0(h)}^x p(u,h)\,\dl u \, ,
    \end{equation}
    where $p(x,h)$ denotes the principal branch of
    \begin{equation}
        p(x,h)=\arccos(q(x,h)) \, ,
    \end{equation}
    and the power $z^{3/2}$ is taken with the branch satisfying $(-1)^{3/2}=\iu$.
\end{lemma}

\begin{proof}
    For ease of notation, we suppress the dependence on $h$ throughout the proof.

    By \cref{hyp:app-qfun}, the function $q(x)$ satisfies Hypotheses~2.2 and~2.3 of \cite{fedotov-complexwkbmethod-2019}. In particular, $q(x)\neq1$ for $x\neq x_0$ and $q'(x_0)\neq0$. 
    Moreover, by assumption, $q(x)<1$ for $x<x_0$, while $q(x)>1$ for $x>x_0$. Hence, $\arccos(q(x))\in\rnum_+$ for $x<x_0$, whereas $\arccosh(q(x))\in\rnum_+$ for $x>x_0$. It follows from \cref{def:app-langer} that $\rho(x)\in\rnum$, with $\rho(x)<0$ for $x<x_0$, $\rho(x)>0$ for $x>x_0$, and $\rho(x_0)=0$.

    Using the branch of $z^{3/2}$ satisfying $(-1)^{3/2}=\iu$, we obtain
    \begin{align}
        \rho(x)^{3/2}
        &=
        \iu \frac{3}{2}
        \int_x^{x_0}
        \arccos(q(u))\,\dl u ,
        \qquad x<x_0 ,
        \\
        \rho(x)^{3/2}
        &=
        -\frac{3}{2}
        \int_{x_0}^x
        \arccosh(q(u))\,\dl u ,
        \qquad x>x_0 .
    \end{align}
    By \cite[\S4.23,\S4.37]{NIST:DLMF}, the principal branches satisfy
    \begin{align}
        \arccos(z)
        &=
        \int_z^1 \frac{\dl t}{(1-t^2)^{1/2}} \, ,
        \\
        \arccosh(z)
        &=
        \int_1^z \frac{\dl t}{(t^2-1)^{1/2}} \, ,
    \end{align}
    where the square root is taken on its principal branch. For $t>1$, we have
    \begin{equation}
        (t^2-1)^{1/2}
        =
        \iu (1-t^2)^{1/2} \, ,
    \end{equation}
    and therefore
    \begin{equation}
        \arccosh(z)
        =
        \iu \arccos(z) \, .
    \end{equation}

    Combining the previous identities yields
    \begin{align}
        \rho(x)^{3/2}
        &=
        \frac{3}{2\iu}
        \int_{x_0}^x
        \arccos(q(u))\,\dl u ,
        \qquad x<x_0 ,
        \\
        \rho(x)^{3/2}
        &=
        \frac{3}{2\iu}
        \int_{x_0}^x
        \arccos(q(u))\,\dl u ,
        \qquad x>x_0 ,
    \end{align}
    where in the first line we exchanged the integration bounds.
\end{proof}

Hence, our definition of $\rho(x,h)$ coincides, for fixed $h$, with the Langer transform introduced by Fedotov and Klopp~\cite{fedotov-complexwkbmethod-2019}. The following lemma explains the role of the Langer transform and shows that it converts the local behavior near the turning point into the Airy-type structure underlying the WKB approximation; cf.~\cite[Eq.~2.5]{fedotov-complexwkbmethod-2019} and \cite[Section~11.3]{olver-asymptoticsspecialfunctions-2010} for the continuous setting.

\begin{lemma}
    \label{lem:app-langer-derivative}
    Let $\rho(x,h)$ be the Langer transform from \cref{def:app-langer}. Then
    \begin{align}
        \label{eq:app-langer-derivative-neg}
        \rho(x,h)^{1/2}\rho'(x,h)
        &= -\iu \arccos(q(x,h))
        \qquad x \leq x_0(h)\, , \\
        \label{eq:app-langer-derivative-pos}
        \rho(x,h)^{1/2}\rho'(x,h)
        &= -\arccosh(q(x,h))
        \qquad x \geq x_0(h)\, ,
    \end{align}
    where the branch of the square root fixed in \cref{lem:app-langer-fedotov} is used.
\end{lemma}

\begin{proof}
    Again, throughout the proof we suppress the dependence on $h$ for notational simplicity.

    We first consider the case $x \geq x_0$ and define
    \begin{equation}
        \label{eq:proof-langer-pos-F}
        F(x)
        \coloneqq
        \frac{3}{2}
        \int_{x_0}^{x}
        \arccosh(q(u))
        \,\dl u .
    \end{equation}
    By definition of the Langer transform,
    \begin{equation}
        \rho(x)=F(x)^{2/3} .
    \end{equation}
    Since $\rho(x)\geq0$ in this regime, the chosen branch gives
    \begin{equation}
        \rho(x)^{1/2}=-F(x)^{1/3} .
    \end{equation}
    Differentiating yields
    \begin{equation}
        \rho'(x)
        =
        \frac{2}{3}F(x)^{-1/3}F'(x)
        =
        F(x)^{-1/3}\arccosh(q(x)),
    \end{equation}
    and therefore
    \begin{equation}
        \rho(x)^{1/2}\rho'(x)
        =
        -\arccosh(q(x)),
    \end{equation}
    proving \cref{eq:app-langer-derivative-pos}.

    Next, let $x<x_0$ and define
    \begin{equation}
        \label{eq:proof-langer-neg-F}
        G(x)
        \coloneqq
        \frac{3}{2}
        \int_x^{x_0}
        \arccos(q(u))
        \,\dl u .
    \end{equation}
    Then
    \begin{equation}
        \rho(x)=-G(x)^{2/3}.
    \end{equation}
    Since $\rho(x)<0$, the branch convention from \cref{lem:app-langer-fedotov} implies
    \begin{equation}
        \rho(x)^{1/2}
        =
        (-1)^{1/2}G(x)^{1/3}
        =
        -\iu G(x)^{1/3}.
    \end{equation}
    Moreover,
    \begin{equation}
        \rho'(x)
        =
        -\frac{2}{3}G(x)^{-1/3}G'(x)
        =
        G(x)^{-1/3}\arccos(q(x)),
    \end{equation}
    where we used
    \begin{equation}
        G'(x)
        =
        -\frac{3}{2}\arccos(q(x)).
    \end{equation}
    Combining the previous identities yields
    \begin{equation}
        \rho(x)^{1/2}\rho'(x)
        =
        -\iu \arccos(q(x)),
    \end{equation}
    proving \cref{eq:app-langer-derivative-neg}.
\end{proof}

We come back to the construction of two linearly independent solutions $U^{(1)}$ and $U^{(2)}$ of the difference equation~\eqref{eq:app-difference-equation} by means of Airy-type approximations near the turning point $x_0(h)$. 
The Langer transform $\rho(x,h)$ introduced above plays the role of a canonical local coordinate in which the difference equation asymptotically reduces to the Airy equation. 
Accordingly, the solutions $U^{(1)}$ and $U^{(2)}$ will be approximated by Airy functions evaluated at $h^{-2/3}\rho(x,h)$.

To formulate the approximation precisely, we introduce the weight function
\begin{equation}
    \label{eq:app-weightfunction}
    g(x,h) =
    \left(
        \frac{
            \rho(x,h)
        }{
            A(x-h/2,h)^2
            \sinh^2\!\left(
                \rho(x,h)^{1/2}\rho'(x,h)
            \right)
        }
    \right)^{1/4},
\end{equation}
where the derivative is taken with respect to $x$.
Using \cref{lem:app-langer-derivative}, this definition is equivalent to \cite[Eq.~2.9]{geronimo-wkbturningpoint-2004}.

Geronimo, Bruno, and Van Assche impose the following additional regularity assumptions on the Langer transform and the associated weight function~\cite[120]{geronimo-wkbturningpoint-2004}:
\begin{hypothesis}
    \label{hyp:geronimo}
    Let $\rho(x,h)$ and $g(x,h)$ be as from \cref{def:app-langer,eq:app-weightfunction}, and suppose that the following holds:
    \begin{enumerate}[(i)]
        \item $\rho,g \in C^\infty([x_i,x_e] \times [0,H])$;
        \item $(x,h)\mapsto \rho(x,h)/(x-x_0(h)) \in C^2([x_i,x_e]) \times C^0([0,H])$ and is uniformly bounded away from $0$;
        \item for fixed $h$, the function $\Real \rho(x,h)^{3/2}$ is non-decreasing in $x$.
    \end{enumerate}
\end{hypothesis}
Under these assumptions, one constructs approximate solutions of the difference equation~\eqref{eq:app-difference-equation} by combining Airy functions with the Langer transform and the weight function $g(x,h)$. More precisely, one considers functions of the form
\begin{equation}
    \label{eq:app-psi}
    \Psi(x,h)
    =
    g(x,h)\,
    \chi\!\left(h^{-2/3}\rho(x,h)\right),
\end{equation}
where $\chi(z)$ is a solution of the Airy equation~\eqref{eq:airy}. 
The scaling $h^{-2/3}\rho(x,h)$ reflects the characteristic Airy behavior near the turning point $x_0(h)$.

The following theorem shows that the functions $\Psi(x,h)$ satisfy the difference equation associated with the coefficients $A,B$ up to a controlled error term.

\begin{theorem}[{\cite[Theorem~4.1]{geronimo-wkbturningpoint-2004}}]
    \label{thm:app-geronimo-1}
    Let $\chi(z)$ be an entire solution of the Airy equation~\eqref{eq:airy}, and define
    \begin{equation}
        \Psi(x,h)
        =
        g(x,h)\chi\!\left(h^{-2/3}\rho(x,h)\right).
    \end{equation}
    Suppose that Hypotheses \labelcref{hyp:app-qfun,hyp:geronimo} hold.
    Then there exists a constant $C>0$, independent of $(x,h)\in [x_i,x_e]\times[0,H]$, such that
    \begin{multline}
        \label{eq:app-geronimo-equation}
        A(x,h)\Psi(x+h,h)
        -2A(x-h/2,h)\cosh\!\left(\rho(x,h)^{1/2}\rho'(x,h)\right)\Psi(x,h)
        \\
        +A(x-h,h)\Psi(x-h,h)
        =
        \beta^{(\chi)}(x,h),
    \end{multline}
    where
    \begin{equation}
        \label{eq:app-geronimo-error}
        |\beta^{(\chi)}(x,h)|
        \leq
        Ch^2
    \end{equation}
    for all $h>0$ and $x_i+h\leq x\leq x_e-h$.
\end{theorem}

Using \cref{lem:app-langer-derivative}, the coefficient involving the hyperbolic cosine can be rewritten in terms of $q(x,h)$. Indeed, for $x\leq x_0(h)$ we have
\begin{equation}
    \rho(x,h)^{1/2}\rho'(x,h)
    =
    -\iu \arccos(q(x,h)),
\end{equation}
while for $x\geq x_0(h)$,
\begin{equation}
    \rho(x,h)^{1/2}\rho'(x,h)
    =
    -\arccosh(q(x,h)).
\end{equation}
Thus, using the identities $\cosh(-\iu z)=\cos(-z)=\cos(z)$, $\cos(\arccos z)=z$, and $\cosh(-\arccosh z)=z$, we obtain
\begin{equation}
    \cosh\!\left(\rho(x,h)^{1/2}\rho'(x,h)\right)
    =
    q(x,h)
\end{equation}
for all $x\in[x_i,x_e]$.
Recalling the definition
\begin{equation}
    q(x,h)
    =
    -\frac{B(x,h)}{2A(x-h/2,h)},
\end{equation}
it follows that \cref{eq:app-geronimo-equation} can equivalently be written as
\begin{equation}
    \label{eq:app-geronimo-equation-b}
    A(x,h)\Psi(x+h,h)
    +
    B(x,h)\Psi(x,h)
    +
    A(x-h,h)\Psi(x-h,h)
    =
    \beta^{(\chi)}(x,h).
\end{equation}
Thus, the functions $\Psi(x,h)$ satisfy the difference equation associated with the coefficients $A,B$ up to an error of order $h^2$.

However, the existence of an approximate solution $\Psi(x,h)$ satisfying
\cref{eq:app-geronimo-equation-b} does not by itself imply the existence of an exact solution of the corresponding difference equation which remains close to $\Psi(x,h)$. A fortiori, it does not imply that the original difference equation~\eqref{eq:app-difference-equation} with coefficients $A_1,B_1$ admits such a solution. Establishing this requires a detailed error analysis. Such an analysis was carried out independently by Geronimo, Bruno and Van Assche~\cite{geronimo-wkbturningpoint-2004} and by Fedotov and Klopp~\cite{fedotov-complexwkbmethod-2019}. Here we summarize the approach of Geronimo et al.~\cite{geronimo-wkbturningpoint-2004}.

To this end, we introduce two distinguished approximate solutions constructed from the Airy functions $\Ai_j$ defined in \cref{eq:def-complex-airy}, cf.~\cite[Eqs.~4.8--4.9]{geronimo-wkbturningpoint-2004}:
\begin{equation}
    \label{eq:app-def-approx-solution}
    \Psi^{(1)}(x,h)
    =
    g(x,h)\Ai_0\!\left(h^{-2/3}\rho(x,h)\right),
    \qquad
    \Psi^{(2)}(x,h)
    =
    g(x,h)\Ai_1\!\left(h^{-2/3}\rho(x,h)\right).
\end{equation}
Since $\Ai_0$ possesses zeros on $\rnum_-$, it is not suitable for measuring relative errors of the form
\begin{equation}
    U_n^{(1)}
    =
    \Psi^{(1)}(nh,h)
    +
    \Psi^{(1)}(nh,h)\sigma_n^{(1,h)},
\end{equation}
as the denominator may vanish. Following \cite{geronimo-wkbturningpoint-2004}, one therefore introduces the auxiliary functions
\begin{equation}
    \label{eq:app-def-error-helper}
    W^{(1)}(x,h)
    =
    g(x,h)\tilde{w}_1\!\left(h^{-2/3}\rho(x,h)\right),
    \qquad
    W^{(2)}(x,h)
    =
    g(x,h)\Ai_1\!\left(h^{-2/3}\rho(x,h)\right),
\end{equation}
where $\tilde{w}_1$ is the helper function introduced in \cref{def:hankel-helper}. In contrast to $\Ai_0$, the function $\tilde{w}_1$ has no zeros on the real line while retaining the same asymptotic behavior on $\rnum_-$.

We denote the corresponding residual errors from \cref{thm:app-geronimo-1} by
\begin{equation}
    \label{eq:app-beta-errors}
    \beta^{(1)}
    =
    \beta^{\Ai_0},
    \qquad
    \beta^{(2)}
    =
    \beta^{\Ai_1},
    \qquad
    \beta_1^{(1)}
    =
    \beta^{\tilde{w}_1},
    \qquad
    \beta_1^{(2)}
    =
    \beta^{\Ai_1}.
\end{equation}
Since both Airy functions and the auxiliary function $\tilde{w}_1$ satisfy the Airy equation, \cref{thm:app-geronimo-1} applies to all four cases above. Consequently, it follows from \cite[119]{geronimo-wkbturningpoint-2004} that there exists a constant $C>0$ such that
\begin{equation}
    \label{eq:app-beta-errors-estimate}
    |\beta^{(j)}|
    \leq
    Ch^2,
    \qquad
    |\beta_1^{(j)}|
    \leq
    Ch^2,
    \qquad j=1,2,
\end{equation}
for all $h>0$ and $x_i+h\leq x\leq x_e-h$.

Finally, using the functions $W^{(j)}$, Geronimo et al.~define propagation operators $G^{(j)}(n,m,h)$, $j=1,2$, cf.~\cite[Eqs.~3.39 and 3.42]{geronimo-wkbturningpoint-2004}, which are used to control the accumulated error between approximate and exact solutions. Their explicit form is not needed for our purposes. The following estimate will be sufficient.

\begin{lemma}[{\cite[Lemmas~4.2--4.3]{geronimo-wkbturningpoint-2004}}]
    \label{lem:app-error-propagator}
    Suppose that Hypothesis~\ref{hyp:geronimo} holds and that
    $\rho(x,h) \in S_0 \cup S_1$ for all
    $(x,h) \in [x_i,x_e]\times[0,H]$, where the sectors
    $S_j \subset \mathbb{C}$ are defined in
    \cref{eq:def-complex-airy-regions}.
    Furthermore, assume that
    $A(x,h) \in C^\infty([x_i,x_e]\times[0,H])$.    
    Then there exist positive functions
    $G^{(j)}(n,h)$, $j=1,2$, and a constant
    $c_1>0$, independent of $A_1,B_1$, such that the following estimates hold for all
    $h \in (0,H)$, $N_1<N_2 \in \nnum$ with
    $N_1 h, N_2 h \in [x_i,x_e]$, and
    $N_1 \leq n,m \leq N_2$:
    \begin{align}
        |G^{(j)}(n,m,h)|
        & \leq G^{(j)}(n,h),
        \qquad j=1,2 \, ,
        \\
        \sum_{n=N_1}^{N_2} G^{(j)}(n,h)
        & \leq c_1 h^{-1},
        \qquad j=1,2 \, .
    \end{align}
\end{lemma}

These error estimates are crucial for the following theorem, which shows that the approximate solutions $\Psi^{(j)}$ can indeed be promoted to genuine solutions of the difference equation~\eqref{eq:app-difference-equation}.

\begin{theorem}[{\cite[Theorem~4.4]{geronimo-wkbturningpoint-2004}}]
    \label{thm:app-geronimo-2}
    Let $\Psi^{(j)},W^{(j)}$, $j=1,2$ be defined as in \cref{eq:app-def-approx-solution,eq:app-def-error-helper}.
    Assume that \cref{hyp:geronimo} holds and that $\rho(x,h) \in S_0 \cup S_1$, as defined in \cref{eq:def-complex-airy-regions} for all $(x,h) \in [x_i,x_e] \times [0,H]$.
    Suppose further that $A(x,h) \in C^\infty([x_i,x_e] \times [0,H])$ is strictly positive, that $A_1(x,h),B_1(x,h) \in C^0([x_i,x_e] \times [0,H])$ and that $A_1(x,h)$ is strictly positive.
    For $h \in (0,H]$ and $N_1<N_2 \in \nnum$, let $[(N_1+1)h,(N_2-1)h] \subset [x_i,x_e]$.
    For $j=1,2$, let $G^{(j)}(n,h)$, $\beta^{(j)}(x,h)$ and $\beta^{(j)}_1(x,h)$ be given as in \cref{lem:app-error-propagator,eq:app-beta-errors}.
    Furthermore, let
    \begin{equation}
        K^{(j)}(i,h) = c_2 G^{(j)}(i) \left(\Delta \left[\frac{A(i h,h)}{A(i h-h,h)}\right] +\Delta \left[\frac{B(i h,h)}{A(i h-h,h)}\right]\right) \quad j =1,2
    \end{equation}
    with the definitions from \cref{eq:app-def-perturbation} and some constant $c_2>0$ independent of $A_1,B_1$ given by \cite{geronimo-wkbturningpoint-2004}.
    Then there exists two linearly independent solutions $U^{(j,h)}$, $j=1,2$ of \cref{eq:app-difference-equation} and a constant $C>0$ independent of $A_1,B_1$ such that
    \begin{equation}
        U^{(j,h)}_n = \Psi^{(j)}(nh,h)+W^{(j)}(nh,h) \sigma^{(j,h)}_n \quad j =1,2\,\quad N_1 \leq n \leq N_2 \, ,
    \end{equation}
    where the relative error $\sigma^{(j,h)}_n$ is bounded by
    \begin{equation}
        \sigma^{(j,h)}_n \leq C \sum_{i = N_1}^{N_2} \left(\left|K^{(j)}(i,h)\right| + G^{(j)}(i,h) \left| \beta^{(j)}(ih,h)\right|\right) \exp\left(\sum_{i = N_1}^{N_2} \left(\left|K^{(j)}(i,h)\right| + G^{(j)}(i,h) \left| \beta^{(j)}_1(ih,h)\right|\right)\right) \, .
    \end{equation}
\end{theorem}
\begin{remark}
    Under additional assumptions, \cite[Theorem~4.4]{geronimo-wkbturningpoint-2004} yields considerably simpler error bounds.
    While the independence of the constants $c_1$, $c_2$, and $C$ from the perturbative coefficients $A_1,B_1$ is not stated explicitly in \cite[Theorem~4.4]{geronimo-wkbturningpoint-2004}, it follows from the proof.
    In particular, the argument provides explicit procedures for estimating these constants from above.
\end{remark}

Geronimo et al.'s framework thus provides a method for constructing approximate solutions of \cref{eq:app-difference-equation} in terms of Airy functions and the Langer transform.
To apply these results, one must verify Hypothesis~\labelcref{hyp:app-qfun,hyp:geronimo}, together with the remaining assumptions of \cref{thm:app-geronimo-2}.
The following lemma shows that many of these properties already follow from Hypothesis~\ref{hyp:app-qfun}.

\begin{lemma}
    \label{lem:app-hypothesis}
    Assume that Hypothesis~\ref{hyp:app-qfun} holds with open region $O \subset \cnum$.
    Then the following statements hold for all
    $(x,h) \in [x_i,x_e]\times[0,H]$:
    \begin{enumerate}[(i)]
        \item \label{item:app-langer-sj}
        $\rho(x,h) \in \rnum$.
        In particular,
        $\rho(x,h) \in S_0 \subset S_0 \cup S_1$,
        where the sectors $S_j$ are defined in
        \cref{eq:def-complex-airy-regions};

        \item \label{item:app-langer-posneg}
        $\rho(x,h) < 0$ for $x < x_0(h)$, and
        $\rho(x,h) > 0$ for $x > x_0(h)$;

        \item \label{item:app-langer-increasing}
        $\mathrm{Re}\,\rho(x,h)^{3/2}$ is strictly increasing in $x$;

        \item \label{item:app-langer-analytic}
        $\rho(x,h) \in C^\infty([x_i,x_e]\times[0,H])$.
        Moreover, for fixed $h \in [0,H]$,
        the map $x \mapsto \rho(x,h)$ is analytic in $O$;

        \item \label{item:app-langer-lowerbound}
        $\rho(x,h)/(x-x_0(h))
        \in C^2([x_i,x_e]\times[0,H])$
        and is uniformly bounded away from $0$;

        \item \label{item:app-gfunc-analytic}
        $g(x,h) \in C^\infty([x_i,x_e]\times[0,H])$,
        and for fixed $h \in [0,H]$, $x \mapsto g(x,h)$ is analytic in $O$;

        \item \label{item:app-gfunc-bound}
        $A(x-h/2,h)^{1/2} g(x,h)$
        is uniformly bounded above and below away from $0$
        on $[x_i,x_e]\times[0,H]$.
    \end{enumerate}
    In particular, Hypothesis~\ref{hyp:geronimo} is satisfied.
\end{lemma}

\begin{proof}
By assumption, $q(x,h)\in[-1,1]$ for $x\leq x_0(h)$, hence
$\arccos(q(x,h))\in \rnum_{\ge 0}$ in this region.
Similarly, $q(x,h)\ge 1$ for $x\ge x_0(h)$, so
$\arccosh(q(x,h))\in \rnum_{\ge 0}$.
Statements \ref{item:app-langer-sj}--\ref{item:app-langer-increasing}
follow directly from \cref{def:app-langer} and \cref{eq:def-complex-airy-regions}.

We prove \ref{item:app-langer-analytic}.
The function $q(x,h)$ is $C^\infty$ in $(x,h)$ and analytic in $x$ for fixed $h$.
Moreover, for each $h$, the equation~$q(x,h)=1$ has a unique solution $x_0(h)$ with $q'(x_0(h),h)\neq 0$.
By the implicit function theorem, $x_0(h)$ is locally $C^\infty$ in $h$; global $C^\infty$-regularity on $[0,H]$ follows from uniqueness and compactness.

By Lemma~\ref{lem:app-langer-fedotov},
\begin{equation}
\rho(x,h)=\left(\frac{3}{2\iu}\int_{x_0(h)}^x \arccos q(u,h)\,\mathrm{d}u\right)^{2/3},
\end{equation}
with the branch choice as specified there.
This representation implies analyticity in $x$, and together with smooth dependence of $x_0(h)$ yields
$\rho \in C^\infty([x_i,x_e]\times[0,H])$.

Next we prove \ref{item:app-langer-lowerbound}.
For fixed $h$, $\rho(x,h)$ is analytic in $x$, hence
\begin{equation}
\rho(x,h)=\sum_{n=0}^\infty \frac{\partial_x^n\rho(x_0(h),h)}{n!}(x-x_0(h))^n.
\end{equation}
Since $\rho(x_0(h),h)=0$ and $\partial_x\rho(x_0(h),h)\neq 0$ (cf. Fedotov--Klopp~\cite[4416]{fedotov-complexwkbmethod-2019}),
we obtain a non-degenerate linear term.
Thus
\begin{equation}
\frac{\rho(x,h)}{x-x_0(h)}
=\sum_{n=1}^\infty \frac{\partial_x^n\rho(x_0(h),h)}{n!}(x-x_0(h))^{n-1}
\end{equation}
extends continuously (and in fact $C^2$) to $x=x_0(h)$.
Moreover, $\partial_x\rho(x_0(h),h)$ depends continuously on $h$ and is nowhere vanishing on the compact interval $[0,H]$, hence it is uniformly bounded away from zero. 
This yields the claim.

For the remaining statements we follow Fedotov and Klopp~\cite{fedotov-complexwkbmethod-2019}.
Define
\begin{equation}
\tilde g(x,h)=\frac{\sinh(\rho(x,h)^{1/2}\rho'(x,h))}{\rho(x,h)^{1/2}}.
\end{equation}
Then $g(x,h)=(A(x-h/2,h)\tilde g(x,h))^{-1/2}$.
By Fedotov--Klopp, $\tilde g$ is analytic and nonvanishing in $x\in O$ for fixed $h \geq 0$~\cite[4417]{fedotov-complexwkbmethod-2019}.
Moreover, near $x_0(h)$,
$\sinh z \sim z$ implies
\begin{equation}
\tilde g(x,h) \sim \rho'(x_0(h),h)\neq 0.
\end{equation}
Thus $\tilde g$ is uniformly bounded away from zero on $[x_i,x_e]\times[0,H]$.
This implies the boundedness properties of $g$ and completes the proof.
\end{proof}
To apply the results of Geronimo et al., it is therefore sufficient to verify the assumptions on $q(x,h)$ stated in Hypothesis~\ref{hyp:app-qfun}. The corresponding analysis is carried out in the main text in \cref{sec:turning-point}, where we follow the framework developed in this appendix.

\printbibliography

@article{fischer2025wrong,
  title = {Quantum Particle in the Wrong Box (or: The Perils of Finite-Dimensional Approximations)},
  shorttitle = {Quantum Particle in the Wrong Box (Or},
  author = {Fischer, Felix and Burgarth, Daniel and Lonigro, Davide},
  date = {2026-01-27},
  journaltitle = {Quantum},
  volume = {10},
  pages = {1985},
  publisher = {Verein zur Förderung des Open Access Publizierens in den Quantenwissenschaften},
  doi = {10.22331/q-2026-01-27-1985}
}

@book{akhiezer-classicalmomentproblem-2020,
  title = {The {{Classical Moment Problem}} and {{Some Related Questions}} in {{Analysis}}},
  author = {Akhiezer, N. I.},
  date = {2020-01},
  publisher = {{Society for Industrial and Applied Mathematics}},
  location = {Philadelphia, PA},
  doi = {10.1137/1.9781611976397},
  isbn = {978-1-61197-638-0 978-1-61197-639-7}
}

@book{albeverio-singularperturbationsdifferential-2000,
  title = {Singular Perturbations of Differential Operators: Solvable {{Schrödinger}} Type Operators},
  shorttitle = {Singular Perturbations of Differential Operators},
  author = {Albeverio, Sergio},
  namea = {Kurasov, P.},
  nameatype = {collaborator},
  date = {2000},
  publisher = {Cambridge University Press},
  location = {Cambridge},
  isbn = {978-0-521-77912-8 978-1-107-10128-9},
  pagetotal = {429}
}

@article{allahverdiev-extensionsdilationsfunctional-2005,
  title = {Extensions, {{Dilations}} and {{Functional Models}} of {{Infinite Jacobi Matrix}}},
  author = {Allahverdiev, B. P.},
  date = {2005-09},
  journaltitle = {Czechoslovak Mathematical Journal},
  shortjournal = {Czech Math J},
  volume = {55},
  number = {3},
  pages = {593--609},
  issn = {0011-4642, 1572-9141},
  doi = {10.1007/s10587-005-0048-3}
}

@article{ammann-relativeoscillationtheory-2012,
  title = {Relative {{Oscillation Theory}} for {{Jacobi Matrices Extended}}},
  author = {Ammann, Kerstin},
  date = {2014},
 journal = {Operators and Matrices},
 volume = {8},
 pages = {99–115},
  doi = {10.7153/oam-08-04},
  pubstate = {prepublished}
}

@article{aptekarev-measuresorthogonalpolynomials-2016,
  title = {Measures for Orthogonal Polynomials with Unbounded Recurrence Coefficients},
  author = {Aptekarev, A. I. and Geronimo, J. S.},
  date = {2016-07-01},
  journaltitle = {Journal of Approximation Theory},
  shortjournal = {Journal of Approximation Theory},
  volume = {207},
  pages = {339--347},
  issn = {0021-9045},
  doi = {10.1016/j.jat.2016.02.009}
}

@article{ashhab-finitedimensionalapproximationsgeneralized-2026,
  title = {Finite-Dimensional Approximations of Generalized Squeezing},
  author = {Ashhab, Sahel and Fischer, Felix and Lonigro, Davide and Braak, Daniel and Burgarth, Daniel},
  date = {2026-01-02},
  journaltitle = {Physical Review A},
  shortjournal = {Phys. Rev. A},
  volume = {113},
  number = {1},
  pages = {013703},
  issn = {2469-9926, 2469-9934},
  doi = {10.1103/9vwp-f35c}
}

@online{ashhab-fractionalsqueezingspectra-2026,
  title = {Fractional Squeezing: Spectra and Dynamics from Generalized Squeezing {{Hamiltonian}} with Fractional Orders},
  shorttitle = {Fractional Squeezing},
  author = {Ashhab, Sahel},
  date = {2026-01-22},
  eprint = {2601.15693},
  eprinttype = {arXiv},
  eprintclass = {quant-ph},
  doi = {10.48550/arXiv.2601.15693},
  pubstate = {prepublished}
}

@article{ashhab-propertiesdynamicsgeneralized-2025,
  title = {Properties and Dynamics of Generalized Squeezed States},
  author = {Ashhab, Sahel and Ayyash, Mohammad},
  date = {2025-05-01},
  journaltitle = {New Journal of Physics},
  shortjournal = {New J. Phys.},
  volume = {27},
  number = {5},
  pages = {054104},
  publisher = {IOP Publishing},
  issn = {1367-2630},
  doi = {10.1088/1367-2630/add7fc}
}

@article{ashhab-replycommentproperties-2026,
  title = {Reply to {{Comment}} on ‘{{Properties}} and Dynamics of Generalized Squeezed States’},
  author = {Ashhab, Sahel and Ayyash, Mohammad},
  date = {2026-02},
  journaltitle = {New Journal of Physics},
  shortjournal = {New J. Phys.},
  volume = {28},
  number = {2},
  pages = {028001},
  publisher = {IOP Publishing},
  issn = {1367-2630},
  doi = {10.1088/1367-2630/ae437b}
}

@article{ayyash-dispersiveregimemultiphoton-2025,
  title = {Dispersive regime of multiphoton qubit-oscillator interactions},
  author = {Ayyash, Mohammad and Ashhab, Sahel},
  journal = {Phys. Rev. A},
  volume = {112},
  issue = {2},
  pages = {023713},
  numpages = {19},
  year = {2025},
  month = {8},
  publisher = {American Physical Society},
  doi = {10.1103/hrc2-7nqg},
  url = {https://link.aps.org/doi/10.1103/hrc2-7nqg}
}

@article{bawin-singularinversesquare-2003a,
  title = {Singular Inverse Square Potential, Limit Cycles, and Self-Adjoint Extensions},
  author = {Bawin, M. and Coon, S. A.},
  date = {2003-04-23},
  journaltitle = {Physical Review A},
  shortjournal = {Phys. Rev. A},
  volume = {67},
  number = {4},
  pages = {042712},
  issn = {1050-2947, 1094-1622},
  doi = {10.1103/PhysRevA.67.042712}
}

@article{beane-singularpotentialslimit-2001,
  title = {Singular Potentials and Limit Cycles},
  author = {Beane, S. R. and Bedaque, P. F. and Childress, L. and Kryjevski, A. and McGuire, J. and Van Kolck, U.},
  date = {2001-09-10},
  journaltitle = {Physical Review A},
  shortjournal = {Phys. Rev. A},
  volume = {64},
  number = {4},
  pages = {042103},
  issn = {1050-2947, 1094-1622},
  doi = {10.1103/PhysRevA.64.042103}
}

@book{behrndt-boundaryvalueproblems-2020,
  title = {Boundary {{Value Problems}}, {{Weyl Functions}}, and {{Differential Operators}}},
  author = {Behrndt, Jussi and Hassi, Seppo and De Snoo, Henk},
  date = {2020},
  series = {Monographs in {{Mathematics}}},
  volume = {108},
  publisher = {Springer International Publishing},
  location = {Cham},
  doi = {10.1007/978-3-030-36714-5},
  isbn = {978-3-030-36713-8 978-3-030-36714-5}
}

@article{bencheikh-demonstratingquantumproperties-2022,
  title = {Demonstrating Quantum Properties of Triple Photons Generated by $\chi^3$ processes},
  author = {Bencheikh, Kamel and Cenni, Marina F. B. and Oudot, Enky and Boutou, Véronique and Félix, Corinne and Prades, Joel Compte and Vernay, Augustin and Bertrand, Julien and Bassignot, Florent and Chauvet, Mathieu and Bussières, Félix and Zbinden, Hugo and Levenson, Ariel and Boulanger, Benoît},
  date = {2022-10-10},
  journaltitle = {The European Physical Journal D},
  shortjournal = {Eur. Phys. J. D},
  volume = {76},
  number = {10},
  pages = {186},
  issn = {1434-6079},
  doi = {10.1140/epjd/s10053-022-00514-3}
}

@article{berg-indeterminatejacobioperators-2025,
  title = {Indeterminate {{Jacobi}} Operators},
  author = {Berg, Christian and Szwarc, Ryszard},
  date = {2025-04-30},
  journaltitle = {Journal of Operator Theory},
  shortjournal = {J. Operator Theory},
  volume = {93},
  number = {1},
  eprint = {2301.00586},
  eprinttype = {arXiv},
  eprintclass = {math},
  pages = {123--145},
  issn = {03794024, 18417744},
  doi = {10.7900/jot.2023ian02.2404}
}

@article{bernhardbeckermann-complexjacobimatrices-2001,
  title = {Complex {{Jacobi}} Matrices},
  author = {{Bernhard Beckermann}},
  date = {2001-01},
  journaltitle = {Journal of Computational and Applied Mathematics},
  shortjournal = {Journal of Computational and Applied Mathematics},
  volume = {127},
  number = {1--2},
  pages = {17--65},
  issn = {03770427},
  doi = {10.1016/S0377-0427(00)00492-1}
}

@article{birkhoff-analytictheorysingular-1933,
  title = {Analytic Theory of Singular Difference Equations},
  author = {Birkhoff, George D. and Trjitzinsky, W. J.},
  date = {1933},
  journaltitle = {Acta Mathematica},
  shortjournal = {Acta Math.},
  volume = {60},
  number = {0},
  pages = {1--89},
  issn = {0001-5962},
  doi = {10.1007/BF02398269}
}

@article{birkhoff-formaltheoryirregular-1930,
  title = {Formal Theory of Irregular Linear Difference Equations},
  author = {Birkhoff, George D.},
  date = {1930},
  journaltitle = {Acta Mathematica},
  shortjournal = {Acta Math.},
  volume = {54},
  number = {0},
  pages = {205--246},
  issn = {0001-5962},
  doi = {10.1007/BF02547522}
}

@article{borg-pauliapproximationsselfadjoint-2003,
  title = {Pauli Approximations to the Self-Adjoint Extensions of the {{Aharonov}}–{{Bohm Hamiltonian}}},
  author = {Borg, J. L. and Pulé, J. V.},
  date = {2003-10-01},
  journaltitle = {Journal of Mathematical Physics},
  volume = {44},
  number = {10},
  pages = {4385--4410},
  issn = {0022-2488, 1089-7658},
  doi = {10.1063/1.1601298}
}

@article{boutin-effecthigherordernonlinearities-2017,
  title = {Effect of {{Higher-Order Nonlinearities}} on {{Amplification}} and {{Squeezing}} in {{Josephson Parametric Amplifiers}}},
  author = {Boutin, Samuel and Toyli, David M. and Venkatramani, Aditya V. and Eddins, Andrew W. and Siddiqi, Irfan and Blais, Alexandre},
  date = {2017-11-15},
  journaltitle = {Physical Review Applied},
  shortjournal = {Phys. Rev. Applied},
  volume = {8},
  number = {5},
  pages = {054030},
  issn = {2331-7019},
  doi = {10.1103/PhysRevApplied.8.054030}
}

@Inbook{braak-$k$photonquantumrabi-2025,
author="Braak, Daniel",
editor="Takagi, Tsuyoshi
and Wakayama, Masato
and Kunihiro, Noboru
and Tanaka, Keisuke
and Kimoto, Kazufumi
and Kudo, Momonari",
title="The k-Photon Quantum Rabi Model",
bookTitle="Mathematical Foundations for Post-Quantum Cryptography: Crypto-Math CREST",
year="2026",
publisher="Springer Nature Singapore",
address="Singapore",
pages="75--87",
isbn="978-981-96-1218-5",
doi="10.1007/978-981-96-1218-5_5",
url="https://doi.org/10.1007/978-981-96-1218-5_5"
}

@article{braeutigam-deficiencynumbersoperators-2016,
  title = {Deficiency Numbers of Operators Generated by Infinite {{Jacobi}} Matrices},
  author = {Braeutigam, I. N. and Mirzoev, K. A.},
  date = {2016-03},
  journaltitle = {Doklady Mathematics},
  shortjournal = {Dokl. Math.},
  volume = {93},
  number = {2},
  pages = {170--174},
  issn = {1064-5624, 1531-8362},
  doi = {10.1134/S1064562416020137}
}

@article{braunstein-generalizedsqueezing-1987,
  title = {Generalized Squeezing},
  author = {Braunstein, Samuel L. and McLachlan, Robert I.},
  date = {1987-02-01},
  journaltitle = {Physical Review A},
  shortjournal = {Phys. Rev. A},
  volume = {35},
  number = {4},
  pages = {1659--1667},
  issn = {0556-2791},
  doi = {10.1103/PhysRevA.35.1659}
}

@article{braunstein-squeezingirreducibleresource-2005,
  title = {Squeezing as an Irreducible Resource},
  author = {Braunstein, Samuel L.},
  date = {2005-05-31},
  journaltitle = {Physical Review A},
  shortjournal = {Phys. Rev. A},
  volume = {71},
  number = {5},
  pages={055801},
  publisher = {American Physical Society (APS)},
  issn = {1050-2947, 1094-1622},
  doi = {10.1103/physreva.71.055801}
}

@article{budyka-deficiencyindicesblock-2024,
  title = {Deficiency {{Indices}} of {{Block Jacobi Matrices}}: {{Survey}}},
  shorttitle = {Deficiency {{Indices}} of {{Block Jacobi Matrices}}},
  author = {Budyka, V. and Malamud, M. and Mirzoev, K.},
  date = {2024-01},
  journaltitle = {Journal of Mathematical Sciences},
  shortjournal = {J Math Sci},
  volume = {278},
  number = {1},
  pages = {39--54},
  issn = {1072-3374, 1573-8795},
  doi = {10.1007/s10958-024-06904-9}
}

@article{bush-singularboundaryconditions-2023,
  title = {Singular Boundary Conditions for {{Sturm}}–{{Liouville}} Operators via Perturbation Theory},
  author = {Bush, Michael and Frymark, Dale and Liaw, Constanze},
  date = {2023-08},
  journaltitle = {Canadian Journal of Mathematics},
  shortjournal = {Can. J. Math.-J. Can. Math.},
  volume = {75},
  number = {4},
  pages = {1110--1146},
  issn = {0008-414X, 1496-4279},
  doi = {10.4153/S0008414X22000293}
}

@article{cai-bosonicquantumerror-2021,
  title = {Bosonic Quantum Error Correction Codes in Superconducting Quantum Circuits},
  author = {Cai, Weizhou and Ma, Yuwei and Wang, Weiting and Zou, Chang-Ling and Sun, Luyan},
  date = {2021-01},
  journaltitle = {Fundamental Research},
  volume = {1},
  number = {1},
  pages = {50--67},
  publisher = {Elsevier BV},
  issn = {2667-3258},
  doi = {10.1016/j.fmre.2020.12.006}
}

@article{case-singularpotentials-1950,
  title = {Singular {{Potentials}}},
  author = {Case, K. M.},
  date = {1950-12-01},
  journaltitle = {Physical Review},
  shortjournal = {Phys. Rev.},
  volume = {80},
  number = {5},
  pages = {797--806},
  issn = {0031-899X},
  doi = {10.1103/PhysRev.80.797}
}

@article{chang-observationthreephotonspontaneous-2020,
  title = {Observation of {{Three-Photon Spontaneous Parametric Down-Conversion}} in a {{Superconducting Parametric Cavity}}},
  author = {Chang, C. W. Sandbo and Sabín, Carlos and Forn-Díaz, P. and Quijandría, Fernando and Vadiraj, A. M. and Nsanzineza, I. and Johansson, G. and Wilson, C. M.},
  date = {2020-01-16},
  journaltitle = {Physical Review X},
  shortjournal = {Phys. Rev. X},
  volume = {10},
  number = {1},
  pages={011011},
  publisher = {American Physical Society (APS)},
  issn = {2160-3308},
  doi = {10.1103/physrevx.10.011011}
}

@article{charif-perturbationseriesjacobi-2021,
  title = {Perturbation Series for {{Jacobi}} Matrices and the Quantum {{Rabi}} Model},
  author = {Charif, Mirna and Zielinski, Lech},
  date = {2021},
  journaltitle = {Opuscula Mathematica},
  shortjournal = {Opuscula Math.},
  volume = {41},
  number = {3},
  pages = {301--333},
  issn = {1232-9274},
  doi = {10.7494/OpMath.2021.41.3.301}
}

@article{clark-spectralanalysisselfadjoint-1996,
  title = {A {{Spectral Analysis}} for {{Self-Adjoint Operators Generated}} by a {{Class}} of {{Second Order Difference Equations}}},
  author = {Clark, S. L.},
  date = {1996-01-01},
  journaltitle = {Journal of Mathematical Analysis and Applications},
  shortjournal = {Journal of Mathematical Analysis and Applications},
  volume = {197},
  number = {1},
  pages = {267--285},
  issn = {0022-247X},
  doi = {10.1006/jmaa.1996.0020}
}

@article{corona-thirdorderspontaneousparametric-2011,
  title = {Third-Order Spontaneous Parametric down-Conversion in Thin Optical Fibers as a Photon-Triplet Source},
  author = {Corona, María and Garay-Palmett, Karina and U’Ren, Alfred B.},
  date = {2011-09-15},
  journaltitle = {Physical Review A},
  shortjournal = {Phys. Rev. A},
  volume = {84},
  number = {3},
  pages={033823},
  publisher = {American Physical Society (APS)},
  issn = {1050-2947, 1094-1622},
  doi = {10.1103/physreva.84.033823}
}

@article{damanik-analytictheorymatrix-2008,
  title = {The {{Analytic Theory}} of {{Matrix Orthogonal Polynomials}}},
  author = {Damanik, David and Pushnitski, Alexander and Simon, Barry},
  date = {2008},
  journal = {Surveys in Approximation Theory},
  volume = {4},
  pages = {1–85},
  doi = {10.48550/arXiv.0711.2703}
}

@article{damanik-perturbationsorthogonalpolynomials-2010,
  title = {Perturbations of Orthogonal Polynomials with Periodic Recursion Coefficients},
  author = {Damanik, David and Killip, Rowan and Simon, Barry},
  date = {2010-05-01},
  journaltitle = {Annals of Mathematics},
  shortjournal = {Ann. Math.},
  volume = {171},
  number = {3},
  pages = {1931--2010},
  issn = {0003-486X},
  doi = {10.4007/annals.2010.171.1931}
}

@article{damanik-uniformspectralproperties-2000,
  title = {Uniform {{Spectral Properties}} of {{One-Dimensional Quasicrystals}}, {{III}}. $\alpha$-{{Continuity}}},
  author = {Damanik, David and Killip, Rowan and Lenz, Daniel},
  date = {2000-06-01},
  journaltitle = {Communications in Mathematical Physics},
  shortjournal = {Comm Math Phys},
  volume = {212},
  number = {1},
  pages = {191--204},
  issn = {1432-0916},
  doi = {10.1007/s002200000203}
}

@article{delgrosso-controlledsqueezegatesuperconducting-2025,
  title = {Controlled-Squeeze Gate in Superconducting Quantum Circuits},
  author = {Del Grosso, Nicolás F. and Cortiñas, Rodrigo G. and Villar, Paula I. and Lombardo, Fernando C. and Paz, Juan Pablo},
  date = {2025-04-07},
  journaltitle = {Physical Review A},
  shortjournal = {Phys. Rev. A},
  volume = {111},
  number = {4},
  pages = {042606},
  publisher = {American Physical Society},
  doi = {10.1103/PhysRevA.111.042606}
}

@article{dombrowski-resultsabsolutecontinuity-2017,
  title = {Some Results on Absolute Continuity for Unbounded {{Jacobi}} Matrices},
  author = {Dombrowski, Joanne},
  date = {2017},
  journaltitle = {Operators and Matrices},
  volume = {11},
  number = {2},
  pages = {355--362},
  issn = {1846-3886},
  doi = {10.7153/oam-11-23}
}

@article{eckhardt-singularweyltitchmarshkodairatheory-2013,
  title = {Singular {{Weyl-Titchmarsh-Kodaira Theory}} for {{Jacobi Operators}}},
  author = {Eckhardt, Jonathan and Teschl, Gerald},
  date = {2013},
  journaltitle = {Operators and Matrices},
  volume ={7},
  number = {3},
  eprint = {1203.5856},
  eprinttype = {arXiv},
  eprintclass = {math},
  pages = {695--712},
  issn = {1846-3886},
  doi = {10.7153/oam-07-38}
}

@article{ehrhardt-solutionsdiscreteairy-2004,
  title = {Solutions to the Discrete {{Airy}} Equation: Application to Parabolic Equation Calculations},
  shorttitle = {Solutions to the Discrete {{Airy}} Equation},
  author = {Ehrhardt, Matthias and Mickens, Ronald E.},
  date = {2004-11-01},
  journaltitle = {Journal of Computational and Applied Mathematics},
  shortjournal = {Journal of Computational and Applied Mathematics},
  volume = {172},
  number = {1},
  pages = {183--206},
  issn = {0377-0427},
  doi = {10.1016/j.cam.2004.02.011}
}

@article{eichinger-weylmatrixperspective-2025,
  title = {A {{Weyl Matrix Perspective}} on {{Unbounded Non-Self-Adjoint Jacobi Matrices}}},
  author = {Eichinger, Benjamin and Lukić, Milivoje and Young, Giorgio},
  date = {2025-10},
  journaltitle = {Complex Analysis and Operator Theory},
  shortjournal = {Complex Anal. Oper. Theory},
  volume = {19},
  number = {7},
  pages = {194},
  issn = {1661-8254, 1661-8262},
  doi = {10.1007/s11785-025-01804-5}
}

@article{eichler-observationtwomodesqueezing-2011,
  title = {Observation of {{Two-Mode Squeezing}} in the {{Microwave Frequency Domain}}},
  author = {Eichler, C. and Bozyigit, D. and Lang, C. and Baur, M. and Steffen, L. and Fink, J. M. and Filipp, S. and Wallraff, A.},
  date = {2011-09-06},
  journaltitle = {Physical Review Letters},
  shortjournal = {Phys. Rev. Lett.},
  volume = {107},
  number = {11},
  pages = {113601},
  issn = {0031-9007, 1079-7114},
  doi = {10.1103/PhysRevLett.107.113601}
}

@article{eisert-supersonicquantumcommunication-2009,
  title = {Supersonic {{Quantum Communication}}},
  author = {Eisert, J. and Gross, D.},
  date = {2009-06-17},
  journaltitle = {Physical Review Letters},
  shortjournal = {Phys. Rev. Lett.},
  volume = {102},
  number = {24},
  pages = {240501},
  issn = {0031-9007, 1079-7114},
  doi = {10.1103/PhysRevLett.102.240501}
}

@book{elaydi-introductiondifferenceequations-2005,
  title = {An {{Introduction}} to {{Difference Equations}}},
  editor = {Elaydi, Saber N.},
  date = {2005},
  series = {Undergraduate {{Texts}} in {{Mathematics}}},
  edition = {Third Edition},
  publisher = {Springer Science+Business Media, Inc},
  location = {New York, NY},
  doi = {10.1007/0-387-27602-5},
  isbn = {978-0-387-23059-7 978-0-387-27602-1},
  pagetotal = {539}
}

@article{esteve-generalizationhellmannfeynman-2010,
  title = {Generalization of the {{Hellmann}}–{{Feynman}} Theorem},
  author = {Esteve, J.G. and Falceto, Fernando and García Canal, C.},
  date = {2010-01},
  journaltitle = {Physics Letters A},
  shortjournal = {Physics Letters A},
  volume = {374},
  number = {6},
  pages = {819--822},
  issn = {03759601},
  doi = {10.1016/j.physleta.2009.12.005}
}

@incollection{everitt-weylsworksingular-2008,
  title = {Weyl's Work on Singular {{Sturm}}–{{Liouville}} Operators},
  booktitle = {Groups and {{Analysis}}},
  author = {Everitt, W.N. and Kalf, H.},
  editor = {Tent, Katrin},
  date = {2008-10-16},
  edition = {1},
  pages = {63--83},
  publisher = {Cambridge University Press},
  doi = {10.1017/CBO9780511721410.004},
  isbn = {978-0-521-71788-5 978-0-511-72141-0}
}

@article{fedotov-complexwkbmethod-2019,
  title = {The {{Complex WKB Method}} for {{Difference Equations}} and {{Airy Functions}}},
  author = {Fedotov, Alexander and Klopp, Frédéric},
  date = {2019-01},
  journaltitle = {SIAM Journal on Mathematical Analysis},
  shortjournal = {SIAM J. Math. Anal.},
  volume = {51},
  number = {6},
  pages = {4413--4447},
  issn = {0036-1410, 1095-7154},
  doi = {10.1137/18M1228694}
}

@online{fischer-selfadjointrealizationshigherorder-2025,
  title = {Self-Adjoint Realizations of Higher-Order Squeezing Operators},
  author = {Fischer, Felix and Burgarth, Daniel and Lonigro, Davide},
  date = {2025-08-13},
  eprint = {2508.09044},
  eprinttype = {arXiv},
  eprintclass = {math-ph},
  doi = {10.48550/arXiv.2508.09044},
  pubstate = {prepublished}
}

@article{fisher-impossibilitynaivelygeneralizing-1984,
  title = {Impossibility of Naively Generalizing Squeezed Coherent States},
  author = {Fisher, Robert A. and Nieto, Michael Martin and Sandberg, Vernon D.},
  date = {1984-03-15},
  journaltitle = {Physical Review D},
  shortjournal = {Phys. Rev. D},
  volume = {29},
  number = {6},
  pages = {1107--1110},
  publisher = {American Physical Society (APS)},
  issn = {0556-2821},
  doi = {10.1103/physrevd.29.1107}
}

@article{frymark-spectralpropertiessingular-2023,
  title = {Spectral Properties of Singular {{Sturm}}–{{Liouville}} Operators via Boundary Triples and Perturbation Theory},
  author = {Frymark, Dale and Liaw, Constanze},
  date = {2023-08-05},
  journaltitle = {Journal of Differential Equations},
  shortjournal = {Journal of Differential Equations},
  volume = {363},
  pages = {391--421},
  issn = {0022-0396},
  doi = {10.1016/j.jde.2023.03.022}
}

@article{ganapathy-broadbandquantumenhancement-2023,
  title = {Broadband {{Quantum Enhancement}} of the {{LIGO Detectors}} with {{Frequency-Dependent Squeezing}}},
  author = {Ganapathy, D. and Jia, W. and Nakano, M. and Xu, V. and Aritomi, N. and Cullen, T. and Kijbunchoo, N. and Dwyer, S. E. and Mullavey, A. and McCuller, L. and Abbott, R. and Abouelfettouh, I. and Adhikari, R. X. and Ananyeva, A. and Appert, S. and Arai, K. and Aston, S. M. and Ball, M. and Ballmer, S. W. and Barker, D. and Barsotti, L. and Berger, B. K. and Betzwieser, J. and Bhattacharjee, D. and Billingsley, G. and Biscans, S. and Bode, N. and Bonilla, E. and Bossilkov, V. and Branch, A. and Brooks, A. F. and Brown, D. D. and Bryant, J. and Cahillane, C. and Cao, H. and Capote, E. and Clara, F. and Collins, J. and Compton, C. M. and Cottingham, R. and Coyne, D. C. and Crouch, R. and Csizmazia, J. and Dartez, L. P. and Demos, N. and Dohmen, E. and Driggers, J. C. and Effler, A. and Ejlli, A. and Etzel, T. and Evans, M. and Feicht, J. and Frey, R. and Frischhertz, W. and Fritschel, P. and Frolov, V. V. and Fulda, P. and Fyffe, M. and Gateley, B. and Giaime, J. A. and Giardina, K. D. and Glanzer, J. and Goetz, E. and Goetz, R. and Goodwin-Jones, A. W. and Gras, S. and Gray, C. and Griffith, D. and Grote, H. and Guidry, T. and Hall, E. D. and Hanks, J. and Hanson, J. and Heintze, M. C. and Helmling-Cornell, A. F. and Holland, N. A. and Hoyland, D. and Huang, H. Y. and Inoue, Y. and James, A. L. and Jennings, A. and Karat, S. and Karki, S. and Kasprzack, M. and Kawabe, K. and King, P. J. and Kissel, J. S. and Komori, K. and Kontos, A. and Kumar, R. and Kuns, K. and Landry, M. and Lantz, B. and Laxen, M. and Lee, K. and Lesovsky, M. and Llamas, F. and Lormand, M. and Loughlin, H. A. and Macas, R. and MacInnis, M. and Makarem, C. N. and Mannix, B. and Mansell, G. L. and Martin, R. M. and Mason, K. and Matichard, F. and Mavalvala, N. and Maxwell, N. and McCarrol, G. and McCarthy, R. and McClelland, D. E. and McCormick, S. and McRae, T. and Mera, F. and Merilh, E. L. and Meylahn, F. and Mittleman, R. and Moraru, D. and Moreno, G. and Nelson, T. J. N. and Neunzert, A. and Notte, J. and Oberling, J. and O’Hanlon, T. and Osthelder, C. and Ottaway, D. J. and Overmier, H. and Parker, W. and Pele, A. and Pham, H. and Pirello, M. and Quetschke, V. and Ramirez, K. E. and Reyes, J. and Richardson, J. W. and Robinson, M. and Rollins, J. G. and Romel, C. L. and Romie, J. H. and Ross, M. P. and Ryan, K. and Sadecki, T. and Sanchez, A. and Sanchez, E. J. and Sanchez, L. E. and Savage, R. L. and Schaetzl, D. and Schiworski, M. G. and Schnabel, R. and Schofield, R. M. S. and Schwartz, E. and Sellers, D. and Shaffer, T. and Short, R. W. and Sigg, D. and Slagmolen, B. J. J. and Soike, C. and Soni, S. and Srivastava, V. and Sun, L. and Tanner, D. B. and Thomas, M. and Thomas, P. and Thorne, K. A. and Torrie, C. I. and Traylor, G. and Ubhi, A. S. and Vajente, G. and Vanosky, J. and Vecchio, A. and Veitch, P. J. and Vibhute, A. M. and Von Reis, E. R. G. and Warner, J. and Weaver, B. and Weiss, R. and Whittle, C. and Willke, B. and Wipf, C. C. and Yamamoto, H. and Zhang, L. and Zucker, M. E. and {LIGO O4 Detector Collaboration}},
  date = {2023-10-30},
  journaltitle = {Physical Review X},
  shortjournal = {Phys. Rev. X},
  volume = {13},
  number = {4},
  pages={041021},
  publisher = {American Physical Society (APS)},
  issn = {2160-3308},
  doi = {10.1103/physrevx.13.041021}
}

@article{gardas-initialstatesqubit-2013,
  title = {Initial States of Qubit–Environment Models Leading to Conserved Quantities},
  author = {Gardas, Bartłomiej and Dajka, Jerzy},
  date = {2013-06-14},
  journaltitle = {Journal of Physics A: Mathematical and Theoretical},
  shortjournal = {J. Phys. A: Math. Theor.},
  volume = {46},
  number = {23},
  pages = {235301},
  publisher = {IOP Publishing},
  issn = {1751-8113, 1751-8121},
  doi = {10.1088/1751-8113/46/23/235301}
}

@article{gardas-multiphotonrabimodel-2013,
  title = {Multi-Photon {{Rabi}} Model: {{Generalized}} Parity and Its Applications},
  shorttitle = {Multi-Photon {{Rabi}} Model},
  author = {Gardas, Bartłomiej and Dajka, Jerzy},
  date = {2013-12},
  journaltitle = {Physics Letters A},
  volume = {377},
  number = {44},
  pages = {3205--3208},
  publisher = {Elsevier BV},
  issn = {0375-9601},
  doi = {10.1016/j.physleta.2013.10.011}
}

@article{gardas-replycommentmultiphoton-2014,
  title = {Reply to “{{Comment}} on: ‘{{Multi-photon Rabi}} Model: {{Generalized}} Parity and Its Applications’ [{{Phys}}. {{Lett}}. {{A}} 377 (2013) 3205]” [{{Phys}}. {{Lett}}. {{A}} 378 (2014) 1969]},
  shorttitle = {Reply to “{{Comment}} On},
  author = {Gardas, Bartłomiej and Dajka, Jerzy},
  date = {2014-05},
  journaltitle = {Physics Letters A},
  volume = {378},
  number = {28--29},
  pages = {1970},
  publisher = {Elsevier BV},
  issn = {0375-9601},
  doi = {10.1016/j.physleta.2014.04.053}
}

@article{geronimo-spectrainfinitedimensionaljacobi-1988,
  title = {On the Spectra of Infinite-Dimensional {{Jacobi}} Matrices},
  author = {Geronimo, J.S},
  date = {1988-06},
  journaltitle = {Journal of Approximation Theory},
  shortjournal = {Journal of Approximation Theory},
  volume = {53},
  number = {3},
  pages = {251--265},
  issn = {00219045},
  doi = {10.1016/0021-9045(88)90022-6}
}

@article{geronimo-wkbliouvillegreenanalysis-1992,
  title = {{{WKB}} ({{Liouville-Green}}) Analysis of Second Order Difference Equations and Applications},
  author = {Geronimo, Jeffrey S and Smith, Dale T},
  date = {1992-06},
  journaltitle = {Journal of Approximation Theory},
  shortjournal = {Journal of Approximation Theory},
  volume = {69},
  number = {3},
  pages = {269--301},
  issn = {00219045},
  doi = {10.1016/0021-9045(92)90003-7}
}

@incollection{geronimo-wkbturningpoint-2004,
  title = {{{WKB}} and {{Turning Point Theory}} for {{Second-order Difference Equations}}},
  booktitle = {Spectral {{Methods}} for {{Operators}} of {{Mathematical Physics}}},
  author = {Geronimo, Jeffrey S. and Bruno, Oscar and Van Assche, Walter},
  editor = {Janas, Jan and Kurasov, Pavel and Naboko, Sergei},
  date = {2004},
  pages = {101--138},
  publisher = {Birkhäuser Basel},
  location = {Basel},
  doi = {10.1007/978-3-0348-7947-7_7},
  isbn = {978-3-0348-9632-0 978-3-0348-7947-7}
}

@article{geronimo-wkbturningpoint-2009,
  title = {{{WKB}} and {{Turning Point Theory}} for {{Second Order Difference Equations}}: {{External Fields}} and {{Strong Asymptotics}} for {{Orthogonal Polynomials}}},
  shorttitle = {{{WKB}} and {{Turning Point Theory}} for {{Second Order Difference Equations}}},
  author = {Geronimo, Jeffrey S.},
  date = {2004},
  journal = {Operator Theory: Advances and Applications},
  volume = {154},
  pages = {101--138},
  doi = {10.48550/arXiv.0905.1684}
}

@article{gesztesy-commutationmethodsjacobi-1996,
  title = {Commutation {{Methods}} for {{Jacobi Operators}}},
  author = {Gesztesy, F. and Teschl, G.},
  date = {1996-06},
  journaltitle = {Journal of Differential Equations},
  shortjournal = {Journal of Differential Equations},
  volume = {128},
  number = {1},
  pages = {252--299},
  issn = {00220396},
  doi = {10.1006/jdeq.1996.0095}
}

@article{gesztesy-jacobioperator$$11$$-2024,
  title = {The {{Jacobi Operator}} on $(-1,1)$ and {{Its Various}} m-{{Functions}}},
  author = {Gesztesy, Fritz and Littlejohn, Lance L. and Piorkowski, Mateusz and Stanfill, Jonathan},
  date = {2024-10},
  journaltitle = {Complex Analysis and Operator Theory},
  shortjournal = {Complex Anal. Oper. Theory},
  volume = {18},
  number = {7},
  pages = {155},
  issn = {1661-8254, 1661-8262},
  doi = {10.1007/s11785-024-01576-4}
}

@article{gesztesy-mfunctionsinversespectral-1997,
  title = {M-{{Functions}} and Inverse Spectral Analysis for Finite and Semi-Infinite {{Jacobi}} Matrices},
  author = {Gesztesy, Fritz and Simon, Barry},
  date = {1997-12},
  journaltitle = {Journal d'Analyse Mathématique},
  shortjournal = {J. Anal. Math.},
  volume = {73},
  number = {1},
  pages = {267--297},
  issn = {0021-7670, 1565-8538},
  doi = {10.1007/BF02788147}
}

@article{gordillo-hachuel-commentpropertiesdynamics-2026,
  title = {Comment on ‘{{Properties}} and Dynamics of Generalized Squeezed States’},
  author = {Gordillo-Hachuel, Rubén and Puebla, Ricardo},
  date = {2026-02},
  journaltitle = {New Journal of Physics},
  shortjournal = {New J. Phys.},
  volume = {28},
  number = {2},
  pages = {028002},
  publisher = {IOP Publishing},
  issn = {1367-2630},
  doi = {10.1088/1367-2630/ae4379}
}

@online{gordillo-hachuel-quantummetrologicaladvantage-2026,
  title = {Quantum Metrological Advantage of High-Order Squeezed States},
  author = {Gordillo-Hachuel, Rubén and Torrontegui, Erik and Dios, Cristina and Puebla, Ricardo},
  date = {2026-04-10},
  eprint = {2604.09958},
  eprinttype = {arXiv},
  eprintclass = {quant-ph},
  doi = {10.48550/arXiv.2604.09958},
  pubstate = {prepublished}
}

@article{gorska-squeezingarbitraryorder-2014,
  title = {Squeezing of Arbitrary Order: The Ups and Downs},
  shorttitle = {Squeezing of Arbitrary Order},
  author = {Górska, Katarzyna and Horzela, Andrzej and Szafraniec, Franciszek Hugon},
  date = {2014-12-08},
  journaltitle = {Proceedings of the Royal Society A: Mathematical, Physical and Engineering Sciences},
  shortjournal = {Proc. R. Soc. A.},
  volume = {470},
  number = {2172},
  pages = {20140205},
  issn = {1364-5021, 1471-2946},
  doi = {10.1098/rspa.2014.0205}
}

@article{gottesman-encodingqubitoscillator-2001,
  title = {Encoding a Qubit in an Oscillator},
  author = {Gottesman, Daniel and Kitaev, Alexei and Preskill, John},
  date = {2001-06-11},
  journaltitle = {Physical Review A},
  shortjournal = {Phys. Rev. A},
  volume = {64},
  number = {1},
  pages={012310},
  publisher = {American Physical Society (APS)},
  issn = {1050-2947, 1094-1622},
  doi = {10.1103/physreva.64.012310}
}

@online{gregory-foursixphotonstimulated-2026,
  title = {Four- and Six-Photon Stimulated {{Raman}} Transitions for Coherent Qubit and Qudit Operations},
  author = {Gregory, Gabriel J. and Ritchie, Evan R. and Quinn, Alex and Brudney, Sean and Wineland, David J. and Allcock, David T. C. and O'Reilly, Jameson},
  date = {2026-02-20},
  eprint = {2602.18567},
  eprinttype = {arXiv},
  eprintclass = {quant-ph},
  doi = {10.48550/arXiv.2602.18567},
  pubstate = {prepublished}
}

@article{harrat-asymptoticexpansionlarge-2020,
  title = {Asymptotic Expansion of Large Eigenvalues for a Class of Unbounded {{Jacobi}} Matrices},
  author = {Harrat, Ayoub and Zerouali, El Hassan and Zielinski, Lech},
  date = {2020},
  journaltitle = {Opuscula Mathematica},
  shortjournal = {Opuscula Math.},
  volume = {40},
  number = {2},
  pages = {241--270},
  issn = {1232-9274},
  doi = {10.7494/OpMath.2020.40.2.241}
}

@article{hinton-spectralanalysissecond-1978,
  title = {Spectral Analysis of Second Order Difference Equations},
  author = {Hinton, Don B and Lewis, Roger T},
  date = {1978-04-01},
  journaltitle = {Journal of Mathematical Analysis and Applications},
  shortjournal = {Journal of Mathematical Analysis and Applications},
  volume = {63},
  number = {2},
  pages = {421--438},
  issn = {0022-247X},
  doi = {10.1016/0022-247X(78)90088-4}
}

@article{iantchenko-periodicjacobioperator-2012,
  title = {Periodic {{Jacobi}} Operator with Finitely Supported Perturbations},
  author = {Iantchenko, Alexei and Korotyaev, Evgeny},
  date = {2012-04},
  journaltitle = {Journal of Mathematical Analysis and Applications},
  shortjournal = {Journal of Mathematical Analysis and Applications},
  volume = {388},
  number = {2},
  eprint = {1006.1538},
  eprinttype = {arXiv},
  eprintclass = {math},
  pages = {1239--1253},
  issn = {0022247X},
  doi = {10.1016/j.jmaa.2011.11.016}
}

@article{ibort-selfadjointextensionssymmetries-2015,
  title = {On {{Self-adjoint}} Extensions and Symmetries in {{Quantum Mechanics}}},
  author = {Ibort, Alberto and Lledó, Fernando and Pérez-Pardo, Juan Manuel},
  date = {2015-10},
  journaltitle = {Annales Henri Poincaré},
  shortjournal = {Ann. Henri Poincaré},
  volume = {16},
  number = {10},
  eprint = {1402.5537},
  eprinttype = {arXiv},
  eprintclass = {math-ph},
  pages = {2367--2397},
  issn = {1424-0637, 1424-0661},
  doi = {10.1007/s00023-014-0379-4}
}

@book{izaac-computationalquantummechanics-2018,
  title = {Computational {{Quantum Mechanics}}},
  author = {Izaac, Joshua and Wang, Jingbo},
  date = {2018},
  series = {Undergraduate {{Lecture Notes}} in {{Physics}}},
  publisher = {Springer International Publishing},
  location = {Cham},
  doi = {10.1007/978-3-319-99930-2},
  isbn = {978-3-319-99929-6 978-3-319-99930-2}
}

@article{janas-asymptoticbehaviorgeneralized-2009,
  title = {Asymptotic {{Behavior}} of {{Generalized Eigenvectors}} of {{Jacobi Matrices}} in the {{Critical}} (“{{Double Root}}”) {{Case}}},
  author = {Janas, J. and Naboko, Sergey and Sheronova, E.},
  date = {2009-12-23},
  journaltitle = {Zeitschrift für Analysis und ihre Anwendungen},
  shortjournal = {Z. Anal. Anwend.},
  volume = {28},
  number = {4},
  pages = {411--430},
  issn = {0232-2064, 1661-4534},
  doi = {10.4171/zaa/1391}
}

@article{janas-decayboundseigenfunctions-2009,
  title = {Decay {{Bounds}} on {{Eigenfunctions}} and the {{Singular Spectrum}} of {{Unbounded Jacobi Matrices}}},
  author = {Janas, J. and Naboko, S. and Stolz, G.},
  date = {2009-01-02},
  journaltitle = {International Mathematics Research Notices},
  shortjournal = {International Mathematics Research Notices},
  volume ={2009},
  pages = {736–764},
  doi = {10.1093/imrn/rnn144}
}

@article{janas-jacobimatricespowerlike-1999,
  title = {Jacobi {{Matrices}} with {{Power-like Weights}}—{{Grouping}} in {{Blocks Approach}}},
  author = {Janas, Jan and Naboko, Serguei},
  date = {1999-08},
  journaltitle = {Journal of Functional Analysis},
  shortjournal = {Journal of Functional Analysis},
  volume = {166},
  number = {2},
  pages = {218--243},
  issn = {00221236},
  doi = {10.1006/jfan.1999.3434}
}

@article{janas-spectralpropertiesjacobi-2003,
  title = {Spectral Properties of {{Jacobi}} Matrices by Asymptotic Analysis},
  author = {Janas, Jan and Moszyński, Marcin},
  date = {2003-02},
  journaltitle = {Journal of Approximation Theory},
  shortjournal = {Journal of Approximation Theory},
  volume = {120},
  number = {2},
  pages = {309--336},
  issn = {00219045},
  doi = {10.1016/S0021-9045(02)00038-2}
}

@article{judge-eigenvaluesperturbedperiodic-2016,
  title = {Eigenvalues for {{Perturbed Periodic Jacobi Matrices}} by the {{Wigner-von Neumann Approach}}},
  author = {Judge, Edmund and Naboko, Sergey and Wood, Ian},
  date = {2016-07-01},
  journaltitle = {Integral Equations and Operator Theory},
  shortjournal = {Integr. Equ. Oper. Theory},
  volume = {85},
  number = {3},
  pages = {427--450},
  issn = {1420-8989},
  doi = {10.1007/s00020-016-2302-5}
}

@article{judge-spectralresultsperturbed-2018,
  title = {Spectral Results for Perturbed Periodic {{Jacobi}} Matrices Using the Discrete {{Levinson}} Technique},
  author = {Judge, Edmund and Naboko, Sergey and Wood, Ian},
  date = {2018},
  journaltitle = {Studia Mathematica},
  shortjournal = {Studia Math.},
  volume = {242},
  number = {2},
  eprint = {1703.10223},
  eprinttype = {arXiv},
  eprintclass = {math},
  pages = {179--215},
  issn = {0039-3223, 1730-6337},
  doi = {10.4064/sm170325-23-8}
}

@article{kaluzhny-preservationabsolutelycontinuous-2010,
  title = {Preservation of Absolutely Continuous Spectrum of Periodic {{Jacobi}} Operators under Perturbations of Square--Summable Variation},
  author = {Kaluzhny, U. and Shamis, M.},
  date = {2012},
  journal = {Constructive Approximation},
  volume={35},
  pages={89–105},
  doi = {10.1007/s00365-011-9126-y},
  pubstate = {prepublished}
}

@book{kato-perturbationtheorylinear-1995,
  title = {Perturbation Theory for Linear Operators},
  author = {Kato, Tosio},
  date = {1995},
  series = {Classics in Mathematics},
  publisher = {Springer},
  location = {Berlin},
  doi = {10.1007/978-3-642-66282-9},
  isbn = {978-3-540-58661-6},
  pagetotal = {619}
}

@online{lakaev-thresholdvirtualstates-2026,
  title = {Threshold {{Virtual States}} of a {{Jacobi}} Operator},
  author = {Lakaev, Saidakhmat N. and Makarov, Konstantin A.},
  date = {2026-04-05},
  eprint = {2604.04019},
  eprinttype = {arXiv},
  eprintclass = {math},
  doi = {10.48550/arXiv.2604.04019},
  pubstate = {prepublished}
}

@article{last-eigenfunctionstransfermatrices-1999,
  title = {Eigenfunctions, Transfer Matrices, and Absolutely Continuous Spectrum of One-Dimensional {{Schrodinger}} Operators},
  author = {Last, Yoram and Simon, Barry},
  date = {1999-01-15},
  journaltitle = {Inventiones Mathematicae},
  shortjournal = {Inventiones Mathematicae},
  volume = {135},
  number = {2},
  eprint = {math-ph/9907023},
  eprinttype = {arXiv},
  pages = {329--367},
  issn = {0020-9910, 1432-1297},
  doi = {10.1007/s002220050288}
}

@article{lee-exactsolutionsfamily-2011,
  title = {Exact Solutions for a Family of Spin-Boson Systems},
  author = {Lee, Yuan-Harng and Links, Jon and Zhang, Yao-Zhong},
  date = {2011-07-01},
  journaltitle = {Nonlinearity},
  volume = {24},
  number = {7},
  eprint = {1008.3738},
  eprinttype = {arXiv},
  eprintclass = {math-ph},
  pages = {1975--1986},
  issn = {0951-7715, 1361-6544},
  doi = {10.1088/0951-7715/24/7/004}
}

@book{levitan-sturmliouvilledirac-1991,
  title = {Sturm—{{Liouville}} and {{Dirac Operators}}},
  author = {Levitan, B. M. and Sargsjan, I. S.},
  date = {1991},
  publisher = {Springer Netherlands},
  location = {Dordrecht},
  doi = {10.1007/978-94-011-3748-5},
  isbn = {978-94-010-5667-0 978-94-011-3748-5}
}

@article{lo-commentinitialstates-2014,
  title = {Comment on ‘{{Initial}} States of Qubit–Environment Models Leading to Conserved Quantities’},
  author = {Lo, C F},
  date = {2014-04-23},
  journaltitle = {Journal of Physics A: Mathematical and Theoretical},
  shortjournal = {J. Phys. A: Math. Theor.},
  volume = {47},
  number = {16},
  pages = {168001},
  publisher = {IOP Publishing},
  issn = {1751-8113, 1751-8121},
  doi = {10.1088/1751-8113/47/16/168001}
}

@article{lo-commentmultiphotonrabi-2014,
  title = {Comment on: “{{Multi-photon Rabi}} Model: {{Generalized}} Parity and Its Applications” by {{B}}. {{Gardas}} and {{J}}. {{Dajka}} [{{Phys}}. {{Lett}}. {{A}} 377 (2013) 3205]},
  shorttitle = {Comment On},
  author = {Lo, C.F.},
  date = {2014-05},
  journaltitle = {Physics Letters A},
  volume = {378},
  number = {28--29},
  pages = {1969},
  publisher = {Elsevier BV},
  issn = {0375-9601},
  doi = {10.1016/j.physleta.2014.04.044}
}

@article{lo-commentsolvingtwomode-2014,
  title = {Comment on ‘{{Solving}} the Two-Mode Squeezed Harmonic Oscillator and The{\mkbibemph{k}}th-Order Harmonic Generation in {{Bargmann}}–{{Hilbert}} Spaces’},
  author = {Lo, C F},
  date = {2014-02-21},
  journaltitle = {Journal of Physics A: Mathematical and Theoretical},
  shortjournal = {J. Phys. A: Math. Theor.},
  volume = {47},
  number = {7},
  pages = {078001},
  publisher = {IOP Publishing},
  issn = {1751-8113, 1751-8121},
  doi = {10.1088/1751-8113/47/7/078001}
}

@article{lo-multiquantumjaynescummingsmodel-1998,
  title = {The Multiquantum {{Jaynes-Cummings}} Model with the Counter-Rotating Terms},
  author = {Lo, C. F and Liu, K. L and Ng, K. M},
  date = {1998-04-01},
  journaltitle = {Europhysics Letters (EPL)},
  shortjournal = {Europhys. Lett.},
  volume = {42},
  number = {1},
  pages = {1--6},
  publisher = {IOP Publishing},
  issn = {0295-5075, 1286-4854},
  doi = {10.1209/epl/i1998-00544-3}
}

@article{lukic-generalizedprufervariables-2016,
  title = {Generalized {{Prüfer}} Variables for Perturbations of {{Jacobi}} and {{CMV}} Matrices},
  author = {Lukic, Milivoje and Ong, Darren C.},
  date = {2016-12-15},
  journaltitle = {Journal of Mathematical Analysis and Applications},
  shortjournal = {Journal of Mathematical Analysis and Applications},
  volume = {444},
  number = {2},
  pages = {1490--1514},
  issn = {0022-247X},
  doi = {10.1016/j.jmaa.2016.07.036}
}

@article{malcolmbrown-kreinfriedrichsextensions-2005,
  title = {On the {{Krein}} and {{Friedrichs}} Extensions of a Positive {{Jacobi}} Operator},
  author = {Malcolm Brown, B. and Christiansen, Jacob S.},
  date = {2005-06},
  journaltitle = {Expositiones Mathematicae},
  shortjournal = {Expositiones Mathematicae},
  volume = {23},
  number = {2},
  pages = {179--186},
  issn = {07230869},
  doi = {10.1016/j.exmath.2005.01.020}
}

@article{mcculler-frequencydependentsqueezingadvanced-2020,
  title = {Frequency-{{Dependent Squeezing}} for {{Advanced LIGO}}},
  author = {McCuller, L. and Whittle, C. and Ganapathy, D. and Komori, K. and Tse, M. and Fernandez-Galiana, A. and Barsotti, L. and Fritschel, P. and MacInnis, M. and Matichard, F. and Mason, K. and Mavalvala, N. and Mittleman, R. and Yu, Haocun and Zucker, M. E. and Evans, M.},
  date = {2020-04-28},
  journaltitle = {Physical Review Letters},
  shortjournal = {Phys. Rev. Lett.},
  volume = {124},
  number = {17},
  pages ={171102},
  publisher = {American Physical Society (APS)},
  issn = {0031-9007, 1079-7114},
  doi = {10.1103/physrevlett.124.171102}
}

@article{menard-emissionphotonmultiplets-2022,
  title = {Emission of {{Photon Multiplets}} by a Dc-{{Biased Superconducting Circuit}}},
  author = {Ménard, G. C. and Peugeot, A. and Padurariu, C. and Rolland, C. and Kubala, B. and Mukharsky, Y. and Iftikhar, Z. and Altimiras, C. and Roche, P. and Le Sueur, H. and Joyez, P. and Vion, D. and Esteve, D. and Ankerhold, J. and Portier, F.},
  date = {2022-04-08},
  journaltitle = {Physical Review X},
  shortjournal = {Phys. Rev. X},
  volume = {12},
  number = {2},
  pages={021006},
  publisher = {American Physical Society (APS)},
  issn = {2160-3308},
  doi = {10.1103/physrevx.12.021006}
}

@online{miguel-matrixvaluedschrodinger-2026,
  title = {On {{Matrix Valued Schrödinger Operators}} on the {{Discrete Real Line}}: {{Resolvent Boundary Values}}, {{Limiting Absorption Principle}}, {{Hölder Regularity}} and {{Dispersive Estimates}}},
  shorttitle = {On {{Matrix Valued Schrödinger Operators}} on the {{Discrete Real Line}}},
  author = {Miguel, Ballesteros and Gerardo, Franco Córdova and Jonathan, Gil and Naumkin, Ivan},
  date = {2026-04-01},
  eprint = {2604.01391},
  eprinttype = {arXiv},
  eprintclass = {math},
  doi = {10.48550/arXiv.2604.01391},
  pubstate = {prepublished}
}

@article{milburn-quantumteleportationsqueezed-1999,
  title = {Quantum Teleportation with Squeezed Vacuum States},
  author = {Milburn, G. J. and Braunstein, Samuel L.},
  date = {1999-08-01},
  journaltitle = {Physical Review A},
  shortjournal = {Phys. Rev. A},
  volume = {60},
  number = {2},
  pages = {937--942},
  publisher = {American Physical Society (APS)},
  issn = {1050-2947, 1094-1622},
  doi = {10.1103/physreva.60.937}
}

@article{mirrahimi-dynamicallyprotectedcatqubits-2014,
  title = {Dynamically Protected Cat-Qubits: A New Paradigm for Universal Quantum Computation},
  shorttitle = {Dynamically Protected Cat-Qubits},
  author = {Mirrahimi, Mazyar and Leghtas, Zaki and Albert, Victor V and Touzard, Steven and Schoelkopf, Robert J and Jiang, Liang and Devoret, Michel H},
  date = {2014-04-22},
  journaltitle = {New Journal of Physics},
  shortjournal = {New J. Phys.},
  volume = {16},
  number = {4},
  pages = {045014},
  publisher = {IOP Publishing},
  issn = {1367-2630},
  doi = {10.1088/1367-2630/16/4/045014}
}

@online{monvel-oscillatorybehaviorlarge-2018,
  title = {Oscillatory Behavior of Large Eigenvalues in Quantum {{Rabi}} Models},
  author = {Monvel, Anne Boutet and Zielinski, Lech},
  date = {2018-10-12},
  eprint = {1711.03366},
  eprinttype = {arXiv},
  eprintclass = {math-ph},
  doi = {10.48550/arXiv.1711.03366},
  pubstate = {prepublished}
}

@inproceedings{nagel-higherpowersqueezed-1997,
    author = {Nagel, Bengt},
    title = {Higher Power Squeezed States, {{Jacobi}} Matrices, and the {{Hamburger}} Moment Problem},
    booktitle = {Proceedings of the Fifth International Conference on Squeezed States and Uncertainty Relations,},
    year = 1998,
    editors = {D. Han, J. Janaszky, Y. S. Kim, and V. I. Man’ko},
    publisher = {NASA, Washington, DC},
    pages = {43--48}
}

@book{oliveira-intermediatespectraltheory-2009,
  title = {Intermediate Spectral Theory and Quantum Dynamics},
  author = {Oliveira, César R.},
  date = {2009},
  series = {Progress in Mathematical Physics},
  number = {v. 54},
  publisher = {Birkhäuser},
  location = {Basel; Boston},
  isbn = {978-3-7643-8794-5 978-3-7643-8795-2},
  pagetotal = {410},
}

@book{olver-asymptoticsspecialfunctions-2010,
  title = {Asymptotics and Special Functions},
  author = {Olver, Frank W. J.},
  date = {2010},
  series = {{{AKP}} Classics},
  edition = {Reprinted},
  publisher = {CRC Press},
  location = {Boca Raton, Fla.},
  isbn = {978-1-56881-069-0},
  pagetotal = {572}
}

@article{pruckner-densityspectrumjacobi-2018,
  title = {Density of the Spectrum of {{Jacobi}} Matrices with Power Asymptotics},
  author = {Pruckner, Raphael},
  date = {2019},
  journal ={Asymptotic Analysis},
  volume ={117},
  pages = {199-213},
  doi = {10.48550/arXiv.1704.06789},
  pubstate = {prepublished}
}

@article{puri-engineeringquantumstates-2017,
  title = {Engineering the Quantum States of Light in a {{Kerr-nonlinear}} Resonator by Two-Photon Driving},
  author = {Puri, Shruti and Boutin, Samuel and Blais, Alexandre},
  date = {2017-04-19},
  journaltitle = {npj Quantum Information},
  shortjournal = {npj Quantum Inf},
  volume = {3},
  number = {1},
  pages={18},
  publisher = {{Springer Science and Business Media LLC}},
  issn = {2056-6387},
  doi = {10.1038/s41534-017-0019-1}
}

@book{reed-mmmp1-funkana-1980,
  title = {Methods of Modern Mathematical Physics. 1: {{Functional}} Analysis},
  shorttitle = {Methods of Modern Mathematical Physics. 1},
  author = {Reed, Michael and Simon, Barry},
  date = {1980},
  edition = {Rev. and enlarged ed., [Nachdr.]},
  publisher = {Academic Press},
  location = {San Diego, Calif.},
  isbn = {978-0-12-585050-6},
  pagetotal = {400}
}

@book{reed-mmmp2-fourier-1975,
  title = {Methods of Modern Mathematical Physics. 2: {{Fourier}} Analysis, Self-Adjointness},
  shorttitle = {Methods of Modern Mathematical Physics. 2},
  author = {Reed, Michael and Simon, Barry},
  date = {1975},
  edition = {Nachdr.},
  publisher = {Acad. Pr},
  location = {San Diego},
  isbn = {978-0-12-585002-5},
  pagetotal = {361}
}

@book{reed-mmmp4-operators-2005,
  title = {Methods of Modern Mathematical Physics. 4: {{Analysis}} of Operators},
  shorttitle = {Methods of Modern Mathematical Physics. 4},
  author = {Reed, Michael and Simon, Barry},
  date = {2005},
  edition = {Nachdr.},
  publisher = {Acad. Pr},
  location = {New York},
  isbn = {978-0-12-585004-9},
  pagetotal = {396}
}

@article{remling-absolutelycontinuousspectrum-2011,
  title = {The Absolutely Continuous Spectrum of {{Jacobi}} Matrices},
  author = {Remling, Christian},
  date = {2011-07-01},
  journaltitle = {Annals of Mathematics},
  shortjournal = {Ann. Math.},
  volume = {174},
  number = {1},
  pages = {125--171},
  issn = {0003-486X},
  doi = {10.4007/annals.2011.174.1.4}
}

@article{rio-inverseproblemsjacobi-2013,
  title = {Inverse Problems for {{Jacobi}} Operators {{II}}: {{Mass}} Perturbations of Semi-Infinite Mass-Spring Systems},
  shorttitle = {Inverse Problems for {{Jacobi}} Operators {{II}}},
  author = {del Rio, Rafael and Kudryavtsev, Mikhail and Silva, Luis O.},
  date = {2013},
  journal = {Journal of Mathematical Physics, Analysis, Geometry},
  volume ={9},
  pages={165--190},
  doi = {10.48550/arXiv.1106.4598}
}

@book{schmudgen-momentproblem-2017,
  title = {The {{Moment Problem}}},
  author = {Schmüdgen, Konrad},
  date = {2017},
  series = {Graduate {{Texts}} in {{Mathematics}}},
  volume = {277},
  publisher = {Springer International Publishing},
  location = {Cham},
  doi = {10.1007/978-3-319-64546-9},
  isbn = {978-3-319-64545-2 978-3-319-64546-9}
}

@book{schmudgen-unboundedselfadjointoperators-2012,
  title = {Unbounded {{Self-adjoint Operators}} on {{Hilbert Space}}},
  author = {Schmüdgen, Konrad},
  date = {2012},
  series = {Graduate {{Texts}} in {{Mathematics}}},
  volume = {265},
  publisher = {Springer Netherlands},
  location = {Dordrecht},
  doi = {10.1007/978-94-007-4753-1},
  isbn = {978-94-007-4752-4 978-94-007-4753-1}
}

@book{simon-basiccomplexanalysis-2015,
  title = {Basic {{Complex Analysis}}},
  author = {Simon, Barry},
  date = {2015},
  publisher = {American Mathematical Society},
  location = {Providence, Rhode Island},
  doi = {10.1090/simon/002.1},
  isbn = {978-1-4704-1100-8 978-1-4704-2757-3}
}

@article{simon-classicalmomentproblem-1998,
  title = {The {{Classical Moment Problem}} as a {{Self-Adjoint Finite Difference Operator}}},
  author = {Simon, Barry},
  date = {1998-07},
  journaltitle = {Advances in Mathematics},
  volume = {137},
  number = {1},
  pages = {82--203},
  publisher = {Elsevier BV},
  issn = {0001-8708},
  doi = {10.1006/aima.1998.1728}
}

@article{simon-operatorssingularcontinuous-1995,
  title = {Operators with {{Singular Continuous Spectrum}}: {{I}}. {{General Operators}}},
  shorttitle = {Operators with {{Singular Continuous Spectrum}}},
  author = {Simon, Barry},
  date = {1995-01},
  journaltitle = {The Annals of Mathematics},
  shortjournal = {The Annals of Mathematics},
  volume = {141},
  number = {1},
  eprint = {2118629},
  eprinttype = {jstor},
  pages = {131--145},
  issn = {0003486X},
  doi = {10.2307/2118629}
}

@book{simon-szegostheoremits-2011,
  title = {Szego's {{Theorem}} and {{Its Descendants}}: {{Spectral Theory}} for {{L2 Perturbations}} of {{Orthogonal Polynomials}}},
  shorttitle = {Szego's {{Theorem}} and {{Its Descendants}}},
  author = {Simon, Barry},
  date = {2011},
  series = {Porter {{Lectures}}},
  publisher = {Princeton University Press},
  location = {Princeton, N.J},
  doi = {10.1515/9781400837052},
  isbn = {978-0-691-14704-8 978-1-4008-3705-2 978-1-282-82115-6},
  pagetotal = {650}
}

@article{stovicek-infinitejacobimatrices-2019,
  title = {On Infinite {{Jacobi}} Matrices with a Trace Class Resolvent},
  author = {Stovicek, Pavel},
  date = {2020-01},
  journal = {Journal of Approximation Theory},
  volume = {249},
  pages = {105306},
  eprint = {1904.13199},
  eprinttype = {arXiv},
  eprintclass = {math},
  doi = {10.48550/arXiv.1904.13199},
  pubstate = {prepublished}
}

@article{swiderski-asymptoticzerosdistribution-2025,
  title = {Asymptotic Zeros' Distribution of Orthogonal Polynomials with Unbounded Recurrence Coefficients},
  author = {Świderski, Grzegorz and Trojan, Bartosz},
  date = {2025-12-15},
  journaltitle = {Journal of Functional Analysis},
  shortjournal = {Journal of Functional Analysis},
  volume = {289},
  number = {12},
  pages = {111162},
  issn = {0022-1236},
  doi = {10.1016/j.jfa.2025.111162}
}

@article{swiderski-periodicperturbationsunbounded-2017,
  title = {Periodic Perturbations of Unbounded {{Jacobi}} Matrices {{I}}: {{Asymptotics}} of Generalized Eigenvectors},
  shorttitle = {Periodic Perturbations of Unbounded {{Jacobi}} Matrices {{I}}},
  author = {Świderski, Grzegorz and Trojan, Bartosz},
  date = {2017-04},
  journaltitle = {Journal of Approximation Theory},
  shortjournal = {Journal of Approximation Theory},
  volume = {216},
  pages = {38--66},
  issn = {00219045},
  doi = {10.1016/j.jat.2017.01.003}
}

@article{swiderski-periodicperturbationsunbounded-2017a,
  title = {Periodic Perturbations of Unbounded {{Jacobi}} Matrices {{II}}: {{Formulas}} for Density},
  shorttitle = {Periodic Perturbations of Unbounded {{Jacobi}} Matrices {{II}}},
  author = {Świderski, Grzegorz},
  date = {2017-04},
  journaltitle = {Journal of Approximation Theory},
  shortjournal = {Journal of Approximation Theory},
  volume = {216},
  pages = {67--85},
  issn = {00219045},
  doi = {10.1016/j.jat.2017.01.004}
}

@article{swiderski-periodicperturbationsunbounded-2018,
  title = {Periodic Perturbations of Unbounded {{Jacobi}} Matrices {{III}}: {{The}} Soft Edge Regime},
  shorttitle = {Periodic Perturbations of Unbounded {{Jacobi}} Matrices {{III}}},
  author = {Świderski, Grzegorz},
  date = {2018-09},
  journaltitle = {Journal of Approximation Theory},
  shortjournal = {Journal of Approximation Theory},
  volume = {233},
  pages = {1--36},
  issn = {00219045},
  doi = {10.1016/j.jat.2018.04.006}
}

@article{swiderski-spectralpropertiesblock-2018,
  title = {Spectral {{Properties}} of {{Block Jacobi Matrices}}},
  author = {Świderski, Grzegorz},
  date = {2018-10},
  journaltitle = {Constructive Approximation},
  shortjournal = {Constr Approx},
  volume = {48},
  number = {2},
  pages = {301--335},
  issn = {0176-4276, 1432-0940},
  doi = {10.1007/s00365-018-9420-z}
}

@article{swiderski-spectralpropertiesunbounded-2016,
  title = {Spectral {{Properties}} of {{Unbounded Jacobi Matrices}} with {{Almost Monotonic Weights}}},
  author = {Świderski, Grzegorz},
  date = {2016-08},
  journaltitle = {Constructive Approximation},
  shortjournal = {Constr Approx},
  volume = {44},
  number = {1},
  pages = {141--157},
  issn = {0176-4276, 1432-0940},
  doi = {10.1007/s00365-015-9308-0}
}

@article{tamura-resolventconvergencenorm-2003,
  title = {Resolvent Convergence in Norm for {{Dirac}} Operator with {{Aharonov}}–{{Bohm}} Field},
  author = {Tamura, Hideo},
  date = {2003-07-01},
  journaltitle = {Journal of Mathematical Physics},
  volume = {44},
  number = {7},
  pages = {2967--2993},
  issn = {0022-2488, 1089-7658},
  doi = {10.1063/1.1580200}
}

@book{teschl-jacobioperatorscompletely-1999,
  title = {Jacobi {{Operators}} and {{Completely Integrable Nonlinear Lattices}}},
  author = {Teschl, Gerald},
  date = {1999-10-05},
  series = {Mathematical {{Surveys}} and {{Monographs}}},
  volume = {72},
  publisher = {American Mathematical Society},
  location = {Providence, Rhode Island},
  doi = {10.1090/surv/072},
  isbn = {978-0-8218-1940-1 978-1-4704-1299-9}
}

@book{teschl-mathematicalmethodsquantum-2009,
  title = {Mathematical {{Methods}} in {{Quantum Mechanics}}},
  author = {Teschl, Gerald},
  date = {2009},
  series = {Graduate {{Studies}} in {{Mathematics}}},
  volume = {99},
  location = {Providence, Rhode Island},
  doi = {10.1090/gsm/157}
}

@article{teschl-oscillationtheoryrenormalized-1996,
  title = {Oscillation {{Theory}} and {{Renormalized Oscillation Theory}} for {{Jacobi Operators}}},
  author = {Teschl, Gerald},
  date = {1996-08},
  journaltitle = {Journal of Differential Equations},
  shortjournal = {Journal of Differential Equations},
  volume = {129},
  number = {2},
  pages = {532--558},
  issn = {00220396},
  doi = {10.1006/jdeq.1996.0126}
}

@article{teschl-traceformulasinverse-1998,
  title = {Trace {{Formulas}} and {{Inverse Spectral Theory}} for {{Jacobi Operators}}},
  author = {Teschl, Gerald},
  date = {1998-08-01},
  journaltitle = {Communications in Mathematical Physics},
  shortjournal = {Communications in Mathematical Physics},
  volume = {196},
  number = {1},
  pages = {175--202},
  issn = {0010-3616, 1432-0916},
  doi = {10.1007/s002200050419}
}

@article{tse-quantumenhancedadvancedligo-2019,
  title = {Quantum-{{Enhanced Advanced LIGO Detectors}} in the {{Era}} of {{Gravitational-Wave Astronomy}}},
  author = {Tse, M. and Yu, Haocun and Kijbunchoo, N. and Fernandez-Galiana, A. and Dupej, P. and Barsotti, L. and Blair, C. D. and Brown, D. D. and Dwyer, S. E. and Effler, A. and Evans, M. and Fritschel, P. and Frolov, V. V. and Green, A. C. and Mansell, G. L. and Matichard, F. and Mavalvala, N. and McClelland, D. E. and McCuller, L. and McRae, T. and Miller, J. and Mullavey, A. and Oelker, E. and Phinney, I. Y. and Sigg, D. and Slagmolen, B. J. J. and Vo, T. and Ward, R. L. and Whittle, C. and Abbott, R. and Adams, C. and Adhikari, R. X. and Ananyeva, A. and Appert, S. and Arai, K. and Areeda, J. S. and Asali, Y. and Aston, S. M. and Austin, C. and Baer, A. M. and Ball, M. and Ballmer, S. W. and Banagiri, S. and Barker, D. and Bartlett, J. and Berger, B. K. and Betzwieser, J. and Bhattacharjee, D. and Billingsley, G. and Biscans, S. and Blair, R. M. and Bode, N. and Booker, P. and Bork, R. and Bramley, A. and Brooks, A. F. and Buikema, A. and Cahillane, C. and Cannon, K. C. and Chen, X. and Ciobanu, A. A. and Clara, F. and Cooper, S. J. and Corley, K. R. and Countryman, S. T. and Covas, P. B. and Coyne, D. C. and Datrier, L. E. H. and Davis, D. and Di Fronzo, C. and Driggers, J. C. and Etzel, T. and Evans, T. M. and Feicht, J. and Fulda, P. and Fyffe, M. and Giaime, J. A. and Giardina, K. D. and Godwin, P. and Goetz, E. and Gras, S. and Gray, C. and Gray, R. and Gupta, Anchal and Gustafson, E. K. and Gustafson, R. and Hanks, J. and Hanson, J. and Hardwick, T. and Hasskew, R. K. and Heintze, M. C. and Helmling-Cornell, A. F. and Holland, N. A. and Jones, J. D. and Kandhasamy, S. and Karki, S. and Kasprzack, M. and Kawabe, K. and King, P. J. and Kissel, J. S. and Kumar, Rahul and Landry, M. and Lane, B. B. and Lantz, B. and Laxen, M. and Lecoeuche, Y. K. and Leviton, J. and Liu, J. and Lormand, M. and Lundgren, A. P. and Macas, R. and MacInnis, M. and Macleod, D. M. and Márka, S. and Márka, Z. and Martynov, D. V. and Mason, K. and Massinger, T. J. and McCarthy, R. and McCormick, S. and McIver, J. and Mendell, G. and Merfeld, K. and Merilh, E. L. and Meylahn, F. and Mistry, T. and Mittleman, R. and Moreno, G. and Mow-Lowry, C. M. and Mozzon, S. and Nelson, T. J. N. and Nguyen, P. and Nuttall, L. K. and Oberling, J. and Oram, R. J. and O’Reilly, B. and Osthelder, C. and Ottaway, D. J. and Overmier, H. and Palamos, J. R. and Parker, W. and Payne, E. and Pele, A. and Perez, C. J. and Pirello, M. and Radkins, H. and Ramirez, K. E. and Richardson, J. W. and Riles, K. and Robertson, N. A. and Rollins, J. G. and Romel, C. L. and Romie, J. H. and Ross, M. P. and Ryan, K. and Sadecki, T. and Sanchez, E. J. and Sanchez, L. E. and Saravanan, T. R. and Savage, R. L. and Schaetzl, D. and Schnabel, R. and Schofield, R. M. S. and Schwartz, E. and Sellers, D. and Shaffer, T. J. and Smith, J. R. and Soni, S. and Sorazu, B. and Spencer, A. P. and Strain, K. A. and Sun, L. and Szczepańczyk, M. J. and Thomas, M. and Thomas, P. and Thorne, K. A. and Toland, K. and Torrie, C. I. and Traylor, G. and Urban, A. L. and Vajente, G. and Valdes, G. and Vander-Hyde, D. C. and Veitch, P. J. and Venkateswara, K. and Venugopalan, G. and Viets, A. D. and Vorvick, C. and Wade, M. and Warner, J. and Weaver, B. and Weiss, R. and Willke, B. and Wipf, C. C. and Xiao, L. and Yamamoto, H. and Yap, M. J. and Yu, Hang and Zhang, L. and Zucker, M. E. and Zweizig, J.},
  date = {2019-12-05},
  journaltitle = {Physical Review Letters},
  shortjournal = {Phys. Rev. Lett.},
  volume = {123},
  number = {23},
  pages = {231107},
  publisher = {American Physical Society (APS)},
  issn = {0031-9007, 1079-7114},
  doi = {10.1103/physrevlett.123.231107}
}

@article{tur-jaynescummingsmodel-2000,
  title = {Jaynes-{{Cummings Model}}: {{Solutions}} without {{Rotating-Wave Approximation}}},
  author = {Tur, Eduard},
  journal={Optics and Spectroscopy},
  date = {2016},
  volume = {89},
  pages = {574–588}
}

@online{tur-jaynescummingsmodelrotating-2002,
  title = {Jaynes-{{Cummings}} Model without Rotating Wave Approximation. {{Asymptotics}} of Eigenvalues},
  author = {Tur, Eduard},
  date = {2002-11-22},
  eprint = {math-ph/0211055},
  eprinttype = {arXiv},
  doi = {10.48550/arXiv.math-ph/0211055},
  pubstate = {prepublished}
}

@article{vaartjes-strongmicrowavesqueezing-2024,
  title = {Strong Microwave Squeezing above 1 {{Tesla}} and 1 {{Kelvin}}},
  author = {Vaartjes, Arjen and Kringhøj, Anders and Vine, Wyatt and Day, Tom and Morello, Andrea and Pla, Jarryd J.},
  date = {2024-05-18},
  journaltitle = {Nature Communications},
  shortjournal = {Nat Commun},
  volume = {15},
  number = {1},
  pages = {4229},
  issn = {2041-1723},
  doi = {10.1038/s41467-024-48519-3}
}

@article{walls-squeezedstateslight-1983,
  title = {Squeezed States of Light},
  author = {Walls, D. F.},
  date = {1983-11},
  journaltitle = {Nature},
  volume = {306},
  number = {5939},
  pages = {141--146},
  publisher = {{Springer Science and Business Media LLC}},
  issn = {0028-0836, 1476-4687},
  doi = {10.1038/306141a0}
}

@article{wang-asymptoticexpansionssecondorder-2003,
  title = {Asymptotic Expansions for Second-Order Linear Difference Equations with a Turning Point},
  author = {Wang, Z. and Wong, R.},
  date = {2003-03-01},
  journaltitle = {Numerische Mathematik},
  shortjournal = {Numerische Mathematik},
  volume = {94},
  number = {1},
  pages = {147--194},
  issn = {0029-599X, 0945-3245},
  doi = {10.1007/s00211-002-0416-y}
}

@inbook{wang-lineardifferenceequations-,
  title = {{{Linear difference equations with transition points}}},
  author = {Wang, Z and Wong, R},
  booktitle = {The Selected Works of Roderick S. C. Wong},
chapter = {},
pages = {1054-1078},
doi = {10.1142/9789814656054_0049},
URL = {https://www.worldscientific.com/doi/abs/10.1142/9789814656054_0049},
eprint = {https://www.worldscientific.com/doi/pdf/10.1142/9789814656054_0049},
}

@article{wang-uniformasymptoticexpansion-2002,
  title = {Uniform Asymptotic Expansion of $J_\nu(\nu a)$ via a Difference Equation},
  author = {Wang, Z. and Wong, R.},
  date = {2002-03-01},
  journaltitle = {Numerische Mathematik},
  shortjournal = {Numer. Math.},
  volume = {91},
  number = {1},
  pages = {147--193},
  issn = {0945-3245},
  doi = {10.1007/s002110100316}
}

@book{wasow-asymptoticexpansionsordinary-2002,
  title = {Asymptotic Expansions for Ordinary Differential Equations},
  author = {Wasow, Wolfgang Richard},
  date = {2002},
  series = {Dover {{Phoenix}} Editions},
  edition = {Republ. of the Dover ed. first publ. in 1987},
  publisher = {Dover Publ},
  location = {Mineola, NY},
  isbn = {978-0-486-49518-7},
  pagetotal = {374}
}

@article{webb-spectrajacobioperators-2020,
  title = {Spectra of {{Jacobi}} Operators via Connection Coefficient Matrices},
  author = {Webb, Marcus and Olver, Sheehan},
  date = {2021},
  journal = {Communications in Mathematical Physics},
  volume = {382},
  pages = {657–707},
  doi = {10.48550/arXiv.1702.03095},
  pubstate = {prepublished}
}

@article{welstead-boundaryconditionsinfinity-1982,
  title = {Boundary Conditions at Infinity for Difference Equations of Limit-Circle Type},
  author = {Welstead, Stephen T},
  date = {1982-10},
  journaltitle = {Journal of Mathematical Analysis and Applications},
  shortjournal = {Journal of Mathematical Analysis and Applications},
  volume = {89},
  number = {2},
  pages = {442--461},
  issn = {0022247X},
  doi = {10.1016/0022-247X(82)90112-3}
}

@article{wimp-resurrectingasymptoticslinear-1985,
  title = {Resurrecting the Asymptotics of Linear Recurrences},
  author = {Wimp, Jet and Zeilberger, Doron},
  date = {1985-10},
  journaltitle = {Journal of Mathematical Analysis and Applications},
  shortjournal = {Journal of Mathematical Analysis and Applications},
  volume = {111},
  number = {1},
  pages = {162--176},
  issn = {0022247X},
  doi = {10.1016/0022-247X(85)90209-4}
}

@article{wong-asymptoticexpansionssecondorder-1992,
  title = {Asymptotic Expansions for Second-Order Linear Difference Equations},
  author = {Wong, R and Li, H},
  date = {1992},
  journaltitle = {Journal of Computational and Applied Mathematics},
  volume = {41},
  pages={65-94}
}

@article{wong-asymptoticexpansionssecondorder-1992a,
  title = {Asymptotic {{Expansions}} for {{Second}}‐{{Order Linear Difference Equations}}, {{II}}},
  author = {Wong, R. and Li, H.},
  date = {1992-11},
  journaltitle = {Studies in Applied Mathematics},
  shortjournal = {Stud Appl Math},
  volume = {87},
  number = {4},
  pages = {289--324},
  issn = {0022-2526, 1467-9590},
  doi = {10.1002/sapm1992874289}
}

@article{wong-recentadvancesasymptotic-2022,
  title = {Recent {{Advances}} in {{Asymptotic Analysis}}},
  author = {Wong, R. and Zhao, Yu-Qiu},
  date = {2022-11},
  journaltitle = {Analysis and Applications},
  shortjournal = {Anal. Appl.},
  volume = {20},
  number = {06},
  eprint = {2204.09305},
  eprinttype = {arXiv},
  eprintclass = {math},
  pages = {1103--1146},
  issn = {0219-5305, 1793-6861},
  doi = {10.1142/S0219530522400012}
}

@article{yafaev-analyticscatteringtheory-2018,
  title = {Analytic Scattering Theory for {{Jacobi}} Operators and {{Bernstein-Szegö}} Asymptotics of Orthogonal Polynomials},
  author = {Yafaev, D. R.},
  date = {2018-09},
  journaltitle = {Reviews in Mathematical Physics},
  shortjournal = {Rev. Math. Phys.},
  volume = {30},
  number = {08},
  eprint = {1711.05029},
  eprinttype = {arXiv},
  eprintclass = {math},
  pages = {1840019},
  issn = {0129-055X, 1793-6659},
  doi = {10.1142/S0129055X18400196}
}

@article{yafaev-asymptoticbehaviororthogonal-2020,
  title = {Asymptotic Behavior of Orthogonal Polynomials without the {{Carleman}} Condition},
  author = {Yafaev, D.R.},
  date = {2020-10},
  journaltitle = {Journal of Functional Analysis},
  shortjournal = {Journal of Functional Analysis},
  volume = {279},
  number = {7},
  pages = {108648},
  issn = {00221236},
  doi = {10.1016/j.jfa.2020.108648}
}

@article{yafaev-asymptoticbehaviororthogonal-2021,
  title = {Asymptotic Behavior of Orthogonal Polynomials. {{Singular}} Critical Case},
  author = {Yafaev, D.R.},
  date = {2021-02},
  journaltitle = {Journal of Approximation Theory},
  shortjournal = {Journal of Approximation Theory},
  volume = {262},
  pages = {105506},
  issn = {00219045},
  doi = {10.1016/j.jat.2020.105506}
}

@article{yafaev-selfadjointjacobioperators-2021,
  title = {Self-Adjoint {{Jacobi}} Operators in the Limit Circle Case},
  author = {Yafaev, D. R.},
  date = {2023},
  journal = {Journal of Operator Theory},
  volume ={89},
  pages ={87-103},
  doi = {10.48550/arXiv.2104.13609},
  pubstate = {prepublished}
}

@article{yafaev-spectralanalysisjacobi-2022,
  title = {Spectral Analysis of {{Jacobi}} Operators and Asymptotic Behavior of Orthogonal Polynomials},
  author = {Yafaev, D. R.},
  date = {2022-12},
  journaltitle = {Bulletin of Mathematical Sciences},
  shortjournal = {Bull. Math. Sci.},
  volume = {12},
  number = {03},
  pages = {2250002},
  issn = {1664-3607, 1664-3615},
  doi = {10.1142/S1664360722500023}
}

@article{yafaev-spectraltheoryjacobi-2024,
  title = {Spectral Theory of {{Jacobi}} Operators with Increasing Coefficients. {{The}} Critical Case},
  author = {Yafaev, Dimitri},
  date = {2024-02-22},
  journaltitle = {Journal of Spectral Theory},
  shortjournal = {J. Spectr. Theory},
  volume = {13},
  number = {4},
  pages = {1393--1444},
  issn = {1664-039X, 1664-0403},
  doi = {10.4171/jst/481}
}

@book{zettl-sturmliouvilletheory-2005,
  title = {Sturm-{{Liouville}} Theory},
  author = {Zettl, Anton},
  date = {2005},
  series = {Mathematical Surveys and Monographs},
  number = {v. 121},
  publisher = {American Mathematical Society},
  location = {Providence, R.I},
  isbn = {978-0-8218-3905-8},
  pagetotal = {328}
}

@article{zhang-2modekphotonquantum-2017,
  title = {On the 2-Mode and k-Photon Quantum {{Rabi}} Models},
  author = {Zhang, Yao-Zhong},
  date = {2017-05},
  journaltitle = {Reviews in Mathematical Physics},
  shortjournal = {Rev. Math. Phys.},
  volume = {29},
  number = {04},
  pages = {1750013},
  publisher = {World Scientific Pub Co Pte Lt},
  issn = {0129-055X, 1793-6659},
  doi = {10.1142/s0129055x17500131}
}

@article{zhang-solvingtwomodesqueezed-2013,
  title = {Solving the Two-Mode Squeezed Harmonic Oscillator and The{\mkbibemph{k}}th-Order Harmonic Generation in {{Bargmann}}–{{Hilbert}} Spaces},
  author = {Zhang, Yao-Zhong},
  date = {2013-11-15},
  journaltitle = {Journal of Physics A: Mathematical and Theoretical},
  shortjournal = {J. Phys. A: Math. Theor.},
  volume = {46},
  number = {45},
  pages = {455302},
  publisher = {IOP Publishing},
  issn = {1751-8113, 1751-8121},
  doi = {10.1088/1751-8113/46/45/455302}
}

@misc{NIST:DLMF,
             key = "{\relax DLMF}",
           title = "{\it NIST Digital Library of Mathematical Functions}",
    howpublished = "\url{https://dlmf.nist.gov/}, Release 1.2.6 of 2026-03-15",
             url = "https://dlmf.nist.gov/",
            note = "F.~W.~J. Olver, A.~B. {Olde Daalhuis}, D.~W. Lozier, B.~I. Schneider,
                    R.~F. Boisvert, C.~W. Clark, B.~R. Miller, B.~V. Saunders,
                    H.~S. Cohl, and M.~A. McClain, eds."
}

@article{eberle-quantumenhancementzeroarea-2010,
  title = {Quantum {{Enhancement}} of the {{Zero-Area Sagnac Interferometer Topology}} for {{Gravitational Wave Detection}}},
  author = {Eberle, Tobias and Steinlechner, Sebastian and Bauchrowitz, Jöran and Händchen, Vitus and Vahlbruch, Henning and Mehmet, Moritz and Müller-Ebhardt, Helge and Schnabel, Roman},
  date = {2010-06-22},
  journaltitle = {Physical Review Letters},
  shortjournal = {Phys. Rev. Lett.},
  volume = {104},
  number = {25},
  pages = {251102},
  issn = {0031-9007, 1079-7114},
  doi = {10.1103/PhysRevLett.104.251102}
}

@article{vahlbruch-observationsqueezedlight-2008,
  title = {Observation of {{Squeezed Light}} with 10-{{dB Quantum-Noise Reduction}}},
  author = {Vahlbruch, Henning and Mehmet, Moritz and Chelkowski, Simon and Hage, Boris and Franzen, Alexander and Lastzka, Nico and Goßler, Stefan and Danzmann, Karsten and Schnabel, Roman},
  date = {2008-01-23},
  journaltitle = {Physical Review Letters},
  shortjournal = {Phys. Rev. Lett.},
  volume = {100},
  number = {3},
  pages = {033602},
  issn = {0031-9007, 1079-7114},
  doi = {10.1103/PhysRevLett.100.033602}
}

@article{wu-generationsqueezedstates-1986,
  title = {Generation of {{Squeezed States}} by {{Parametric Down Conversion}}},
  author = {Wu, Ling-An and Kimble, H. J. and Hall, J. L. and Wu, Huifa},
  date = {1986-11-17},
  journaltitle = {Physical Review Letters},
  shortjournal = {Phys. Rev. Lett.},
  volume = {57},
  number = {20},
  pages = {2520--2523},
  issn = {0031-9007},
  doi = {10.1103/PhysRevLett.57.2520}
}

@article{andersen-continuousvariablequantuminformation-2010,
  title = {Continuous‐variable Quantum Information Processing},
  author = {Andersen, U.L. and Leuchs, G. and Silberhorn, C.},
  date = {2010-04-28},
  journaltitle = {Laser \& Photonics Reviews},
  shortjournal = {Laser \&amp; Photonics Reviews},
  volume = {4},
  number = {3},
  pages = {337--354},
  issn = {1863-8880, 1863-8899},
  doi = {10.1002/lpor.200910010}
}

@article{lu-recentprogresscoherent-2023,
  title = {Recent Progress on Coherent Computation Based on Quantum Squeezing},
  author = {Lu, Bo and Liu, Lu and Song, Jun-Yang and Wen, Kai and Wang, Chuan},
  date = {2023-03-01},
  journaltitle = {AAPPS Bulletin},
  shortjournal = {AAPPS Bull.},
  volume = {33},
  number = {1},
  pages = {7},
  issn = {2309-4710},
  doi = {10.1007/s43673-023-00077-4}
}

@article{schlegel-quantumerrorcorrection-2022,
  title = {Quantum Error Correction Using Squeezed {{Schrödinger}} Cat States},
  author = {Schlegel, David S. and Minganti, Fabrizio and Savona, Vincenzo},
  date = {2022-08-25},
  journaltitle = {Physical Review A},
  shortjournal = {Phys. Rev. A},
  volume = {106},
  number = {2},
  pages = {022431},
  issn = {2469-9926, 2469-9934},
  doi = {10.1103/PhysRevA.106.022431}
}

@online{zeng-quantumerrorcorrection-2025,
  title = {Quantum {{Error Correction}} with {{Superpositions}} of {{Squeezed Fock States}}},
  author = {Zeng, Yexiong and Quijandría, Fernando and Gneiting, Clemens and Nori, Franco},
  date = {2025-10-05},
  eprint = {2510.04209},
  eprinttype = {arXiv},
  eprintclass = {quant-ph},
  doi = {10.48550/arXiv.2510.04209},
  pubstate = {prepublished}
}

@article{braunstein-quantuminformationcontinuous-2005,
  title = {Quantum Information with Continuous Variables},
  author = {Braunstein, Samuel L and van Loock, Peter},
  date = {2005},
  journaltitle = {Rev. Mod. Phys.},
  volume = {77},
  number = {2},
  pages = {513--577}
}

@article{jitomirskaya-metalinsulatortransitionalmost-1999,
  title = {Metal-{{Insulator Transition}} for the {{Almost Mathieu Operator}}},
  author = {Jitomirskaya, Svetlana Ya.},
  date = {1999-11},
  journaltitle = {The Annals of Mathematics},
  shortjournal = {The Annals of Mathematics},
  volume = {150},
  number = {3},
  eprint = {121066},
  eprinttype = {jstor},
  pages = {1159--1175},
  issn = {0003486X},
  doi = {10.2307/121066}
}

@article{thouless-bandwidthsquasiperiodictightbinding-1983,
  title = {Bandwidths for a Quasiperiodic Tight-Binding Model},
  author = {Thouless, D. J.},
  date = {1983-10-15},
  journaltitle = {Physical Review B},
  shortjournal = {Phys. Rev. B},
  volume = {28},
  number = {8},
  pages = {4272--4276},
  issn = {0163-1829},
  doi = {10.1103/PhysRevB.28.4272}
}

@article{zhang-andersonlocalizationblock-2024,
  title = {Anderson Localization for Block {{Jacobi}} Operators with Quasi‐periodic Meromorphic Potential},
  author = {Zhang, Xiaojian},
  date = {2024-11-15},
  journaltitle = {Mathematical Methods in the Applied Sciences},
  shortjournal = {Math Methods in App Sciences},
  volume = {47},
  number = {16},
  pages = {12816--12832},
  issn = {0170-4214, 1099-1476},
  doi = {10.1002/mma.10182}
}

@inproceedings{bellisard-spectralproperties-1990,
  title = {Spectral Properties of Schrödinger's Operator with a Thue-Morse Potential},
  booktitle = {Number Theory and Physics},
  author = {Bellissard, J.},
  editor = {Luck, Jean-Marc and Moussa, Pierre and Waldschmidt, Michel},
  date = {1990},
  pages = {140--150},
  publisher = {Springer Berlin Heidelberg},
  location = {Berlin, Heidelberg},
  isbn = {978-3-642-75405-0}
}

@book{berezanskii-expansioneigenfunctionsselfadjoint-1968,
  title = {Expansion in {{Eigenfunctions}} of {{Self-adjoint Operators}}},
  author = {Berezanskii, Ju M},
  date = {1968},
  series = {Translations of {{Mathematical}}                         {{Monographs}}},
  volume = {17},
  publisher = {American Mathematical Society}
}

@article{allakhverdiev-spectraltheorydissipative-1990,
  title = {{{On the spectral theory of dissipative difference operators of second order}}},
  author = {Allakhverdiev, B P and Guseĭnov, G Sh},
  date = {1990-02-28},
  journaltitle = {Mathematics of the USSR-Sbornik},
  shortjournal = {Math. USSR Sb.},
  volume = {66},
  number = {1},
  pages = {107--125},
  issn = {0025-5734},
  doi = {10.1070/SM1990v066n01ABEH002081}
}

@article{allahverdiev-spectralproblemsjacobi-,
  title={Spectral problems of Jacobi operators in limit-circle case},
  author={Bilender Pasaoğlu Allahverdiev},
  year={2015},
  journaltitle={Mathematical Reports},
  volume={17},
  pages={81--89}
}

@article{bazavan-squeezingtrisqueezingquadsqueezing-2026,
  title = {Squeezing, Trisqueezing and Quadsqueezing in a Hybrid Oscillator–Spin System},
  author = {Băzăvan, O. and Saner, S. and Webb, D. J. and Ainley, E. M. and Drmota, P. and Nadlinger, D. P. and Araneda, G. and Lucas, D. M. and Ballance, C. J. and Srinivas, R.},
  date = {2026-05-01},
  journaltitle = {Nature Physics},
  shortjournal = {Nat. Phys.},
  issn = {1745-2473, 1745-2481},
  doi = {10.1038/s41567-026-03222-6},
  pages = {757--762}
}

@article{eriksson-universalcontrolbosonic-2024,
  title = {Universal Control of a Bosonic Mode via Drive-Activated Native Cubic Interactions},
  author = {Eriksson, Axel M. and Sépulcre, Théo and Kervinen, Mikael and Hillmann, Timo and Kudra, Marina and Dupouy, Simon and Lu, Yong and Khanahmadi, Maryam and Yang, Jiaying and Castillo-Moreno, Claudia and Delsing, Per and Gasparinetti, Simone},
  date = {2024-03-21},
  journaltitle = {Nature Communications},
  shortjournal = {Nat Commun},
  volume = {15},
  number = {1},
  pages = {2512},
  issn = {2041-1723},
  doi = {10.1038/s41467-024-46507-1}
}

@article{lloyd-quantumcomputationcontinuous-1999,
  title = {Quantum {{Computation}} over {{Continuous Variables}}},
  author = {Lloyd, Seth and Braunstein, Samuel L},
  date = {1999},
  journaltitle = {Physical Review Letters},
  volume = {82},
  number = {8},
  pages = {1784--1787}
}

@book{mandel-opticalcoherencequantum-1995,
  title = {Optical {{Coherence}} and {{Quantum Optics}}},
  author = {Mandel, Leonard and Wolf, Emil},
  date = {1995-09-29},
  edition = {1},
  publisher = {Cambridge University Press},
  doi = {10.1017/CBO9781139644105},
  isbn = {978-0-521-41711-2 978-1-139-64410-5}
}

@article{blanes-magnusexpansionits-2009,
  title = {The {{Magnus}} Expansion and Some of Its Applications},
  author = {Blanes, S. and Casas, F. and Oteo, J.A. and Ros, J.},
  date = {2009-01},
  journaltitle = {Physics Reports},
  shortjournal = {Physics Reports},
  volume = {470},
  number = {5--6},
  pages = {151--238},
  issn = {03701573},
  doi = {10.1016/j.physrep.2008.11.001}
}

@article{bukov-universalhighfrequencybehavior-2015,
  title = {Universal High-Frequency Behavior of Periodically Driven Systems: From Dynamical Stabilization to {{Floquet}} Engineering},
  shorttitle = {Universal High-Frequency Behavior of Periodically Driven Systems},
  author = {Bukov, Marin and D'Alessio, Luca and Polkovnikov, Anatoli},
  date = {2015-03-04},
  journaltitle = {Advances in Physics},
  shortjournal = {Advances in Physics},
  volume = {64},
  number = {2},
  pages = {139--226},
  issn = {0001-8732, 1460-6976},
  doi = {10.1080/00018732.2015.1055918}
}
\end{document}